\documentclass[acmsmall,screen%
  ,nonacm=true,timestamp=true%
]{acmart}%
\usepackage[utf8]{inputenc}

\AtBeginDocument{%
  \providecommand\BibTeX{{%
    \normalfont B\kern-0.5em{\scshape i\kern-0.25em b}\kern-0.8em\TeX}}}

\def\etc{\emph{etc.}\@\xspace}%
\def\eg{e.g.\@\xspace}%
\def\Eg{E.g.\@\xspace}%
\def\ie{i.e.\@\xspace}%
\def\Ie{I.e.\@\xspace}%
\def\vs{vs.\@\xspace}%
\def\wrt{wrt.\@\xspace}%
\newcommand{\hlight}[1]{{\setlength{\fboxsep}{0pt}\colorbox{yellow}{#1}}}%

\makeatletter
\@ifclasswith{acmart}{review}{%
  \makeatletter
  \patchcmd{\@addmarginpar}{\ifodd\c@page}{\ifodd\c@page\@tempcnta\m@ne}{}{}
  \makeatother
  \reversemarginpar
  \newcommand{\change}[2]{{\marginpar{{\color{red}\scriptsize\textbf{#1}}}}{{\color{red}{#2}}}}%
  \newcommand{\changeNoMargin}[1]{{\color{red}{#1}}}%
}{%
  \newcommand{\change}[2]{{#2}}%
  \newcommand{\changeNoMargin}[1]{{#1}}%
}
\makeatother

\newcommand{\dlsqb}{[\![}
\newcommand{\drsqb}{]\!]}

\newcommand{\ts}[1]{\mathbb{#1}}%
\newcommand{\T}{{\mathbb T}}
\newcommand{\TT}{{\mathsf T}}
\newcommand{\W}{{\mathbb W}}
\newcommand{\WT}{{\mathsf W}}
\newcommand{\UU}{{\mathbb U}}
\newcommand{\UT}{{\mathsf U}}
\newcommand{\V}{{\mathbb V}}
\newcommand{\VT}{{\mathsf V}}

\newcommand{\ST}{{\mathsf S}}

\newcommand{\Rel}{{\mathcal{R}}}
\newcommand{\MM}{{\mathsf M}}

\newcommand{\zero}{\texttt 0}

\newcommand{\rulename}[1]{\text{\scriptsize[\textsc{#1}]}}

\newcommand{\role}[1]{\sf {#1}}%

\newcommand{\pp}{{\sf p}}
\newcommand{\q}{\pq}
\newcommand{\pq}{{\sf q}}
\newcommand{\pr}{{\sf r}}

\newcommand{\e}{\kf{e}}
\newcommand{\x}{x}
\newcommand{\y}{y}
\newcommand{\h}{h}
\newcommand{\val}{\kf{v}}
\newcommand{\valn}{\kf{n}}

\newcommand{\valr}{\kf{i}}

\newcommand{\sep}{\ensuremath{~~\mathbf{|\!\!|}~~ }}
\newcommand{\kf}[1]{\ensuremath{\mathsf{#1}}}
\newcommand{\pc}{\ensuremath{~|~}}

\newcommand{\ty}{\textbf{t}}

\newcommand{\N}{\ensuremath{\mathcal M}}
\newcommand{\M}{\ensuremath{\mathcal M}}

\newcommand{\pa}[2]{#1 \triangleleft  #2}
\newcommand{\set}[1]{\{#1\}}
\newcommand{\eval}[2]{#1 \downarrow #2}
\newcommand{\noteval}[2]{#1 \mathbin{\not{\downarrow}} #2}
\newcommand{\true}{\kf{true}}
\newcommand{\false}{\kf{false}}

\newcommand{\rln}[1]{\textsc{#1}}

\newcommand{\CP}[1]{ {\mathcal P}\!\left(#1\right)}

\newcommand{\participant}[1]{\mathtt{pt}(#1)}

\newcommand{\proj}[2]{ #1 \upharpoonright #2}
\newcommand{\projOut}[2]{ #1 \mathbin{\upharpoonright^{!}} #2}
\newcommand{\projIn}[2]{ #1 \mathbin{\upharpoonright^{?}} #2}
\newcommand{\sub}[2]{\set{#1/#2}}
\newcommand{\actions}[1]{\mathtt{act}(#1)}

\newcommand{\subst}[2]{\{\nicefrac{#1}{#2}\}}%

\newcommand{\mapUpdate}[2]{\!\left\{{#1} \mapsto {#2}\right\}}%

\newcommand{\single}[1]{\llbracket #1 \rrbracket_\SI}

\newcommand{\msgLabel}[1]{\mathit{{#1}}}%
\newcommand{\procin}[3]{#1 ? #2.#3}%
\newcommand{\procinNoSuf}[2]{{#1} ? {#2}}

\newcommand{\procout}[4]{#1 ! #2\langle #3 \rangle.#4}%
\newcommand{\procoutNoSuf}[3]{{#1} ! {#2}\langle{#3}\rangle}

\newcommand{\error}{\kf{error}}
\newcommand{\PP}{\ensuremath{P}}
\newcommand{\Q}{\PQ}
\newcommand{\PQ}{\ensuremath{Q}}
\newcommand{\cond}[3]{\kf{if}~ #1 ~\kf{then} ~#2 ~\kf{else}~#3}
\newcommand{\inact}{\mathbf{0}}
\newcommand{\internal}{\oplus}

\newcommand{\emptyqueue}{\varnothing}

\newcommand{\exprt}[2]{\text{\rm\textbf{\color{blue}\underline{expr$\llparenthesis{\color{black}#1},{\color{black}#2}\rrparenthesis$}}}}
\newcommand{\valt}[1]{\text{\rm\textbf{\color{blue}\underline{val$\llparenthesis{\color{black}#1}\rrparenthesis$}}}}

\newcommand{\tend}{\mathtt{end}}
\newcommand{\tstr}{\mathtt{string}}
\newcommand{\tbool}{\mathtt{bool}}
\newcommand{\tnat}{\mathtt{nat}}
\newcommand{\treal}{\mathtt{real}}
\newcommand{\tint}{\mathtt{int}}
\newcommand{\tunit}{\mathtt{unit}}
\newcommand{\tin}[3]{#1?#2(#3)}
\newcommand{\tout}[3]{#1!#2(#3)}

\DeclareMathOperator*{\tinternal}{%
  \vphantom{\sum}\mathlarger{\mathlarger{\mathlarger{\mathlarger{\raisebox{-3pt}{$\oplus$}}}}}%
}%
\DeclareMathOperator*{\texternal}{%
  \vphantom{\sum}\mathlarger{\mathlarger{\mathlarger{\raisebox{-3pt}{\textnormal\&}}}}%
}%

\newcommand{\tqueue}{\sigma}
\newcommand{\temptyqueue}{\epsilon}

\renewcommand{\S}{\ST}%

\newcommand{\ttreeSym}{\mathcal{T}}%
\newcommand{\ttree}[1]{\operatorname{\ttreeSym}\!\left({#1}\right)}%

\newcommand{\Econtext}{\mathcal{E}}

\newcommand{\subt}{\leqslant}
\newcommand{\subttt}{\lesssim}
\newcommand{\nsubttt}{\not\lesssim}%
\newcommand{\subs}{\leq\vcentcolon}

\newcommand{\red}{\longrightarrow}
\newcommand{\recvLabel}[3]{{#1}{:}{#2}?{#3}}
\newcommand{\sendLabel}[3]{{#1}{:}{#2}!{#3}}
\newcommand{\redLabel}[1]{\xrightarrow{#1}}
\newcommand{\redSend}[3]{\xrightarrow{\sendLabel{#1}{#2}{#3}}}
\newcommand{\redRecv}[3]{\xrightarrow{\recvLabel{#1}{#2}{#3}}}
\newcommand{\reds}{\mathrel{\longrightarrow^{\!*}}}

\newcommand{\fsqrt}[1]{{\sf inv}(#1)}
\newcommand{\finv}[1]{{\sf inv}(#1)}
\newcommand{\fneg}{\fsqrt}
\newcommand{\fsucc}[1]{{\sf succ}(#1)}

\newcommand{\fv}{\mathsf{fv}}

\newcommand{\cyclic}[2]{{\tt cyclic}(#1,#2,\pr)}

\newcommand{\regU}[4]{{\tt reg}_\SO(#1,#2,#3,#4)}
\newcommand{\mufree}[1]{{\tt mu}^-(#1)}

\newcommand{\cinferrule}[3][]{
  \infer=[#1]{#3}{#2}%
}
\newcommand{\inferrule}[3][]{
  \infer[\!\!{#1}]{#3}{#2}%
}

\newcommand{\inferruleR}[3][]{
  \infer[\!\!{#1}]{#3}{#2}%
}

\newcommand{\cinfer}[3][]{
  \infer=[#1]{#3}{#2}%
}%

\newcommand{\AContext}[1]{\mathcal{A}^{(#1)}}
\newcommand{\BContext}[1]{\mathcal{B}^{(#1)}}
\newcommand{\CContext}[1]{\mathcal{C}^{(#1)}}
\newcommand{\DContext}[1]{\mathcal{D}^{(#1)}}

\newcommand{\AContextt}[2]{\mathcal{A}^{(#1)}_{#2}}
\newcommand{\BContextt}[2]{\mathcal{B}^{(#1)}_{#2}}
\newcommand{\ACon}[2]{\mathcal{A}^{(#1)}.#2}
\newcommand{\BCon}[2]{\mathcal{B}^{(#1)}.#2}
\newcommand{\DCon}[2]{\mathcal{D}^{(#1)}.#2}
\newcommand{\ACont}[3]{\mathcal{A}^{(#1)}_{#2}.#3}
\newcommand{\BCont}[3]{\mathcal{B}^{(#1)}_{#2}.#3}
\newcommand{\BC}{\mathcal{B}}
\newcommand{\AC}{\mathcal{A}}

\newcommand{\dom}[1]{\ensuremath{dom( #1)}}

\newcommand{\EmptyQueue}{\varnothing}
\newcommand{\Queue}{h}
\newcommand{\msg}[3]{(#1,#2(#3))}

\definecolor{ceca}{rgb}{1,0.5,0}

\newcommand{\SO}{\text{\tiny\sf{SO}}}%
\newcommand{\SI}{\text{\tiny\sf{SI}}}%

\newcommand{\Source}{\text{\color{blue}{Source}}\xspace}
\newcommand{\Sink}{\text{\color{blue}{Sink}}\xspace}

\newcommand{\Processor}{\text{\color{blue}{Processor}}\xspace}
\newcommand{\ProcessorOne}{\text{\color{blue}{Processor~1}}\xspace}
\newcommand{\ProcessorTwo}{\text{\color{blue}{Processor~2}}\xspace}

\newcommand{\Control}{\text{\color{blue}{Control}}\xspace}

\usepackage[UKenglish]{babel}

\usepackage{color}
\usepackage[inline]{enumitem}%
\usepackage[mathcal]{euscript}%
\usepackage{stmaryrd}
\usepackage{thmtools,thm-restate}
\usepackage{nicefrac}%
\usepackage{proof}%

\usepackage{mathtools}

\usepackage{subcaption}%

\usepackage{xifthen}%
\usepackage{ifdraft}%
\usepackage[all]{xy}%

\hypersetup{final}%

\usepackage{relsize}

\usepackage{wrapfig}%

\usepackage{xspace}%

\usepackage{tikz}%
\usetikzlibrary{arrows,snakes,backgrounds, decorations.markings}
\usetikzlibrary{positioning}

\ifdefined\needsPoormanCleveref%
  \usepackage[poorman]{cleveref}%
\else
  \usepackage{cleveref}%
\fi%
 
\allowdisplaybreaks

\newtoggle{techreport}%
\toggletrue{techreport}

\newtoggle{anon}%
\togglefalse{anon}%
\setcopyright{none}
\iftoggle{techreport}{
  \acmBooktitle{}
}{}

\AtEndPreamble{
\theoremstyle{acmdefinition}
\newtheorem{remark}[theorem]{Remark}
}

\makeatletter%
\begin{document}

\newcommand{\myTitle}{Precise Subtyping for Asynchronous Multiparty Sessions%
  \iftoggle{techreport}{ %
    (Extended Version)%
  }{}%
}
\title[\myTitle]{\myTitle}

\author{Silvia Ghilezan}
\email{gsilvia@uns.ac.rs}
\orcid{0000-0003-2253-8285}
\author{Jovanka Pantović}
\email{pantovic@uns.ac.rs}
\orcid{0000-0002-3974-5064}
\author{Ivan Prokić}
\email{prokic@uns.ac.rs}
\orcid{0000-0001-5420-1527}
\affiliation{%
  \institution{Univerzitet u Novom Sadu}
  \city{Novi Sad}
  \country{Serbia}
}
\author{Alceste Scalas}
\email{alcsc@dtu.dk}
\orcid{0000-0002-1153-6164}
\affiliation{%
  \institution{Technical University of Denmark}
  \streetaddress{Richard Petersens Plads, Bygning 324}
  \city{Kongens Lyngby}
  \country{DK}
  \postcode{2800}
}
\affiliation{%
  \institution{Aston University}
  \city{Birmingham}
  \country{UK}
}
\author{Nobuko Yoshida}
\email{n.yoshida@imperial.ac.uk}
\orcid{0000-0002-3925-8557}
\affiliation{%
  \institution{Imperial College London}
  \streetaddress{South Kensington Campus}
  \city{London}
  \country{UK}
  \postcode{SW7 2AZ}
}

\begin{abstract}
  Session subtyping is a cornerstone of refinement of communicating processes:
  a process implementing a session type (\ie, a communication
  protocol) $T$ %
  can be safely used whenever a process implementing %
  one of its supertypes $T'$ is expected, in any context, %
  without introducing deadlocks nor other communication errors. %
  As a consequence, whenever $T \subt T'$ holds,
  it is safe to replace an implementation of\, $T'$ with an
  implementation of the subtype $T$,
  which may allow for more optimised communication patterns.

  This paper presents the first formalisation %
  of the \emph{precise} subtyping relation for \emph{asynchronous
    multiparty} sessions.
  We show that our subtyping relation is \emph{sound} %
  (\ie, guarantees safe process replacement, as outlined above) %
  and also \emph{complete}: any extension of the relation is unsound.
  To achieve our results, we develop a novel
  \emph{session decomposition} technique, from \emph{full}
  session types (including internal/external choices) into
  \emph{single input/output session trees} (without choices).

  Previous work studies precise subtyping for %
  \emph{binary} sessions (with just two participants), %
  or multiparty sessions (with any number of participants) and
  \emph{synchronous} interaction. %
  Here, we cover \emph{multiparty} sessions with \emph{asynchronous}
  interaction, where messages are transmitted via FIFO queues (as in
  the TCP/IP protocol), and prove that our subtyping is both
  operationally and denotationally precise. %
  In the asynchronous multiparty setting, finding the precise
  subtyping relation is a highly complex task: %
  this is because, under some conditions, %
  participants can permute the order of their inputs and outputs, by
  sending some messages earlier or receiving some later, %
  without causing errors; the precise subtyping relation must capture
  \emph{all} such valid permutations %
  --- and consequently, its formalisation, reasoning and proofs become
  challenging. %
  Our session decomposition technique overcomes this complexity,
  expressing the subtyping relation %
  as a composition of refinement relations between single input/output
  trees, and providing a simple reasoning principle for asynchronous
  message optimisations.
\end{abstract}


\maketitle
\section{Introduction}\label{sec:intro}
Modern software systems are routinely designed and developed %
as ensembles of concurrent and distributed components, %
interacting via message-passing %
according to pre-determined \emph{communication protocols}. %
A key challenge lies in ensuring that %
each component abides by the desired protocol, %
thus avoiding run-time failures %
due to, \eg, communication errors and deadlocks. %
One of the most successful approaches to this problem %
are \textbf{session types} \cite{THK,HVK}.
In their original formulation,
session types allow formalising two-party protocols
(\eg, for client-server interaction),
whose structure includes sequencing, choices, and recursion;
they were later extended to
\emph{multiparty} protocols \cite{HYC08,HYC16}.
By describing \emph{(multiparty) protocols as types}, %
session types provide a type-based methodology to statically verify
whether a given process implements a desired protocol.
Beyond the theoretical developments, %
multiparty session types have been implemented in mainstream programming languages
such as Java, Python, Go, Scala, C, TypeScript, F$\sharp$,
OCaml, Haskell, Erlang \cite{BTTrends,bettytoolbook}.%

\paragraph{Precise subtyping} The substitution
principle of \citeN{Liskov:1994:BNS} %
establishes a general notion of \emph{subtyping}: %
if $T$ is subtype of $T'$, then
an object of type $T$ can always replace an object of type $T'$,
in any context.
Similar notions arise in the realm of process calculi, %
since \citet{PierceSangiorgi95} first introduced
IO-subtyping for input and output channel capabilities in the $\pi$-calculus. %
As session types are protocols,
the notion of \emph{session subtyping}
(first introduced by \citeN{GH05}) can be interpreted as %
\textbf{\emph{protocol refinement}}:
given two types/protocols $T$ and $T'$, %
if $T$ is a subtype (or refinement) of $T'$, %
then a process that implements $T$ can be used %
whenever a process implementing $T'$ is needed. %
In general, when a type system is equipped with a subtyping (subsumption) rule,
then we can enlarge the set of typable programs
by enlarging its subtyping relation.
On the one hand, a larger subtyping can be desirable,
since it makes the type system more flexible,
and the verification more powerful;
however, a subtyping relation that is \emph{too} large
makes the type system unsound:
\eg, if we consider $\tstr$ as a subtype of $\treal$,
then expressions like\; $1 + \mathtt{``foo"}$ \;become typable,
and typed programs can crash at run-time. %
Finding the ``right subtyping'' (not too strict, nor too lax)
leads to the problem of finding a canonical, \emph{precise} subtyping
for a given type system --- \ie, a subtyping relation that is sound
(``typed programs never go wrong'') and cannot be further enlarged
(otherwise, the type system would become unsound).
The problem has been widely studied for the $\lambda$-calculus %
\cite{BHLN12,BHLN17,DG14,SIAM}; %
several papers have also %
tackled the problem in the realm of session types
\cite{CDSY2017,GhilezanJPSY19}. 
A session subtyping relation $\subt$ is \textbf{\emph{precise}} when it is both %
\emph{sound} and \emph{complete}:

\begin{description}
\item[soundness] means that, if we have a context $C$ %
  expecting some process $P'$ of type $T'$, then $T \subt T'$ implies that %
  any process $P$ of type $T$ can be placed into $C$ %
  without causing ``bad behaviours'' %
  (\eg, communication errors or deadlocks); %
\item[completeness] means that %
  $\subt$ cannot be extended without becoming unsound. %
  More accurately: if $T \not\subt T'$, then %
  we can find a process $P$ of type $T$, %
  and a context $C$ expecting a process of type $T'$, %
  such that if we place $P$ in $C$, %
  it will cause ``bad behaviours.''
\end{description}

\paragraph{Asynchronous Multiparty Session Subtyping}
This work tackles the problem of finding
the precise subtyping relation $\subt$ %
for \emph{multiparty asynchronous} session types. %
The starting point is a type system for
processes that %
\begin{enumerate*}[label=\emph{(\arabic*)}]
\item
  implement types/protocols with 2 or more participants, and
\item
  communicate through a medium %
  that buffers messages %
  while preserving their order %
  --- as in TCP/IP sockets,
  and akin to the original papers on multiparty session types
  \cite{HYC08,BCDLDY2008}.
\end{enumerate*}

For example, consider a scenario where participant $\pr$ %
waits for the outcome of a computation from $\pp$, %
and notifies $\pq$ on whether to continue the calculation %
(we omit part of the system ``$\cdots$''):
\[
\begin{array}{c}
\pa\pr{P_\pr}%
\;\pc\;
\pa\pp{P_\pp}%
\;\pc\;
\pa\pq{P_\pq}%
\;\pc\;%
\cdots%
\\[2mm]%
\text{where }\;%
P_\pr \;=\;
\sum\left\{%
  \begin{array}{@{}l@{}}
    \procin{\pp}{\msgLabel{success}(x)}{
      \cond{(x > \zero)}{\procout{\pq}{\msgLabel{cont}}{x}{\inact}}
                   {\procout{\pq}{\msgLabel{stop}}{}{\inact}}
    }%
    \\%
    \procin{\pp}{\msgLabel{error}(\mathit{fatal})}{
      \cond{(\neg\mathit{fatal})}{\procout{\pq}{\msgLabel{cont}}{42}{\inact}}
                            {\procout{\pq}{\msgLabel{stop}}{}{\inact}}
    }%
  \end{array}
\right\}
\end{array}
\]

\noindent%
Above, $\pa\pr{P_\pr}$ denotes a process $P_\pr$ %
executed by participant $\pr$, %
$\procinNoSuf{\pp}{\ell(x)}$ is an input %
of message $\ell$ with payload value $x$ from participant $\pp$, %
and $\procoutNoSuf{\pq}{\ell}{5}$ is an output of message $\ell$ %
with payload $5$ to participant $\pq$. %
In the example, %
$\pr$ waits to receive %
either $\msgLabel{success}$ or $\msgLabel{error}$ %
from $\pp$. In case of $\msgLabel{success}$, %
$\pr$ checks whether the message payload $x$
is greater than $\zero$ (zero), and tells $\pq$ to either $\msgLabel{cont}$inue
(forwarding the payload $x$) or $\msgLabel{stop}$, %
then terminates ($\inact$); %
in case of $\msgLabel{error}$, %
$\pr$ checks whether the error is non-fatal, %
and then tells $\pq$ to either $\msgLabel{cont}$inue
(with a constant value $42$), or $\msgLabel{stop}$. %
Note that $\pr$ is blocked until a message is sent by $\pp$, %
and correspondingly, $\pq$ is waiting for $\pr$, who is waiting for $\pp$.
Yet, depending on the application, %
one might attempt to \emph{locally optimise} $\pr$, %
by replacing $P_\pr$ above with the process:
\[
P'_\pr \;\;=\;\;%
  \cond{(\ldots)}{%
    \procout{\pq}{\msgLabel{cont}}{42}{%
      \sum\left\{%
        \begin{array}{@{}l@{}}
          \procin{\pp}{\msgLabel{success}(x)}{\inact}%
          \\%
          \procin{\pp}{\msgLabel{error}(y)}{\inact}%
        \end{array}
      \right\}%
    }%
  }{%
    \procout{\pq}{\msgLabel{stop}}{}{%
      \sum\left\{%
        \begin{array}{@{}l@{}}
          \procin{\pp}{\msgLabel{success}(x)}{\inact}%
          \\%
          \procin{\pp}{\msgLabel{error}(y)}{\inact}%
        \end{array}
      \right\}%
    }%
  }%
\]
\noindent%
Process $P'_\pr$ internally decides (with an omitted condition ``$\ldots$'') %
whether to tell $\pq$ to $\msgLabel{cont}$inue with a constant value $42$, %
or $\msgLabel{stop}$. %
Then, $\pr$ receives the $\msgLabel{success}$/$\msgLabel{error}$ message
from $\pp$, and does nothing with it. %
As a result, $\pq$ can start its computation immediately, %
without waiting for $\pp$. %
Observe that this optimisation permutes the order of inputs and outputs %
in the process of $\pr$: is it ``correct''? %
\Ie, could this permutation %
introduce any deadlock or communication error in the system? %
Do we have enough information to determine it, %
or do we need to know the behaviour of $P_\pp$ and $P_\pq$, %
and the omitted part (``$\cdots$'') of the system? %
If this optimisation never causes deadlocks or communication errors,
then a session subtyping relation should allow for it: %
\ie, %
if $T'$ is the type of $P'_\pr$, %
and $T$ is the type of $P_\pr$, %
we should have $T' \subt T$ %
--- hence, the type system should let %
$P'_\pr$ be used in place of $P_\pr$. %
(We illustrate such types later on, in Example~\ref{ex:sub-intro}.) %
Due to practical needs, similar program optimisations  %
have been implemented 
for various programming languages \cite{NYH12,NCY2015,H2017,CY2020,CY2020B,YoshidaVPH08}. %
Yet, this optimisation is \emph{not} allowed %
by \emph{synchronous} multiparty session subtyping \cite{GhilezanJPSY19}: %
in fact, under synchronous communication, %
there are cases where replacing $P_\pr$ with $P'_\pr$
would introduce deadlocks. %
However, most real-world distributed and concurrent systems %
use \emph{asynchronous} communication: %
does asynchrony make the optimisation above \emph{always} safe, %
and should the subtyping allow for it? %
If we prove that this optimisation, and others, %
are indeed sound under asynchrony, %
it would be possible %
to check them locally, at the type-level, %
for each individual participant in a multiparty session.%

\paragraph{Contributions}%
We present the \textbf{first formalisation of the \emph{precise} subtyping
relation $\subt$ for multiparty asynchronous session types}. %
We introduce the relation (\Cref{subsec:subtyping:formal})
and prove that it is \textbf{operationally
  precise} (\Cref{lem:completeness}),
\ie, satisfies the notions of ``soundness'' and
``completeness'' outlined above. Then, we use this result as a
stepping stone to prove that $\subt$ is also
\textbf{denotationally precise} (\Cref{thm:dpreciseness}),
and \textbf{precise \wrt liveness} (\Cref{thm:subt-precise-liveness}).

A key element of our contribution is a novel approach based on
\textbf{session decomposition}: %
given two session types $T$ and $T'$, we formalise the subtyping
$T \subt T'$ as a composition of refinement relations $\subttt$ over
\emph{single-input, single-output (SISO) trees} extracted from $T$ and
$T'$. (The idea and its motivation are explained in \Cref{sec:session-trees-refinement}.)

We base our development on a recent advancement of
the multiparty session types theory \cite{ScalasY19}: %
this way, we achieve not only %
a precise subtyping relation, 
but also a simpler formulation, %
and more general results than previous work
on multiparty session subtyping \cite{CDSY2017,MostrousY15,GhilezanJPSY19}, %
supporting the verification of a larger set of concurrent and distributed
processes and protocols.

To demonstrate the tractability and generality of our subtyping
relation, we discuss various examples --- including one that is
\emph{not} supported by a sound algorithm for asynchronous
\emph{binary} session subtyping (\Cref{ex:CONCUR19-maintext}), and
another where we prove the correctness of a messaging optimisation
(based on \emph{double buffering} \cite[Section 3.2]{doublebuffer},
adapted from \cite{mostrous_yoshida_honda_esop09,CY2020B,YoshidaVPH08}) applied to a
distributed data processing scenario (\Cref{sec:batch-processing}).

\paragraph{Overview}%
Section~\ref{sec:msc} formalises %
the asynchronous multiparty session calculus. 
Section~\ref{sec:tsu} presents our asynchronous multiparty session subtyping
relation, with its decomposition technique.
Section~\ref{sec:tsy} introduces the typing system, proving its soundness. %
Section~\ref{sec:op} proves the completeness and preciseness of our subtyping.
Section~\ref{sec:batch-processing} applies our subtyping
to prove the correctness of an asynchronous
optimisation of a distributed system.
Section~\ref{sec:preciseness-extra} proves additional preciseness results:
denotational preciseness, and preciseness \wrt liveness.
Related work is in Section~\ref{sec:related}. %
\iftoggle{techreport}{%
  Proofs and additional examples
  are available in the appendices.%
}{%
  \change{Tech report}{%
    Proofs and additional examples
    are available in a separate technical report \cite{POPL21TR}.
  }%
}

\section{Asynchronous Multiparty Session Calculus}\label{sec:msc}
This section formalises the syntax and operational semantics of an
asynchronous multiparty session calculus. %
Our formulation is a streamlined presentation of the session calculus
by \citeN{BCDLDY2008},
omitting some elements (in particular, session creation and shared channels)
to better focus on subtyping.
The same design is adopted, \eg, by \citeN{GhilezanJPSY19}
--- but here we include message queues,
for asynchronous (FIFO-based) communication.

\begin{table}[t]
  \[
  \begin{array}{c}
    \val \;\Coloneqq\; 
    \valr \sep \true \sep \false \sep
    \changeNoMargin{()}
    \qquad
    \changeNoMargin{\valr \;\Coloneqq\; {\zero} \sep \valn \sep -\valn}
    \qquad
    \changeNoMargin{\valn \;\Coloneqq\; 1 \sep \valn+1}
    \\
    \e \;\Coloneqq\; x \sep \val \sep \fsucc{\e} \sep \finv{\e} \sep \neg{\e} \sep \changeNoMargin{\e > \zero \sep \e\approx ()}
  \end{array}
  \]
\[
  \begin{array}{@{}rcll@{}}
    \N & \Coloneqq &  & \text{\em\bf Sessions} \\
       &         & \pa\pp\PP \pc \pa\pp \h & \text{\em individual participant} \\
       & \sep & \N\pc\N  & \text{\em parallel} \\
       & \sep & \error & \text{\em error} \\[2mm]
    \h & \Coloneqq & & \text{\em\bf Message queues} \\
       &         & \emptyqueue & \text{\em empty queue} \\
       & \sep & \left(\q,\ell(\val)\right) & \text{\em message} \\
       & \sep & \h\cdot \h  &\text{\em concatenation}
  \end{array}
  \qquad%
  \begin{array}{@{}rcll@{}}
    \PP,\PQ & \Coloneqq & \text{\em\bf Processes} \\
            & & \sum_{i\in I}\procin{\pp}{\ell_i(\x_i)}{\PP_i} & \text{\em external choice}\\
            & \sep & \procout{\pp}{\ell}{\e}{\PP} & \text{\em output} \\
            & \sep & \cond{\e} \PP \PQ & \text{\em conditional} \\ 
            & \sep & X & \text{\em variable} \\   
            & \sep & \mu X.\PP & \text{\em recursion} \\ 
            & \sep &  \inact  & \text{\em inaction}
  \end{array}
  \]

\caption{\label{tab:sessions}\label{tab:expressions}Syntax of values ($\val$), expressions ($\e$), sessions, processes, and queues.} %
\vspace{-5mm}%
\end{table}

\subsection{Syntax}%

The syntax of our calculus is defined in Table~\ref{tab:sessions}.
Values and expressions are standard:
a \textbf{value} $\val$ can be 
\changeNoMargin{an integer $\valr$ (positive $\valn$, negative $-\valn$, or zero $\zero$)}
, a boolean $\true$/$\false$,
\change{Clarify \Cref{tab:expressions} and $\tunit$}{%
  or unit $()$ (that we will often omit, for brevity)%
};
an \textbf{expression} $\e$ can be a
variable, a value, or a term built from expressions by applying the
operators ${\tt succ}, {\tt inv}, \neg$, or the relations \changeNoMargin{$>, \approx$}. 

\textbf{Asynchronous multiparty sessions}
(ranged over by $\N, \N',\ldots$) 
are parallel compositions of individual \textbf{participants} (ranged over by $\pp,\pq,\ldots$) %
associated with their own \textbf{process} $\PP$ and \textbf{message queue} $\h$ %
(notation: $\pa\pp\PP \pc \pa\pp \h$). %
In the processes syntax, %
the \textbf{external choice} $\sum_{i\in I}\procin{\pp}{\ell_i(\x_i)}{\PP_i}$ %
denotes the input from participant $\pp$ %
of a message with label $\ell_i$ carrying value $x_i$, for any $i \!\in\! I$;
instead, $\procout{\pp}{\ell}{\e} \PP$ %
denotes the \textbf{output} towards participant $\pp$ %
of a message with label $\ell$ %
carrying the value returned by expression $\e$. 
The \textbf{conditional}\; $\cond{\e} \PP \PQ$ \;is standard. %
\change{\#B, \#D: queues of sent vs.~received messages}{}%
The term $\pa\pp \h$ states that $\h$ is the \textbf{output message queue} of participant $\pp$; if a message $(\q,\ell(\val))$ is in the queue of participant $\pp$, it means that $\pp$ has sent $\ell(\val)$ to $\pq$.%
\footnote{%
  \changeNoMargin{%
    Alternatively, we could formalise a calculus
    with actor-style \emph{input} message queues,
    at the cost of some additional notation;
    the semantics would be equivalent \cite{DemangeonY15} %
    and would not influence subtyping.%
  }%
}
Messages are consumed by their recipients %
on a FIFO (first in, first out) basis.
The rest of the syntax is standard \cite{GhilezanJPSY19}.
We assume that in recursive processes, %
recursion variables are guarded by external choices and/or outputs.

We also define the set $\actions\PP$,
containing the \textbf{input and output actions} of $\PP$;
its elements have the form $\pp?$ or $\pp!$,
representing an input or an output from/to participant $\pp$, respectively:

\[
\begin{array}{c}
\actions{\inact}=\emptyset \qquad
\actions{\procout{\pp}{\ell}{\e}{\PP}} = \{\pp!\} \cup \actions{\PP} \qquad
\actions{\sum_{i\in I}\procin{\pp}{\ell_i(\x_i)}{\PP_i}}
= \{\pp?\} \cup \bigcup_{i \in I}\actions{\PP_i} \qquad
\\
\actions{\mu X.\PP'}=\actions{\PP'} \qquad
\actions{\cond{\e} {\PP_1}  {\PP_2}} =
\actions{\PP_1} \cup \actions{\PP_2}
\end{array}
\]

\subsection{Reductions and Errors}
First, we give the operational semantics of expressions.
To this purpose, we define an \textbf{evaluation context} $\Econtext$
as an expression (from \Cref{tab:expressions}) with exactly one hole $[\,]$, %
\change{\#D: clarify $\Econtext$}{%
  given by the grammar:
  \[
    \Econtext \;\Coloneqq\;
    [\,] \sep \fsucc{\Econtext} \sep \finv{\Econtext} \sep \neg{\Econtext} \sep \Econtext > \zero
  \]
}%
We write $\Econtext(\e)$ for the expression obtained from
context $\Econtext$ by filling its unique hole with $\e$.
The \textbf{value of an expression} is computed as defined in %
Table~\ref{tab:evaluation}:
the notation  $\eval\e\val$ means that expression $\e$ evaluates to $\val$.
Notice that the successor operation ${\tt succ}$ is
defined on natural numbers {\changeNoMargin{(i.e. positive integers)}, the inverse operation ${\tt inv}$ is
defined on integers, and negation $\neg$ is 
defined on boolean values;
\change{Clarify $\tunit$ and evaluation}{moreover, the evaluation of\; $\e>\zero$ \;is defined only if $\e$ evaluates to some integer, while the evaluation of\; $\e\approx ()$ \;is defined only if $\e$ is exactly the unit value.}

The operational semantics of our calculus
is defined in Table~\ref{tab:main:reduction}.
By $\rulename{r-send},$ a participant $\pp$ sends $\ell\langle \e\rangle$ to a participant $\pq$, enqueuing the 
message $\msg\pq\ell{\val}$.
Rule $\rulename{r-rcv}$ lets participant $\pp$ %
receive a message from $\pq$: %
if one of the input labels $\ell_k$ %
matches a queued message $\msg\pp{\ell_k}{\val}$ previously sent by $\pq$ %
(for some $k \in I$), %
the message is dequeued, %
and the continuation $\PP_k$ proceeds %
with value $\val$ substituting $\x_k$. %
The rules for conditionals %
are standard.
Rule \rulename{r-struct} defines
the reduction modulo a standard %
\textbf{structural congruence} $\equiv$, %
defined in Table \ref{tab:congruence}. 

\begin{table}[t]{}
\centerline{
\(
\small
 \begin{array}[t]{@{}c@{}}
\eval{\fsucc\valn}(\valn +1)
 \qquad
 \changeNoMargin{\eval{\finv\valn}{-\valn}
  \qquad
 \eval{\finv{-\valn}}{\valn}
  \qquad
 \eval{\finv\zero}{\zero}}
 \\
 \eval{\neg\true}\false \qquad \eval{\neg\false}\true
 \qquad 
 \changeNoMargin{\eval{\big(()\approx ()\big)}{\true}}
 \qquad \eval\val\val
 \\[2mm]
\changeNoMargin
 { \inferrule[]{\eval{\e_1}{\valn}}{\eval{(\e_1 > \zero)}{\true}}
 \qquad
 \inferrule[]{\eval{\e_1}{-\valn}}{\eval{(\e_1 > \zero)}{\false}}
  \qquad
 \inferrule[]{\eval{\e_1}{\zero}}{\eval{(\e_1 > \zero)}{\false}}
}
 \qquad
 \inferrule[]{\eval{\e}{\val}\quad\eval{\Econtext(\val)}{\val'}}{\eval{\Econtext(\e)}{\val'}}
\end{array}
\)
}
\vspace{1mm}
\caption{\label{tab:evaluation}Evaluation rules for expressions.}

\centerline{\(%
\small
\begin{array}{@{}llr@{}}
\rulename{r-send} & 
{\pa\pp\procout{\pq}{\ell}{\e}{\PP} \pc \pa\pp\h_{\pp} \pc \N \;\red\; \pa\pp\PP \pc \pa\pp\h_\pp\cdot(\q,\ell(\val)) \pc \N} & (\eval{\e} \val)       
\\[1mm]
\rulename{r-rcv}
&
\pa\pp\sum_{i\in I} \procin\pq{\ell_i(\x_i)}\PP_i \pc \pa\pp\h_\pp\pc \pa\q\Q \pc \pa\pq (\pp,\ell_k(\val)) \cdot \h \pc \N\hspace{-10mm}
 & (k\in I)
\\
& \hspace{3.2cm}\red\; \pa\pp \PP_k\subst{\val}{\x_k} \pc \pa\pp\h_\pp\pc \pa\q\Q \pc \pa\pq \h \pc \N
\hspace{-5mm}
\\[1mm]
\rulename{r-cond-T} & 
\pa\pp{\cond{\e}{\PP}{\Q}}  \pc \pa\pp\h \pc \N \red \pa\pp\PP \pc\pa\pp\h \pc  \N & (\eval{\e}{\true})
\\[1mm]
\rulename{r-cond-F} & 
\pa\pp{\cond{\e}{\PP}{\Q}}  \pc \pa\pp\h \pc \N \red \pa\pp\Q \pc \pa\pp\h \pc \N
& (\eval{\e}{\false})
\\[1mm]
\rulename{r-struct} & 
\N_1\equiv \N_1' \;\;\;\text{and}\;\;\; \N_1'\red \N_2' \;\;\;\text{and}\;\;\; \N_2'\equiv\N_2 \quad\implies\quad 
\N_1\red\N_2
\hspace{-15mm}\\[1mm]
\rulename{err-mism}&
   \pa\pp\sum_{i\in I} \procin\pq{\ell_i(\x_i)}\PP_i \pc \pa\pp\h_\pp  \pc\pa\pq \PQ \pc\pa\pq \msg{\pp}{\ell}{\val}\cdot\h \pc \M \;\red\; \error
     \hspace{-15mm}%
   &   (\forall i\in I. \ell_i\neq \ell)
   \\[1mm]
\rulename{err-ophn}&  \pa\pp\PP \pc \pa\pp \h_\pp \pc \pa\pq \PQ\pc \pa\pq\msg{\pp}{\ell}{\val}\cdot\h \pc \M \;\red\; \error &
({\pq?} \not\in \actions{\PP})\\[1mm]
\rulename{err-strv} & 
    \pa\pp\sum_{i\in I} \procin\pq{\ell_i(\x_i)}\PP_i  \pc \pa\pp\h_\pp \pc 
    \pa\pq\PQ \pc \pa\pq\h_\pq \pc \M \;\red\; \error%
&
\hspace{-4cm}
({\pp!} \not\in \actions{\PQ}, 
\h_\pq \not\equiv \msg{\pp}{-}{-} \cdot \h_\pq')
\\[1mm]
\rulename{err-eval}&
{\pa\pp \cond{\e}{\PP}{\PQ}  \pc \pa\pp\h \pc \M  \;\red\; \error}
& %
(\noteval{\e}{\true} \;\text{ and }\; \noteval{\e}{\false})
\\[1mm]
\rulename{err-eval2} & 
{\pa\pp\procout{\pq}{\ell}{\e}{\PP} \pc \pa\pp\h \pc \N \;\red\; \error}
& ({\not\exists \val: \eval\e \val})%
\\[1mm]
\rulename{err-dlock}& 
 \changeNoMargin{\prod_{j\in J} \left( \pa{\pp_j}{\sum_{i_j\in I_j} \procin{\pq_j}{\ell_{i_j}(\x_{i_j})}\PP_{i_j}} \pc \pa{\pp_{j}}{\h_{\pp_{j}}}\right)
  \pc
  \prod_{j\in J'} \left( \inact \pc \pa{\pp_{j}}{\h_{\pp_{j}}}\right)}
  & \changeNoMargin{ \hspace{-1.3cm} (\forall j \!\in\! J \!\neq\! \emptyset:
    \forall i \!\in\! J \!\cup\! J':}%
    \\
  &  \hspace{3.2cm}\red \error
  & \changeNoMargin{\h_{\pp_i} \!\not\equiv\! (\pp_j,-(-)) \cdot \h'_{\pp_i})}
\end{array}
\)}%
\smallskip%
\vspace{2mm}
\caption{\label{tab:main:reduction}\label{tab:error}Reduction relation on sessions.}%
\small
$\begin{array}{c}
\h_1\cdot \msg{\pq_1}{\ell_1}{\val_1} \cdot \msg{\q_2}{\ell_2}{\val_2}\cdot \h_2  \equiv  \h_1 \cdot \msg{\q_2}{\ell_2}{\val_2}\cdot \msg{\pq_1}{\ell_1}{\val_1}\cdot\h_2 \quad\text{(if $\q_1\neq \q_2$)}\\[1mm]
 \EmptyQueue \cdot\Queue \equiv  \Queue \qquad\qquad
 \Queue\cdot \EmptyQueue \equiv  \Queue \qquad\qquad
 \Queue_1\cdot (\Queue_2\cdot\Queue_3) \equiv (\Queue_1 \cdot \Queue_2)\cdot\Queue_3
\qquad\qquad
 \mu X.\PP \equiv \PP\subst{\mu X.\PP}{X}  
 \\[1mm]
 \pa\pp\inact \pc \pa\pp\EmptyQueue \pc \N \equiv \N \qquad \N_1\pc \N_2 \equiv  \N_2\pc \N_1 \qquad 
 (\N_1\pc \N_2) \pc \N_3 \,\equiv\,  \N_1\pc (\N_2 \pc \N_3) \qquad 
 \\[1mm]
 \PP\equiv \Q \;\text{ and }\;\h_1 \equiv \h_2 \;\;\implies\;\; \pa\pp\PP\pc \pa\pp\h_1 \pc \N \,\equiv\, \pa\pp\Q\pc \pa\pp\h_2 \pc \N
\end{array}
$\\
\vspace{2mm}
\caption{\label{tab:congruence}Structural congruence rules for queues, processes, and sessions.}
\end{table}

Table~\ref{tab:error} also formalises \textbf{error reductions}, %
modelling the following scenarios: %
in \rulename{err-mism}, 
a process tries to read a queued message with an unsupported label;
in \rulename{err-orph}, %
there is a queued message from $\pq$ to $\pp$, %
but $\pp$'s process does not contain any input from $\pq$,  
hence the message is orphan;
in \rulename{err-strv}, 
$\pp$ is waiting for a message from $\pq$, %
but no such message is queued, %
and $\pq$'s process does not contain any output for $\pp$, %
hence $\pp$ will starve;
in \rulename{err-eval}, a condition does not evaluate to a boolean value;
in \rulename{err-eval2}, %
an expression like ``$\fsucc{\true}$'' cannot reduce to any value;
in \rulename{err-dlock}, 
\change{Clarify \rulename{err-dlock}}{%
  the session cannot reduce further, %
  but at least one participant is expecting an input.%
}%

\begin{example}[Reduction relation]
  \label{ex:reduction} We now describe the operational semantics using the example from the Introduction. Consider the session:
\[
\small
\pa\pr\sum\left\{%
  \begin{array}{@{}l@{}}
    \procin{\pp}{\msgLabel{success}(x)}{
      \cond{(x > \zero)}{\procout{\pq}{\msgLabel{cont}}{x}{\inact}}
                   {\procout{\pq}{\msgLabel{stop}}{}{\inact}}
    }%
    \\%
    \procin{\pp}{\msgLabel{error}(\mathit{fatal})}{
      \cond{(\neg\mathit{fatal})}{\procout{\pq}{\msgLabel{cont}}{42}{\inact}}
                            {\procout{\pq}{\msgLabel{stop}}{}{\inact}}
    }%
  \end{array}
\right\}
\;\pc\;
\pa\pr\EmptyQueue
\;\pc\;
\pa\pp\PP_\pp
\;\pc\;
\pa\pp\EmptyQueue
\;\pc\;
\pa\pq\PP_\pq
\;\pc\;%
\cdots%
\]
In this session, the process of $\pr$ is blocked until a message is sent by $\pp$ to $\pr$. If this cannot happen (\eg, because $\PP_\pp$ is $\inact$),
the session will reduce to $\error$ by $\rulename{err-starv}$. %
Now, consider the session above optimised with the process:
\[
\pa\pr{%
  \cond{(\e)}{%
    \procout{\pq}{\msgLabel{cont}}{42}{%
      \sum\left\{%
        \begin{array}{@{}l@{}}
          \procin{\pp}{\msgLabel{success}(x)}{\inact}%
          \\%
          \procin{\pp}{\msgLabel{error}(y)}{\inact}%
        \end{array}
      \right\}%
    }%
  }{%
    \procout{\pq}{\msgLabel{stop}}{}{%
      \sum\left\{%
        \begin{array}{@{}l@{}}
          \procin{\pp}{\msgLabel{success}(x)}{\inact}%
          \\%
          \procin{\pp}{\msgLabel{error}(y)}{\inact}%
        \end{array}
      \right\}%
    }%
  }%
}%
\]
This session could reduce to $\error$ if $\e$
does \emph{not} evaluate to $\true$ or $\false$ (by $\rulename{err-eval}$).
Instead, if $\eval{\e}{\true}$, then $\pr$ can fire $\rulename{r-cond-T}$
and $\rulename{r-send}$, reducing the session to:
\[
\pa\pr{%
  {%
      \sum\left\{%
        \begin{array}{@{}l@{}}
          \procin{\pp}{\msgLabel{success}(x)}{\inact}%
          \\%
          \procin{\pp}{\msgLabel{error}(y)}{\inact}%
        \end{array}
      \right\}%
    }%
}%
\pc \pa \pr (\q,\msgLabel{cont}{(42)})
\;\pc\;
\pa\pp\PP_\pp
\;\pc\;
\pa\pp\EmptyQueue
\;\pc\;
\pa\pq\PP_\pq
\;\pc\;%
\cdots%
\]
Otherwise, if $\eval{\e}{\false}$, then $\pr$ can fire $\rulename{r-cond-F}$
and $\rulename{r-send}$, reducing the session to:
\[
\pa\pr{%
  {%
      \sum\left\{%
        \begin{array}{@{}l@{}}
          \procin{\pp}{\msgLabel{success}(x)}{\inact}%
          \\%
          \procin{\pp}{\msgLabel{error}(y)}{\inact}%
        \end{array}
      \right\}%
    }%
}%
\pc \pa \pr (\q,\msgLabel{stop}{( )})
\;\pc\;
\pa\pp\PP_\pp
\;\pc\;
\pa\pp\EmptyQueue
\;\pc\;
\pa\pq\PP_\pq
\;\pc\;%
\cdots%
\]
In both cases, $\pr$ reduces to $\inact$ (by $\rulename{r-rcv}$)
if it receives a %
$\msgLabel{success}$/$\msgLabel{error}$ message from $\pp$; %
meanwhile, if $\pq$ is ready to receive an input from $\pp$,
then $\pq$ can continue by consuming a message from $\pp$'s output queue.
This kind of optimisation will be verified by means of subtyping in the following section.
\end{example}

\section{Multiparty Session Types and Asynchronous Subtyping}\label{sec:tsu}
This section formalises multiparty session types,
and introduces our asynchronous session subtyping relation. %
We begin with the standard definition of (local) session types.

\begin{definition}\label{def:types}\label{def:sorts}
The \emph{sorts} $\ST$ and \emph{session types} $\T$ are defined as follows:
\[
  \begin{array}{@{}r@{\;}c@{\;}l@{\hskip 0.8cm}r@{\;\;}c@{\;\;}l@{}}
    \ST &\Coloneqq& \tnat \!\!\sep\!\! \tint \!\!\sep\!\! \tbool \!\!\sep\!\! \tunit
    &
    \T  &\Coloneqq& \texternal_{i\in I}\tin\pp{\ell_i}{\ST_i}.\T_i   \sep   \tinternal_{i\in I}\tout\pp{\ell_i}{\ST_i}.\T_i 
    \sep \tend  \sep      \mu\ty. \T   \sep     \ty
  \end{array}
\]
with \change{Clarify}{$I \!\neq\! \emptyset$},
and $\forall i,j\!\in\! I{:}\; %
i \!\neq\! j \Rightarrow \ell_i \!\neq\! \ell_j$. %
We assume guarded recursion. %
We define $\equiv$ as the least congruence %
such that $\mu\ty. \T \equiv \T\subst{\mu\ty. \T}{\ty}$.
We define $\participant{\T}$ as the set of participants occurring in $\T$.%
\end{definition}
Sorts are the types of values ($\tnat$urals,  $\tint$egers, $\tbool$eans, \ldots).
A session type $\T$ describes the behaviour of a participant in a multiparty session. 
The \textbf{branching type} (or \textbf{external choice})\; $\texternal_{i\in I}\tin\pp{\ell_i}{\ST_i}.\T_i$ \;denotes waiting for a message from participant $\pp$, 
where (for some $i\in I$) the message has label $\ell_i$ and carries a payload value of sort $\ST_i$; then, the interaction continues by following $\T_i$. 
The \textbf{selection type} (or \textbf{internal choice})\; $\tinternal_{i\in I}\tout\pp{\ell_i}{\ST_i}.\T_i$ %
\;denotes an output toward participant $\pp$ of a message with label $\ell_i$ and payload of sort $\ST_i$, after which the interaction follows $\T_i$ (for some $i\in I$). %
Type $\mu\ty. \T$ provides \textbf{recursion}, %
binding the \emph{recursion variable} $\ty$ in $\T$; %
the guarded recursion assumption means: in $\mu\ty.\T$, we have $\T\not=\ty'$ for any $\ty'$ (which ensures \emph{contractiveness}).
Type $\tend$ denotes that the participant has concluded its interactions.
For brevity, %
we often omit %
branch/selection symbols for singleton inputs/outputs, %
\change{Clarify $\tunit$}{%
  payloads of sort $\tunit$%
},
unnecessary parentheses, and trailing $\tend$s.%

\subsection{Session Trees and Their Refinement}%
\label{sec:session-trees-refinement}

Our subtyping is defined in two phases:
\begin{enumerate} 
\item%
  we introduce a \emph{refinement relation} $\subttt$ for \emph{session trees} %
  having only singleton choices in all branchings and selections, %
  called \emph{single-input-single-output (SISO) trees} 
  (definition below);
\item%
  then, we consider trees that have %
  only singleton choices in branchings %
  (called \emph{single-input (SI) trees}), %
  or in selections (\emph{single-output (SO) trees}), %
  and we define the session subtyping $\subt$ over all session types %
  by considering their decomposition into SI, SO, and SISO trees.
\end{enumerate}
This two-phases approach is crucial %
to capture all input/output reorderings needed %
by the \emph{precise} subtyping relation, %
while taming the technical complexity of its formulation.
In essence, our session decomposition
is a ``divide and conquer'' technique to separately tackle the main
sources of complications in the definition of the subtyping, and in
the proofs of preciseness:
\begin{itemize}
\item%
  on the one hand, the SISO trees refinement $\subttt$ focuses on
  capturing safe permutations and alterations of input/output messages,
  that never cause deadlocks or communication errors under asynchrony;
\item%
  on the other hand, the subtyping relation $\subt$ focuses on
  reconciling the SISO tree refinement $\subttt$ with the branching
  structures (\ie, the choices) occurring in session types.
\end{itemize}

\newcommand{\payloadTag}{\mathsf{P}}%
\newcommand{\contTag}{\mathsf{C}}%

\begin{wrapfigure}[7]{r}{4cm}
  \vspace{-5mm}%
  \;\scalebox{0.9}{%
    \begin{minipage}{1.1\linewidth}%
\centerline{\(%
  \xymatrix@R-1.9pc@C-2pc@L0pt{%
    &&{\&\pp}\ar@{-}[dll]_{\ell_1^{\payloadTag}}\ar@{-}[ddl]_{\ell_1^{\contTag}}\ar@{-}[ddr]^{\ell_2^{\payloadTag}}\ar@{-}[drr]^{\ell_2^{\contTag}}&&&\\
      \tbool&&&&\tend\\
      &{\oplus\pq}\ar@{-}[dl]_{\ell_3^{\payloadTag}}\ar@{-}[dd]_{\ell_3^{\contTag}}\ar@{-}[ddr]^{\ell_4^{\payloadTag}}\ar@{-}[drr]^{\ell_4^{\contTag}}&&\tnat\\
     \tint&&&\tend&\\
     &_{\vdots}{\&\pp}_{\vdots}&\treal%
  }%
\)}%
    \end{minipage}
  }%
\end{wrapfigure}

\paragraph{Session trees}
To define our subtyping relation, %
we use (possibly infinite) \emph{session trees} %
with the standard formulation %
of~\cite[Appendix A.2]{GhilezanJPSY19}, %
based on \citeN{PierceBC:typsysfpl}. %
The diagram on the right depicts a session tree: %
its internal nodes represent %
branching ($\&\pp$) or selection ($\oplus\pq$) from/to a participant; %
leaf nodes are either payload sorts or $\tend$; %
edge annotations are either $\ell^\payloadTag$ or $\ell^\contTag$, %
respectively linking an internal node to the payload or continuation %
for message $\ell$. %
A type $\T$ yields a tree $\ttree{\T}$: 
the diagram above shows the (infinite) tree of the type\; %
$\mu\ty.\texternal\big\{%
    \tin\pp{\ell_1}{\tbool}.\tinternal\left\{%
      \tout\pq{\ell_3}{\tint}.\ty,%
      \tout\pq{\ell_4}{\treal}.\tend
    \right\}%
    \,,\, \tin\pp{\ell_2}{\tnat}.\tend%
\big\}$.\; %
Notably,
\change{\#D: clarify session trees, coinductive grammar}{%
  the tree of a recursive type $\mu\ty.\T$ %
  coincides with the tree of its unfolding $\T\sub{\mu\ty. \T}{\ty}$.%
} %
We will write $\TT$ to denote a session tree, %
and we will represent it using the \emph{coinductive} syntax:
\[
  \textstyle%
  \TT \;\;\Coloneqq\;\; \tend \sep%
  \texternal_{i\in I}\tin\pp{\ell_i}{\ST_i}.\TT_i \sep%
  \tinternal_{i\in I}\tout\pp{\ell_i}{\ST_i}.\TT_i%
\]
\change{}{%
  The above coinductive definition means that $\TT$ can be an infinite term, generated by infinite applications of the productions: this approach (previously adopted \eg by \citeN{CastagnaWebServices09}) provides a compact way to represent (possibly infinite) trees.
}

\paragraph{SISO trees}%
A SISO tree $\WT$ 
only has singleton choices %
(\ie, one pair of payload+continuation edges) %
in all its branchings and selections. %
We represent $\WT$ with the \emph{coinductive} syntax:
\[
  \WT \;\;\Coloneqq\;\; \tend  \sep \tin\pp{\ell}{\ST}.\WT \sep \tout\pp{\ell}{\ST}.\WT %
\]
We will write $\W$ to denote %
a SISO session type (\ie, having singleton choice in all its branchings and selections), such that $\ttree{\W}$ yields a SISO tree.
We coinductively define the set $\actions\WT$ over a tree $\WT$ as the set of participant names together with actions $?$ (input) or $!$ (output), as:
\[
  \actions\tend=\emptyset \qquad
  \actions{\tin\pp{\ell}{\ST}.\WT'}=\{\pp?\} \cup  \{\actions{\WT'}\}
  \qquad
  \actions{\tout\pp{\ell}{\ST}.\WT'}=\{\pp!\} \cup  \{\actions{\WT'}\}
\]
By extension, we also define  $\actions{\W}=\actions{\ttree{\W}}$.

\paragraph{SISO trees refinement}%
As discussed in \Cref{sec:intro}, %
the asynchronous subtyping should support %
the reordering the input/output actions of a session type: %
\change{\#B: meaning of ``anticipate''}{%
  the intuition is that, under certain conditions,
  the subtype could \emph{anticipate} some input/output actions
  occurring in the supertype,
  by performing them earlier than prescribed.
} %
This is crucial to achieve %
the most flexible %
and \emph{precise} subtyping. %
\change{\#A: ``meat'' of subtyping}{
More in detail, such reorderings can have two forms: %
\begin{enumerate}[label={R\arabic*.},ref={R\arabic*}]
\item\label{item:reordering:in}%
  anticipating a branching from participant $\pp$ %
  before a finite number of branchings which are \emph{not} from $\pp$;%
\item\label{item:reordering:in-out}%
  anticipating a selection toward participant $\pp$ %
  before a finite number of branchings (from any participant), %
  and \emph{also} before other selections %
  which are  \emph{not} toward participant $\pp$,%
\end{enumerate}
}%

\noindent%
To characterise such reorderings of actions %
we define two kinds of \emph{finite} sequences of inputs/outputs:%
\change{\#A: ``meat'' of subtyping}{%
\begin{itemize}
\item%
  $\AContext\pp$, containing only inputs %
  from participants distinct from $\pp$ %
  (we will use it to formalise reordering \ref{item:reordering:in});
\item%
  $\BContext\pp$, containing inputs from any participant %
  and/or outputs to participants distinct from $\pp$
  \change{}{%
    (we will use it to formalise reordering \ref{item:reordering:in-out})%
  }.%
\end{itemize}
}%

\noindent%
Such sequences are inductively defined by the following productions:
\[
  \AContext\pp \Coloneqq  \tin\q{\ell}{\ST} \sep \tin\q{\ell}{\ST}. \AContext\pp \qquad 
  \BContext\pp \Coloneqq  \tin\pr{\ell}{\ST} \sep \tout\pq{\ell}{\ST} \sep  \tin\pr{\ell}{\ST}.\BContext\pp \sep \tout\pq{\ell}{\ST}.\BContext\pp \ \ (\pq\!\neq\!\pp)
\]
We will use the sequences $\AContext\pp$ and $\BContext\pp$
as prefixes for SISO trees;
notice that the base cases require the sequences to have at least one element.

\begin{definition}
  \label{tab:ref}\label{def:ref}\label{tab:sub2}%
  We define \emph{subsorting} $\subs$ as 
  the least reflexive binary relation on sorts (\Cref{def:sorts})
  such that\; $\tnat\subs\tint$.
  \;The \emph{SISO tree refinement relation} $\subttt$ %
  is coinductively defined as:

\noindent%
\hspace{-5mm}\scalebox{0.95}{\rm%
\;\;\begin{minipage}{1.1\linewidth}
\[
  \small
\begin{array}{@{}c@{}}
    \cinfer[\rulename{ref-in}]{\ST' \subs \ST  \quad  \WT \subttt \WT'}{
  \tin\pp{\ell}{\ST}.\WT
  \subttt
   \tin\pp{\ell}{\ST'}.\WT'
 }
 \qquad 
 \cinfer[\rulename{ref-$\AC$}]{\ST' \subs \ST  \quad  \WT \subttt \ACon\pp{\WT'} \quad \actions{\WT}=\actions{\ACon\pp{\WT'}}}{
  \tin\pp{\ell}{\ST}.\WT
  \subttt
   \ACon\pp{{\tin\pp{\ell}{\ST'}.\WT'}}
 }
\\[1mm]
\cinfer[\rulename{ref-out}]{\ST \subs \ST'  \quad  \WT \subttt \WT'}
{
  \tout\pp{\ell}{\ST}.\WT
  \subttt
 \tout\pp{\ell}{\ST'}.\WT'
 }
 \quad
\cinfer[\rulename{ref-$\BC$}]{\ST \subs \ST'  \quad  \WT \subttt \BCon\pp{\WT'} \quad \actions{\WT}=\actions{\BCon\pp{\WT'}}}
{
  \tout\pp{\ell}{\ST}.\WT
  \subttt
\BCon\pp {\tout\pp{\ell}{\ST'}.\WT'}
 }
\quad 
\cinferrule[\rulename{ref-end}]{}
{\tend \subttt \tend}
\end{array}
\]
\end{minipage}
}%
\end{definition}

\change{\#D: clarify subtyping as simulation}{%
  \Cref{def:ref} above formalises a simulation between SISO trees,
  as the largest relation closed backward under the given rules
  \cite{Sangiorgi2011}.%
} %
Rule \rulename{ref-in} relates trees beginning with inputs from a same participant, and having equal message labels; the subtyping between carried sorts must be contravariant, and the continuation trees must be related.
Rule \rulename{ref-out} relates trees beginning with outputs to the same participant, and having equal message labels; the subtyping between carried sorts must be covariant, and the continuation trees must be related.
Rule \rulename{ref-$\AC$}
\change{}{captures reordering \ref{item:reordering:in}}:
it enables anticipating an input from participant $\pp$ before a finite number of inputs from any other participant;
the two payload sorts and the rest of the trees satisfy the same conditions as in rule \rulename{ref-in}, %
while ``$\actions{\WT}=\actions{\ACon\pp{\WT'}}$'' ensures soundness: without such a condition, the tree refinement relation could ``forget'' some inputs.
\change{\#B: why clauses on $\actions{...}$?}{%
  (See \Cref{ex:subt-action-clauses} below.)%
}
Rule \rulename{ref-$\BC$}
\change{}{captures reordering \ref{item:reordering:in-out}}:
it enables anticipating an output to participant $\pp$ before a finite number of inputs from any participant and/or 
outputs to any other participant; the payload types and the rest of the two trees are related similarly to rule \rulename{ref-out}, while ``$\actions{\WT}=\actions{\BCon\pp{\WT'}}$'' ensures that inputs or outputs are not ``forgotten,'' %
similarly to rule~\rulename{ref-$\AC$}.
\change{}{%
  (See \Cref{ex:subt-action-clauses} below.)%
}%

\begin{lemma}
  \label{lem:transitivityW}%
  The refinement relation $\subttt$ over SISO trees is %
  reflexive and transitive.
\end{lemma}

\iftoggle{techreport}{
 \begin{proof}
   See \Cref{sec:app:transitivity}.
 \end{proof}
}{}

\begin{example}[SISO tree refinement]
\label{ex:ref}
Consider the session types:
\change{\#A: \rulename{ref-$\AC$}}{
\[
  \W_1 \;=\; \mu \ty.\tin\pp{\ell}{\S}.\tin\pq{\ell'}{\S'}.\ty
  \qquad\qquad
  \W_2\;=\;\mu \ty.\tin\pq{\ell'}{\S'}.\tin\pp{\ell}{\S}.\ty
\]
}%
\change{\#D: clarify coinductive derivations}{%
  Their trees are related by the following infinite coinductive derivation
  (notice that the coinductive premise and conclusion coincide);%
} %
we \hlight{highlight} the inputs matched by rule \rulename{ref-$\AC$}:
\[
  \small%
  \infer=[\mbox{\rulename{ref-$\AC$} \,with\, $\AContext\pp \!=\! \tin\pq{\ell'}{\S'}$}]{%
      \ttree{\W_1} \;=\;
      \hlight{$\tin\pp{\ell}{\S}$}.\tin\pq{\ell'}{\S'}.\ttree{\W_1}
      \;\subttt\;
      \tin\pq{\ell'}{\S'}.\hlight{$\tin\pp{\ell}{\S}$}.\ttree{\W_2}
      \;=\;\ttree{\W_2}
  }{%
    \infer=[\rulename{ref-in}]{%
      \tin\pq{\ell'}{\S'}.\ttree{\W_1} \;\subttt\; \tin\pq{\ell'}{\S'}.\ttree{\W_2}
    }{%
      \ttree{\W_1} \;\subttt\; \ttree{\W_2}
    }%
  }%
\]
\change{}{%
  An example using rule \rulename{ref-$\BC$} is available later on
  (\Cref{ex:sub-intro}).
}%
\end{example}

\change{\#B: why clauses on $\actions{...}$?}{
\begin{example}
  \label{ex:subt-action-clauses}
  We now illustrate why we need the clauses on $\actions{...}$
  in rules \rulename{ref-$\AC$} and \rulename{ref-$\BC$} (\Cref{def:ref}).
  Consider the following session types:
  \[
    \T \;=\; \mu\ty.\tin\pp{\ell}{\ST}.\ty
    \qquad\qquad%
    \T' \;=\; \tin\pq{\ell'}{\ST'}.\T \;=\; \tin\pq{\ell'}{\ST'}.\mu\ty.\tin\pp{\ell}{\ST}.\ty%
  \]
  Observe that $\T$ is ``forgetting'' to perform the input $\tin\pq{\ell'}{\ST'}$
  occurring in $\T'$.

  If we omit the clause on $\actions{...}$ in rule \rulename{ref-$\AC$},
  we can construct the following infinite coinductive derivation %
  (notice that the coinductive premise and conclusion coincide;
  we \hlight{highlight} the inputs matched by rule \rulename{ref-$\AC$}):
  \[
    \infer=[\mbox{\rulename{ref-$\AC$} \,with\, $\AContext\pp = \tin\pq{\ell'}{\ST'}$}]{
      \ttree{\T} \;=\; \hlight{$\tin\pp{\ell}{\ST}$}.\ttree{\T}
      \;\subttt\;
      \tin\pq{\ell'}{\ST'}.\hlight{$\tin\pp{\ell}{\ST}$}.\ttree{\T}
      \;=\; \ttree{\T'}
    }{
      \ttree{\T} \;\subttt\; \tin\pq{\ell'}{\ST'}.\ttree{\T}
      \;=\; \ttree{\T'}
    }
  \]
  Were we to admit\; $\ttree{\T} \subttt \ttree{\T'}$,
  \;then later on (\Cref{def:subtyping})
  we would consider $\T$ a subtype of $\T'$,
  which would be unsound: in fact, this would allow our type system
  (\Cref{sec:tsy}) to type-check processes
  that ``forget'' to perform inputs and cause orphan message errors
  (rule \rulename{err-ophn} in \Cref{tab:error}).

  A similar example can be constructed for rule \rulename{ref-$\BC$},
  with the following session types:
  \[
    \T \;=\; \mu\ty.\tout\pp{\ell}{\ST}.\ty
    \qquad\qquad%
    \T' \;=\; \tout\pq{\ell'}{\ST'}.\T \;=\; \tout\pq{\ell'}{\ST'}.\mu\ty.\tout\pp{\ell}{\ST}.\ty%
  \]
  Now, $\T$ is ``forgetting'' to perform the output $\tout\pq{\ell'}{\ST'}$.
  If we omit the clause on $\actions{...}$ in rule \rulename{ref-$\BC$},
  we could derive\; $\ttree{\T} \subttt \ttree{\T'}$
  \;through an infinite sequence of instances of \rulename{ref-$\BC$}
  with $\BContext\pp = \tout\pq{\ell'}{\ST'}$;
  were we to admit this, we would later be able to type-check processes
  that ``forget'' to perform outputs
  and cause starvation errors
  (rule \rulename{err-strv} in \Cref{tab:error}).
\end{example}
}%

\begin{remark}[Prefixes \vs $n$-hole contexts]\label{rem:context}
\label{remark:prefix-vs-n-hole-ctx}%
The binary asynchronous subtyping by
\citeN{cdy14,CDSY2017} uses the $n$-hole branching type context\;
$\AC \Coloneqq [\ ]^n  \!\!\sep\!\! \&_{i\in I} \tin{}{\ell_i}{\ST_i}.{\AC_i}$,
\;which complicates the rules and reasoning
(see Fig.2, Fig.3 in \citeN{CDSY2017}). %
Our $\AContext\pp$ and $\BContext\pp$ have a similar purpose, %
but they are simpler %
(just sequences of inputs or outputs), %
and cater for multiple participants.
\end{remark}

\paragraph{SO trees and SI trees}
To formalise our subtyping, we need two more kinds of session trees: %
\emph{single-output (SO) trees}, denoted $\UT$, %
have only singleton choices in their selections; %
dually, \emph{single-input (SI) trees}, denoted $\VT$, %
have only singleton branchings. %
We represent them with a \emph{coinductive} syntax:
\[
  \UT \;\Coloneqq\; \tend  \sep \texternal_{i\in I}
  \tin\pp{\ell_i}{\ST_i}.\UT_i \sep \tout\pp{\ell}{\ST}.\UT %
     \qquad\quad%
  \VT \;\Coloneqq\; \tend  \sep  \tin\pp{\ell}{\ST}.\VT \sep \tinternal_{i\in I}\tout\pp{\ell_i}{\ST_i}.\VT_i %
\]
\noindent%
We will write %
$\UU$ (resp. $\V$) to denote a SO (resp. SI) session type, \ie, with only singleton selections (resp. branchings), such that $\ttree{\UU}$ (resp.~$\ttree{\V}$) yields a SO (resp.~SI) tree.

We decompose session trees into %
their SO/SI subtrees, with the functions %
$\llbracket \cdot \rrbracket_\SO$ / $\llbracket \cdot \rrbracket_\SI$:
\label{page:tree-decomp}%
\[
  \begin{array}{@{}r@{\;\;}c@{\;\;}l@{}}
    \llbracket \tend \rrbracket_\SO &=&  \{ \tend \}
  \end{array}
  \qquad%
  \begin{array}{@{}r@{\;\;}c@{\;\;}l@{}}
    \llbracket \tinternal_{i\in I} \tout\pp{\ell_i}{\ST_i}.\TT_i\rrbracket_\SO  &=&  \{\tout\pp{\ell_i}{\ST_i}.\UT: \UT\in \llbracket \TT_i\rrbracket_\SO , i\in I \}\\
    \llbracket \texternal_{i\in I}\tin\pp{\ell_i}{\ST_i}.\TT_i\rrbracket_\SO &=&  \{\texternal_{i\in I}\tin\pp{\ell_i}{\ST_i}.\UT: \UT\in \llbracket \TT_i\rrbracket_\SO\}
  \end{array}
\]
\[
  \begin{array}{@{}r@{\;\;}c@{\;\;}l@{}}
    \llbracket \tend \rrbracket_\SI  \, \, &=&  \{ \tend \}
  \end{array}
  \qquad \   %
  \begin{array}{@{}r@{\,\;\;}c@{\;\;}l@{}}
    \llbracket \tinternal_{i\in I}\tout\pp{\ell_i}{\ST_i}.\TT_i\rrbracket_\SI &=&  \{\tinternal_{i\in I}\tout\pp{\ell_i}{\ST_i}.\VT: \VT\in \llbracket \TT_i\rrbracket_\SI \}\\
    \llbracket \texternal_{i\in I} \tin\pp{\ell_i}{\ST_i}.\TT_i\rrbracket_\SI  &=&  \{\tin\pp{\ell_i}{\ST_i}.\VT: \VT\in \llbracket \TT_i\rrbracket_\SI , i\in I \}
  \end{array}
\]
\noindent%
Hence, when $\llbracket \cdot \rrbracket_\SO$ is applied to a session tree $\TT$, it gives the set of all SO trees obtained by taking only a single choice from each selection in $\TT$ (\ie, we take a single continuation edge and the corresponding payload edge starting in a selection node).
The function $\llbracket \cdot \rrbracket_\SI$ is dual:
it takes a single selection from each branching in $\TT$. %
Notice that for any SO tree $\UT$, and SI tree $\VT$, both\; %
$\llbracket \UT \rrbracket_\SI$ %
\;and\; %
$\llbracket \VT \rrbracket_\SO$ %
\;yield SISO trees.
We will provide an example shortly (see \Cref{ex:sub-intro} below).

\subsection{Asynchronous Multiparty Session Subtyping}%
\label{subsec:subtyping:formal}
We can now define our asynchronous session subtyping $\subt$: %
it relates two session types %
by decomposing them into their SI, SO, and SISO trees, %
and checking their refinements.

\begin{definition}\label{def:subtyping} 
The \emph{asynchronous subtyping relation $\subt$} %
over session trees is defined as:
\[
  \inferrule%
  {
   \forall \UT\in\llbracket \TT\rrbracket_\SO \quad \forall \VT' \in \llbracket \TT' \rrbracket_\SI \quad \exists \WT \in\llbracket \UT \rrbracket_\SI \quad \exists \WT' \in \llbracket \VT' \rrbracket_\SO \qquad \WT \subttt\WT'
  }
  {\TT \subt \TT'}
\]
The subtyping relation for session types is defined as\;
$\T\subt \T'$ \;iff\; $\ttree{\T}\subt \ttree{\T'}$.
\end{definition}

Definition~\ref{def:subtyping} says that a session tree $\TT$ is subtype of $\TT'$ if, for all 
SO decompositions of $\TT$ and all SI decompositions of $\TT'$, %
there are paths (i.e., SISO decompositions) related by $\subttt$.

\begin{restatable}{lemma}{lemTransitivity}
  \label{lem:transitivity}%
  The asynchronous subtyping relation $\subt$ is reflexive and transitive.
\end{restatable}
\iftoggle{techreport}{
  \begin{proof}
    See \Cref{sec:app:transitivity}.
  \end{proof}
}{}

We now illustrate the relation with two examples: %
we reprise the scenario in the introduction %
(\Cref{ex:sub-intro}), %
and we discuss a case %
from \cite{BravettiCLYZ19,BravettiCLYZ19L} %
(\Cref{ex:CONCUR19-maintext}).
\iftoggle{techreport}{%
  \emph{More examples are available in Appendix~\ref{sec:subtyping-further-examples}}%
}{}%

\newcommand{\introTypes}{%
\[
\small
\T' = \tinternal\pq!\left\{
  \begin{array}{@{}l@{}}
  {\msgLabel{cont}}{(\tint)}.\texternal\pp?\left\{
  		\begin{array}{@{}l@{}}
  		{\msgLabel{success}}{(\tint)}.\tend\\
  		{\msgLabel{error}}{(\tbool)}.\tend
  		\end{array} \right.\\
  {\msgLabel{stop}}{(\tunit)}.\texternal\pp?\left\{
  		\begin{array}{@{}l@{}}
  		{\msgLabel{success}}{(\tint)}.\tend\\
  		{\msgLabel{error}}{(\tbool)}.\tend
  		\end{array} \right.
  \end{array} \right.
\quad
\T = \texternal\pp?\left\{
  \begin{array}{@{}l@{}}
  {\msgLabel{success}}{(\tint).}\tinternal\pq!\left\{
  		\begin{array}{@{}l@{}}
  		{\msgLabel{cont}}{(\tint)}.\tend\\
  		{\msgLabel{stop}}{(\tunit)}.\tend
  		\end{array} \right.\\
  {\msgLabel{error}}{(\tbool)}.\tinternal\pq!\left\{
  		\begin{array}{@{}l@{}}
  		{\msgLabel{cont}}{(\tint)}.\tend\\
  		{\msgLabel{stop}}{(\tunit)}.\tend
  		\end{array} \right.
  \end{array} \right.
\]
}%

\begin{example}\label{ex:sub-intro}
Consider the opening example in Section~\ref{sec:intro}. %
The following types describe the interactions %
of $\PP'_\pr$ and $\PP_\pr$, respectively:
\introTypes%
In order to derive $\T'\subt\T$, we first show the two SO trees such that $\llbracket\ttree{\T'}\rrbracket_\SO = \left\{\UT_1, \UT_2\right\}$:
\[
\UT_1= \tout\pq{\msgLabel{cont}}{\tint}.\texternal\pp?\left\{
  		\begin{array}{@{}l@{}}
  		{\msgLabel{success}}{(\tint)}.\tend\\
  		{\msgLabel{error}}{(\tbool)}.\tend
  		\end{array} \right.\\
\quad
\UT_2= \tout\pq{\msgLabel{stop}}{\tunit}.\texternal\pp?\left\{
  		\begin{array}{@{}l@{}}
  		{\msgLabel{success}}{(\tint)}.\tend\\
  		{\msgLabel{error}}{(\tbool)}.\tend
  		\end{array} \right.
\]
and these are the two SI trees such that $\llbracket\ttree{\T}\rrbracket_\SI = \left\{\VT_1, \VT_2\right\}$:
\[
\VT_1= \tin\pp{\msgLabel{success}}{\tint}.\tinternal\pq!\left\{
  		\begin{array}{@{}l@{}}
  		{\msgLabel{cont}}{(\tint)}.\tend\\
  		{\msgLabel{stop}}{(\tunit)}.\tend
  		\end{array} \right.
\quad
\VT_2= \tin\pp{\msgLabel{error}}{\tbool}.\tinternal\pq!\left\{
  		\begin{array}{@{}l@{}}
  		{\msgLabel{cont}}{(\tint)}.\tend\\
  		{\msgLabel{stop}}{(\tunit)}.\tend
  		\end{array} \right.
\]

Therefore, for all $i,j \in \{1,2\}$, we can find 
$\WT'\in\llbracket\UT_i\rrbracket_\SI$ and 
$\WT\in\llbracket\VT_j\rrbracket_\SO$ such that 
$\WT'\subttt\WT$ can be derived using \rulename{ref-${\BC}$}. 
For instance, since we have:
\begin{align*}
   \llbracket\UT_1\rrbracket_\SI &=  
      \big\{
      \tout\pq{\msgLabel{cont}}{\tint}.\tin\pp{\msgLabel{success}}{\tint}.\tend\,,\;
  	  \tout\pq{\msgLabel{cont}}{\tint}.\tin\pp{\msgLabel{error}}{\tbool}.\tend
      \big\} \\
    \llbracket\VT_1\rrbracket_\SO &= 
      \big\{
      \tin\pp{\msgLabel{success}}{\tint}.\tout\pq{\msgLabel{cont}}{\tint}.\tend\,,\;
      \tin\pp{\msgLabel{success}}{\tint}.\tout\pq{\msgLabel{stop}}{\tunit}.\tend
      \big\}
\end{align*}
we can show that\;
$
\tout\pq{\msgLabel{cont}}{\tint}.\tin\pp{\msgLabel{success}}{\tint}.\tend
\subttt
\tin\pp{\msgLabel{success}}{\tint}.\tout\pq{\msgLabel{cont}}{\tint}.\tend
$, 
\;by the following coinductive derivation:
(we \hlight{highlight} the outputs matched by rule \rulename{ref-$\BC$})
\[
  \infer=[\mbox{\rulename{ref-$\BC$}, {\scriptsize $\BContext\pq = \tin\pp{\msgLabel{success}}{\tint}$}}]{%
  \hlight{$\tout\pq{\msgLabel{cont}}{\tint}$}.\tin\pp{\msgLabel{success}}{\tint}.\tend
\subttt
\tin\pp{\msgLabel{success}}{\tint}.\hlight{$\tout\pq{\msgLabel{cont}}{\tint}$}.\tend}
  {
  \infer=[\mbox{\rulename{ref-in}}]{
    \tin\pp{\msgLabel{success}}{\tint}.\tend
\subttt
\tin\pp{\msgLabel{success}}{\tint}.\tend}
  {
  \infer=[\mbox{\rulename{ref-end}}]{    
  \tend\subttt\tend}
  {}
  }
  }
\]
Hence, by \Cref{def:subtyping}, we conclude that $\T'\subt\T$ holds.
\end{example} 

\begin{example}[Example 3.21 by \citeN{BravettiCLYZ19L}]
  \label{ex:CONCUR19-maintext}%
  This example demonstrates how our subtyping, and its underlying
  session decomposition approach, apply to a complex case, where the
  subtyping proof requires infinite, non-cyclic derivations.
  Consider the following session types:
\[
    \begin{array}{ll}
      \begin{array}{l}
        \T = \mu \ty_1. \texternal\pp?\left\{
        \begin{array}{@{}l@{}}
          {\msgLabel{\ell_1}}{(\S_1).}\tout\pp{\ell_3}{\S_3}. \tout\pp{\ell_3}{\S_3}. \tout\pp{\ell_3}{\S_3}. \ty_1\\
          {\msgLabel{\ell_2}}{(\S_2)}.\mu \ty_2. \tout\pp{\ell_3}{\S_3} .\ty_2
        \end{array} \right.
      \end{array}
      &
        \begin{array}{l}
          \T' =  \mu \ty_1. \texternal\pp?\left\{
          \begin{array}{@{}l@{}}
            {\msgLabel{\ell_1}}{(\S_1).}\tout\pp{\ell_3}{\S_3}. \ty_1\\
            {\msgLabel{\ell_2}}{(\S_2)}.\mu \ty_2. \tout\pp{\ell_3}{\S_3} .\ty_2
          \end{array} \right.
        \end{array}
    \end{array}
    \]
We now prove that $\T \subt \T'$.
Notably, if we omit the participant $\pp$, we obtain binary session types
that are related under the \emph{binary} asynchronous subtyping relation
by \citeN{CDSY2017} %
--- but due to its complexity, the relation cannot be proved %
using the binary asynchronous subtyping algorithm %
of \citeN{BravettiCLYZ19,BravettiCLYZ19L}.%

Letting  $\TT = \ttree{\T}$ and $\TT' = \ttree{\T'}$, 
by \Cref{def:subtyping} we need to show that:
\begin{equation}
  \label{eq:CONCUR19-subtyping-goal}%
   \forall \UT\in\llbracket \TT\rrbracket_\SO \quad \forall \VT' \in \llbracket \TT' \rrbracket_\SI \quad \exists \WT \in\llbracket \UT \rrbracket_\SI \quad \exists \WT' \in \llbracket \VT' \rrbracket_\SO \qquad \WT \subttt\WT'
\end{equation}
\noindent%
 Observe that both $\TT$ and $\TT'$ are SO trees. Therefore, we have
 $\llbracket \TT \rrbracket_\SO = \{ \TT \}$;
 moreover, all $\VT' \in \llbracket \TT' \rrbracket_\SI$ are SISO trees,
 which means that, in \eqref{eq:CONCUR19-subtyping-goal},
 for all such $\VT'$ we have
 $\llbracket \VT' \rrbracket_\SO = \{\VT'\}$. %
 These singleton sets allow for simplifying the quantifications
 in \eqref{eq:CONCUR19-subtyping-goal}, as follows:
\begin{equation}
  \label{eq:CONCUR19-subtyping-goal-simpl}%
   \forall \WT' \in \llbracket \TT' \rrbracket_\SI \quad \exists \WT \in\llbracket \TT \rrbracket_\SI \qquad \WT \subttt \WT'
 \end{equation}
 \noindent%
 and to prove $\T \subt \T'$,
 it is enough to prove \eqref{eq:CONCUR19-subtyping-goal-simpl}.
 Before proceeding, we introduce
 the following abbreviations (where $\WT_i$ are SISO trees,
 and $\pi_i$ are sequences of inputs and outputs, used to prefix a SISO tree):
   \[
 \begin{array}{c}
  \WT_1 = \ttree{\mu \ty. \tin\pp{\ell_1}{\ST_1}. \tout\pp{\ell_3}{\ST_3}.\ty }
   \qquad
   \WT_2 = \tin\pp{\ell_2}{\ST_2}.\ttree{ \mu \ty. \tout\pp{\ell_3}{\ST_3}.\ty} \\
  \WT_3 = \ttree{\mu \ty. \tin\pp{\ell_1}{\ST_1}. \tout\pp{\ell_3}{\ST_3}. \tout\pp{\ell_3}{\ST_3}. \tout\pp{\ell_3}{\ST_3}.\ty} %
   \\[-5mm]
   \pi_1 = \tin\pp{\ell_1}{\ST_1}. \tout\pp{\ell_3}{\ST_3}
   \qquad
   \pi_3 = \tin\pp{\ell_1}{\ST_1}. \tout\pp{\ell_3}{\ST_3}. \tout\pp{\ell_3}{\ST_3}. \tout\pp{\ell_3}{\ST_3}
   \qquad
   \pi_i^n = \overbrace{\pi_i. \dots . \pi_i}^\text{\text{$n$ times}}
 \end{array} 
\] %
\noindent%
 Then, we have:
 \[
  \begin{array}{lcl}
    \llbracket \TT \rrbracket_\SI & = & \big\{
   \WT_3,\; \WT_2,\;   \pi_3. \WT_2,\; \pi_3^2.\WT_2,\; \ldots,\; \pi_3^n.\WT_2,\; \dots
   \big\} \\
 \llbracket \TT' \rrbracket_\SI & = & \big\{
 \WT_1,\; \WT_2,\;   \pi_1. \WT_2,\; \pi_1^2.\WT_2,\; \ldots,\; \pi_1^n.\WT_2,\; \dots
 \big\}\\
 \end{array}
 \]
 \noindent%
 and we can prove \eqref{eq:CONCUR19-subtyping-goal-simpl} by showing that:
 
 \begin{enumerate}[label=\emph{(\roman*)}]
 \item\label{item:CONCUR19-1} $\WT_3 \subttt \WT_1$ \;and
 \item\label{item:CONCUR19-2} for all $n \ge 1$,\; $\pi_3^n. \WT_2 \subttt \pi_1^n. \WT_2$
 \end{enumerate}
 Here we develop item~\ref{item:CONCUR19-1}.
 \iftoggle{techreport}{
   \emph{(The complete proof of item~\ref{item:CONCUR19-2} and more details are given in Appendix~\ref{sec:subtyping-further-examples}, \Cref{ex:CONCUR19})}
 }{}
 The relation $\WT_3 \subttt \WT_1$ is proved by the following coinductive derivation:
\[
  \infer=[\mbox{\rulename{ref-in}}]{%
   \WT_3 \subttt \WT_1
  }{%
  \infer=[\mbox{\rulename{ref-out}}]{%
   (\tout\pp{\ell_3}{\ST_3})^3.\WT_3 \subttt\tout\pp{\ell_3}{\ST_3}.\WT_1
  }{%
   \infer=[\mbox{\rulename{ref-$\BC$}, {\scriptsize $\BContext\pp = \tin\pp{\ell_1}{\S_1}$}}]{%
  (\tout\pp{\ell_3}{\ST_3})^2. \WT_3 \subttt \WT_1
  }{%
    \infer=[\mbox{\rulename{ref-$\BC$}, {\scriptsize $\BContext\pp = \tin\pp{\ell_1}{\S_1}.\tin\pp{\ell_1}{\S_1}$}}]{%
    \tout\pp{\ell_3}{\ST_3}.\WT_3 \subttt \tin\pp{\ell_1}{\ST_1}.\WT_1
    }{%
       \WT_3 \subttt \tin\pp{\ell_1}{\ST_1}.\tin\pp{\ell_1}{\ST_1}.\WT_1
    }%
  }%
 }%
}%
\]
where the topmost relation
holds by the following
\change{\#D: clarify coinductive derivations}{%
  infinite coinductive derivation,
  for all $n \!\geq\! 1$%
}:
\[
  \infer=[\mbox{\rulename{ref-in}}]{%
   \WT_3 \subttt  (\tin\pp{\ell_1}{\ST_1}.\tin\pp{\ell_1}{\ST_1})^{n}.\WT_1
  }{%
    \infer=[\mbox{\rulename{ref-$\BC$}, {\scriptsize $\BContext\pp = (\tin\pp{\ell_1}{\S_1}.\tin\pp{\ell_1}{\S_1})^{n}$}}]{%
   (\tout\pp{\ell_3}{\ST_3})^3.\WT_3 \subttt  (\tin\pp{\ell_1}{\ST_1}.\tin\pp{\ell_1}{\ST_1})^{n}.\tout\pp{\ell_3}{\ST_3}.\WT_1
  }{%
    \infer=[\mbox{\rulename{ref-$\BC$}, {\scriptsize $\BContext\pp = (\tin\pp{\ell_1}{\S_1}.\tin\pp{\ell_1}{\S_1})^{n}.\tin\pp{\ell_1}{\S_1}$}}]{%
  (\tout\pp{\ell_3}{\ST_3})^2.\WT_3 \subttt  (\tin\pp{\ell_1}{\ST_1}.\tin\pp{\ell_1}{\ST_1})^{n}.\WT_1
  }{%
\infer=[\mbox{\rulename{ref-$\BC$}, {\scriptsize $\BContext\pp = (\tin\pp{\ell_1}{\S_1}.\tin\pp{\ell_1}{\S_1})^{n+1}$}}]{%
    \tout\pp{\ell_3}{\ST_3}.\WT_3 \subttt  (\tin\pp{\ell_1}{\ST_1}.\tin\pp{\ell_1}{\ST_1})^{n}.\tin\pp{\ell_1}{\ST_1}.\WT_1
    }{%
       \WT_3 \subttt (\tin\pp{\ell_1}{\ST_1}.\tin\pp{\ell_1}{\ST_1})^{n+1}.\WT_1
    }%
  }%
 }%
}%
\]
\end{example} 
 
\section{Typing System and Type Safety}\label{sec:tsy}
\label{subsec:typesystem}
Our multiparty session typing system %
blends the one by \citeN{GhilezanJPSY19} %
with the one in \cite[Section 7]{ScalasY19}: %
like the latter, we type multiparty sessions without need for global types, %
thus simplifying our formalism %
and generalising our results. %
The key differences are our asynchronous subtyping %
(\Cref{def:subtyping}) and our choice of %
\emph{typing environment liveness} %
(\Cref{def:env-liveness}): %
their interplay yields our preciseness results.

\subsection{Typing System}

Before proceeding, we need to formalise \textbf{queue types} for message queues, %
extending Def.~\ref{def:types}:

\smallskip%
\centerline{\(%
 \tqueue \;\;\Coloneqq\;\;%
 \temptyqueue \sep \tout\pp\ell\ST \sep \tqueue \cdot \tqueue
\)}%
\smallskip%

\noindent%
Type $\temptyqueue$ denotes an empty queue; %
$\tout\pp\ell\ST$ denotes a queued message with recipient $\pp$, label $\ell$, and payload of sort $\ST$; they are concatenated as $\tqueue \cdot \tqueue'$.

\begin{definition}[Typing system]
\label{def:type-system}%
  The type system uses 4 judgments:%

  \smallskip%
  \begin{minipage}{0.5\linewidth}%
    \begin{itemize}
    \item%
      for expressions:\; $\Theta \vdash \e:\ST$
    \item%
      for queues:\; $\vdash \h:\tqueue$
    \end{itemize}
  \end{minipage}
  \begin{minipage}{0.5\linewidth}%
    \begin{itemize}
    \item%
      for processes:\; $\Theta  \vdash \PP:\T$
    \item%
      for sessions:\; $\Gamma \vdash \N$
    \end{itemize}
  \end{minipage}

\noindent%
where the typing environments $\Gamma$ and $\Theta$ are defined as:

  \smallskip%
  \centerline{\(%
    \Gamma \;\;\Coloneqq\;\; \emptyset \sep \Gamma, \pp:(\tqueue, \T)%
    \qquad\qquad%
    \Theta \;\;\Coloneqq\;\; \emptyset \sep \Theta, X:\T \sep \Theta, \x:\ST
  \)}%

\noindent%
  The typing system is inductively defined by the rules %
  in \Cref{figure:typesystem}. 
\end{definition}
The judgment for expressions is standard:\; %
$\Theta \vdash \e:\ST$ \;means that, 
given the variables and sorts in environment $\Theta$,
expression $\e$ is of the sort $\ST$.
The judgment for queues means that queue $\h$ has queue type $\tqueue$. 
The judgment for processes states that, 
given the types of the variables in $\Theta$, 
process $\PP$ behaves as prescribed by $\T$. 
The judgment for sessions states that multiple participants and queues %
behave as prescribed by $\Gamma$, which maps each participant $\pp$ %
to the pairing of a queue type %
(for $\pp$'s message queue) and a session type (for $\pp$'s process).
If $\Theta=\emptyset$ we write $ \vdash \e:\ST$ and  $\vdash \PP:\T.$

\begin{table}[t!]
 \[
 \begin{array}{c}
   \inferrule[]{ }{\Theta \vdash \valn:\tnat}
    \qquad 
  \changeNoMargin{\inferrule[]{ }{\Theta \vdash (-\valn):\tint}
     \qquad 
  \inferrule[]{ }{\Theta \vdash \zero:\tint}}
   \quad
   \inferrule[]{ }{\Theta\vdash \true:\tbool} \quad \inferrule[]{ }{\Theta\vdash \false:\tbool}
  \\
   \changeNoMargin{\inferrule[]{ }{\Theta \vdash ():\tunit}}
    \quad
    \inferrule[]{ }{\Theta, x:\ST \vdash x:\ST}
  \quad
   \inferrule[]
   {\Theta \vdash \e:\tnat}{\Theta\vdash \fsucc\e:\tnat}
   \quad
   \inferrule{\Theta\vdash \e: \tint}{\Theta\vdash \fsqrt\e:\tint}
    \\[2mm]
   \inferrule{\Theta \vdash \e:\tbool}{\Theta\vdash \neg\e:\tbool}{}
   \qquad
 \changeNoMargin{\inferrule{\Theta \vdash \e:\tint}{\Theta \vdash \e > \zero : \tbool}}
    \qquad
 \changeNoMargin{\inferrule{\Theta \vdash \e:\tunit}{\Theta \vdash \e \approx () : \tbool}}
     \qquad
    \inferrule{\Theta \vdash \e:\ST \quad \ST\subs\ST'}{\Theta \vdash \e:\ST'}
   \\[2mm]
   \inferrule{ }{\vdash \emptyqueue: \temptyqueue}
       \qquad
   \inferrule{\ \vdash\val:\ST}{\vdash  (\pq,\ell(\val)):\tout\pq\ell\ST}
       \qquad
   \inferrule{\vdash \h_1:\tqueue_1 \quad \vdash \h_2:\tqueue_2}{\vdash \h_1 \cdot \h_2: \tqueue_1\cdot\tqueue_2}
   \end{array}
 \]
 \\[2mm]
\[
\begin{array}{@{}c@{}}
  \inferruleR[\rulename{t-$\inact$}]{$\;$}{\Theta \vdash \inact: \tend}
  \quad\;
  \inferruleR[\rulename{t-var}]{$\;$}{\Theta, X:\T  \vdash X:\T}
  \quad\;
  \inferruleR[\rulename{t-rec}]{\Theta, X:\T  \vdash \PP:\T}{\Theta  \vdash  \mu X.\PP: \T}
  \quad\;%
\inferruleR[\rulename{t-out}]{\Theta \vdash \e:\ST \quad \Theta \vdash \PP:\T}{\Theta  \vdash \procout\pq\ell\e\PP: \tout\pq\ell\ST.\T}
  \\[3mm]%
\inferruleR[\rulename{t-ext}]{\forall i\in I\;\;\; \Theta, x_i:\ST_i \vdash \PP_i:\T_i }{\Theta \vdash  \sum_{i\in I}\procin{\pq}{\ell_i(\x_i)}{\PP_i}: \texternal_{i\in I}\tin\pq{\ell_i}{\ST_i}.{\T_i}}
  \qquad
\inferruleR[\rulename{t-cond}]{\Theta \vdash \e:\tbool
\quad \Theta  \vdash \PP_i:\T \ \text{\tiny $(i=1,2)$}
}{\Theta \vdash \cond\e{\PP_1}{\PP_2}:\T}
\\[3mm]
\inferruleR[\rulename{t-sub}]{\Theta \vdash \PP:\T \quad \T\subt \T' }{\Theta \vdash \PP:\T'}
\qquad
\inferruleR[\rulename{t-sess}]{\Gamma=\{\pp_i{:}\,(\tqueue_i, \T_i) \;|\; i\in I\} \quad %
  \forall i\in I\;\; \vdash \PP_i:\T_i \;\;  \vdash \h_i:\tqueue_i }{\Gamma \vdash \prod_{i\in I} (\pa{\pp_i}{\PP_i} \pc \pa{\pp_i}{\h_i})}
\end{array}  
\]

\caption{\label{figure:typesystem}%
  Typing rules for expressions and queues (top 3 rows),
  and for processes and sessions (bottom 4 rows).%
}%
\end{table}

We now comment the rules for processes and sessions %
(other rules are self-explanatory).
Rule \rulename{t-$\inact$} types a terminated process. 
Rule \rulename{t-var} types a process variable %
with the assumption in the environment. %
By \rulename{t-rec}, %
a recursive process is typed with $\T$ %
if the process variable $X$ and the body $\PP$ have the same type $\T$. 
By \rulename{t-out}, %
an output process is typed with a singleton selection type, %
if the message being sent is of the correct sort, %
and the process continuation has the continuation type.
By \rulename{t-ext}, a process external choice is typed as a branching type
with matching participant $\pp$ and labels $\ell_i$ (for all $i \!\in\! I$); %
in the rule premise, each continuation process $\PP_i$ must be typed with the corresponding continuation type $\T_i$, assuming the bound variable $\x_i$ is of sort $\ST_i$ (for all $i \!\in\! I$). %
By \rulename{t-cond}, a conditional has type $\T$ %
if its expression has sort $\tbool$, %
and its ``$\mathsf{then}$'' and ``$\mathsf{else}$'' branches have type $\T$.
Rule \rulename{t-sub} is the \emph{subsumption rule}: %
it states that a process of type $\T$ is also typed  by any supertype of $\T$ %
(and since $\subt$ relates types up-to unfolding, %
this rule makes our type system \emph{equi-recursive} %
\cite{PierceBC:typsysfpl}). %
By \rulename{t-sess}, %
a session $\N$ is typed by environment $\Gamma$ %
if all the participants in $\N$ have processes and queues typed %
by the session and queue types pairs in $\Gamma$.

\begin{example}
  \label{ex:typing-intro}%
  The processes $\PP_\pr$ and its optimised version $\PP'_\pr$
  in Section~\ref{sec:intro} %
  are typable using the rules in Def.~\ref{def:type-system},
  and the types $\T, \T'$ in Example~\ref{ex:sub-intro}. %
  In particular, $\PP_\pr$ has type $\T$,
  while $\PP'_\pr$ has type $\T'$.
  Moreover, since $\T' \subt \T$, %
  by rule \rulename{t-sub} %
  the optimised process $\PP'_\pr$ has \emph{also} type $\T$:
  hence, our type system allows using $\PP'_\pr$
  whenever a process of type $\T$ %
  (such as $\PP_\pr$) is expected.
\end{example}

\subsection{Typing Environment Reductions and Liveness}%

To formulate the soundness result for our type system,
we define typing environment reductions. 
The reductions rely on a standard structural congruence relation $\equiv$ %
over queue types,
allowing to reorder messages with different recipients.
Formally, $\equiv$ is the least congruence satisfying:

\smallskip%
\centerline{
$
\begin{array}{c}
\tqueue \cdot \temptyqueue \;\equiv\; \temptyqueue \cdot \tqueue \;\equiv\; \tqueue %
  \qquad\qquad
   \tqueue_1 \cdot (\tqueue_2 \cdot \tqueue_3) \;\equiv\; (\tqueue_1 \cdot \tqueue_2) \cdot \tqueue_3\\[1mm] 
{\tqueue \cdot \tout{\pp_1}{\ell_1}{\ST_1} \cdot \tout{\pp_2}{\ell_2}{\ST_2} \cdot \tqueue' \;\;\equiv\;\; \tqueue\cdot \tout{\pp_2}{\ell_2}{\ST_2} \cdot \tout{\pp_1}{\ell_1}{\ST_1} \cdot \tqueue'}
\quad\text{(if $\pp_1 \neq \pp_2$)}
\end{array}
$}%
\smallskip

\noindent%
For pairs of queue/session types, 
we define 
structural congruence as\; $(\tqueue_1,\T_1) \equiv (\tqueue_2, \T_2)$ 
\;iff\; $\tqueue_1 \equiv \tqueue_2$ and $\T_1 \equiv \T_2$, \; and
subtyping
as\; $(\tqueue_1,\T_1)\subt (\tqueue_2, \T_2)$ \;iff\; 
$\tqueue_1\equiv\tqueue_2$ and $\T_1 \subt \T_2$. %
We extend subtyping and congruence to typing environments,
by requiring all corresponding entries to be related,
and allowing additional unrelated entries of type $(\temptyqueue,\tend)$.
More formally, we write\; %
\label{page:gamma-subt-gammai}%
  $\Gamma\equiv\Gamma'$ (resp. $\Gamma\subt\Gamma'$) \;iff\; $\forall\pp\!\in\! \dom{\Gamma} \cap \dom{\Gamma'}{:}$ $\Gamma(\pp)\!\equiv\Gamma'(\pp)$ (resp. $\Gamma(\pp)\subt\Gamma'(\pp)$) and $\forall\pp\!\in\! \dom{\Gamma} \setminus \dom{\Gamma'}, \pq\!\in\!\dom{\Gamma'} \setminus \dom{\Gamma}{:}$ $\Gamma(\pp) \!\equiv\! (\temptyqueue,\tend) \!\equiv\! \Gamma'(\pq)$. %

\begin{definition}[Typing environment reduction]
\label{def:typing-env-reductions}%
  The reduction\,  $\redLabel{\;\alpha\;}$ %
  \,of asynchronous session typing environments %
  is inductively defined as follows,
  \change{\#D: define $\alpha$}{%
    with $\alpha$ being either\,
    $\recvLabel{\pp}{\pq}{\ell}$
    \,or $\sendLabel{\pp}{\pq}{\ell}$
    \,(for some $\pp,\pq,\ell$):%
  }
\smallskip%
\noindent%
\centerline{%
\small\rm%
\(%
  \begin{array}{@{}l@{\;}l@{\qquad}r@{}}
\rulename{e-rcv} &
\hspace{-3mm}{\pp{:}\,(\tout\pq{\ell_k}{\ST_k'}{\cdot} \tqueue, \T_{\pp}),\pq{:}\,(\tqueue_{\pq}, \texternal_{i\in I}\tin\pp{\ell_i}{\ST_i}.{\T_i}), \Gamma \;\redRecv{\pq}{\pp}{\ell_{k}}\; \pp{:}\,(\tqueue, \T_\pp), \pq{:}\,(\tqueue_{\pq}, \T_k),\Gamma} %
& \hspace{0mm}(k\!\in\! I, \ST_k' \!\subs\! \ST_k)\\
\rulename{e-send}
   & {\pp:(\tqueue, \tinternal_{i\in I}\tout\pq{\ell_i}{\ST_i}.{\T_i}),\Gamma \;\;\redSend{\pp}{\pq}{\ell_{k}}\;\; \pp:(\tqueue{\cdot} \tout\pq{\ell_k}{\ST_k}{\cdot}\temptyqueue, \T_k),\Gamma} &
   {(k\in I) }\\[1mm]
\rulename{e-struct}&
\quad \Gamma \equiv \Gamma_1 \redLabel{\;\alpha\;} \Gamma_1' \equiv \Gamma'%
\;\implies\;%
\Gamma \redLabel{\;\alpha\;} \Gamma'%
   \end{array}
\)%
}%
\noindent%
We often write\; $\Gamma \red \Gamma'$ %
  instead of\; $\Gamma \redLabel{\;\alpha\;} \Gamma'$ %
  \;when $\alpha$ is not important.
We write $\reds$ for the reflexive and transitive closure of $\red$.
\end{definition}

Rule \rulename{e-rcv} says that an environment can take a reduction step if %
participant $\pp$ has a message toward $\pq$ %
with label $\ell_k$ and payload sort $\ST_k'$ at the head of its queue, %
while $\pq$'s type is a branching from $\pp$ 
including label $\ell_k$ and a corresponding sort $\ST_k$ 
supertype of $\ST_k'$; %
the environment evolves with a reduction labelled  
${\pq}{:}{\pp}{?}{\ell_{k}}$, by consuming $\pp$'s message %
and activating the continuation $\T_k$ in $\pq$'s type. %
In rule \rulename{e-send} %
the environment evolves by letting participant $\pp$ %
(having a selection type) %
send a message %
toward $\pq$; the reduction is labelled ${\pp}{:}{\pq}{!}{\ell_{k}}$ %
(with $\ell_k$ being a selection label), %
and places the message at the end of $\pp$'s queue. %
Rule \rulename{e-struct} closes the reduction %
under structural congruence.

Similarly to~\citeN{ScalasY19}, %
we define a behavioral property of typing environments %
(and their reductions) called \emph{liveness}:\footnote{%
  Notably, our definition of liveness is stronger %
  than the ``liveness'' in~\cite[Fig.~5]{ScalasY19}, %
  and is closer to ``liveness\textsuperscript{+}'' therein:
  we adopt it because a weaker ``liveness'' %
  would not allow to achieve %
  Theorem~\ref{thm:subt-precise-liveness} later on.
} %
we will use it as a precondition for typing, %
to ensure that typed processes %
cannot reduce to any error in Table~\ref{tab:error}.%

\begin{definition}[Live typing environment]
  \label{def:env-path-fairness}%
  \label{def:env-liveness}%
  A \emph{typing environment path} %
  is %
  possibly infinite sequence of typing environments %
  $(\Gamma_i)_{i \in I}$, %
  where $I = \{0,1,2,\ldots\}$ is a set of consecutive natural numbers, %
  and, $\forall i \in I$, $\Gamma_i \red \Gamma_{i+1}$. %
  We say that a path $(\Gamma_i)_{i \in I}$ is \emph{fair} iff, %
  $\forall i \in I$:
  \begin{enumerate}[leftmargin=9mm,label={\sf\textbf{(F{\arabic*})}},ref={\sf\textbf{F{\arabic*}}}]
  \item\label{item:fairness:send}%
    whenever $\Gamma_i \redSend{\pp}{\pq}{\ell} \Gamma'$, %
    then $\exists k, \ell'$ %
    such that $I \ni k{+}1 > i$, %
    and $\Gamma_k \redSend{\pp}{\pq}{\ell'} \Gamma_{k+1}$
  \item\label{item:fairness:recv}%
    whenever $\Gamma_i \redRecv{\pp}{\pq}{\ell} \Gamma'$, %
    then $\exists k$ %
    such that $I \ni k{+}1 > i$, %
    and $\Gamma_k \redRecv{\pp}{\pq}{\ell} \Gamma_{k+1}$
  \end{enumerate}

  \noindent%
  We say that a path $(\Gamma_i)_{i \in I}$ is \emph{live} iff, %
  $\forall i \in I$:
  \begin{enumerate}[leftmargin=9mm,label={\sf\textbf{(L{\arabic*})}},ref={\sf\textbf{L{\arabic*}}}]
  \item\label{item:liveness:send}%
    if $\Gamma_i(\pp) \equiv \left(\tout\pq{\ell}\ST \cdot \tqueue \,,\, \T\right)$, %
    then $\exists k$:\; $I \ni k{+}1 > i$ %
    \;and\; %
    $\Gamma_k \redRecv{\pq}{\pp}{\ell} \Gamma_{k+1}$
  \item\label{item:liveness:recv}%
    if $\Gamma_i(\pp) \equiv \left(\tqueue_\pp \,,\, \texternal_{j \in J}{\tin\pq{\ell_j}{\ST_j}.{\T_j}}\right)$, %
    then $\exists k,\ell'$:\; $I \ni k{+}1 > i$ %
    \;and\; %
    $\Gamma_k \redRecv{\pp}{\pq}{\ell'} \Gamma_{k+1}$
  \end{enumerate}
  
  \noindent%
  We say that a typing environment $\Gamma$ is \emph{live} %
  iff %
  all fair paths beginning with $\Gamma$ %
  are live.
\end{definition}

By Def.~\ref{def:env-path-fairness}, %
a path is a (possibly infinite) sequence of reductions %
of a typing environment. %
Intuitively, a fair path represents a ``fair scheduling:'' %
along its reductions, every pending internal choice %
eventually enqueues a message (\ref{item:fairness:send}), %
and every pending message reception %
is eventually performed (\ref{item:fairness:recv}). %
A path is live %
if, along its reductions, %
every queued message is eventually consumed %
(\ref{item:liveness:send}), %
and every waiting external choice eventually consumes a queued message %
(\ref{item:liveness:recv}). %
A typing environment is live if %
it always yields a live path when it is fairly scheduled. %

\begin{example}[Fairness and liveness]
  Consider the typing environment:
  \[
  \Gamma \;=\; \pp{:}(\temptyqueue,\mu\ty.\tout{\pq}{\ell}{\ST}.\ty),\,
  \pq{:}(\temptyqueue,\mu\ty.\tin{\pp}{\ell}{\ST}.\ty)
  \]
  The typing environment $\Gamma$ has an infinite path where $\pp$ keeps
  enqueuing outputs, while $\pq$ never fires a reduction to receive them:
  such a path is unfair, because $\pq$'s message reception is always enabled,
  but never performed. %
  Instead, in all fair paths of $\Gamma$, the messages enqueued by $\pp$
  are eventually consumed by $\pq$, and the inputs of $\pq$
  are eventually fired: hence, $\Gamma$ is live.
  Now consider:
  \[
  \textstyle
  \Gamma' \;=\; \pp{:}\left(\temptyqueue,\mu\ty.\tinternal\!\big\{
    \tout{\pq}{\ell}{\ST}.\ty,\,
    \tout{\pq}{\ell'}{\ST}.\tout{\pr}{\ell'}{\ST}.\ty
    \big\}\right)
    ,\,
    \pq{:}\left(\temptyqueue,\mu\ty.\texternal\!\big\{
      \tin{\pp}{\ell}{\ST}.\ty\,,
      \tin{\pp}{\ell'}{\ST}.\ty
    \big\}\right)
    ,\,
    \pr{:}\left(\temptyqueue,\mu\ty.\tin{\pp}{\ell'}{\ST}.\ty\right)
  \]
  The environment $\Gamma'$ has fair and live paths
  where $\pp$ chooses to send $\ell'$ to $\pq$ and then to $\pr$.
  However, there is also a fair path where $\pp$
  always chooses to send $\ell$ to $\pq$:
  in this case, $\pq$ always eventually receives a message,
  but $\pr$ forever waits for an input that will never arrive.
  Therefore, such a path is fair but not live, hence $\Gamma'$ is not live.
  \hfill$\blacksquare$%
\end{example}

Liveness is preserved by environment reductions and
subtyping, as formalised in \Cref{lem:move-preserves-liveness}
and \Cref{lem:subtyping-preserves-liveness}:
these properties are crucial for
proving subject reduction later on.
\iftoggle{techreport}{%
  \emph{(proofs in \Cref{app:proofs-subtyping-liveness})}%
}{}

\begin{restatable}{proposition}{lemMovePreservesLiveness}
  \label{lem:move-preserves-liveness}%
  If\, $\Gamma$ is live and $\Gamma \red \Gamma'$, %
  then $\Gamma'$ is live.
\end{restatable}

\vspace{-3mm}%
\begin{restatable}{lemma}{lemSubPreservesLiveness}
  \label{lem:subtyping-preserves-liveness}%
  If\, $\Gamma$ is live and $\Gamma' \subt \Gamma$, %
  then $\Gamma'$ is live.
\end{restatable}

We can now state our subject reduction result
(\Cref{thm:SR} below): %
if session $\N$ is typed by a live $\Gamma$, %
and $\N \red \N'$, %
then $\N$ might anticipate some inputs/outputs %
prescribed by $\Gamma$, %
as allowed by subtyping $\subt$. %
Hence, $\N$ reduces by following some $\Gamma'' \!\subt\! \Gamma$, %
which evolves to $\Gamma'$ that types $\N'$.
\iftoggle{techreport}{%
  \emph{(Proofs: \Cref{subsec:app:typesystem}.)}
}{}%

\begin{restatable}[Subject Reduction]{theorem}{theoremSR}
  \label{thm:SR}
  Assume\; $\Gamma \vdash \N$ %
  \;with $\Gamma$ live. %
  If\; $\N \red \N'$, %
  \;then there are live type environments $\Gamma', \Gamma''$ %
  such that\; $\Gamma'' \subt \Gamma$,\; %
  $\Gamma'' \red^* \Gamma'$ \;and\; $\Gamma'\vdash\N'$.
\end{restatable}

\vspace{-3mm}%
\begin{restatable}[Type Safety and Progress]{corollary}{corollaryProgress}
  \label{thm:error-freedom}%
  Let\; $\Gamma \vdash \N$ %
  \;with $\Gamma$ live. %
  Then,\; %
  $\N \red^* \N'$ \;implies\; $\N' \neq \error$; %
  \;also, %
  either\, $\N'\equiv\pa\pp\inact \pc \pa\pp\emptyqueue$, \,or\, %
  $\exists \N''$ such that\, $\N' \red \N'' \neq \error$.
\end{restatable}

Notably, %
since our errors (Table~\ref{tab:error}) include orphan messages, %
deadlocks, and starvation, %
Corollary~\ref{thm:error-freedom} implies %
\emph{session liveness}: %
a typed session will never deadlock, %
all its external choices will be eventually activated, %
all its queued messages will be eventually consumed.\footnote{%
  This assumes that a session is fairly scheduled,
  similarly to the notion of ``fair path'' in \Cref{def:env-path-fairness}.%
}%

\section{Preciseness of Asynchronous Multiparty Session Subtyping}%
\label{sec:op}
We now present our main result.
A subtyping relation $\subt$ is \emph{sound} if it satisfies %
the \citeN{Liskov:1994:BNS} substitution principle: %
if $\T \subt \T'$, %
then a process of type $\T'$ engaged in a well-typed session %
may be safely replaced with a process of type $\T$. %
The reversed implication is called \emph{completeness}: %
if it is always safe to replace
a process of type $\T'$ with a process of type $\T$,
then we should have $\T \subt \T'$. %
If a subtyping $\subt$ is both sound and complete,
then $\subt$ is \emph{precise}. %
This is formalised in \Cref{def:preciseness} below %
(where we use the contrapositive of the completeness implication).

\newcommand{\subGen}{\unlhd}%
\newcommand{\nsubGen}{\mathrel{\not\!\!\unlhd}}%
\begin{definition}[Preciseness]
 \label{def:preciseness}%
 Let $\subGen$ be a preorder over session types. We say that $\subGen$ is:
  \begin{enumerate}[label=(\arabic*),ref=\emph{(\arabic*)}]
    \item a \textbf{sound subtyping} if\; $\T \subGen \T'$ \;implies that, %
      for all $\pr \not\in \participant{\T'}, \M, \PP$,  the following holds:
      \begin{enumerate}
      \item%
        if $\big(\forall \Q: \vphantom{\Theta}\vdash \Q:\T' \; \implies \;\; \Gamma \vdash \pa\pr \Q \pc \pa\pr\EmptyQueue \pc \M$  \;for some live $\Gamma\big)$ then \\
        $\big(\vdash \PP:\T \; \implies \; (\pa\pr\PP \pc \pa\pr\EmptyQueue \pc \M \red^* \M'$ \,implies\, $\M'\neq \error) \big)$
      \end{enumerate}
    \item\label{item:completeness} a \textbf{complete subtyping} if $\T \nsubGen \T'$ implies that there are $\pr \not\in \participant{\T'},\M,\PP$ such that:
    \begin{enumerate}[nosep]
           \item  $\forall Q: \vphantom{\Theta}\vdash \Q:\T' \; \implies \;\; \Gamma \vdash \pa\pr \Q \pc \pa\pr\EmptyQueue \pc \M$  \;for some live $\Gamma$   
           \item  $\vdash \PP:\T$
           \item  $\pa\pr\PP \pc \pa\pr\EmptyQueue \pc \M \red^* \error$.
     \end{enumerate}
    \item\label{item:preciseness} a \textbf{precise subtyping} if it is both sound and complete.
  \end{enumerate}
\end{definition}

As customary, our subtyping relation is embedded in the type system via a subsumption rule, giving soundness as an immediate consequence of the subject reduction property.

\begin{theorem}[Soundness]
  \label{thm:sound}
  The asynchronous multiparty session subtyping $\subt$ is sound.
\end{theorem}

\begin{proof}
  Take any $\T, \T'$ such that $\T \subt \T'$, and $\pr, \M$ satisfying the following condition:
  \begin{align}
       &   \forall Q:\; \vphantom{\Theta}\vdash \Q:\T' \; \implies \;\; \Gamma \vdash \pa\pr \Q \pc \pa\pr\EmptyQueue \pc \M \;\text{ for some live $\Gamma$} \label{eqn:sound1}
  \end{align}
  If\; $\vdash\PP:\T$, we derive by \rulename{t-sub} that\; $\vdash\PP:\T'$ holds.  By \eqref{eqn:sound1},   $\Gamma \vdash \pa\pr \PP \pc \pa\pr\EmptyQueue \pc \M$  for some live $\Gamma.$
  Hence, by Corollary~\ref{thm:error-freedom},\; %
  $\pa\pr \PP \pc \pa\pr\EmptyQueue \pc \M \red^* \M'$ \;implies\; $\M'\neq \error$.
\end{proof}

The proof of completeness of $\subt$ is much more involved.
We show that $\subt$ satisfies %
item~\ref{item:completeness} of \Cref{def:preciseness} %
in 4 steps, that we develop in the next sections:

\begin{description}[itemsep=2mm]
   \item[{[Step~1]}] We define the \emph{negation $\nsubttt$ of the SISO trees refinement relation} by an inductive definition, thus getting a clear %
   characterisation of the complement $\not\subt$ of the subtyping relation, that is necessary for Step 2. In addition, for every pair $\T,\T'$ such that $\T\not\subt \T'$, we choose a pair $\UU,\V'$ satisfying $\UU\not\subt\V'$  and
$\ttree{\UU}\in \llbracket \ttree{\T}\rrbracket_\SO$ and $\ttree{\V'}\in \llbracket \ttree{\T'}\rrbracket_\SI.$
  \item[{[Step~2]}] We define for every $\UU$ a {\em characteristic process} $\CP{\UU}$. Moreover, we prove that if $\ttree{\UU}\in \llbracket \ttree{\T} \rrbracket\SO$ then we have  $\vdash \CP{\UU}: \T$. 
   \item[{[Step~3]}] For every  $\V'$ with $\ttree{\V'}\in\llbracket \ttree{\T'} \rrbracket_\SI,$ and for every participant $\pr\not\in\participant{\V'},$ we define a {\em characteristic session} $\M_{\pr,\V'}$,  which is typable if composed with a process $\PQ$ of type $\T'$:%
\[\forall \PQ:\; \vphantom{\Theta}\vdash \PQ: \T' \; \implies \; \Gamma \vdash \pa\pr\PQ \pc \pa\pr\EmptyQueue \pc \M_{\pr,\V'} \;\text{ for some live $\Gamma$}.\]
 
\item[{[Step~4]}]
   Finally, we show that for all
  $\UU, \V'$ such that  $\UU \not\subt \V'$, %
  the characteristic session $\M_{\pr,\V'}$ (Step~3) %
  reduces to error if composed with the characteristic process of $\UU$ %
  (Step~2):
\[\pa\pr\CP{\UU} \pc \pa\pr\EmptyQueue \pc \M_{\pr,\V'} \;\red^*\; \error.\]
 \end{description}

\noindent%
Hence, we prove the completeness of $\subt$ by showing that, %
for all $\T,\T'$ such that $\T \not\subt \T'$, %
we can find $\pr\not\in\participant{\T'}$, $\PP=\CP{\UU}$ (Steps 1,2), %
and $\M = \M_{\pr, \V'}$ (Step 3) %
satisfying Def.~\ref{def:preciseness}\ref{item:completeness} (Step 4). %
We now illustrate each step in more detail.

\subsection{Step 1: Subtyping Negation}
\label{sec:preciseness-step1}

\begin{table}[t!]
\noindent%
$
{\small
\begin{array}{@{}c@{}}
     \inferruleR[\rulename{n-out}]
  {\pp! \not\in \actions{\WT' }}{\tout\pp\ell\S.\WT \nsubttt \WT'}
\;\;
       \inferruleR[\rulename{n-inp}]
  {\pp? \not\in \actions{\WT' }}{\tin\pp\ell\S.\WT \nsubttt \WT'}
\;\;
       \inferruleR[\rulename{n-out-R}]
  {\pp! \not\in \actions{\WT}}{\WT \nsubttt \tout\pp\ell\S.\WT'}
\;\;
       \inferruleR[\rulename{n-inp-R}]
  {\pp? \not\in \actions{\WT}}{\WT \nsubttt \tin\pp\ell\S.\WT'}
\\[1.5mm]
\inferruleR[\rulename{n-inp-$\ell$}]
     {\ell\neq \ell'}{\tin\pp{\ell}{\ST}.\WT \nsubttt\tin\pp{\ell'}{\ST'}.\WT'}
\;\;
 \inferruleR[\rulename{n-inp-$\S$}]
     {\ST' \not\subs \ST}{\tin\pp{\ell}{\ST}.\WT\nsubttt\tin\pp{\ell}{\ST'}.\WT'}
 \;\;
 \inferruleR[\rulename{n-inp-$\WT$}]
   {\S'\subs\S\;\;\WT \nsubttt \WT'}{\tin\pp{\ell}{\ST}.\WT\nsubttt  \tin\pp{\ell}{\ST'}.\WT'}
 \\[1.5mm]
 \inferruleR[\rulename{n-${\AC}$-$\ell$}]
     {\ell\neq\ell'}{\tin\pp{\ell}{\ST}.\WT \nsubttt \ACon\pp{{\tin\pp{\ell'}{\ST'}.\WT'}}}
\;\;
 \inferruleR[\rulename{n-${\AC}$-$\S$}]
     {\S'\not\subs\S}{\tin\pp{\ell}{\ST}.\WT \nsubttt \ACon\pp{{\tin\pp{\ell}{\ST'}.\WT'}}}
\\[1.5mm]
\inferruleR[\rulename{n-${\AC}$-$\WT$}]
  {\S' \subs\S\;\;\WT\not\subttt \ACon\pp\WT'}{\tin\pp{\ell}{\ST}.\WT \nsubttt \ACon\pp{{\tin\pp{\ell}{\ST'}.\WT'}}}
\;\;
  \inferruleR[\rulename{n-i-o-1}]
{ $\ $}{\tin\pp{\ell}{\ST}.\WT \nsubttt \tout\pq{\ell'}{\ST'}.\WT'}
  \;\;
  \inferruleR[\rulename{n-i-o-2}]
 { $\ $}{\tin\pp{\ell}{\ST}.\WT \nsubttt \ACon\pp\tout\pq{\ell'}{\ST'}.\WT'}
\\[1.5mm]
\inferruleR[\rulename{n-out-$\ell$}]
   {\ell\neq \ell'}{\tout\pp{\ell}{\ST}.\WT\nsubttt\tout\pp{\ell'}{\ST'}.\WT'}
 \;\;
 \inferruleR[\rulename{n-out-$\S$}]
    {\ST \not\subs \ST'}{\tout\pp{\ell}{\ST}.\WT\nsubttt\tout\pp{\ell}{\ST'}.\WT'}
\;\;
 \inferruleR[\rulename{n-out-$\WT$}]
     {\S\subs\S'\;\;\WT \not\nsubttt \WT'}{\tout\pp{\ell}{\ST}.\WT\nsubttt\tout\pp{\ell}{\ST'}.\WT'}
 \\[1.5mm]%
  \inferruleR[\rulename{n-${\BC}$-$\ell$}]
      {\ell\neq\ell'}{\tout\pp{\ell}{\ST}.\WT \nsubttt \BCon\pp{{\tout\pp{\ell'}{\ST'}.\WT'}}}
  \;\;
\inferruleR[\rulename{n-${\BC}$-$\S$}]
  {\S\not\subs\S'}{\tout\pp{\ell}{\ST}.\WT \nsubttt \BCon\pp{{\tout\pp{\ell}{\ST'}.\WT'}}}
 \;\;%
   \inferruleR[\rulename{n-${\BC}$-$\WT$}]
  {\S\subs\S'\;\;\WT\not\subttt \BCon\pp\WT'}{\tout\pp{\ell}{\ST}.\WT \nsubttt \BCon\pp{{\tout\pp{\ell}{\ST'}.\WT'}}}
\end{array}}$
\smallskip%

\caption{\label{tab:negationW}The relation $\nsubttt$ between SISO trees.}
\vspace{-4mm}
\end{table}

In Table~\ref{tab:negationW} we \emph{inductively}
define the relation $\nsubttt$ %
over SISO trees: it contains all pairs of SISO trees that are \emph{not} related by $\subttt$, 
as stated in Lemma~\ref{lem:subttt-negation} below. %
This step is necessary because, in Step 2 (\Cref{sec:preciseness-step2}),
we will need the shape of the types related by $\nsubttt$
in order to generate some corresponding processes.

The first category of rules checks whether
two SISO trees have a direct syntactic mismatch: 
whether their sets of actions are disjoint ($\rulename{n-out},$ $\rulename{n-inp},$ $\rulename{n-out-R},$ $\rulename{n-inp-R})$; 
the label of the LHS  is  
not equal to the label of the RHS ($\rulename{n-inp-$\ell$}$, $\rulename{n-out-$\ell$}$);  
or 
matching labels are followed by mismatching sorts or continuations  ($\rulename{n-inp-$\S$}$, $\rulename{n-out-$\S$},$ $\rulename{n-inp-$\WT$}$, $\rulename{n-out-$\WT$}$).

The second category checks more subtle cases related to 
asynchronous permutations; 
rule $\rulename{n-$\AC$-$\ell$}$ checks a label mismatch 
when the input on the RHS is preceded by a finite number of inputs from other participant; similarly,
rules $\rulename{n-$\AC$-$\S$}$ and $\rulename{n-$\AC$-$\WT$}$
check mismatching sorts or continuations. 
Rules $\rulename{n-i-o-1}$ and $\rulename{n-i-o-2}$
formulate the cases such that 
the top prefix on the LHS is input and  the top sequence of  prefixes on
the RHS consists of a finite number of inputs from other participants
and/or outputs.
Finally, rules $\rulename{N-$\BC$-$\ell$}$,
$\rulename{n-$\BC$-$\S$}$ and $\rulename{n-$\BC$-$\WT$}$
check the cases of label mismatch, or 
matching labels followed by mismatching sorts, or continuations
of the two types with output prefixes targeting a same participant, %
where the RHS is prefixed by a finite number of outputs (to other participants) and/or inputs (to any participant).
\iftoggle{techreport}{%

More in detail, according to the rules in \Cref{tab:negationW},
two types are related by $\nsubttt$ if:
\begin{itemize}
\item
  their sets of actions are disjunctive ($\rulename{n-out},$ $\rulename{n-inp},$ $\rulename{n-out-R},$ $\rulename{n-inp-R})$;
\item
  their top prefixes are both inputs or both outputs, targeting the same participant --- but the label of the LHS  is  
  not equal to the label of the RHS ($\rulename{n-inp-$\ell$}$, $\rulename{n-out-$\ell$}$);  
 \item
  their top prefixes are both inputs or outputs, targeting the same participant, with matching labels followed by mismatching sorts or mismatching continuations  ($\rulename{n-inp-$\S$}$, $\rulename{n-out-$\S$},$ $\rulename{n-inp-$\WT$}$, $\rulename{n-out-$\WT$}$);
    \item
  they both have inputs from a same participant $\pp$, and such input on the RHS is preceded by a finite number of inputs from other participants --- but the input label from $\pp$ on the LHS is not equal to the label on the RHS ($\rulename{n-$\AC$-$\ell$}$);  
  \item
  they both have inputs from a same participant $\pp$, and such input on the RHS is preceded by a finite number of inputs from other participants --- but the inputs from $\pp$ have matching labels followed by mismatching sorts or mismatching continuations  ($\rulename{n-$\AC$-$\S$}$, $\rulename{n-$\AC$-$\WT$}$);
  \item
  the top prefix on the LHS is an input from a participant $\pp$, but the top-level prefixes on the RHS consists of a finite number of inputs ($0$ or more) from other participants, followed by an output  ($\rulename{n-i-o-1}$, $\rulename{n-i-o-2}$);
    \item
 they both have outputs toward a same participant $\pp$, and the output on the RHS is preceded by a finite number of outputs to other participants and/or inputs --- but the output label towards $\pp$ on the LHS  is  
  not equal to the label on the RHS ($\rulename{N-$\BC$-$\ell$}$);  
  \item
  they both have outputs toward a same participant $\pp$, and the output on the RHS is preceded by a finite number of outputs to other participants and/or inputs --- but the outputs toward $\pp$ have matching labels followed by mismatching sorts or mismatching continuations  ($\rulename{n-$\BC$-$\S$}$, $\rulename{n-$\BC$-$\WT$}$).
\end{itemize}
}{}%

\begin{restatable}{lemma}{lemNegationW}
  \label{lem:subttt-negation}%
  Take any pair of SISO trees $\WT$ and $\WT'$. Then,\, $\WT \subttt \WT'$
  is \emph{not} derivable \emph{if and only if}
  $\WT \nsubttt \WT'$ is derivable with the rules in \Cref{tab:negationW}.
\end{restatable}
\begin{proof}
  We adopt an approach inspired by \citeN{BHLN12} and \citeN{BHLN17}.%

  For the ``$\Rightarrow$'' direction of the statement,
  assume that $\WT \subttt \WT'$ is \emph{not} derivable:
  this means that, if we attempt to build a derivation
  by applying the rules in \Cref{tab:ref},
  starting with ``$\WT \subttt \WT'$'' and moving upwards,
  we obtain a \emph{failing derivation} 
  that, after a finite number $n$ of rule applications,
  reaches a (wrong) judgement ``$\WT_n \subttt \WT'_n$''
  on which no rule of \Cref{tab:ref} can be further applied.
  Then, by induction on $n$, we transform such a failing derivation
  into an \emph{actual} derivation based on the rules for $\nsubttt$
  in \Cref{tab:negationW}, which proves\, $\WT \nsubttt \WT'$.

  For the ``$\Leftarrow$'' direction,
  assume that $\WT \nsubttt \WT'$ holds, by the rules in \Cref{tab:negationW}:
  from its derivation, we construct a failing derivation
  that starts with ``$\WT \subttt \WT'$''
  and is based on the rules of \Cref{tab:ref}
  --- which implies that $\WT \subttt \WT'$ is \emph{not} derivable.
  \iftoggle{techreport}{%
    \noindent%
    \emph{(Full proof in \Cref{sec:subttt-negation})}
  }{}%
\end{proof}

It is immediate from \Cref{def:subtyping} and \Cref{lem:subttt-negation}
that %
$\T$ is \emph{not} a subtype of $\T'$, %
written $\T\not\subt\T'$,  if and only if:
\begin{equation}
  \label{eq:not-subt-def}%
  \exists \UT\in \llbracket\ttree{\T} \rrbracket_\SO \;\; \exists \VT'\in \llbracket \ttree{\T'}\rrbracket_\SI\;\;  \forall\WT\in \llbracket\UT\rrbracket_\SI \;\; \forall \WT'\in \llbracket\VT'\rrbracket_\SO\;\; \WT\not\subttt\WT'
\end{equation}
Moreover, we prove that %
whenever $\T \not\subt \T'$, %
we can find \emph{regular, syntax-derived} SO/SI trees %
usable as the witnesses $\UT,\VT'$ in \eqref{eq:not-subt-def}. %
\iftoggle{techreport}{
  \emph{(See Appendix, page~\pageref{sec:regular_representatives_appendix}.)} %
}{}%
Thus, from this result and by \eqref{eq:not-subt-def}, %
$\T \not\subt \T'$ 
implies:
\begin{equation}
  \label{eq:not-subt-def-regular}%
  \exists \UU, \V':\; %
  \ttree{\UU}\!\in\! \llbracket\ttree{\T} \rrbracket_\SO \;\; \ttree{\V'}\!\in\! \llbracket \ttree{\T'}\rrbracket_\SI\;\;  \forall\WT\!\in\! \llbracket\ttree{\UU}\rrbracket_\SI \;\; \forall \WT'\!\in\! \llbracket\ttree{\V'}\rrbracket_\SO\;\; \WT\not\subttt\WT'
\end{equation}

\begin{example}
  \label{ex:not-sub-intro}%
  Consider the example in Section~\ref{sec:intro}, %
  and its types $\T'$ and $\T$ in Example~\ref{ex:sub-intro}:%
  \introTypes%
  We have seen that $\T' \subt \T$ holds (Example~\ref{ex:sub-intro}), %
  and thus, by subsumption, %
  our type system allows to use the optimised process $\PP'_\pr$ %
  in place of  $\PP_\pr$ (Example~\ref{ex:typing-intro}). 
  We now show that the inverse relation does \emph{not} hold, %
  \ie, $\T \not\subt \T'$, hence the inverse process replacement is disallowed.
  Take, \eg, $\UU,\V'$ as follows, %
  noticing that $\ttree{\UU} \!\in\! \llbracket \ttree{\T} \rrbracket_\SO$ %
  and $\ttree{\V'} \!\in\! \llbracket \ttree{\T'} \rrbracket_\SI$:
\[\small%
    \UU = \texternal\pp ?\left\{
      \begin{array}{@{}l@{}}
        \msgLabel{success}(\tint).\pq!\msgLabel{cont}(\tint).\tend\\
        \msgLabel{error}(\tbool).\pq!\msgLabel{stop}(\tunit).\tend
      \end{array} \right.
    \qquad%
    \V' = \tinternal\pq !\left\{
      \begin{array}{@{}l@{}}
        \msgLabel{cont}(\tint).\pp? \msgLabel{success}(\tint).\tend\\
        \msgLabel{stop}(\tunit).\pp?  \msgLabel{error}(\tbool).\tend
      \end{array} \right.
\]
  \noindent%
  For all $\WT\in \llbracket \ttree{\UU} \rrbracket_\SI=\left\{\pp?\msgLabel{success}(\tint).\pq!\msgLabel{cont}(\tint).\tend\,,\, \pp?\msgLabel{error}(\tbool).\pq!\msgLabel{stop}(\tunit).\tend\right\}$ %
and all $\WT'\in \llbracket \ttree{\V'} \rrbracket_\SO = \left\{\pq! \msgLabel{cont}(\tint).\pp? \msgLabel{success}(\tint).\tend\,,\, \pq!\msgLabel{stop}(\tunit).\pp?  \msgLabel{error}(\tbool).\tend\right\}$ %
 we get by \rulename{n-i-o-1} that
  $\WT\not\subttt\WT'$.  Therefore, we conclude $\T \not\subt \T'$.
 \end{example}

\subsection{Step 2: Characteristic Processes}
\label{sec:preciseness-step2}

For any SO type $\UU$, %
we define a characteristic process $\CP\UU$ %
(Def.~\ref{def:characteristic-process}): %
intuitively, it is a process constructed to %
  communicate as prescribed by $\UU$, %
  and to be typable by $\UU$.%

\begin{definition}
  \label{def:characteristic-process}
  The characteristic process $\CP{\UU}$ of type $\UU$ is defined inductively as follows:%
\[
    \begin{array}{@{}c@{}}
        \CP{\tend} =  \inact \qquad   \CP{\ty} = X_\ty \qquad \CP{\mu\ty.\UU} = \mu X_\ty.\CP{\UU}    \qquad  \CP{\pp!{\ell}(\ST).\UU}=  \procout \pp {\ell}{\valt{\ST}}{\CP{\UU}}     \\
        \CP{ \texternal_{i\in I} \pp?\ell_i(\ST_i).\UU_i } \;=\;  \sum_{i\in I} \procin  \pp {\ell_i(\x_i)}\cond{\exprt{\x_i}{\ST_i}}{\CP{\UU_i}}{\CP{\UU_i}}
      \\[1mm]%
      \begin{array}{@{}l@{}}
            \text{where: }%
            \\[1mm]%
        \valt{\tnat}=1 \quad \valt{\tint}=-1 \quad \valt{\tbool}=\true \quad   \changeNoMargin{\valt{\tunit}=() 
        \quad \exprt{\x}{\tunit}= (\x \approx ())}
        \\%
        \exprt{\x}{\tbool} = (\neg\x) \quad %
        \exprt{\x}{\tnat}=( \fsucc\x > \zero) \quad %
        \exprt{\x}{\tint}= (\fneg\x  > \zero) 
      \end{array}
   \end{array}
\]
\end{definition}

By \Cref{def:characteristic-process}, %
for every output in $\UU$, the characteristic $\CP{\UU}$ sends a value $\valt{\ST}$ %
of the expected sort $\ST$; %
and for every external choice in $\UU$, %
$\CP{\UU}$ %
performs a branching, and uses any received value $x_i$ of sort $\ST_i$
in a boolean expression $\exprt{x_i}{\ST_i}$:
the expression will cause an $\error$ if the value of $x_i$
is not of sort $\ST_i$.

Crucially, %
for all $\T$ and $\UU$ such that %
$\ttree{\UU} \!\in\! \llbracket \ttree{\T} \rrbracket_\SO$ %
(\eg, from \eqref{eq:not-subt-def-regular} above), %
we have $\UU \subt \T$: %
therefore, $\CP{\UU}$ is also typable by $\T$, %
as per \Cref{cp} below.%
\iftoggle{techreport}{%
\emph{ (Proof: \Cref{subsec:app:preciseness})}
}{}%

\begin{restatable}{proposition}{propCP}
\label{cp}
For all closed types $\T$ and $\UU$, if\; $\ttree{\UU}\!\in\!\llbracket \ttree{\T} \rrbracket_\SO$ \;then\; $\vdash \CP{\UU}:\T$.
\end{restatable}

\subsection{Step 3: Characteristic Session}

The next step to prove  completeness is to define for each session type $\V'$ and participant $\pr\not\in \participant{\V'}$ a \emph{characteristic session} $\M_{\pr, \ts{V'}}$, that is well typed (with a live typing environment) %
when composed with participant $\pr$ %
associated with a process of type $\V'$ and empty queue.  

For a SI type $\V'$ and $\pr\not\in\participant{\V'}=\{\pp_1,\ldots,\pp_m\}$, %
we define $m$ \emph{characteristic SO session types} %
where participants $\pp_1,\ldots,\pp_m$ %
are engaged in a live multiparty interaction with $\pr$, %
and with each other. %
\Cref{def:characteristic-types} %
ensures that after each communication between $\pr$ and some %
$\pp\!\in\!\participant{\V'}$, %
there is a cyclic sequence of communications %
starting with $\pp$, %
involving \emph{all} other $\pq \!\in\! \participant{\V'}$, %
and ending with $\pp$ --- %
with each participant acting both as receiver, and as sender.

\begin{definition}
\label{def:characteristic-types}
Let  $\V'$ be  a SI session type and %
$\pr \!\not\in\! \participant{\V'} \!=\! \{\pp_1,{...},\pp_m\}$. %
For every $k \!\in\! \{1,{...},m\}$, %
if $m \!\ge\! 2$ %
we define a \emph{characteristic SO session type} $\cyclic{\V'}{\pp_k}$ %
as follows:
\[
\small
\begin{array}{@{}l@{\;}c@{\;}l@{}}
     \cyclic{\tend}{\pp_k}&=&\tend \\
     \cyclic{\ty}{\pp_k}&=&\ty \\
     \cyclic{\mu\ty.\V_1''}{\pp_k}&=&\mu\ty. \cyclic{\V_1''}{\pp_k}\\
     \cyclic{\tin{\pp_k}{\ell}{\ST}.\V'}{\pp_k} &= & \tout\pr{\ell}{\ST}.\hlight{$\tout{\pp_{k+1}}{\ell}{\tbool}.\tin{\pp_{k-1}}{\ell}{\tbool}.$}
                                    \cyclic{\V'}{\pp_k} \\
     \cyclic{\tin\pq{\ell}{\ST}.\V'}{\pp_k} &= & \tin{\pp_{k-1}}{\ell}{\tbool}.\tout{\pp_{k+1}}{\ell}{\tbool}.
                   \cyclic{\V'}{\pp_k} \;\;\qquad\qquad\hlight{(if $\pq \!\neq\! \pp_k$)}\\
     \cyclic{\tinternal_{j\in J} \tout{\pp_k}{\ell_j}{\ST_j}.\V'_j}{\pp_k} & = & \texternal_{j\in J} \tin\pr{\ell_j}{\ST_j}.\hlight{$\tout{\pp_{k+1}}{\ell_j}{\tbool}.\tin{\pp_{k-1}}{\ell_j}{\tbool}.$}\cyclic{\V'_j}{\pp_{k}}\\
     \cyclic{ \tinternal_{j\in J} \tout\pq{\ell_j}{\ST_j}.\V'_j}{\pp_k}  & = &\texternal_{j\in J} \tin{\pp_{k-1}}{\ell_j}{\tbool}.\tout{\pp_{k+1}}{\ell_j}{\tbool}.
                                                                                                              \cyclic{\V'_j}{\pp_k} \quad\hlight{(if $\pq \!\neq\! \pp_k$)}                                                                                                           
\end{array}
\]
\noindent%
If $m \!=\! 1$ (\ie, if there is only one participant in $\V'$) %
we define $\cyclic{\V'}{\pp_1}$ as above, %
but we omit the \hlight{(highlighted) cyclic communications}, %
and the cases with \hlight{$\pq \!\neq\! \pp_k$} do not apply.
\end{definition}

\begin{example}[Characteristic session types]
  Consider the following SI type:
\[
    \V' \;=\; \mu\ty.%
    \tinternal{}{\big\{%
        \tout\pq{\ell_2}{\tnat}.\tin\pp{\ell_1}{\tnat}.\ty\,,\,%
        \tout\pq{\ell_3}{\tnat}. \tin\pp{\ell_4}{\tnat}.\ty%
    \big\}}%
\]
  \noindent%
  Let $\pr \!\not\in \participant{\V'}$. %
  The characteristic session types for participants %
  $\pp,\pq \!\in\! \participant{\V'}$ are: %
\[
    \begin{array}{@{}r@{\;\,}c@{\,\;}l@{}}
      \cyclic{\V'}{\pp} &=&%
     \mu\ty.\texternal{}{\left\{
      \begin{array}{@{}l@{}}
         \tin\pq{\ell_2}{\tbool}.\tout{\pq}{\ell_2}{\tbool}.\tout\pr{\ell_1}{\tnat}.\tout\pq{\ell_1}{\tbool}.\tin\pq{\ell_1}{\tbool}.\ty \\
         \tin\pq{\ell_3}{\tbool}.\tout{\pq}{\ell_3}{\tbool}.\tout\pr{\ell_4}{\tnat}.\tout\pq{\ell_4}{\tbool}.\tin\pq{\ell_4}{\tbool}.\ty
     \end{array}\right.
 }\\[3mm]%
     \cyclic{\V'}{\pq} &=&%
      \mu\ty.\texternal{}{\left\{
           \begin{array}{@{}l@{}}
         \tin\pr{\ell_2}{\tnat}.\tout{\pp}{\ell_2}{\tbool}.\tin{\pp}{\ell_2}{\tbool}.\tin\pp{\ell_1}{\tbool}.\tout\pp{\ell_1}{\tbool}.\ty \\
          \tin\pr{\ell_3}{\tnat}.\tout{\pp}{\ell_3}{\tbool}.\tin{\pp}{\ell_3}{\tbool}.\tin\pp{\ell_4}{\tbool}.\tout\pp{\ell_4}{\tbool}.\ty
          \end{array}\right.}
   \end{array}
\]
  \noindent%
  Note that if $\pr$ follows type $\V'$, %
  then it must select and send to $\pq$ one message between $\ell_2$ %
  and $\ell_3$; %
  correspondingly, the characteristic session type for $\pq$ %
  receives the message (with a branching), %
  and propagates it to $\pp$, who sends it back to $\pq$ %
  (cyclic communication). %
  Then, $\pr$ waits for a message from $\pp$ %
  (either $\ell_1$ or $\ell_4$, depending on the previous selection): %
  correspondingly, the characteristic session type for
  $\pp$ will send such a message, %
  and also propagate it to $\pq$ with a cyclic communication.%
  \hfill$\blacksquare$%
\end{example}

Given an SI type $\V'$, %
we can use Def.~\ref{def:characteristic-types} %
to construct the following typing environment:
\begin{align}
  \label{eq:characteristic-gamma}%
  &\Gamma \,=\, \{\pr{:}\,(\temptyqueue, \V')\} \cup \left\{\pp{:}\,(\temptyqueue, \UU_{\pp}) \;\middle|\; \pp\!\in\!\participant{\V'} \right\}%
 \quad%
    \text{where\, $\forall \pp \!\in\! \participant{\V'}{:}\, \UU_{\pp}= \cyclic{\V'}{\pp}$}%
\end{align}
\noindent%
\ie, we compose $\V'$ %
with the characteristic session types of all its participants. %
The cyclic communications of Def.~\ref{def:characteristic-types} %
ensure that $\Gamma$ is live. %
We can use $\Gamma$ to type the composition of %
a process for $\pr$, of type $\V'$, %
together with the characteristic processes %
of the characteristic session types of each participants in $\V'$: %
we call such processes the \emph{characteristic session} $\M_{\pr, \V'}$. %
This is formalised in Def.~\ref{def:characteristic-session} %
and Prop.~\ref{prop:live} below.%

\begin{definition}
  \label{def:characteristic-session}%
  For any SI type $\V'$ and $\pr\!\not\in\!\participant{\V'}$, %
  we define the \emph{characteristic session}:

  \smallskip\centerline{\(%
    \displaystyle%
    \M_{\pr, \V'} \;=\;%
    \prod_{\pp \in \participant{\V'}} \Big(\pa{\pp} \CP{\UU_{\pp}} \pc \pa{\pp} \EmptyQueue\Big)
    \qquad%
    \text{where\, $\forall \pp \!\in\! \participant{\V'}:\; \UU_{\pp}= \cyclic{\V'}{\pp}$}%
\)}%
\vspace{-3mm}%
\end{definition}

\begin{restatable}{proposition}{propCharacteristicSession}
\label{prop:live}
  Let $\V'$ be a SI type and $\pr\!\not\in\!\participant{\V'}$. %
  Let $Q$ be a process such that $\vphantom{\Theta} \vdash \Q:\V'$. %
  Then, there is a live typing environment $\Gamma$ %
  (see \eqref{eq:characteristic-gamma}) %
  such that\; %
  $\Gamma \vdash \pa\pr \Q \pc \pa\pr\EmptyQueue \pc \M_{\pr, \V'}$.
\end{restatable}

Crucially, %
for all $\T'$ and $\V'$ such that %
$\ttree{\V'} \!\in\! \llbracket \ttree{\T'} \rrbracket_\SI$ %
(\eg, from \eqref{eq:not-subt-def-regular} above), %
we have $\T' \subt \V'$. %
Thus, by subsumption, %
$\M_{\pr,\V'}$ is also typable with a process of type $\T'$ %
(\Cref{prop:live-ti}).

\begin{restatable}{proposition}{propCharacteristicSessionTi}
\label{prop:live-ti}
  Take any $\T'$, $\pr\!\not\in\!\participant{\T'}$, SI type $\V'$
  such that $\ttree{\V'} \!\in\! \llbracket \ttree{\T'} \rrbracket_\SI$, %
  and $Q$ such that $\vphantom{\Theta} \vdash \Q:\T'$. %
  Then, there is a live $\Gamma$ %
  (see \eqref{eq:characteristic-gamma}) %
  such that\; %
  $\Gamma \vdash \pa\pr \Q \pc \pa\pr\EmptyQueue \pc \M_{\pr, \V'}$.
\end{restatable}

\subsection{Step 4: Completeness}

This final step of our completeness proof %
encompasses all elements introduced thus far.%
\iftoggle{techreport}{%
  \emph{ (Proofs in \Cref{subsec:app:preciseness})}
}{}

\begin{proposition}
  \label{prop:completeness-maintext}
  Let $\T$ and $\T'$ be session types such that $\T \not\subt\T'$. %
  Take any\, $\pr \!\not\in\! \participant{\T'}$. %
  Then, there are $\UU$ and $\V'$ %
  with $\ttree{\UU} \!\in\! \llbracket \ttree{\T} \rrbracket_\SO$ %
  and $\ttree{\V'} \!\in\! \llbracket \ttree{\T'} \rrbracket_\SI$ %
  and $\UU \not\subt \V'$ %
  such that:
  \begin{enumerate}
  \item\label{item:prop:completeness-gamma}%
    $\forall Q: \vphantom{\Theta}\vdash \Q:\T' \,\implies\, \Gamma \vdash \pa\pr \Q \pc \pa\pr\EmptyQueue \pc \M_{\pr, \V'}$  for some live $\Gamma$;%
    \hfill(by \eqref{eq:characteristic-gamma} %
    and Prop.~\ref{prop:live-ti})%
  \item%
    $\vdash \CP{\UU}:\T$;%
    \hfill(by Prop.~\ref{cp})%
  \item\label{item:prop:completeness-error}%
    $\pa\pr\CP{\UU} \pc \pa\pr\EmptyQueue \pc \M_{\pr, \V'} \;\red^*\; \error$.
  \end{enumerate}
\end{proposition}%

Intuitively, we obtain item~\ref{item:prop:completeness-error} %
of \Cref{prop:completeness-maintext} because %
the characteristic session $\M_{\pr, \V'}$ %
expects to interact with a process of type $\V'$ (or a subtype, like $\T'$); %
however, when a process that behaves like $\UU$ is inserted, %
the cyclic communications and/or the expressions
of $\M_{\pr, \V'}$ %
(given by Def.~\ref{def:characteristic-process} %
and \ref{def:characteristic-types}) %
are disrupted: %
this is because $\UU \not\subt \V'$, %
and the (incorrect) message reorderings and mutations %
allowed by $\not\subttt$ (Table~\ref{tab:negationW}) %
cause the errors in Table~\ref{tab:error}.

We now conclude with our main results.

\begin{restatable}{theorem}{thCompleteness}
  \label{lem:completeness}%
  The asynchronous multiparty session subtyping $\subt$ is complete.
\end{restatable}

\begin{proof}
  Direct consequence of \Cref{prop:completeness-maintext}: %
  by taking $\pr$ and letting $\M = \M_{\pr, \V'}$ and $\PP = \CP{\UU}$ %
  from its statement, %
  we satisfy item \ref{item:completeness} %
  of Def.~\ref{def:preciseness}.%
\end{proof}

\begin{theorem}
  \label{lem:preciseness}%
  The asynchronous multiparty session subtyping $\subt$ is precise.
\end{theorem}

\begin{proof}
  Direct consequence of Theorems~\ref{thm:sound} and \ref{lem:completeness}, %
  which satisfy item \ref{item:preciseness} %
  of Def.~\ref{def:preciseness}.%
\end{proof}

\section{Example: Distributed Batch Processing}
\label{sec:batch-processing}
In this section we illustrate our 
subtyping relation by showing the correctness of a messaging optimisation
in a distributed processing scenario.
We adapt a multiparty protocol from the
\emph{double-buffering algorithm}~\cite{doublebuffer,mostrous_yoshida_honda_esop09,CY2020B,YoshidaVPH08},
that is widely used, \eg,
in streaming media applications and computer graphics, 
to regulate and speed up asynchronous communication
and data processing.

\subsection{Basic Unoptimised Protocol}
\label{sec:basic-batch-protocol}

\begin{figure}[tpb]
  \centering
  \subcaptionbox{%
    \footnotesize%
    \Control tells \Source it should send the next batch of data
    to \ProcessorOne.%
    \label{fig:batch-processing-1}%
  }{%
    \includegraphics[width=0.28\textwidth]{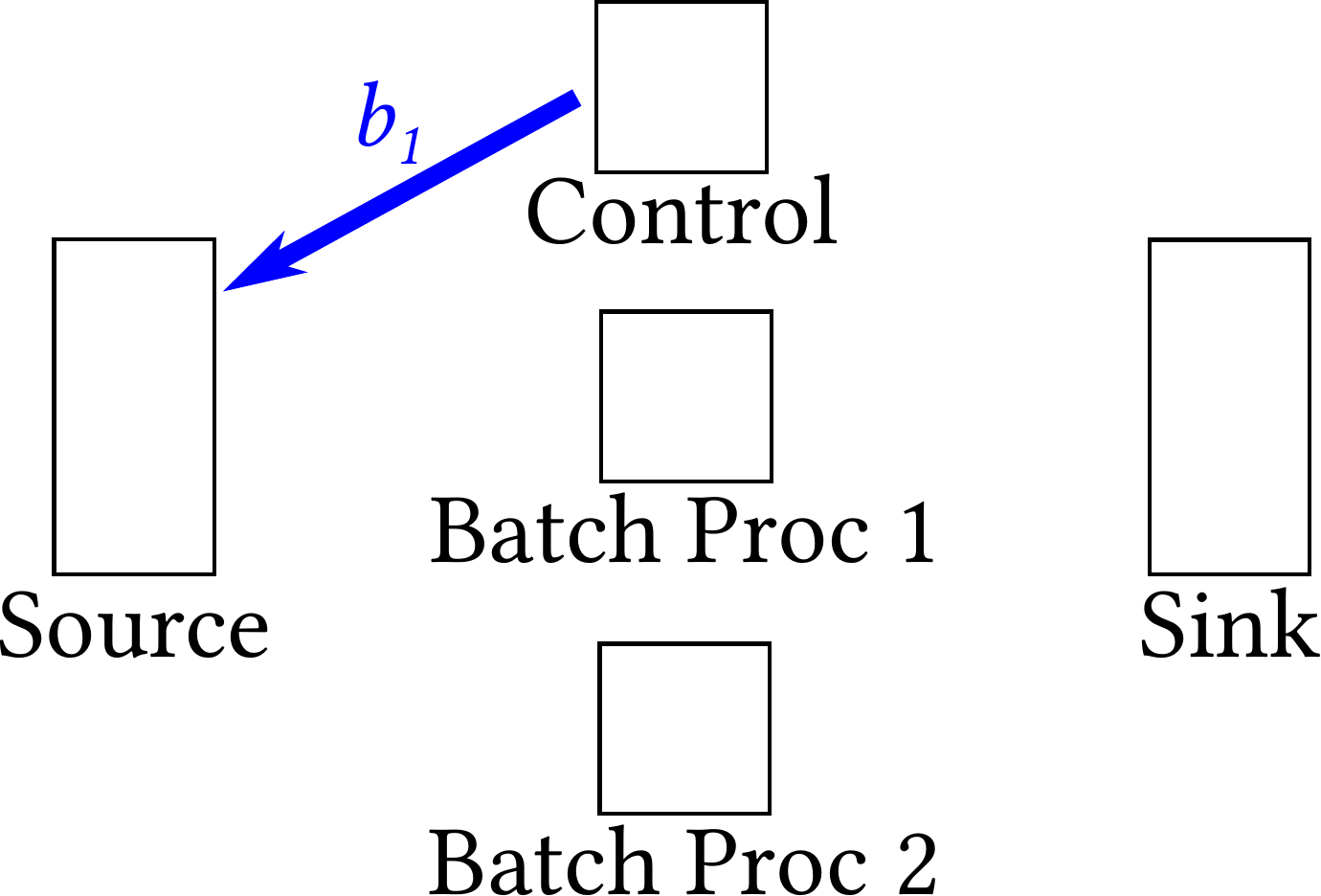}
  }
  \hspace{7mm}
  \subcaptionbox{%
    \footnotesize%
    \Source starts sending data to \ProcessorOne.
    Meanwhile, \Sink asynchronously tells \Control it is ready to receive
    data from \ProcessorOne.%
    \label{fig:batch-processing-2}%
  }{%
    \includegraphics[width=0.28\textwidth]{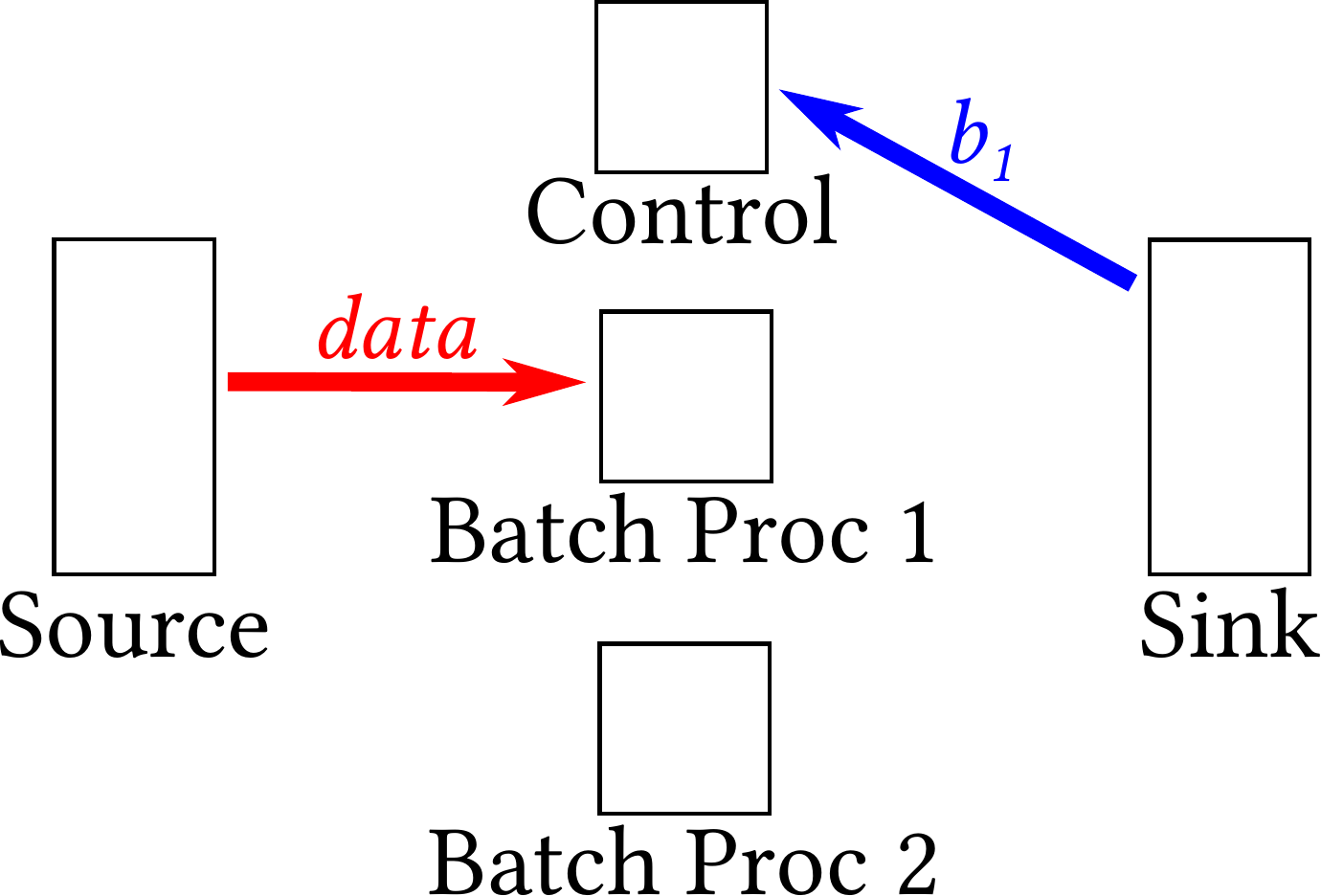}
  }
  \hspace{7mm}
  \subcaptionbox{%
    \footnotesize%
    \ProcessorOne finishes processing its data
    and sends the result to \Sink; meanwhile,
    \Control tells \Source that it should send the next data
    to \ProcessorTwo.
    \label{fig:batch-processing-3}%
  }{%
    \includegraphics[width=0.28\textwidth]{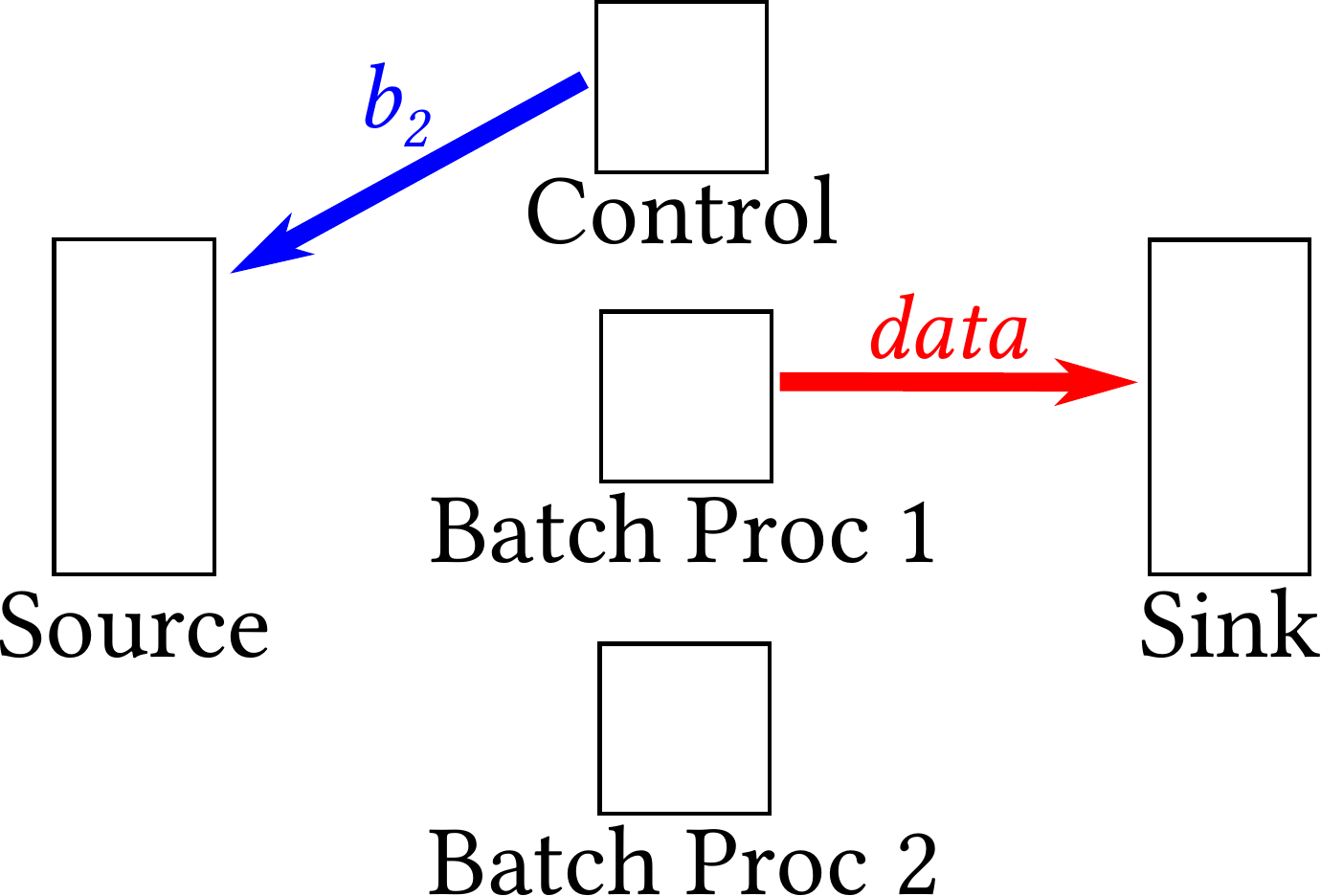}
  }
  \\[3mm]%
  \subcaptionbox{%
    \footnotesize%
    \ProcessorOne has sent all its data to \Sink;
    \Source starts sending data to \ProcessorTwo,
    while \Sink asynchronously
    tells \Control it is ready to receive data from \ProcessorTwo.
    \label{fig:batch-processing-4}%
  }{%
    \includegraphics[width=0.28\textwidth]{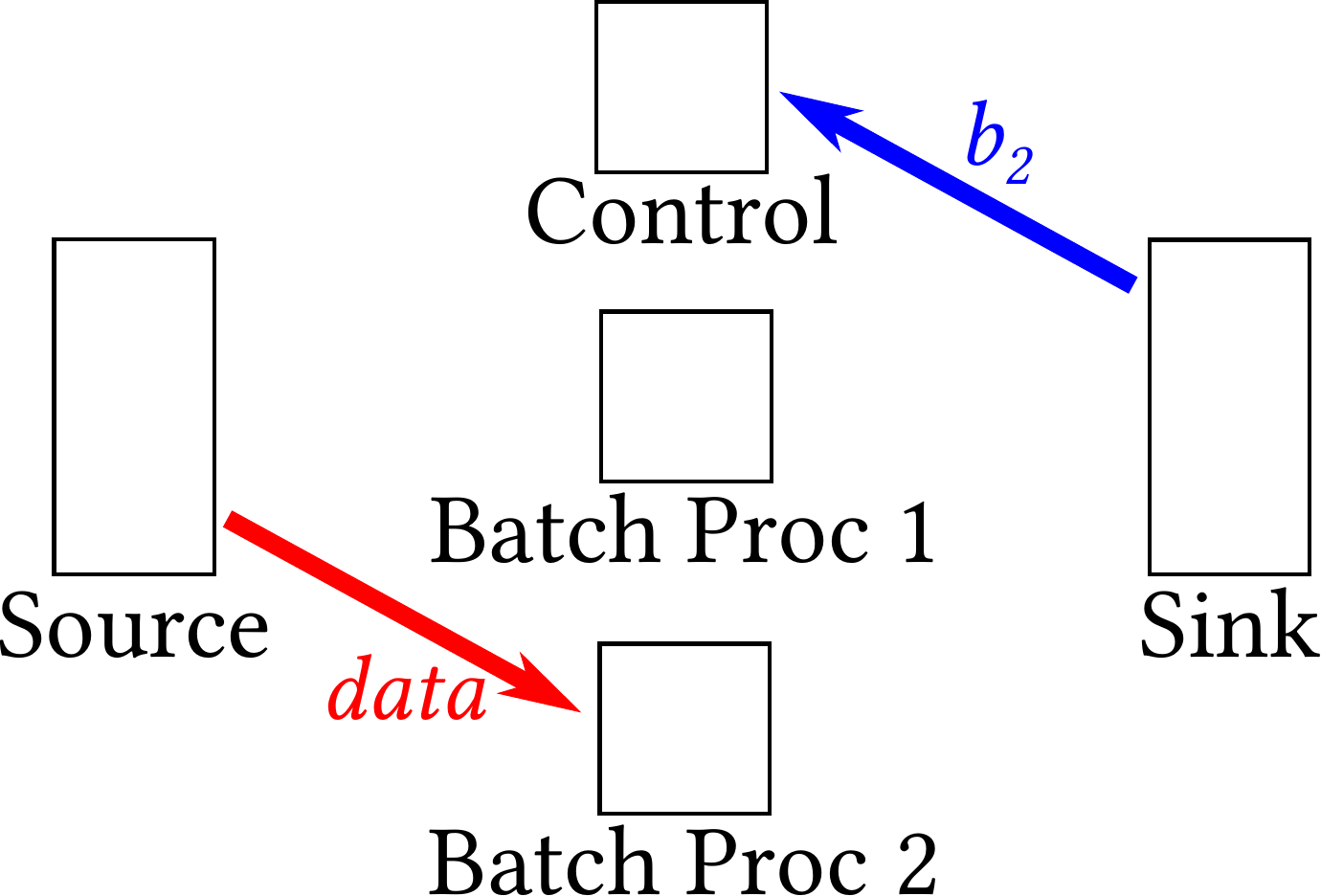}
  }
  \hspace{7mm}
  \subcaptionbox{%
    \footnotesize%
    \ProcessorTwo is sending its resulting data to \Sink.
    \ProcessorTwo completes the transmission,
    we go back to step (\subref{fig:batch-processing-1}).%
    \label{fig:batch-processing-5}%
  }{%
    \includegraphics[width=0.28\textwidth]{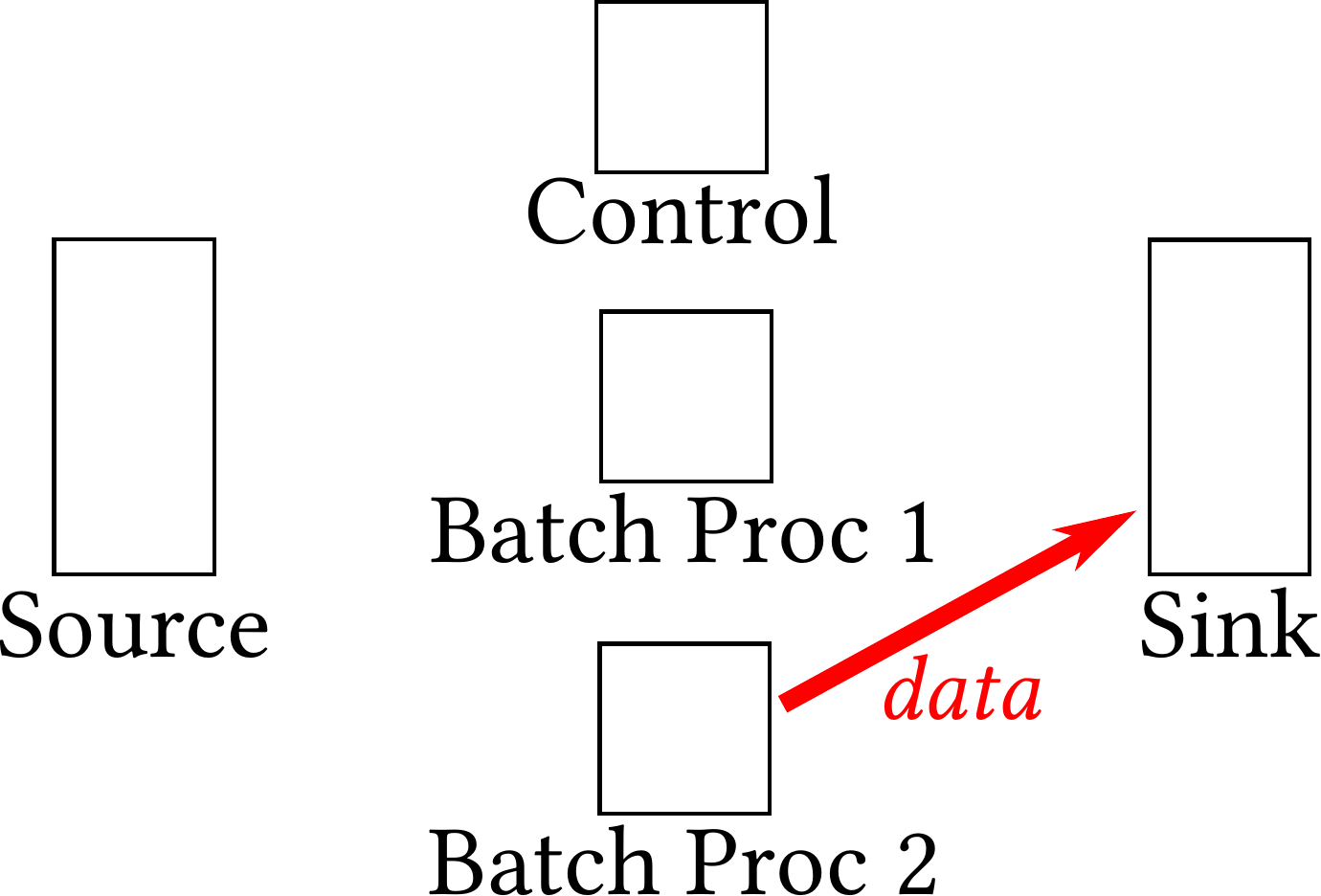}
  }
  \caption{An execution of the basic unoptimised distributed batch processing
    protocol.}
  \label{fig:batch-processing-basic}
\end{figure}

In a distributed processing system,
two \emph{Batch Processors} are always ready to
receive data from a \Source, perform some computation, and send
the results to a \Sink. A \Control process regulates their
interactions, by telling the Source/Sink where to send/receive new data;
consequently, the Control can influence when the Batch \Processor{}s are
running (\eg, all at once, or one a time), and can ensure that the
\Source is not sending too much data (thus overwhelming the Batch
\Processor{}s and the \Sink), and the \Sink is expecting data from an
active \Processor.

\Cref{fig:batch-processing-basic} provides an overview of the
communications between the five parties (\Source, \Sink, \Control, \ProcessorOne and \ProcessorTwo) involved in the system.
Their communication protocols can be formalised as multiparty session
types. %
The types $\T_{\role{bp1}}$ and $T_{\role{bp2}}$ below provide the specification
of the two batch processors. They are very simple:
\change{\#A: we changed the subscripts for ``source'' and ``sink''}{%
  they recursively read data (of sort $\ST$) from the \Source ($\role{src}$)
  and send the processing result to the \Sink ($\role{sk}$).
}%
\[
  \T_{\role{bp1}} \;=\; \T_{\role{bp2}} \;=\;
  \mu\ty.\tin{\role{src}}{\mathit{data}}{\ST}.
  \tout{\role{sk}}{\mathit{data}}{\ST}.\ty
\]

We now present the types
$\T_{\role{src}}, \T_{\role{ctl}}$, and $\T_{\role{sk}}$,
which provide the specifications of \Source, \Control and \Sink, respectively:
\[
\begin{array}{rcl}
    \T_{\role{src}} &=& \mu\ty.\texternal\big\{
       \tin{\role{ctl}}{\mathit{b}_1}{}.\tout{\role{ctl}}{\mathit{b}_1}{}.\tout{\role{bp_1}}{\mathit{data}}{\ST}.\ty\;,\,
       \tin{\role{ctl}}{\mathit{b}_2}{}.\tout{\role{ctl}}{\mathit{b}_2}{}.\tout{\role{bp_2}}{\mathit{data}}{\ST}.\ty\big\} \\
    \T_{\role{sk}} &=& \mu\ty.\tout{\role{ctl}}{\mathit{b}_1}{}.\tin{\role{ctl}}{\mathit{b}_1}{}.\tin{\role{bp_1}}{\mathit{data}}{\ST}.\tout{\role{ctl}}{\mathit{b}_2}{}.\tin{\role{ctl}}{\mathit{b}_2}{}.\tin{\role{bp_2}}{\mathit{data}}{\ST}.\ty \\
    \T_{\role{ctl}} &=& \mu\ty. \tout{\role{src}}{\mathit{b}_1}{}.\tin{\role{src}}{\mathit{b}_1}{}.\tin{\role{sk}}{\mathit{b}_1}{}.\tout{\role{sk}}{\mathit{b}_1}{}.
    \tout{\role{src}}{\mathit{b}_2}{}.\tin{\role{src}}{\mathit{b}_2}{}.\tin{\role{sk}}{\mathit{b}_2}{}.\tout{\role{sk}}{\mathit{b}_2}{}.\ty 
\end{array}
\]
The above type specifications are summarised as follows:
\begin{itemize}
\item
  the \Source (denoted as participant $\role{src}$) 
  expects to be told by the \Control ($\role{ctl}$)
  where to send the next data --- either to the \ProcessorOne or \ProcessorTwo
  (denoted with message labels $\mathit{b}_1$ and $\mathit{b}_2$, respectively).
  Then, the \Source acknowledges
  (by replying $\mathit{b}_1$ or $\mathit{b}_2$ to the Control)
  sends the data to the corresponding \Processor
  ($\role{bp_1}$ or $\role{bp_2}$), and loops;
\item
  instead, the \Sink (denoted as participant $\role{sk}$)
  notifies the \Control that it is willing to read the output from
  \ProcessorOne (message $\mathit{b}_1$), expects an acknowledgement
  from \Control (with the same message $\mathit{b}_1$),
  and proceeds with reading the data.
  Then, it performs a similar sequence of interactions to
  notify \Control and read data from \ProcessorTwo, and loops;
\item
  finally, the \Control (denoted as $\role{ctl}$)
  regulates the interactions of \Source, \Sink, \ProcessorOne and \ProcessorTwo,
  in a loop where it first
  tells Source to send to \ProcessorOne, and then waits for Sink 
  to be ready to receive from \ProcessorTwo,
  and then repeats the above for \ProcessorTwo.
\end{itemize}

\begin{figure}[tpb]
  \centering
  \subcaptionbox{%
    \footnotesize%
    \Control tells \Source it should send the next batch of data
    to \ProcessorOne.%
    \label{fig:double-buffering-1}%
  }{%
    \includegraphics[width=0.28\textwidth]{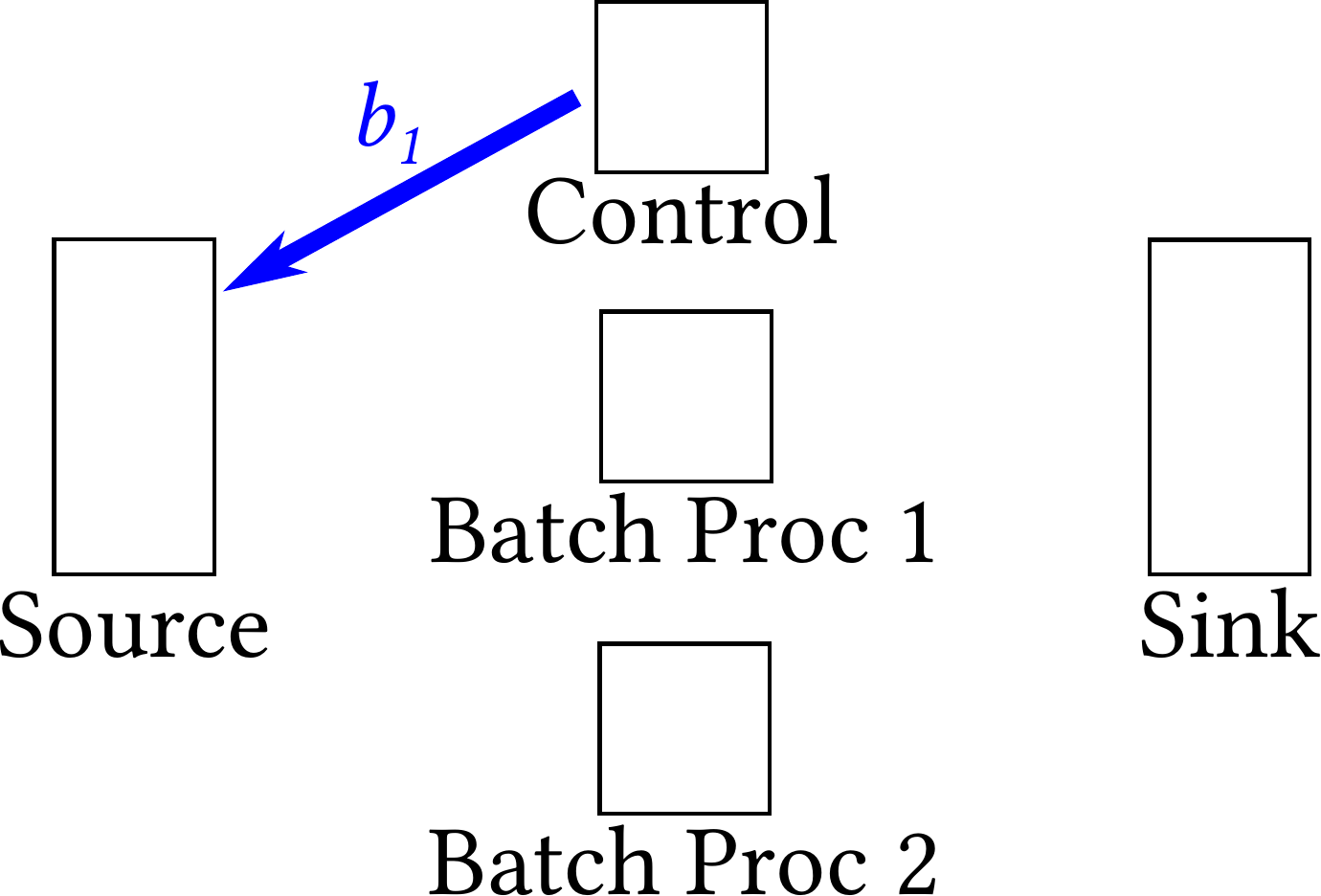}
  }
  \hspace{7mm}
  \subcaptionbox{%
    \footnotesize%
    \Source starts sending data to \ProcessorOne;
    asynchronously, \Control tells \Source it should \emph{also}
    send data to \ProcessorTwo.
    Meanwhile, \Sink asynchronously tells \Control it is ready to receive
    data from \ProcessorOne.%
    \label{fig:double-buffering-2}%
  }{%
    \includegraphics[width=0.28\textwidth]{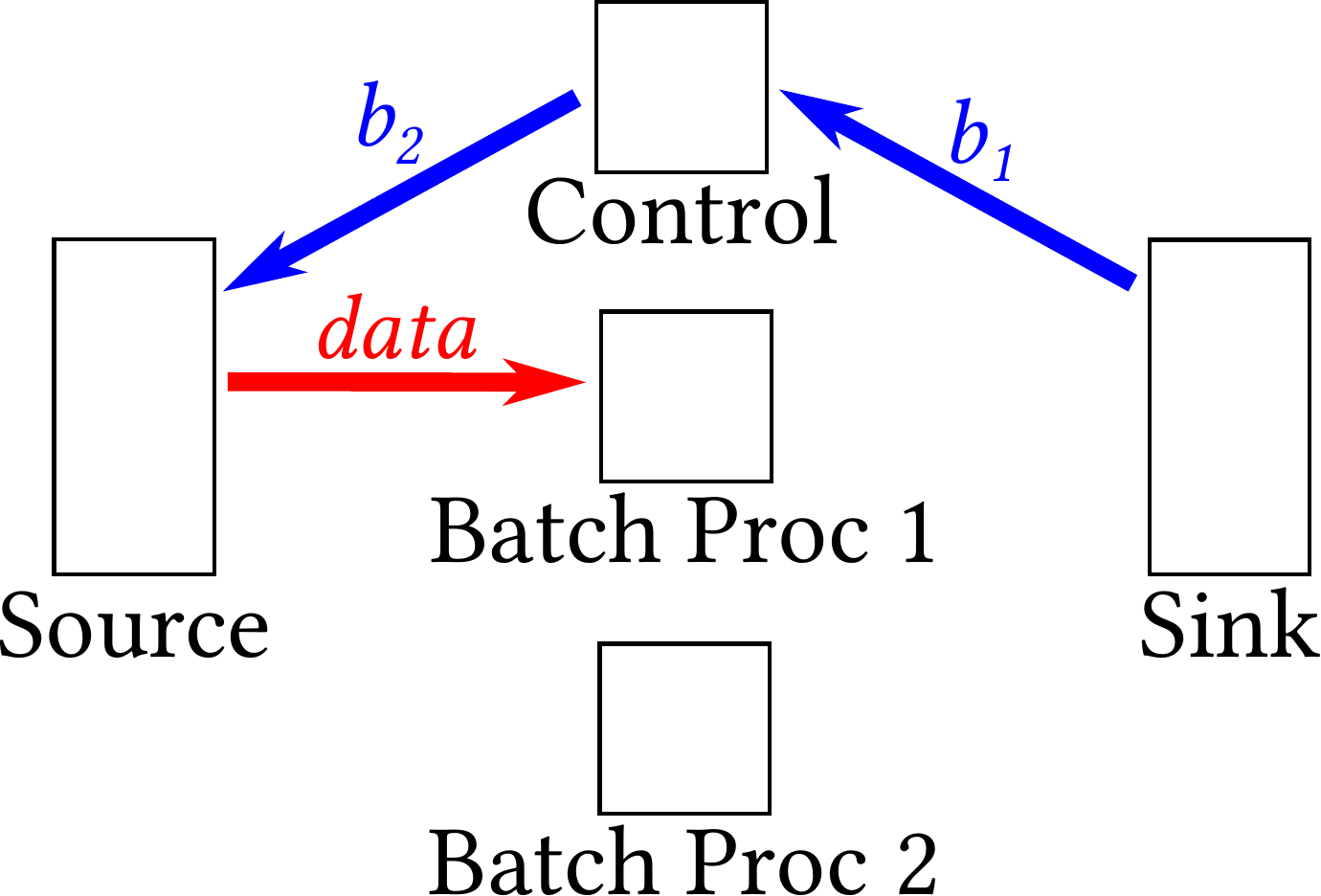}
  }
  \hspace{7mm}
  \subcaptionbox{%
    \footnotesize%
    \ProcessorOne finishes processing its data
    and sends the result to \Sink; meanwhile,
    \Source is sending data to \ProcessorTwo.
    \label{fig:double-buffering-3}%
  }{%
    \includegraphics[width=0.28\textwidth]{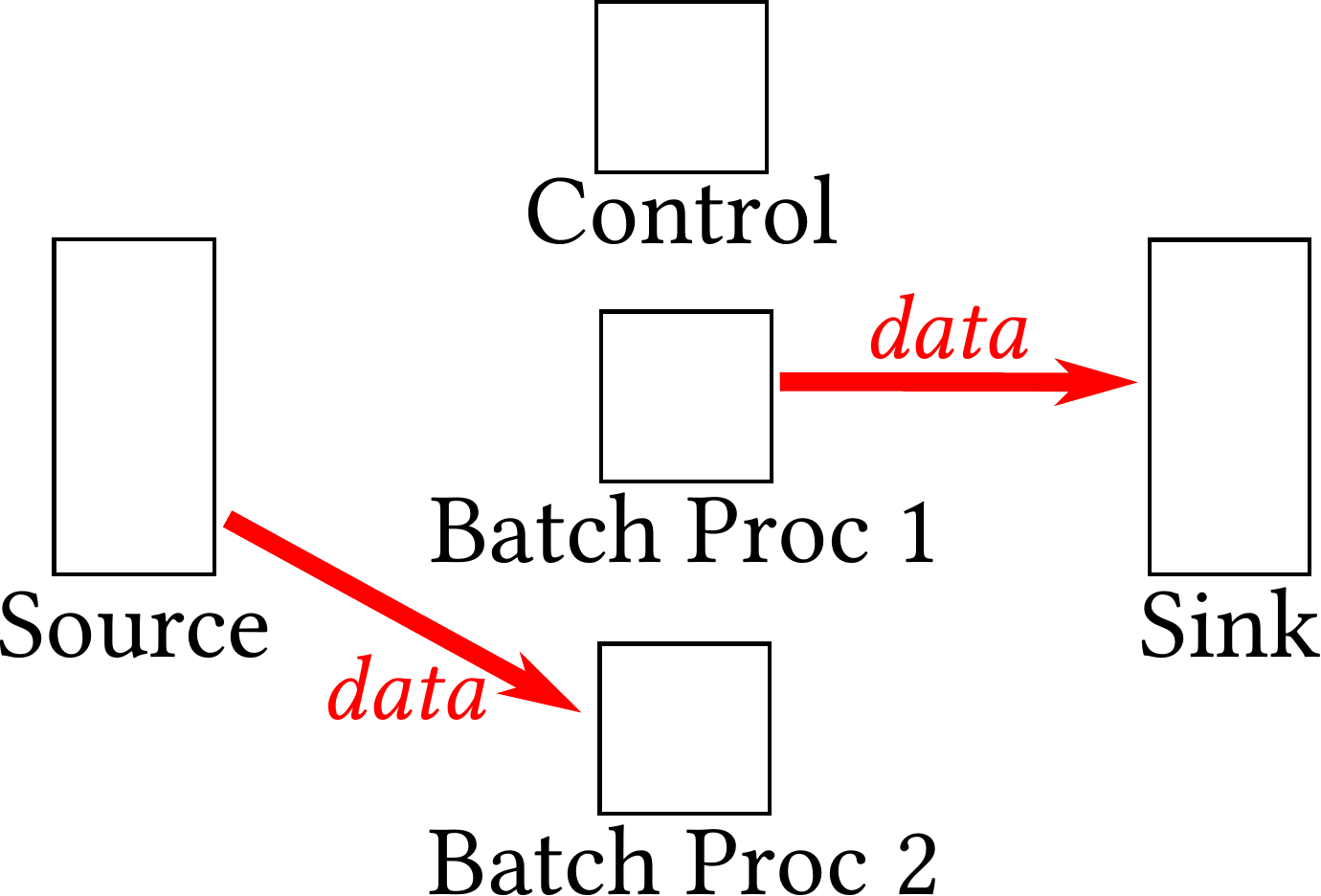}
  }
  \\[3mm]%
  \subcaptionbox{%
    \footnotesize%
    \ProcessorOne has sent its data to \Sink;
    \Source is still sending to \ProcessorTwo;
    \Control asynchronously tells \Source to send
    data to \ProcessorOne; \Sink asynchronously
    tells \Control it can now receive %
    from \ProcessorTwo.
    \label{fig:double-buffering-4}%
  }{%
    \includegraphics[width=0.28\textwidth]{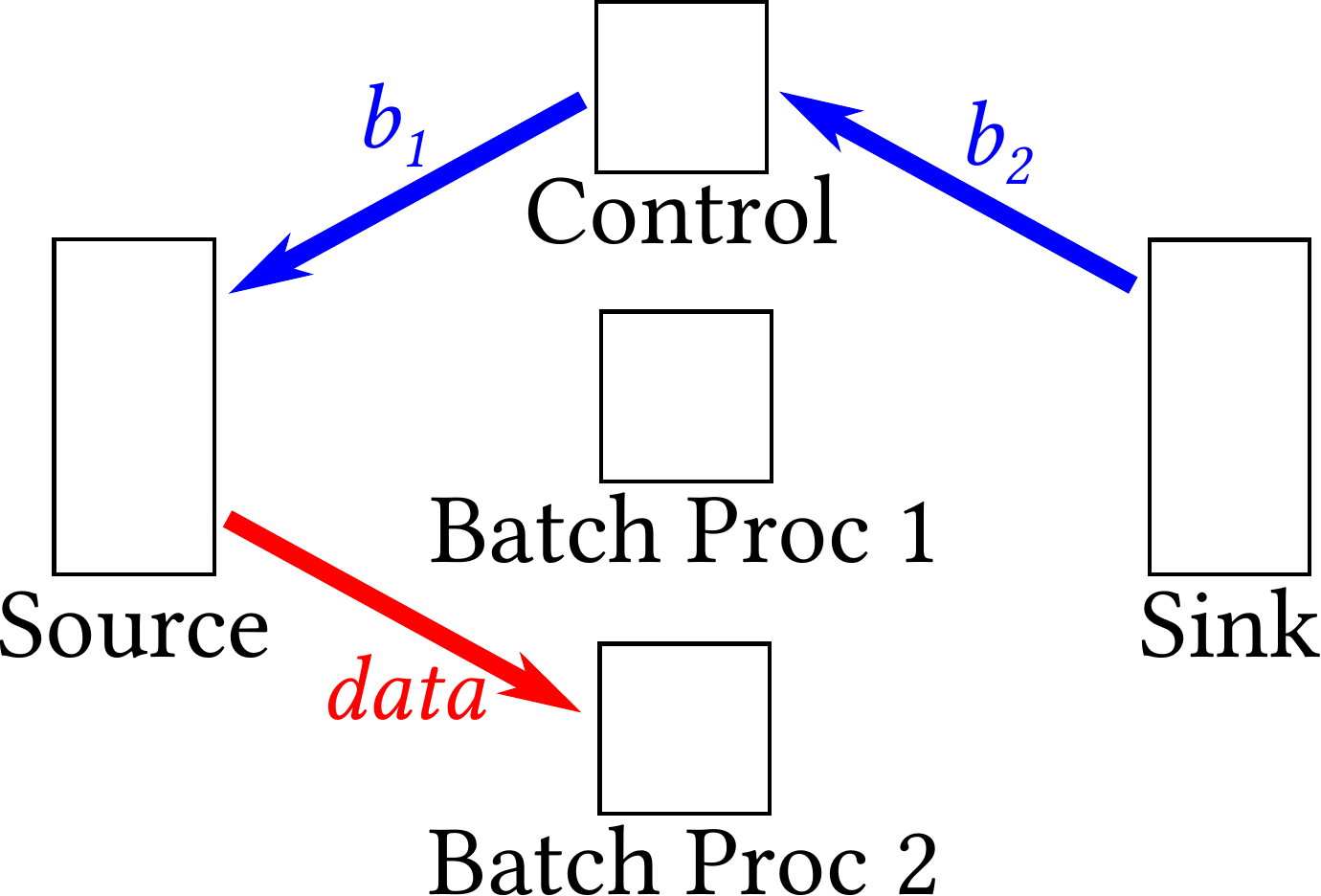}
  }
  \hspace{7mm}
  \subcaptionbox{%
    \footnotesize%
    \Source sends data to \ProcessorOne,
    while \ProcessorTwo is sending its resulting data to Sink.
    After \ProcessorTwo completes the transmission,
    we go back to step (\subref{fig:double-buffering-2}).%
    \label{fig:double-buffering-5}%
  }{%
    \includegraphics[width=0.28\textwidth]{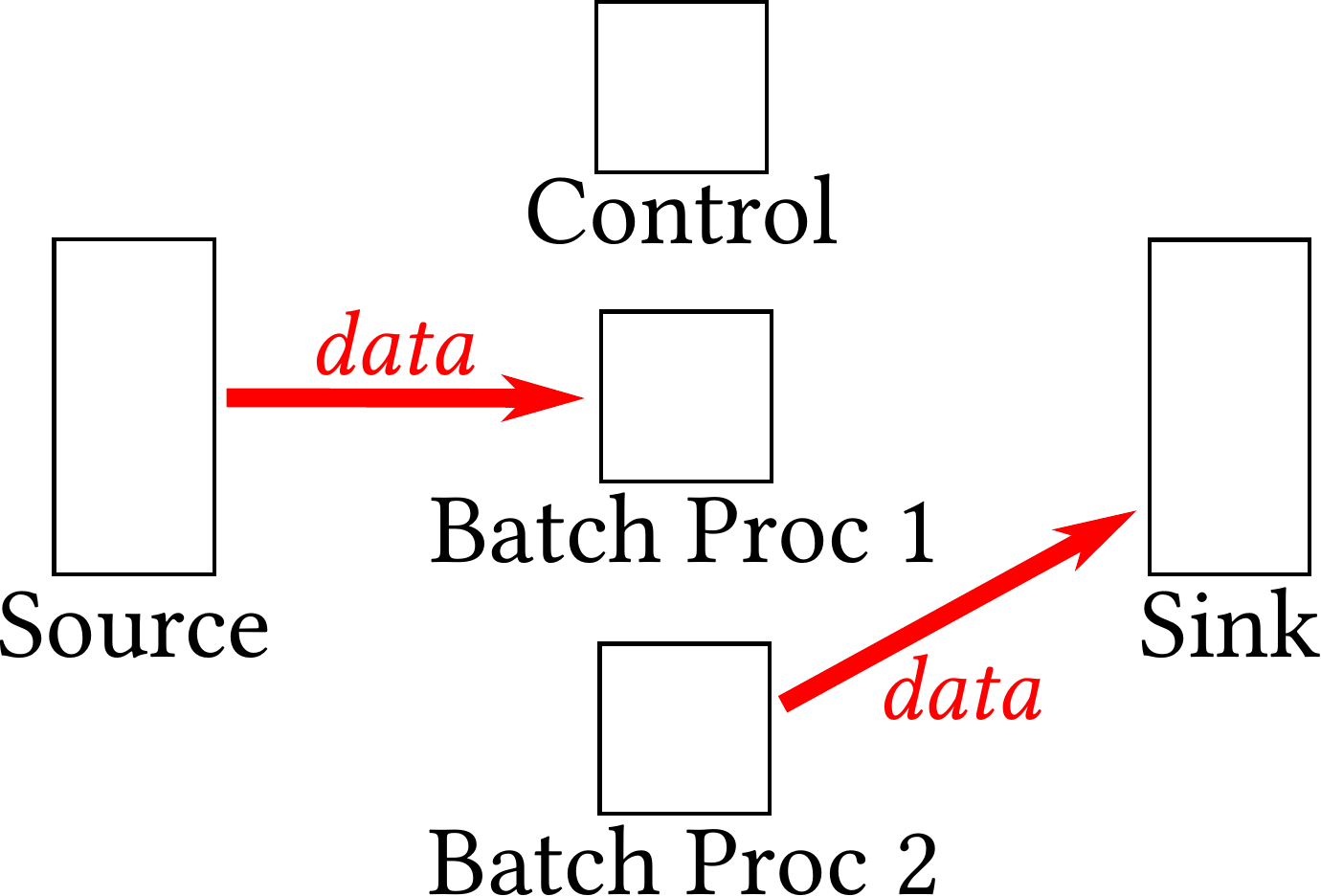}
  }
  \caption{An execution of the optimised distributed batch processing
    protocol. Note that in steps~(\subref{fig:batch-processing-3}) and
    (\subref{fig:batch-processing-5}) the two batch processors are
    active at the same time, and they can perform their computations
    in parallel (depending on the time required by the computation,
    and the time it takes for the source to send its data.)}
  \label{fig:batch-processing-double-buf}
\end{figure}
\vspace{-3mm}

\subsection{Optimised Protocol}
\label{sec:optim-batch-protocol}

The type specifications illustrated in \Cref{sec:basic-batch-protocol}
produce a correct (live) interaction. However, the interaction is also
rather sequential:
\begin{itemize}
\item after \ProcessorOne ($\role{bp_1}$) receives its data, 
  \ProcessorTwo ($\role{bp_2}$) is idle, waiting for data that is only
  made available after \Sink sends a message $\mathit{b}_1$ to \Control;
\item later on, while $\role{bp_2}$ is processing, $\role{bp_1}$
  becomes idle, waiting for data that is only made available after
  \Sink sends a message $\mathit{b}_2$ to \Control.
\end{itemize}
When implementing the system, it can be appropriate to optimise it by
leveraging asynchronous communications: we can aim at a higher degree
of parallelism for $\role{bp_1}$ and $\role{bp_2}$ by making
their input data available earlier --- without
compromising the overall correctness of the system.

One way to achieve such an optimisation is to
take inspiration from the
\emph{double-buffering algorithm}
(a version implemented in C is found in \citeN[Section 3.2]{doublebuffer}),
and implement the following
alternative type $\T^*_{\role{ctl}}$ for the \Control %
that changes the order of some of
the actions of $\T_{\role{ctl}}$:
\[
\T^*_{\role{ctl}} \;=\; \tout{\role{src}}{\mathit{b}_1}{}.\tout{\role{src}}{\mathit{b}_2}{}.\mu\ty. \tin{\role{src}}{\mathit{b}_1}{}.\tin{\role{sk}}{\mathit{b}_1}{}.\tout{\role{sk}}{\mathit{b}_1}{}.\tout{\role{src}}{\mathit{b}_1}{}.
    \tin{\role{src}}{\mathit{b}_2}{}.\tin{\role{sk}}{\mathit{b}_2}{}.\tout{\role{sk}}{\mathit{b}_2}{}.\tout{\role{src}}{\mathit{b}_2}{}.\ty
\]

According to type $\T^*_{\role{ctl}}$,  \Control behaves as follows:
\begin{itemize}
\item first, it signals \Source that it should send data to
  $\role{bp_1}$, then to $\role{bp_2}$,
  and then enters the loop;
\item inside the loop, \Control receives the notifications from Source
  and Sink, on their intention to send/receive data to/from
  $\role{bp_1}$; then, the \Control \emph{immediately notifies Source to
    send again at $\role{bp_1}$}. This refinement allows \Source to
  start sending data earlier, and activate $\role{bp_1}$ immediately.
\end{itemize}
A resulting execution of the protocol is illustrated in
\Cref{fig:batch-processing-double-buf}.

We now want to ensure that $\T^*_{\role{ctl}}$ represents a correct
optimisation of $\T_{\role{ctl}}$ --- \ie, if a process implementing
the former is used in a system where a process implementing the latter is
expected, then the optimisation does not introduce deadlocks, orphan
messages, or starvation errors. To this purpose, we show that
$\T^*_{\role{ctl}}$ is an asynchronous subtype of $\T_{\role{ctl}}$, by
\Cref{def:subtyping}. We take the session trees\;
$\TT^*_{\role{ctl}} = \ttree{\T^*_{\role{ctl}}}$
\;and\; $\TT_{\role{ctl}} = \ttree{\T_{\role{ctl}}}$, and we define:
\[
\begin{array}{c}
\TT_\mathit{R} = \ttree{\T_\mathit{R}}
\\\text{where\; }%
\T_\mathit{R}=\mu\ty. \tin{\role{src}}{\mathit{b}_1}{}.\tin{\role{sk}}{\mathit{b}_1}{}.\tout{\role{sk}}{\mathit{b}_1}{}.\tout{\role{src}}{\mathit{b}_1}{}.\tin{\role{src}}{\mathit{b}_2}{}.\tin{\role{sk}}{\mathit{b}_2}{}.\tout{\role{sk}}{\mathit{b}_2}{}.\tout{\role{src}}{\mathit{b}_2}{}.\ty
\end{array}
\]
Hence, we have $\TT^*_{\role{ctl}}=\tout{\role{src}}{\mathit{b}_1}{}.\tout{\role{src}}{\mathit{b}_2}{}.\TT_\mathit{R}$.
We can now prove $\T^*_{\role{ctl}}\subt\T_{\role{ctl}}$:
since they are both SISO session types, by \Cref{def:subtyping}
it is enough to build a coinductive derivation
showing\; $\TT^*_{\role{ctl}} \subttt \TT_{\role{ctl}}$,
\;with the SISO refinement rules in \Cref{def:ref}. %
By using the abbreviation\;
$\pi=\tin{\role{src}}{\mathit{b}_2}{}.\tin{\role{sk}}{\mathit{b}_2}{}.\tout{\role{sk}}{\mathit{b}_2}{}$, 
\;and \hlight{highlighting} matching pairs of prefixes, %
we have
\change{\#D: clarify coinductive derivations}{%
  the following infinite coinductive derivation:
  (notice that the topmost refinement matches the second from the bottom)
}%
\[
{\footnotesize%
\infer=[\mbox{\rulename{ref-out}}]{%
\TT^*_{\role{ctl}} =
\hlight{$\tout{\role{src}}{\mathit{b}_1}{}$}.\tout{\role{src}}{\mathit{b}_2}{}.\TT_\mathit{R} 
\subttt 
\hlight{$\tout{\role{src}}{\mathit{b}_1}{}$}.\tin{\role{src}}{\mathit{b}_1}{}.\tin{\role{sk}}{\mathit{b}_1}{}.\tout{\role{sk}}{\mathit{b}_1}{}.
    \tout{\role{src}}{\mathit{b}_2}{}.\pi.\TT_{\role{ctl}}
 = \TT_{\role{ctl}}
}
{
\infer=[\mbox{\rulename{ref-$\BC$}}]{
\hlight{$\tout{\role{src}}{\mathit{b}_2}{}$}.\TT_\mathit{R} 
\subttt 
\tin{\role{src}}{\mathit{b}_1}{}.\tin{\role{sk}}{\mathit{b}_1}{}.\tout{\role{sk}}{\mathit{b}_1}{}.
    \hlight{$\tout{\role{src}}{\mathit{b}_2}{}$}.\pi.\TT_{\role{ctl}}
}
{
\infer=[\mbox{\rulename{ref-in}}]{
\hlight{$\tin{\role{src}}{\mathit{b}_1}{}$}.\tin{\role{sk}}{\mathit{b}_1}{}.\tout{\role{sk}}{\mathit{b}_1}{}.\tout{\role{src}}{\mathit{b}_1}{}.\pi.\tout{\role{src}}{\mathit{b}_2}{}.\TT_\mathit{R} 
\subttt 
\hlight{$\tin{\role{src}}{\mathit{b}_1}{}$}.\tin{\role{sk}}{\mathit{b}_1}{}.\tout{\role{sk}}{\mathit{b}_1}{}.\pi.\TT_{\role{ctl}}
}
{
\infer=[\mbox{\rulename{ref-in}}]{
\hlight{$\tin{\role{sk}}{\mathit{b}_1}{}$}.\tout{\role{sk}}{\mathit{b}_1}{}.\tout{\role{src}}{\mathit{b}_1}{}.\pi.\tout{\role{src}}{\mathit{b}_2}{}.\TT_\mathit{R}   
\subttt 
\hlight{$\tin{\role{sk}}{\mathit{b}_1}{}$}.\tout{\role{sk}}{\mathit{b}_1}{}.\pi.\TT_{\role{ctl}}
}
{
\infer=[\mbox{\rulename{ref-out}}]{
\hlight{$\tout{\role{sk}}{\mathit{b}_1}{}$}.\tout{\role{src}}{\mathit{b}_1}{}.\pi.\tout{\role{src}}{\mathit{b}_2}{}.\TT_\mathit{R}
\subttt 
\hlight{$\tout{\role{sk}}{\mathit{b}_1}{}$}.\pi.\TT_{\role{ctl}}
}
{
\infer=[\mbox{\rulename{ref-$\BC$}}]{
\hlight{$\tout{\role{src}}{\mathit{b}_1}{}$}.\pi.\tout{\role{src}}{\mathit{b}_2}{}.\TT_\mathit{R}    
\subttt
\pi.\hlight{$\tout{\role{src}}{\mathit{b}_1}{}$}.\tin{\role{src}}{\mathit{b}_1}{}.\tin{\role{sk}}{\mathit{b}_1}{}.\tout{\role{sk}}{\mathit{b}_1}{}.
    \tout{\role{src}}{\mathit{b}_2}{}.\pi.\TT_{\role{ctl}}
}
{
\infer=[\mbox{\rulename{ref-in}}]{
\hlight{$\tin{\role{src}}{\mathit{b}_2}{}$}.\tin{\role{sk}}{\mathit{b}_2}{}.\tout{\role{sk}}{\mathit{b}_2}{}.\tout{\role{src}}{\mathit{b}_2}{}.\TT_\mathit{R} \subttt 
\hlight{$\tin{\role{src}}{\mathit{b}_2}{}$}.\tin{\role{sk}}{\mathit{b}_2}{}.\tout{\role{sk}}{\mathit{b}_2}{}.\tin{\role{src}}{\mathit{b}_1}{}.\tin{\role{sk}}{\mathit{b}_1}{}.\tout{\role{sk}}{\mathit{b}_1}{}.
\tout{\role{src}}{\mathit{b}_2}{}.\pi.\TT_{\role{ctl}}
}
{
\infer=[\mbox{\rulename{ref-in}}]{
\hlight{$\tin{\role{sk}}{\mathit{b}_2}{}$}.\tout{\role{sk}}{\mathit{b}_2}{}.\tout{\role{src}}{\mathit{b}_2}{}.\TT_\mathit{R}   
\subttt
\hlight{$\tin{\role{sk}}{\mathit{b}_2}{}$}.\tout{\role{sk}}{\mathit{b}_2}{}.\tin{\role{src}}{\mathit{b}_1}{}.\tin{\role{sk}}{\mathit{b}_1}{}.\tout{\role{sk}}{\mathit{b}_1}{}.
    \tout{\role{src}}{\mathit{b}_2}{}.\pi.\TT_{\role{ctl}}
}{
\infer=[\mbox{\rulename{ref-out}}]{
\hlight{$\tout{\role{sk}}{\mathit{b}_2}{}$}.\tout{\role{src}}{\mathit{b}_2}{}.\TT_\mathit{R}   
\subttt
\hlight{$\tout{\role{sk}}{\mathit{b}_2}{}$}.\tin{\role{src}}{\mathit{b}_1}{}.\tin{\role{sk}}{\mathit{b}_1}{}.\tout{\role{sk}}{\mathit{b}_1}{}.
    \tout{\role{src}}{\mathit{b}_2}{}.\pi.\TT_{\role{ctl}}
    }
    {
    \hlight{$\tout{\role{src}}{\mathit{b}_2}{}$}.\TT_\mathit{R} 
\subttt 
\tin{\role{src}}{\mathit{b}_1}{}.\tin{\role{sk}}{\mathit{b}_1}{}.\tout{\role{sk}}{\mathit{b}_1}{}.
    \hlight{$\tout{\role{src}}{\mathit{b}_2}{}$}.\pi.\TT_{\role{ctl}}
    }
}
}
}
}
}
}
}
}
}
\]

Thus, we obtain the following result.

\begin{proposition}%
  \label{ex:lem:subt}%
 $\T^*_{\role{ctl}}\subt\T_{\role{ctl}}$ holds.
\end{proposition}

The types
$\T_{\role{bp1}}, \T_{\role{src}}, \T^*_{\role{ctl}}$ 
and $\T_{\role{sk}}$, can be implemented
as processes for \Processor, \Source, \Control and \Sink, such as:
\[
  \begin{array}{r@{\;\;}c@{\;\;}l}
   \PP_{\role{bp1}} &=&     \PP_{\role{bp2}} =
  \mu X.\procin{{\role{src}}}{\mathit{data}(x)}
  \procout{{\role{sk}}}{\mathit{data}}{f(\x)} X \\
    \PP_{\role{src}} &=& \mu X.\sum\big\{
       \procin{\role{ctl}}{\mathit{b}_1()}
       \procout{\role{ctl}}{\mathit{b}_1}{}
       \tout{\role{bp_1}}{\mathit{data}}{\mathit{datum}}. X\;,\,
       \procin{\role{ctl}}{\mathit{b}_2()}
       \procout{\role{ctl}}{\mathit{b}_2}{}
       \procout{\role{bp_2}}{\mathit{data}}{\mathit{datum}} X\big\} \\
    \PP_{\role{ctl}} &=& 
    \procout{\role{src}}{\mathit{b}_1}{} \procout{\role{src}}{\mathit{b}_2}{} \mu X. \procin{\role{src}}{\mathit{b}_1()} \procin{\role{sk}}{\mathit{b}_1()} \procout{\role{sk}}{\mathit{b}_1}{} \procout{\role{src}}{\mathit{b}_1}{}
    \procin{\role{src}}{\mathit{b}_2()} \procin{\role{sk}}{\mathit{b}_2()} \procout{\role{sk}}{\mathit{b}_2}{} \procout{\role{src}}{\mathit{b}_2}{} X \\
    \PP_{\role{sk}}   &=& \mu X.
    \procout{{\role{ctl}}}{\mathit{b}_1}{}
    \procin{{\role{ctl}}}{\mathit{b}_1()}
    \procin{{\role{bp_1}}}{\mathit{data}(\y_1)}
    \procout{{\role{ctl}}}{\mathit{b}_2}{}
    \procin{{\role{ctl}}}{\mathit{b}_2()}
    \procin{{\role{bp_2}}}{\mathit{data}(\y_2)} X
\end{array}
\]
where  $f$ (used by $ \PP_{\role{bp1}}$ and $ \PP_{\role{bp2}}$) represents processing of the data received by $\role{src}$, 
and where the variables $\x,\y_1$ and $y_2$, and objects $f(\x)$ and $\mathit{datum}$ are all of sort $\ST$.

\begin{proposition}[Correctness of Optimisation]
Let:
\begin{itemize}
\item $\Gamma = 
\role{bp_1}:(\temptyqueue, \T_{\role{bp1}}), 
\role{bp_2}:(\temptyqueue, \T_{\role{bp2}}),
\role{src}:(\temptyqueue, \T_{\role{src}}), 
\role{ctl}:(\temptyqueue, \T_{\role{ctl}}), 
\role{sk}:(\temptyqueue, \T_{\role{sk}})$ \;\;and 
\item 
$
\N = 
\pa{\role{bp_1}} {\PP_{\role{bp1}}} \pc \pa{\role{bp_1}} \emptyqueue \pc
\pa{\role{bp_2}} {\PP_{\role{bp2}}} \pc \pa{\role{bp_2}} \emptyqueue \pc
\pa{\role{src}} {\PP_{\role{src}}} \pc \pa{\role{src}} \emptyqueue \pc 
\pa {\role{ctl}} {\PP_{\role{ctl}}} \pc \pa{\role{ctl}} \emptyqueue \pc 
\pa{\role{sk}} {\PP_{\role{sk}}} \pc \pa{\role{sk}} \emptyqueue
$
\end{itemize}
Then, we have that $\Gamma$ is live and $\Gamma \vdash \N$.
Also, $\N$ will never reduce to error.
\end{proposition}

Observe that $\Gamma$ uses the basic unoptimised protocol
$\T_{\role{ctl}}$ for the \Control, while $\N$ implements the
optimised protocol $\T^*_{\role{ctl}}$: the typing holds by \Cref{ex:lem:subt},
while the absence of errors is consequence of \Cref{thm:error-freedom}.
 
\section{Additional Preciseness Results}
\label{sec:preciseness-extra}

We now use our main result %
(\Cref{lem:preciseness}) to prove two additional results:
our subtyping $\subt$ is \emph{denotationally precise} (\Cref{sec:denotation})
and also \emph{precise \wrt liveness} (\Cref{sec:preciseness-liveness}).

\subsection{Denotational Preciseness of Subtyping}
\label{sec:denotation}

The approach to preciseness of subtyping described in the previous sections
is nowadays dubbed {\em operational preciseness}.
In prior work, the canonical approach to preciseness of subtyping
for a given calculus has been \emph{denotational}:
a type $\T$ is interpreted (with notation $\dlsqb  \T \drsqb$)
as the set that describes the meaning of $\T$, according
to denotations of the expressions (terms, processes, \etc)
of the calculus. 
\Eg, $\lambda$-calculus types are usually interpreted as
subsets of the domains of $\lambda$-models~\citep{bcd83,H83}.
With this approach, subtyping is interpreted as set-theoretic inclusion.

Under this denotational viewpoint,
a subtyping $\subGen$ is sound when\; $\T \subGen \T' \ \text{ \,implies\, }\ \dlsqb  \T \drsqb \subseteq\dlsqb  \T' \drsqb$, and complete when $\dlsqb  \T \drsqb \subseteq\dlsqb  \T' \drsqb \ \text{ \,implies\, }\  \T \subGen \T' $.
Hence, a subtyping $\subGen$ is precise when:
\[
\T \subGen \T' \ \text{ \;if and only if\; }\ \dlsqb  \T \drsqb \subseteq\dlsqb  \T' \drsqb
\]
This notion is nowadays dubbed {\em denotational preciseness}.

We now adapt the denotational approach to our setting.
Let us interpret a session type $\T$ as the set of closed processes typed by $\T$, i.e.
\[%
\dlsqb  \T \drsqb \;=\; \{\PP~ \mid \vphantom{x} \vdash \PP: \T\}
\] %
We now show that our subtyping is denotationally precise. The denotational soundness follows from the subsumption rule \rln{[t-sub]} (\Cref{figure:typesystem}) and the interpretation of subtyping as set-theoretic inclusion. In order to prove that denotational completeness holds, we show a more general result (along the lines of~\citeN{GhilezanJPSY19}): by leveraging characteristic processes and sessions, we show that operational completeness implies denotational completeness.

\begin{theorem}
  The multiparty asynchronous subtyping $\subt$ is denotationally complete.
\end{theorem}
\begin{proof}
We proceed by contradiction.
Suppose that denotational completeness does \emph{not} hold, \ie,
there are $\T$ and $\T'$ such that
$\dlsqb \T \drsqb \subseteq \dlsqb \T' \drsqb$ but $\T \not\subt \T'$. 
\begin{itemize}
\item From $\T \not\subt \T'$, by Step 1 (\Cref{sec:preciseness-step1}, eq.~\eqref{eq:not-subt-def-regular}), there are $\UU$ and $\V'$ satisfying $\UU\not\subt\V'$ and $\ttree{\UU}\in \llbracket \ttree{\T}\rrbracket_\SO$ and $\ttree{\V'}\in \llbracket \ttree{\T'}\rrbracket_\SI$. Then  by Step 4 (\Cref{prop:completeness-maintext}(3)) we have:
   \begin{align}
 \pa\pr\CP{\UU} \pc \pa\pr\EmptyQueue \pc \M_{\pr, \V'} \;\red^*\; \error \label{eqn:dcomp1}
 \end{align}
 where $\M_{\pr, \V'}$, (with $\pr\not\in\participant{\V'}$) is the characteristic  
session of Step 3 (\Cref{def:characteristic-session}) and $\CP{\UU} $ is the characteristic process of $\UU$ of Step
 2 (\Cref{def:characteristic-process}).
\item By Step 2 (\Cref{cp}), the characteristic process of the above $\UU$ is typable by $\T$, \ie, we have $\vdash \CP{\UU}: \T$; therefore, from  $\dlsqb \T \drsqb \subseteq \dlsqb \T' \drsqb$ we also have $\vdash \CP{\UU}: \T'$. Hence, by Step 3 (\Cref{def:characteristic-session}, \Cref{prop:live-ti}) for the above $\V'$ the characteristic session $\M_{\pr,\V'}$ is typable if composed with $\CP{\UU}: \T'$, \ie:
   \begin{align}
\Gamma \vdash \pa\pr\CP{\UU} \pc \pa\pr\EmptyQueue \pc \M_{\pr,\V'} \quad \mbox{for some live} \; \Gamma \label{eqn:dcomp2}
\end{align}
\end{itemize}
But then, \eqref{eqn:dcomp2} and \eqref{eqn:dcomp1}
contradict the soundness of the type system (\Cref{thm:SR}).
Therefore, we conclude that %
for all $\T$ and $\T'$, 
$\dlsqb \T \drsqb \subseteq \dlsqb \T' \drsqb$ implies $\T \subt \T'$,
which is the thesis.
\end{proof}

As a consequence, we obtain the denotational preciseness of subtyping.

\begin{theorem}[Denotational preciseness]\label{thm:dpreciseness}
  The subtyping relation is denotationally precise.
\end{theorem}

\subsection{Preciseness of Subtyping \wrt Liveness}
\label{sec:preciseness-liveness}

Our results also provide a stepping stone to show that
that \emph{our multiparty asynchronous subtyping is precise \wrt liveness} %
(\Cref{def:env-liveness}),
as formalised in \Cref{thm:subt-precise-liveness} below.%

Notably, this result can be lifted to the realm of
Communicating Finite-State Machines (CFSMs) \cite{Brand:1983:CFSM},
by the following interpretation. %
\emph{Communicating session automata (CSA)}
(introduced by \citeN{DY12}) %
are a class of CFSMs
with a finite control based on internal/external choices,
corresponding to a session type; %
a typing environment $\Gamma$ corresponds to a system of CSA,
and reduces with the same semantics (\Cref{def:typing-env-reductions}).
Therefore, our notion of typing environment liveness (\Cref{def:env-liveness})
guarantees deadlock-freedom, orphan message-freedom, and starvation-freedom
of the corresponding system of CSA %
--- similarly to the notion of \emph{$\infty$-multiparty compatibility}
between CSA, by \citeN{LangeY19}. %
Correspondingly, by \Cref{thm:subt-precise-liveness},
our subtyping $\subt$ is a sound and complete preorder %
allowing to refine any live system of CSA
(by replacing one or more automata), %
while preserving the overall liveness of the system.

To state the result, we adopt the notation\;
$\Gamma\mapUpdate{\pp}{(\tqueue,\T)}$ %
\;to represent the typing environment obtained from $\Gamma$ %
by replacing the entry for $\pp$ with $(\tqueue,\T)$.

\begin{theorem}
  \label{thm:subt-precise-liveness}%
  For all session types $\T$ and $\T'$, %
  multiparty asynchronous subtyping $\subt$ is:
  \begin{description}
  \item[sound \wrt liveness:]%
    if\, $\T \subt \T'$,\; %
    then $\forall\Gamma$ with $\Gamma(\pr) \!=\! (\temptyqueue,\!\T')$, %
    when\, $\Gamma$ is live, %
    $\Gamma\mapUpdate{\pr\!}{\!(\temptyqueue,\!\T)}$ is live;%
  \item[complete \wrt liveness:]%
    if\, $\T \not\subt \T'$,\; %
    then $\exists \Gamma$ live %
    with $\Gamma(\pr) \!=\! (\temptyqueue,\!\T')$, %
    but $\Gamma\mapUpdate{\pr\!}{\!(\temptyqueue,\!\T)}$ is not live.%
  \end{description}
\end{theorem}

\begin{proof}%
Soundness of $\subt$ follows by \Cref{lem:subtyping-preserves-liveness}.
Completeness, instead, %
descends from \Cref{prop:completeness-maintext}: %
for any $\T,\T'$ such that $\T \!\not\subt\! \T'$, %
we can build a live typing context $\Gamma$ %
(see \eqref{eq:characteristic-gamma}) %
with $\Gamma(\pr) \!=\! (\temptyqueue, \T')$ %
and cyclic communications %
(item~\ref{item:prop:completeness-gamma}). %
Observe that the environment %
$\Gamma\mapUpdate{\pr}{(\temptyqueue, \T)}$ %
types the session of %
\Cref{prop:completeness-maintext}(\ref{item:prop:completeness-error}), %
that reduces to $\error$: %
hence, by the contrapositive of \Cref{thm:error-freedom}, %
we conclude that $\Gamma\mapUpdate{\pr}{(\temptyqueue, \T)}$ is not live.
\end{proof}
 
\section{Related Work, Future Work, and Conclusion}\label{sec:related}
\paragraph{Precise Subtyping for $\lambda$-Calculus and Semantic Subtyping}
The notion of preciseness has been adopted
as a criterion to justify the canonicity of subtyping relations,
in the context of both functional and concurrent calculi. 
Operational preciseness of subtyping %
was first introduced by \citeN{BHLN12} (and later published
in \cite{BHLN17}), %
and applied to $\lambda$-calculus with iso-recursive types. %
Later, \citeN{DG14} adapted the idea of \citeN{BHLN12} %
to the setting of the concurrent $\lambda$-calculus %
with intersection and union types by \citeN{SIAM}, and proved both operational
and denotational preciseness. %
In the context of the $\lambda$-calculus, 
a similar framework, \emph{semantic subtyping}, was proposed by
\citeN{CF05}: each type $T$ is interpreted as the
set of values having type $T$, %
and subtyping is defined as subset inclusion between type interpretations. %
This gives a precise subtyping as long as the calculus 
allows to operationally distinguish values of different types. 
Semantic subtyping was also studied
by~\citeN{CNV08} (for a $\pi$-calculus with a patterned input and IO-types), %
and by \citeN{castagna09foundations} %
(for a $\pi$-calculus with binary session types); %
in both works, types include %
union, intersection and negation.
Semantic subtyping is precise for the
calculi of \cite{CNV08,castagna09foundations,FCB08}: %
this is due to the type case constructor in~\cite{FCB08}, %
and to the blocking of inputs for values of ``wrong'' types %
in~\cite{CNV08,castagna09foundations}.

\paragraph{Precise Subtyping for Session Types}
In the context of \emph{binary} session types, %
the first general formulation of precise subtyping %
(synchronous and asynchronous) %
is given by
\citeN{%
CDSY2017}, %
for a $\pi$-calculus %
typed by assigning session types to channels %
(as in the work by \citeN{HVK}). %
The first result by 
\citeN{%
CDSY2017} is that the well-known
branching-selection subtyping \cite{GH05,DemangeonH11} %
is sound and complete for the \emph{synchronous} binary session $\pi$- calculus.
\citeN{%
CDSY2017} also examine an \emph{asynchronous} %
binary session $\pi$-calculus, %
and introduce a subtyping relation %
(restricting the subtyping %
for the higher-order $\pi$-calculus by \citeN{MostrousY15}) %
that is also proved precise. %
\change{Answer B.1}{
  Our results imply that the binary session subtyping of \citeN{CDSY2017}
  (without delegation) is a subset of our multiparty subtyping, hence
  the subtyping of \citeN{CDSY2017} is
  expressible via tree decomposition.
  We compare our subtyping and theirs
  (specifically, the use $n$-hole contexts)
  in \Cref{remark:prefix-vs-n-hole-ctx};
  further, we prove \Cref{ex:CONCUR19-maintext} via tree decomposition
  --- and such a proof would be more difficult
  with the rules of \citeN{CDSY2017}.
}

In the context of \emph{multiparty} session types, %
\citeN{GhilezanJPSY19}
adopt an approach similar to \citeN{%
CDSY2017} %
to prove that the \emph{synchronous} multiparty extension %
of binary session subtyping \cite{GH05} %
is sound and complete, hence precise.
\change{Answer B.2}{
  Our subtyping becomes equivalent to \citeN{GhilezanJPSY19} by removing rules
  $\rulename{ref-$\AC$}$ and $\rulename{ref-$\BC$}$ 
  from \Cref{def:ref};
  however, such a formulation would not be much simpler than the one by
  \citeN{GhilezanJPSY19}, nor would simplify its algorithm:
  lacking asynchrony, their subtyping is already quite simple, and decidable.
}

An \emph{asynchronous} subtyping relation for multiparty session types %
was proposed by \citeN{mostrous_yoshida_honda_esop09};
notably, the paper claims that the relation is decidable,
but this was later disproved by \citeN{BravettiCZ18}
(see next paragraph). %
The main crucial difference between our subtyping
and the one by \citeN{mostrous_yoshida_honda_esop09}
is that the latter is \emph{unsound} in our setting:
it does not guarantee orphan message freedom,
hence it accepts processes that reduce to $\error$
by rule \rulename{err-ophn} in \Cref{tab:error}
(see the counterexamples in \citeN[page 46]{CDSY2017},
also applicable in the multiparty setting).
A recent work by \citeN{Horne20}
introduces a novel formalisation of subtyping
for \emph{synchronous} multiparty session types
equipped with parallel composition.
Horne's work does not address the problem of precise subtyping,
and its extension to asynchronous communication seems non-trivial:
in fact, the subtyping appears to be unsound in our setting, %
because (as in \citeN{mostrous_yoshida_honda_esop09})
it focuses on deadlock-freedom, and does not address, \eg,
orphan message freedom.
Studying such extensions and issues in the setting of \citeN{Horne20}
can lead to intriguing future work.

\paragraph{Session Types and Communicating Automata}
Asynchronous session subtyping was shown to be undecidable, even for
binary sessions, by \citeN{LangeY17} and \citeN{BravettiCZ18},
by leveraging a correspondence between session types and
communicating automata theories
--- a correspondence first established by \citeN{DY12} %
with the notion of \emph{session automata}
based on Communicating Finite-State Machines \cite{Brand:1983:CFSM}.
This undecidability result prompted the research on
various limited classes of (binary)
session types for which asynchronous subtyping is decidable
\cite{BravettiCLYZ19,BravettiCZ18,LangeY17}:
for example, if we have binary session types without choices
(which, in our setting, corresponds to single-input, single-output types
interacting with a single other participant),
then asynchronous subtyping becomes decidable. 
-%
The aim of our paper 
is \emph{not} finding a decidable approximation %
of asynchronous \emph{multiparty} session subtyping, %
but finding a canonical, \emph{precise} subtyping. 
Interestingly, our SISO decomposition technique leads to: %
\begin{enumerate}[label=\emph{(\arabic*)}]
\item%
  intuitive but general %
  refinement rules (see Example~\ref{ex:CONCUR19-maintext}, %
  where $\subt$ proves an example %
  not supported by the subtyping algorithm in
  \cite{BravettiCLYZ19}); and %
\item%
  preciseness of $\subt$ \wrt
  liveness (Theorem~\ref{thm:subt-precise-liveness})
  which is directly usable to define 
  the precise multiparty asynchronous refinement relation
  \wrt liveness in communicating session automata \cite{DY12,DY13}.%
\end{enumerate}

\paragraph{Future Work}
We plan to investigate
precise subtyping for richer multiparty session $\pi$-calculi, %
\eg with multiple session initiations and delegation,
\change{Answers C.2, D.1, \#A (session creation)}{
  with a formalisation similar to \citeN{ScalasY19}.
  We expect that, even in this richer setting,
  our subtyping relation will remain the same,
  and its preciseness will be proved in the same way
  --- except that session types will support session types as payloads,
  and Table~{\ref{tab:negationW}}
  and Def.~{\ref{def:characteristic-process}} will have more cases.
  This can be observed in the work by {\citeN{CDSY2017}}: they support multiple sessions, process spawning, delegation --- and yet, such features don't impact the precise subtyping relation (except for having session types as payloads).%
}%

\change{Answers A.2, C.1}{
  We also plan to study the precise subtyping for
  typing environment properties
  other than our liveness ({Def.~\ref{def:env-liveness}}).
  For example, we may consider a weaker deadlock freedom property,
  which can be sufficient for correct session typing
  (as shown by \citeN{ScalasY19})
  as long as the session calculus does {\emph{not}} include error reductions
  in case of starvation or orphan messages
  (rules {$\rulename{err-strv}$} and {$\rulename{err-ophn}$}
  in {\Cref{tab:error}}).
  We expect that, if we lift the liveness requirement,
  the resulting precise subtyping relation will be
  somewhat surprising and counter-intuitive, and harder to formulate.
  For example: due to the lack of orphan message errors,
  subtyping should allow processes forget their inputs in some situations,
  hence we should have, {\eg},\; {$\tend \subt \tin\pp{\ell}{\tnat}.\tend$}.%
}%

\change{Answers C.3, E.1, \#A (type checking decidability)}{
  Another compelling topic for future work
  is finding non-trivial decidable approximations of our multiparty
  asynchronous subtyping relation.  As remarked above,
  the relation is inherently undecidable
  --- hence type-checking is also undecidable,
  because rule $\rulename{t-sub}$ (\Cref{figure:typesystem}) is undecidable.
  Our relation is trivially decidable if restricted to
  non-recursive types.
  The known decidable fragments of binary asynchronous subtyping
  \cite{BravettiCLYZ19,BravettiCZ18,LangeY17} are sound \wrt our relation
  --- but they are also limited to two-party sessions,
  and we suspect that they may become undecidable if naively
  generalised to multiparty sessions.
  If a decidable approximation of asynchronous subtyping is adopted
  for rule $\rulename{t-sub}$ (\Cref{figure:typesystem}),
  then type-chcecking becomes decidable.
}

\paragraph{Conclusion} %
Unlike this paper, %
no other published work addresses precise \emph{asynchronous} %
\emph{multiparty} session subtyping. %
A main challenge was the exact formalisation of the subtyping itself, %
which must satisfy many \emph{desiderata}: %
it must capture a wide variety of input/output
reorderings %
performed by different participants, 
without being too strict %
(otherwise, completeness is lost) %
nor too lax (otherwise, soundness is lost); %
moreover, its definition must not be overly complex to understand, %
and tractable in proofs. %
We achieved these \emph{desiderata} with our novel approach, %
based on SISO tree decomposition and refinement, 
which yields a simpler subtyping definition than 
\cite{%
CDSY2017} (see Remark~\ref{rem:context}).  
Moreover, 
our results are much more general than \cite{GhilezanJPSY19}: %
by using live typing environments (Def.~\ref{def:env-liveness}), %
we are not limited to sessions that match some global type; %
our results are also stronger, %
as we prove soundness \wrt a wider range of errors %
(see Table~\ref{tab:error}). %
\begin{acks}                            %
  We thank the POPL reviewers for their insightful comments and suggestions,
  Mariangiola Dezani-Ciancaglini for her initial collaboration,
  and Simon Castellan for the fruitful discussion.
  This work was supported by:
  \grantsponsor{EUH2020}{EU Horizon 2020}{https://ec.europa.eu/programmes/horizon2020/en}
  project \grantnum{EUH2020}{830929}
  (``\href{https://cybersec4europe.eu}{CyberSec4Europe}'');
  \grantsponsor{SP4670}{EU COST}{https://www.cost.eu/} Actions %
  \grantnum{SP4670}{CA15123} %
  (``\href{https://eutypes.cs.ru.nl/}{EUTypes}'')
  and \grantnum{SP4670}{IC1201} %
  (``\href{http://www.dcs.gla.ac.uk/research/betty/www.behavioural-types.eu/}{BETTY}'');
  EPSRC EP/T006544/1, EP/K011715/1, EP/K034413/1, EP/L00058X/1, EP/N027833/1,
  EP/N028201/1, EP/T006544/1, EP/T014709/1 and EP/V000462/1, and NCSS/EPSRC VeTSS; 
    \grantsponsor{MPNTR}{MPNTR}{http://www.mpn.gov.rs/?lng=lat} and
     \grantsponsor{SFRS}{SFRS}{http://fondzanauku.gov.rs/?lang=en} 
      $\#$\grantnum{SFRS}{6526707}
        (``\href{http://www.mi.sanu.ac.rs/novi_sajt/research/projects/AI4TrustBC.php}{AI4TrustBC}''). 
\end{acks}

\bibliographystyle{ACM-Reference-Format}
\bibliography{session}

\newpage 
\appendix
\section{Appendix of Section~\ref{sec:tsu}}
\label{app:types}

\label{sec:app:transitivity}

We say that a binary relation $\Rel$ over single-input-single-output trees is a tree simulation if it complies with the rules given in Definition~\ref{tab:ref}, i.e., if for every $(\WT_1,\WT_2)\in\Rel$ there is a rule with $\WT_1\subttt\WT_2$ in its conclusion and it holds $(\WT'_1,\WT'_2)\in \Rel$ if  $\WT'_1\subttt \WT'_2$  is in the premise of the rule. It is required that all other premises hold as well.

\begin{lemma}
If $\Rel {\subseteq} \subttt$ is a tree simulation and  $(\WT,\WT')\in \Rel$ then $\actions{\WT}=\actions{\WT'}$.
\end{lemma}

\begin{lemma}\label{lemm:trans_contexts_2}
\begin{enumerate}
\item If $\BContextt{\pp}{} \not=\BContextt{\pq}{}$ and there are $\WT_1$ and $\WT_2$ such that 
$\WT= \BCont{\pp}{}{\WT_1}$ and 
$\WT= \BCont{\pq}{}{\WT_2}$, then one of the following holds:
\begin{enumerate}
\item   
$\WT=\BCont{\pp}{}{\BCont{\pq}{1}{\WT_2}}$,
where $\BContextt{\pq}{}=\BCont{\pp}{}{\BContextt{\pq}{1}}$ and $\WT_1=\BCont{\pq}{1}{\WT_2}$;
\item 
$\WT=\BCont{\pq}{}{\BCont{\pp}{1}{\WT_1}}$,
where $\BContextt{\pp}{}=\BCont{\pq}{}{\BContextt{\pp}{1}}$ and $\WT_2=\BCont{\pp}{1}{\WT_1}$.
\end{enumerate}
\item If $\AContextt{\pp}{} \not=\AContextt{\pq}{}$ and there are $\WT_1$ and $\WT_2$ such that
$\WT= \ACont{\pp}{}{\WT_1}$ and 
$\WT=\ACont{\pq}{}{\WT_2}$, then one of the following holds:
\begin{enumerate}
\item  
$\WT=\ACont{\pp}{}{\ACont{\pq}{1}{\WT_2}}$,
where $\AContextt{\pq}{}=\ACont{\pp}{}{\AContextt{\pq}{1}}$ and $\WT_1=\ACont{\pq}{1}{\WT_2}$;
\item 
$\WT=\ACont{\pq}{}{\ACont{\pp}{1}{\WT_1}}$, 
where $\AContextt{\pp}{}=\ACont{\pq}{}{\AContextt{\pp}{1}}$ and $\WT_2=\ACont{\pp}{1}{\WT_1}$.
\end{enumerate}
\item If $\WT=\BCont{\pp}{}{\WT_1}$ and $\WT=\ACont{\pq}{}{\WT_2}$, for some $\WT_1$ and $\WT_2$,
where $\BContextt{\pp}{}$ is not an {I-sequence}, then  $\WT=\ACont{\pq}{}{\BCont{\pp}{1}{\WT_1}}$, 
where $\BContextt{\pp}{}=\AContextt{\pq}{}.{\BContextt{\pp}{1}}$ and $\WT_2=\BCont{\pp}{1}{\WT_1}$.
\end{enumerate}
\end{lemma}

\begin{lemma}\label{lemm:trans_supertype_of_A/BContexts}
Let $\Rel {\subseteq} \subttt$ be a tree simulation.
\begin{enumerate}
\item If $(\BCont{\pp}{}{\tout\pp{\ell}{\ST}.\WT},\WT')\in \Rel$ then
     \[
     \WT'= \BCont{\pp}{1}{\tout\pp{\ell}{\ST_1}.\WT_1} \text{ or }  \WT'= \tout\pp{\ell}{\ST_1}.\WT_1,  
     \text{ where }  \ST \subs \ST_1.
     \]
\item If $(\ACont{\pp}{}{\tin\pp{\ell}{\ST}.\WT},\WT')\in\Rel$ then 
     \[
     \WT'= \ACont{\pp}{1}{\tin\pp{\ell}{\ST_1}.\WT_1} \text{ or } \WT'= \tin\pp{\ell}{\ST_1}.\WT_1, 
     \text{ where } \ST_1 \subs \ST.
     \]
\end{enumerate}
\end{lemma}
\begin{proof}
\begin{enumerate}
\item
The proof is by induction on the structure of context $\BContextt{\pp}{}$. The basis step is included in the induction step if we notice that the lemma holds  by definition of $\subttt$ for $(\tout\pp{\ell}{\ST}.\WT,\WT')\in \Rel.$ 
For the inductive step, we distinguish two cases:
	\begin{enumerate}
	\item Let $\BContextt{\pp}{}=\tout\pq{\ell'}{\ST'}.\BContextt{\pp}{2}$ and $\pp\neq \pq.$

	Then, $(\tout\pq{\ell'}{\ST'}.\BCont{\pp}{2}{\tout\pp{\ell}{\ST}.\WT},\WT')\in \Rel$ 
	could be derived by rules \rulename{ref-out} or \rulename{ref-${\BC}$}. 
		\begin{enumerate}
		
		\item If  the rule applied is \rulename{ref-out} then we have 
		\[
		\WT'=\tout\pq{\ell'}{\ST''}.\WT''  \text{ and }  
		\ST' \subs \ST'' \text{ and } (\BCont{\pp}{2}{\tout\pp{\ell}{\ST}.\WT},\WT'')\in\Rel.
		\]
		By induction hypothesis  
		\[
		\WT''= \BCont{\pp}{1}{\tout\pp{\ell}{\ST_1}.\WT_1} \text{ or }
		\WT''= \tout\pp{\ell}{\ST_1}.\WT_1,\text{ where } \ST \subs \ST_1.
		\] 
		Both cases follow directly.

		\item If the rule applied is \rulename{ref-${\BC}$} then we have 
		\[
		\WT'= \BCont{\pq}{1}{\tout\pq{\ell'}{\ST''}.\WT''} \text{ and } \ST' \subs \ST'' 
		\text{ and } (\BCont{\pp}{2}{\tout\pp{\ell}{\ST}.\WT},\BCont{\pq}{1}{\WT''})\in\Rel.
		\] 
		By induction hypothesis one of the following holds:  
			\begin{enumerate}
			\item 
			If $\BCont{\pq}{1}{\WT''} = \BCont{\pp}{1}{\tout\pp{\ell}{\ST_1}.\WT_1}$ 
			then, 
			\begin{itemize}
			\item either  $\BContextt{\pq}{1}= \BContextt{\pp}{1} \text{ and } \WT''=\tout\pp{\ell}{\ST_1}.\WT_1  \text{ and }$  \\
				 $\begin{array}{l}				
				 \qquad\qquad
				 \WT'= \BContextt{\pq}{1}.\tout\pq{\ell'}{\ST''}.\tout\pp{\ell}{\ST_1}.\WT_1,
				 \end{array}$
                          \item or, if $\BContextt{\pq}{1}\not= \BContextt{\pp}{1}$, by Lemma~\ref{lemm:trans_contexts_2} \\
                                  $\begin{array}{l}
				 \BCont{\pq}{1}{\WT''}= \BCont{\pq}{2}{\BCont{\pp}{3}{\tout\pp{\ell}{\ST_1}.\WT_1}},
				 \text{ where } 
				 \BCont{\pq}{2}{\BContextt{\pp}{3}}=\BContextt{\pp}{1} \text{ and }\\
				 \qquad \qquad \WT'= \BCont{\pq}{2}{\BCont{\pp}{3}{\tout\pp{\ell}{\ST_1}.\WT_1}},
				 \text{ or } \\
				 \BCont{\pq}{1}{\WT''}= \BCont{\pp}{1}{\tout\pp{\ell}{\ST_1}.\BCont{\pq}{3}{\WT''}},
				 \text{ where } 
				 \BContextt{\pq}{1}=\BCont{\pp}{1}{\tout\pp{\ell}{\ST_1}.\BContextt{\pq}{3}} \text{ and }\\
				 \qquad\qquad \WT' = \BCont{\pp}{1}{\tout\pp{\ell}{\ST_1}.\BContextt{\pq}{3}}.\tout\pq{\ell'}{\ST''}.\WT'' .
                                   \end{array}$
	                   \end{itemize}	

			\item If $\BCont{\pq}{1}{\WT''} = \tout\pp{\ell}{\ST_1}.\WT_1$, where $\ST \subs \ST_1$, 
			then, either 
			\[
			\BContextt{\pq}{1}=\tout\pp{\ell}{\ST_1} \text{ and } \WT''=\WT_1
			\]
			or, by Lemma~\ref{lemm:trans_contexts_2} 
			\[
			\BCont{\pq}{1}{\WT''} = \tout\pp{\ell}{\ST_1}.\BCont{\pq}{2}{\WT''}, \text{ where } 
			\BContextt{\pq}{1} =\tout\pp{\ell}{\ST_1}.\BContextt{\pq}{2}. 
			\]
                         In both cases the proof follows directly.
			\end{enumerate}
		\end{enumerate}

	\item Let $\BContextt{\pp}{}=\tin\pq{\ell'}{\ST'}.\BContextt{\pp}{2}$.
	Then, $(\tin\pq{\ell'}{\ST'}.\BCont{\pp}{2}{\tout\pp{\ell}{\ST}.\WT}, \WT')\in\Rel$ 
	could be derived by rules \rulename{ref-in} or \rulename{ref-${\AC}$}.

		\begin{enumerate}
		\item If \rulename{ref-in} is applied then we get this case by the same reasoning 
		as in the first part of the proof.
		\item If the rule applied is \rulename{ref-${\AC}$} then 
		\[
		\WT'= \ACont{\pq}{1}{\tin\pq{\ell'}{\ST''}.\WT''} \text{ and } \ST'' \subs \ST' \text{ and }
		(\BCont{\pp}{2}{\tout\pp{\ell}{\ST}.\WT},\ACont{\pq}{1}{\WT''})\in\Rel.
		\] 
		By induction hypothesis we have only  the case
		$\ACont{\pq}{1}{\WT''} = \BCont{\pp}{1}{\tout\pp{\ell}{\ST_1}.\WT_1}$, 
		where $\ST \subs \ST_1$, since the case 
		$\ACont{\pq}{1}{\WT''} = \tout\pp{\ell}{\ST_1}.\WT_1$ is not possible. 
		By Lemma~\ref{lemm:trans_contexts_2} we have  
		$\ACont{\pq}{1}{\WT''} = \ACont{\pq}{1}{\BCont{\pp}{3}{\tout\pp{\ell}{\ST_1}.\WT_1}}$, 
		where $\ACont{\pq}{1}{\BContextt{\pp}{3}}= \BContextt{\pp}{1}$. 
		\end{enumerate}
	\end{enumerate}
\item The proof is by induction on the structure of context $\AContextt{\pp}{}$ and follows by similar reasoning.
\end{enumerate}
\end{proof}
By definition of $\subttt$ we consider related pairs by peeling left-hand side trees from left to right, i.e. by matching and eliminating always the leftmost tree. The proof of transitivity requires also to consider actions that are somewhere in the middle of the left-hand side tree. For that purpose, we associate each binary relation $\Rel$ over SISO  trees with its extension $\Rel^+$ as follows:
\[
\begin{array}{lclcl}
  \Rel^+ &   =    &  \Rel & \cup & \{(\BCont{\pp}{}{\WT},\BCont{\pp}{1}{\WT_1}): (\BCont{\pp}{}{\tout\pp{\ell}{\ST}.\WT},\BCont{\pp}{1}{\tout\pp{\ell}{\ST_1}.\WT_1})\in\Rel\}\\
             &         &          & \cup & \{(\ACont{\pp}{}{\WT},\ACont{\pp}{1}{\WT_1}): (\ACont{\pp}{}{\tin\pp{\ell}{\ST}.\WT}, \ACont{\pp}{1}{\tin\pp{\ell}{\ST_1}.\WT_1})\in\Rel\}\\
             &         &          & \cup & \{(\BCont{\pp}{}{\WT},\WT_1): (\BCont{\pp}{}{\tout\pp{\ell}{\ST}.\WT},\tout\pp{\ell}{\ST_1}.\WT_1)\in\Rel\}\\
             &         &          & \cup & \{(\ACont{\pp}{}{\WT},\WT_1): (\ACont{\pp}{}{\tin\pp{\ell}{\ST}.\WT}, \tin\pp{\ell}{\ST_1}.\WT_1)\in\Rel\}.
\end{array}
\]
\begin{lemma}
    If $\Rel {\subseteq} \subttt$ is a tree simulation then $\Rel^+$ is a tree simulation.
\end{lemma}
\begin{proof}
We discuss some interesting cases.  Let $(\BCont{\pp}{}{\WT},\BCont{\pp}{1}{\WT_1}) \in \Rel^+\setminus\Rel.$
	\begin{enumerate} 
	\item If $\BContextt{\pp}{}=\tout\pq{\ell'}{\ST'}.\BContextt{\pp}{2}$, then, 
	\[
	     (\tout\pq{\ell'}{\ST'}.\BCont{\pp}{2}{\tout\pp{\ell}{\ST}.\WT}, \BCont{\pp}{1}{\tout\pp{\ell}{\ST_1}.\WT_1})\in\Rel
	\]
	could be derived by two rules:

		\begin{enumerate}
		\item if the rule applied is \rulename{ref-out} then 
		\[
		\BContextt{\pp}{1}= \tout\pq{\ell'}{\ST''}.\BContextt{\pp}{3} \text{ and }
		(\BCont{\pp}{2}{\tout\pp{\ell}{\ST}.\WT},\BCont{\pp}{3}{\tout\pp{\ell}{\ST_1}.\WT_1})\in \Rel,
		\text{ where } \ST' \subs \ST''. 
		\]
		By definition of $\Rel^+,$ we get $(\BCont{\pp}{2}{\WT},\BCont{\pp}{3}{\WT_1})\in \Rel^+$. 

		\item if the rule applied is \rulename{ref-${\BC}$} then 
		$\BCont{\pp}{1}{\tout\pp{\ell}{\ST_1}.\WT_1}=\BCont{\pq}{}{\tout\pq{\ell'}{\ST''}.{\WT''}}$ 
		and by Lemma~\ref{lemm:trans_contexts_2} we distinguish two cases
			\begin{enumerate}
			\item  
			\[\BCont{\pp}{1}{\tout\pp{\ell}{\ST_1}.\WT_1}= 
			\BCont{\pq}{}{\tout\pq{\ell'}{\ST''}.\BCont{\pp}{3}{\tout\pp{\ell}{\ST_1}.\WT_1}}, 
			\text{ where }
			\BContextt{\pp}{1}= \BCont{\pq}{}{\tout\pq{\ell'}{\ST''}.\BContextt{\pp}{3}}
			\]
			Then, by \rulename{ref-${\BC}$} we get 
			$(\BCont{\pp}{2}{\tout\pp{\ell}{\ST}.\WT},\BCont{\pq}{}{\BCont{\pp}{3}{\tout\pp{\ell}{\ST_1}.\WT_1}})\in\Rel$ 
			and $\ST' \subs \ST''$.
			By definition of $\Rel^+,$ we have
			$(\BCont{\pp}{2}{\WT}, \BCont{\pq}{}{\BCont{\pp}{3}{\WT_1}})\in\Rel^+.$ 
			\item
			\[
			\BCont{\pp}{1}{\tout\pp{\ell}{\ST_1}.\WT_1}= 
			\BCont{\pp}{1}{\tout\pp{\ell}{\ST_1}.\BCont{\pq}{2}{\tout\pq{\ell'}{\ST''}.\WT''}}, 
			\text{ where } \BCont{\pp}{1}{\tout\pp{\ell}{\ST_1}.\BContextt{\pq}{2}}= \BContextt{\pq}{}.
			\]
		    Then, we apply similar reasoning as in the first case. 
		    \end{enumerate}
		\end{enumerate}

	\item if $\BContextt{\pp}{}=\tin\pq{\ell'}{\ST'}.\BContextt{\pp}{2}$
	and
	$(\tin\pq{\ell'}{\ST'}.\BCont{\pp}{2}{\tout\pp{\ell}{\ST}.\WT}, \BCont{\pp}{1}{\tout\pp{\ell}{\ST_1}.\WT_1})\in \Rel$, then    by  Lemma~\ref{lemm:trans_supertype_of_A/BContexts} we have 
	$\BCont{\pp}{1}{\tout\pp{\ell}{\ST_1}.\WT_1} = \ACont{\pq}{1}{\tin\pq{\ell'}{\ST''}.\WT''}$ 
	and $\ST'' \subs \ST'$.
	By Lemma~\ref{lemm:trans_contexts_2} 
	\[
	\BCont{\pp}{1}{\tout\pp{\ell}{\ST_1}.\WT_1} = 
	\ACont{\pq}{1}{\tin\pq{\ell'}{\ST''}.\BCont{\pp}{3}{\tout\pp{\ell}{\ST_1}.\WT_1}}, \text{ where } 
	\BContextt{\pp}{1}= \ACont{\pq}{1}{\tin\pq{\ell'}{\ST''}.\BContextt{\pp}{3}}. 
	\]
	Then, by \rulename{ref-${\AC}$} we get 
	$(\BCont{\pp}{2}{\tout\pp{\ell}{\ST}.\WT},\ACont{\pq}{1}{\BCont{\pp}{3}{\tout\pp{\ell}{\ST_1}.\WT_1}})\in\Rel$, 
	and consequently\\
	$(\BCont{\pp}{2}{\WT}, \ACont{\pq}{1}{\BCont{\pp}{3}{\WT_1}})\in \Rel^+$. 
	\end{enumerate}
\end{proof}

\paragraph{Proof of Lemma \ref{lem:transitivityW}}
\begin{proof}
  Reflexivity is straightforward: %
  any SISO tree $\WT$ is related to itself by a coinductive derivation %
  which only uses rules %
  \rulename{ref-in}, \rulename{ref-out}, and \rulename{ref-end} %
  in Def.~\ref{def:ref}.

  We now focus on the proof of transitivity. %
  If $\WT_1\subttt\WT_2$ and $\WT_2\subttt\WT_3$ then there are tree simulations $\Rel_1$ and $\Rel_2$ such that $(\WT_1,\WT_2)\in \Rel_1$ and $(\WT_2,\WT_3)\in \Rel_2.$
  We shall prove that relation
  \[
    \Rel = \Rel_1 \circ \Rel_2^+
  \]
 is a tree simulation that contains $(\WT_1,\WT_3)$. It follows directly from definition of $\Rel^+$ that $(\WT_1,\WT_3)\in \Rel$ and it  remains to prove that $\Rel$ is a tree simulation.
 Assuming that $(\WT'_1,\WT_3')\in \Rel$ we consider the following possible cases for $\WT'_1:$
 \begin{enumerate}
 \item  $\WT'_1=\tend:$ By definition of $\Rel$ there is $\WT_2'$ such that $(\tend,\WT'_2)\in\Rel_1$ and $(\WT'_2,\WT'_3)\in \Rel_2^+.$ Since $\Rel_1$ and $\Rel_2^+$ are tree simulations, it holds by \rulename{ref-end} that $\WT'_2=\tend$ and also $\WT'_3=\tend.$
 \item $\WT'_1=\tout\pp\ell\S.\WT:$ By definition of $\Rel$ there is $\WT_2'$ such that $(\tout\pp\ell\S.\WT ,\WT'_2)\in\Rel_1$ and $(\WT'_2,\WT'_3)\in \Rel_2^+.$ Since $\Rel_1$ is tree simulation, using definition of $\BC$-sequence and applying \rulename{ref-out} or \rulename{ref-$\BC$}, we get three possibilities for $\WT'_2$:
 \begin{enumerate}
 \item $\WT'_2= \tout\pp\ell{\S_1}.\WT'$ with $\S\subs \S_1$ and $(\WT,\WT')\in \Rel_1: $       
      Since $\Rel_2^+$ is tree simulation and $(\WT'_2,\WT'_3)\in \Rel_2^+$, there are two possibilities for $\WT'_3$ (by \rulename{ref-out} or \rulename{ref-$\BC$}):
      \begin{enumerate}
      \item $\WT'_3=\tout\pp\ell{\S_2}.\WT''$ and $\S_1\subs \S_2$ and $(\WT',\WT'')\in\Rel_2^+:$ Then, by transitivity of $\subs$ and definition of $\Rel$ we get
      \[
           \S\subs \S_2 \text{ and }(\WT,\WT'')\in \Rel_1\circ \Rel_2^+=\Rel.
      \]
      \item $\WT'_3=\BCon\pp\tout\pp\ell{\S_2}.\WT''$ and $\S_1\subs \S_2$ and $(\WT',\BCon\pp\WT'')\in\Rel_2:$ Then, by transitivity of $\subs$ and definition of $\Rel$ we get
      \[
           \S\subs \S_2 \text{ and }(\WT,\BCon\pp\WT'')\in \Rel_1\circ \Rel_2^+=\Rel.
      \]

      \end{enumerate}
 \item $\WT'_2=\tout\pq{\ell'}{\S'}.\BCon\pp\tout\pp\ell{\S_1}.\WT'$ with $\S\subs \S_1$  and 
      \begin{equation}\label{eqn_trans2_1}
           (\WT,\tout\pq{\ell'}{\S'}.\BCon\pp{\WT'})\in \Rel_1 \text{ and }\actions{\WT}=\actions{\tout\pq{\ell'}{\S'}.\BCon\pp{\WT'}}: 
      \end{equation}
      Since $(\WT'_2,\WT'_3)\in \Rel_2^+,$ we have two cases (by \rulename{ref-out} or \rulename{ref-$\BC$}):
      \begin{enumerate}
      \item $\WT'_3=\BContextt{\pq}{11}.\tout\pq{\ell'}{\S''}.{\WT''}$ with $\S'\subs\S''$ and $(\BCon\pp{\tout\pp\ell{\S_1}.\WT'},\BContextt{\pq}{11}.{\WT''})\in \Rel_2^+.$
           By Lemma~\ref{lemm:trans_supertype_of_A/BContexts}, we have that
           \[
             \WT'_3=\BContextt\pp{1}.\tout\pp\ell{\S_2}.\WT''' \text{ and }\S_1\subs\S_2.
           \]
           By Lemma~\ref{lemm:trans_contexts_2}, there are two possibilities:
           \begin{enumerate}
               \item $\WT'_3 = \BContextt\pq{11}.\tout\pq{\ell'}{\S''}.\BContextt\pp{2}.\tout\pp\ell{\S_2}.\WT'''$ and $\pp!\not\in\actions{\BContext\pq_{11}}:$  
                        \begin{equation}\label{eqn_trans2_2}
                           \begin{split}
                           (\tout\pq{\ell'}{\S'}.\BCon\pp{\WT'},\BContextt\pq{11}.\tout\pq{\ell'}{\S''}.\BContextt\pp{2}.\WT''') \in \Rel_2^+  \text{ and }\\
                           \actions{\tout\pq{\ell'}{\S'}.\BCon\pp{\WT'}} = \actions{\BContextt\pq{11}.\tout\pq{\ell'}{\S''}.\BContextt\pp{2}.\WT'''}
                           \end{split}
                        \end{equation}
                        Hence, we conclude from \eqref{eqn_trans2_1} and \eqref{eqn_trans2_2}  that
                        \[
                           (\WT,\BContextt\pq{11}.\tout\pq{\ell'}{\S''}.\BContextt\pp{2}.\WT''')\in \Rel_1\circ\Rel_2^+ \text{ and }
                           \actions{\WT}=\actions{\BContextt\pq{11}.\tout\pq{\ell'}{\S''}.\BContextt\pp{2}.\WT'''}.
                        \]
               \item $\WT'_3 = \BContextt\pp{1}.\tout\pp{\ell}{\S_2}.\BContextt\pq{12}.\tout\pq{\ell'}{\S''}.\WT''$ and $\pq!\not\in\actions{\BContext\pp_{1}}:$  
                     The proof is similar to the previous case.
           \end{enumerate}
      \item $\WT'_3=\tout\pq{\ell'}{\S''}.{\WT''}$ with $\S'\subs\S''$ and $(\BCon\pp{\tout\pp\ell{\S_1}.\WT'},\WT'')\in \Rel_2^+.$
      \end{enumerate}
 \item $\WT'_2=\tin\pq{\ell'}{\S'}.\BCon\pp\tout\pp\ell{\S_1}.\WT'$ with $\S \subs \S_1$ and
       \begin{equation}\label{eqn_trans2_3}
            (\WT,\tin\pq{\ell'}{\S'}.\BCon\pp\WT')\in \Rel_1 \text{ and } \actions{\WT}=\actions{\tin\pq{\ell'}{\S'}.\BCon\pp\WT'}: 
        \end{equation}
        Since $(\WT'_2,\WT'_3) \in \Rel_2^+,$ we have two cases (by \rulename{ref-in} and \rulename{ref-$\BC$}):
        \begin{enumerate}
            \item $\WT'_3=\tin\pq{\ell'}{\S''}.\WT''$ and $\S''\subs \S'$ and $(\BCon\pp\tout\pp\ell{\S_1}.\WT', \WT'') \in \Rel_2^+)$:
            \item $\WT'_3=\ACon\pq{\tin\pq{\ell'}{\S''}.\WT''}$ and $\S''\subs \S'$ and $(\BCon\pp\tout\pp\ell{\S_1}.\WT', \ACon\pq{\WT''}) \in \Rel_2^+)$:
        \end{enumerate}
 \end{enumerate}
 \item $\WT'_1=\tin\pp\ell\S.\WT:$ The proof in this case follows similarly.
 \end{enumerate}
\end{proof}

\begin{lemma}\label{lemm:Vanja'sLemma21}
Let $\llbracket \TT \rrbracket_\SI = \{ \VT_j : j\in J \}$ and let $Y=\{ \WT_j : j\in J\}$, where $\WT_j\in \llbracket V_j \rrbracket_\SO$ for  $j\in J$. 
Then, there is $\UT \in \llbracket \TT \rrbracket_\SO$ such that $\llbracket \UT \rrbracket_\SI \subseteq Y$.
\end{lemma}

\begin{proof}
In this proof we consider the coinductive interpretation of  the definition of tree $\TT.$ 
	\begin{enumerate}
	\item $\TT=\tend$: 
		Since $\llbracket \TT \rrbracket_\SI=\{ \tend\}$, by selecting $U=Y=\tend$ the case follows.
	\item $\TT= \texternal_{i\in I}\tin\pp{\ell_i}{\ST_i}.\TT_i$: 
		Assume $\llbracket \TT \rrbracket_\SI = \{ \tin\pp{\ell_i}{\ST_i}.\VT'_{j_i} : \VT'_{j_i}\in \llbracket \TT_i \rrbracket_\SI, j_i\in J_i, i\in I\}$. Then, 
		\[
		Y=\{\WT_j\in \llbracket \VT_j \rrbracket_\SO : j\in J\} = \{ W_{j_i}=\tin\pp{\ell_i}{\ST_i}.\WT'_{j_i}: \WT'_{j_i}\in \llbracket \VT_i \rrbracket_\SO, j_i\in J_i, i\in I\}.
		\]	
		Since for $i\in I$ we have $\llbracket \TT_i \rrbracket_\SI=\{\VT'_{j_i}: j_i\in J_i\}$ and 
		$Y_i=\{ \WT'_{j_i}\in \llbracket \VT'_{j_i} \rrbracket_\SI : j_i\in J_i \}$, we can apply coinductive hypothesis and obtain $\UT_i\in \llbracket \TT_i \rrbracket_\SO$, such that $\llbracket \UT_i \rrbracket_\SI \subseteq Y_i$. 
		Hence, for $\UT = \texternal_{i\in I}\tin\pp{\ell_i}{\ST_i}.U_i$ we have $\UT\in \llbracket \TT \rrbracket_\SO$ and $\llbracket \UT \rrbracket_\SI \subseteq Y$.
	\item $\TT= \tinternal_{i\in I} \tout\pp{\ell_i}{\ST_i}.\TT_i$: 
		Assume $\llbracket \TT \rrbracket_\SI = \{\tinternal_{i\in I} \tout\pp{\ell_i}{\ST_i}.\VT'_{j_i} : \VT'_{j_i}\in \llbracket \TT_i \rrbracket_\SI, j_i\in J_i, i\in I\}$. Then, 
		\[
		Y=\{\WT_j\in \llbracket \VT_j \rrbracket_\SO : j\in J\} = \{ \WT_{k_i}=\tout\pp{\ell_i}{\ST_i}.\WT'_{k_i}: \WT'_{k_i}\in \llbracket \VT_i \rrbracket_\SO, j_i\in K_i\subseteq J_i, i\in N\subseteq I\}
		\]
		We claim that there is $i\in I$ such that for all $j_i\in J_i$ holds $\{\tout\pp{\ell_i}{\ST_i}.\llbracket \VT_{j_i} \rrbracket_\SO\} \cap Y \not= \emptyset$. 
		To prove the claim let us assume the opposite: for all $i\in I$ there is $j_i\in J_i$ such that $\{\tout\pp{\ell_i}{\ST_i}.\llbracket \VT_{j_i} \rrbracket_\SO\} \cap Y = \emptyset$. 
		For such $j_i$'s, let us consider $\VT = \tinternal_{i\in I} \tout\pp{\ell_i}{\ST_i}.\VT_{j_i}$. 
		Since $\VT \in \llbracket \TT \rrbracket_\SI$ and $\llbracket \VT \rrbracket_\SO \cap Y=\emptyset$ we obtain a contradiction with the definition of $Y$.
		
		Let us now fix $i\in I$ for which the above claim holds. 
		Let $Y' = \{ \WT'_{k_i}: \tout\pp{\ell_i}{\ST_i}.\WT'_{k_i}\in Y\}$. For $\TT_i$ we have $\llbracket \TT_i \rrbracket_\SI = \{ \VT'_{j_i}: j_i\in J_i\}$ and for each $j_i\in J_i$ there is $\WT'_{j_i} \in \llbracket \VT'_{j_i} \rrbracket_\SO$ such that $\WT'_{i_j} \in Y'$. 
		Hence, we can apply coinductive hypothesis and get $\UT'\in \llbracket \TT_i \rrbracket_\SO$ for which $\llbracket \UT' \rrbracket_\SI \subseteq Y'$ holds. 
		By taking $\UT = \tout\pp{\ell_i}{\ST_i}.\UT'$ we can conclude $\UT \in \llbracket \TT \rrbracket_\SO$ and $\llbracket \UT \rrbracket_\SI \subseteq \{\tout\pp{\ell_i}{\ST_i}.\WT'_{j_i}: \WT'_{j_i}\in Y'\}\subseteq Y$.
	\end{enumerate}

\end{proof}

\begin{lemma}\label{lem:for-reflexivity-of-subtyping}
For any tree $\TT$ we have 
    \begin{equation*}
    \forall \UT\in\llbracket \TT \rrbracket_\SO  \;
    \forall \VT\in\llbracket \TT \rrbracket_\SI \;\exists\WT
    \quad \text{such that}\quad
    \WT\in\llbracket \UT \rrbracket_\SI \cap \llbracket \VT \rrbracket_\SO.
    \end{equation*}
\end{lemma}
\begin{proof}
By coinduction on the definition of tree $\TT.$ 
    \begin{enumerate}
	\item $\TT=\tend$: 
		Since $\llbracket \TT \rrbracket_\SO=\llbracket \TT \rrbracket_\SI=\{ \tend\}$, 
		the proof follows directly.
	\item $\TT= \texternal_{i\in I}\tin\pp{\ell_i}{\ST_i}.\TT_i$: 
		Then,
	\begin{align*}
	& \UT \in  \{\texternal_{i\in I}\tin\pp{\ell_i}{\ST_i}.\UT': \UT'\in \llbracket \TT_i\rrbracket_\SO\}\\
	& \VT \in  \{\tin\pp\ell{\ST_i}.\VT': \VT'\in \llbracket \TT_i\rrbracket_\SI , i\in I \}
	\end{align*}
	By coinductive hypothesis for all $i\in I$ we have
    \begin{equation*}
    \forall \UT'\in\llbracket \TT_i \rrbracket_\SO  \;
    \forall \VT'\in\llbracket \TT_i \rrbracket_\SI \;\exists\WT'
    \quad \text{such that}\quad
    \WT'\in\llbracket \UT' \rrbracket_\SI \cap \llbracket \VT' \rrbracket_\SO.
    \end{equation*}
    Thus, for all $\UT\in\llbracket \TT \rrbracket_\SO$, $i\in I$ and
    $\VT\in\{\tin\pp\ell{\ST_i}.\VT': \VT'\in \llbracket \TT_i\rrbracket_\SI\}$ 
    we obtain there exists $\WT=\tin\pp\ell{\ST_i}.\WT'$, such that
    \[
    \WT \in \llbracket \UT \rrbracket_\SI \cap \llbracket \VT \rrbracket_\SO
    \]
	\item $\TT= \tinternal_{i\in I} \tout\pp{\ell_i}{\ST_i}.\TT_i$: Follows by a similar reasoning.
	\end{enumerate}
\end{proof}

\lemTransitivity*
\begin{proof}
  Reflexivity is straightforward from 
  Lemma~\ref{lem:for-reflexivity-of-subtyping} 
  and reflexivity of $\subttt$.
  
  We now focus on the proof of transitivity. %
       Assume that $\TT_1\subt \TT_2$ and $\TT_2\subt \TT_3.$
    From $\TT_1\subt \TT_2$, by Definition~\ref{def:subtyping},  we have
    \begin{equation}\label{eq:t1_subtype_t2}
    \forall \UT_1\in\llbracket \TT_1 \rrbracket_\SO  \;
    \forall \VT_2\in\llbracket \TT_2 \rrbracket_\SI \;
    \exists \WT_1\in\llbracket \UT_1 \rrbracket_\SI \;
    \exists \WT_2\in\llbracket \VT_2 \rrbracket_\SO
    \quad \WT_1 \subttt \WT_2
    \end{equation}
    From $\TT_2\subt \TT_3$, by Definition~\ref{def:subtyping},
    \begin{equation}\label{eq:t2_subtype_t3}
    \forall \UT_2\in\llbracket \TT_2 \rrbracket_\SO \;
    \forall \VT_3\in\llbracket \TT_3 \rrbracket_\SI \;
    \exists \WT'_2\in\llbracket \UT_2 \rrbracket_\SI \;
    \exists \WT_3\in\llbracket \VT_3 \rrbracket_\SO
    \quad \WT'_2 \subttt \WT_3
    \end{equation}
    Let us now fix one $\UT_1\in\llbracket \TT_1 \rrbracket_\SO$. 
    By~(\ref{eq:t1_subtype_t2}) we have that  
    \begin{equation}\label{eq:t1_subtype_t2_with_fixed_U1}
     \forall\VT_2\in\llbracket \TT_2 \rrbracket_\SI \; \exists \WT_2\in\llbracket \VT_2 \rrbracket_\SO \; \exists \WT_1\in\llbracket \UT_1 \rrbracket_\SI \text{ such that } \WT_1 \subttt \WT_2
    \end{equation}
    and let $Y$ be the set of all such $\WT_2$'s. By Lemma~\ref{lemm:Vanja'sLemma21}, there exist $\UT\in\llbracket \TT_2 \rrbracket_\SO$ such that $\llbracket \UT \rrbracket_\SI \subseteq Y$.
    Now from~(\ref{eq:t2_subtype_t3}) we have 
    $\forall \VT_3\in\llbracket \TT_3 \rrbracket_\SI \; \exists \WT_2\in\llbracket \UT \rrbracket_\SI \; \exists \WT_3\in\llbracket \VT_3 \rrbracket_\SO$ such that $\WT_2 \subttt \WT_3$. 
    Then, we conclude by transitivity of $\subttt$ that
    \[
        \forall \UT_1\in\llbracket \TT_1 \rrbracket_\SO \;
        \forall \VT_3\in\llbracket \TT_3 \rrbracket_\SI \;
        \exists \WT_1\in\llbracket \UT_1 \rrbracket_\SI \;
        \exists \WT_3\in\llbracket \VT_3 \rrbracket_\SO
        \quad \WT_1 \subttt \WT_3.
    \]

\end{proof}

\subsection{Further Examples}
\label{sec:subtyping-further-examples}%

We illustrate the refinement relation with an example.

\begin{example}
\label{ex:ref}
Consider $\W_1=\mu \ty.\tout\pp{\ell}{\S}.\tin\pq{\ell'}{\S'}.\ty$ and $\W_2=\mu \ty.\tin\pq{\ell'}{\S'}.\tout\pp{\ell}{\S}.\ty$. %
Their trees are related by the following coinductive derivation:
\[
  \infer=[\mbox{\rulename{ref-$\BC$} with $\BContext\pp = \tin\pq{\ell'}{\S'}$}]{%
    \ttree{\W_1} \subttt \ttree{\W_2}
  }{%
    \infer=[\rulename{ref-in}]{%
      \tin\pq{\ell'}{\S'}.\ttree{\W_1} \subttt \tin\pq{\ell'}{\S'}.\ttree{\W_2}
    }{%
      \ttree{\W_1} \subttt \ttree{\W_2}
    }%
  }%
\]
\end{example}

Next we give a simple asynchronous subtyping example. 
\begin{example}[Asynchronous subtyping]
\label{ex:sub}
Let\; %
\(
  \T = \tout\pp{\ell_1}{\ST_1}.\texternal\pq ?\left\{
  \begin{array}{@{}l@{}}
     \ell_3(\ST_3).\tend\\
     \ell_4(\ST_4).\tend
  \end{array} \right.
  \;\;\text{and}\;\;
  \T' = \tinternal\pp!\left\{
  \begin{array}{@{}l@{}}
     \ell_1(\ST_1).\tin\pq{\ell_3}{\ST_3}.\tend\\
     \ell_2(\ST_2).\tend
  \end{array} \right.
\) %
\;We show that $\T\subt \T'$. %
Notice that $\T$ is a SO type %
and $\T'$ is a SI type: %
hence, by Def.~\ref{def:subtyping}, we only need to show that there are 
$\WT\in\llbracket\ttree{\T}\rrbracket_\SI$ and 
$\WT'\in\llbracket\ttree{\T'}\rrbracket_\SO$ such that
$\WT\subttt\WT'$.
Since:
\[
\begin{array}{ccl}
\llbracket\ttree{\T}\rrbracket_\SI 
&=&
\left\{\,
\tout\pp{\ell_1}{\ST_1}.\tin\pq{\ell_3}{\ST_3}.\tend, \;
\tout\pp{\ell_1}{\ST_1}.\tin\pq{\ell_4}{\ST_4}.\tend
\,\right\}\\
\llbracket\ttree{\T'}\rrbracket_\SO
&=&
\left\{\,
\tout\pp{\ell_1}{\ST_1}.\tin\pq{\ell_3}{\ST_3}.\tend, \;
\tout\pp{\ell_2}{\ST_2}.\tend
\,\right\}
\end{array}
\] 
we have that $\WT\subttt\WT'$ holds for
$\WT=\WT'=\tout\pp{\ell_1}{\ST_1}.\tin\pq{\ell_3}{\ST_3}.\tend$, by reflexivity of $\subttt$.
\end{example}

Next we consider a complex example of asynchronous subtyping which
contains branching, selection and recursion.
This example is undecided by the algorithm in 
\cite{BravettiCLYZ19,BravettiCLYZ19L} (it returns ``unknown'')
but we can reason by our rules using the SISO decomposition method
as demonstrated below. 

\begin{example}
\label{ex:CONCUR19}
Consider the examples of session types $M_1$ and $M_2$ from
\cite[Example 3.21]{BravettiCLYZ19L}, here denoted by  $\T$ and $\T'$, respectively, where %
\[
\small
\begin{array}{ll}
\begin{array}{l}
\T = \mu \ty_1. \texternal\pp?\left\{
  \begin{array}{@{}l@{}}
  {\msgLabel{\ell_1}}{(\ST_1).}\tout\pp{\ell_3}{\ST_3}. \tout\pp{\ell_3}{\ST_3}. \tout\pp{\ell_3}{\ST_3}. \ty_1\\
  {\msgLabel{\ell_2}}{(\ST_2)}.\mu \ty_2. \tout\pp{\ell_3}{\ST_3} .\ty_2
\end{array} \right.
\end{array}
&
\begin{array}{l}
\T' =  \mu \ty_1. \texternal\pp?\left\{
  \begin{array}{@{}l@{}}
  {\msgLabel{\ell_1}}{(\ST_1).}\tout\pp{\ell_3}{\ST_3}. \ty_1\\
  {\msgLabel{\ell_2}}{(\ST_2)}.\mu \ty_2. \tout\pp{\ell_3}{\ST_3} .\ty_2
  \end{array} \right.
\end{array}
\end{array}
 \]

\tikzstyle{place}=[circle,draw=black,thick, inner sep=4pt,minimum size=8mm]
\tikzstyle{pre}=[<-,shorten <=1pt, >=stealth', semithick]
\tikzstyle{post}=[->,shorten >=1pt, >=stealth',  semithick]
\begin{tiny}
\begin{tikzpicture}%
\node[place] (q2) {\normalsize $\pp!$};
\node (M1) [left of = q2, node distance = 1cm] {\large $\T$:};
\node[place] (q1) [right of = q2, node distance=1.7cm] {\normalsize $\pp?$}
	edge [post] node[auto,swap] {$\ell_1(\ST_1)$} (q2);
\node[place] (q5) [right of = q1, node distance=1.7cm] {\normalsize $\pp!$}
	edge [pre] node[auto, swap] {$\ell_2(\ST_2)$} (q1)
	edge   [loop right, shorten >=1pt, >=stealth',  semithick] node[auto, swap, right] {$\ell_3(\ST_3)$} (q5);
\node[place] (q3) [below of = q2, node distance=1.5cm] {\normalsize $\pp!$}
	edge [pre, bend left] node[swap, left] {$\ell_3(\ST_3)$} (q2);
\node[place] (q4) [below of = q1, node distance=1.5cm] {\normalsize $\pp!$}
	edge [post, bend right] node[auto, swap] {$\ell_3(\ST_3)$} (q1)
	edge [pre] node[auto, swap] {$\ell_3(\ST_3)$} (q3);
\node (q1') [above of  = q1] {}
	edge [post]  (q1);
\end{tikzpicture}
\begin{tikzpicture}
\node[place] (q2) {\normalsize $\pp!$};
\node (M2) [left of = q2, node distance = 1cm] {\large $\T'$:};
\node[place] (q1) [right of = q2, node distance = 1.7cm] {\normalsize $\pp?$}
	edge [post, bend right] node[auto, swap] {$\ell_1(\ST_1)$} (q2)
	edge [pre, bend left] node[ swap, below] {$\ell_3(\ST_3)$} (q2);
\node[place] (q3) [right of = q1, node distance = 1.7cm] {\normalsize $\pp!$}
	edge [pre] node[auto, swap] {$\ell_2(\ST_2)$} (q1)
	edge [loop right, shorten >=1pt, >=stealth',  semithick] node[right] {$\ell_3(\ST_3)$} (q1);
\node (q1') [above of = q1] {}
	edge [post] node[auto, swap] {} (q1);
\end{tikzpicture}
\end{tiny}

In \cite{BravettiCLYZ19,BravettiCLYZ19L},
if the algorithm returns ``true'' (``false''), then the considered types are (are not) in the subtyping relation. The algorithm can return ``unknown'', meaning that the algorithm cannot check whether the types are in the subtyping relation or not.

For the considered types,
the algorithm in \cite{BravettiCLYZ19,BravettiCLYZ19L}
returns ``unknown'' and thus it can\emph{not} check that $\T \subt \T'$, which is, according to the authors of \cite{BravettiCLYZ19,BravettiCLYZ19L}, due to the complex accumulation patterns of these types which cannot be recognised by their theory. 

Here our approach comes into the picture and we demonstrate that the two decomposition functions into SO and SI trees are sufficiently 
fine-grained to recognise the complex structure of these types and to prove that $\T \subt \T'$.  We show that types $\T$ and $\T'$ are in the subtyping  relation $\T \subt \T'$, i.e., the corresponding session trees are in the subtyping relation \\$\ttree{\T} \subt \ttree{\T'}$ by showing that 
 \[
   \forall \UT\in\llbracket \TT\rrbracket_\SO \quad \forall \VT' \in \llbracket \TT' \rrbracket_\SI \quad \exists \WT \in\llbracket \UT \rrbracket_\SI \quad \exists \WT' \in \llbracket \VT' \rrbracket_\SO \qquad \WT \subttt\WT'
\]
where  $\TT = \ttree{\T}$ and $\TT' = \ttree{\T'}$.
 For the sake of simplicity, let us use the following notations and abbreviations: 
  \[
 \begin{array}{ll}
  \WT_1 = \ttree{\mu \ty. \tin\pp{\ell_1}{\ST_1}. \tout\pp{\ell_3}{\ST_3}.\ty }& \WT_2 = \tin\pp{\ell_2}{\ST_2}.\ttree{ \mu \ty. \tout\pp{\ell_3}{\ST_3}.\ty} \\
  \WT_3 = \ttree{\mu \ty. \tin\pp{\ell_1}{\ST_1}. \tout\pp{\ell_3}{\ST_3}. \tout\pp{\ell_3}{\ST_3}. \tout\pp{\ell_3}{\ST_3}.\ty} & %
 \end{array} 
\] %
 \[
 \begin{array}{ll}
  \pi_1 \equiv \tin\pp{\ell_1}{\ST_1}. \tout\pp{\ell_3}{\ST_3} & \pi_1^n \equiv \underbrace{\pi_1. \dots . \pi_1.}_\text{n}  \\
  \pi_3 \equiv \tin\pp{\ell_1}{\ST_1}. \tout\pp{\ell_3}{\ST_3}. \tout\pp{\ell_3}{\ST_3}. \tout\pp{\ell_3}{\ST_3} & %
   \pi_3^n \equiv \underbrace{\pi_3. \dots . \pi_3.}_\text{n}
 \end{array} 
\] %
 Then $\llbracket \TT \rrbracket_\SO = \{ \TT \}$ since $\TT$ is a SO tree, whereas
 \[
  \begin{array}{lcl}
 \llbracket \TT' \rrbracket_\SI & = & \{
 \WT_1, \WT_2,   \pi_1. \WT_2, \pi_1^2.\WT_2, \ldots, \pi_1^n.\WT_2, \dots
 \}\\
 \end{array}
 \]
 then $\UT = \TT$ and 
 \[
  \begin{array}{lcl}
  \llbracket \UT \rrbracket_\SI & = & \{
   \WT_3, \WT_2,   \pi_3. \WT_2, \pi_3^2.\WT_2, \ldots, \pi_3^n.\WT_2, \dots
 \}\\
\end{array}
 \] 
 Notice that all $\VT' \in \llbracket \TT' \rrbracket_\SI$ are SISO trees %
 hence $\llbracket \VT' \rrbracket_\SO = \{\VT'\}$ and $\WT' = \VT'$.

 We show now that %
 for all $\WT' \in \llbracket \TT' \rrbracket_\SI$ there is a  $\WT \in \llbracket \UT \rrbracket_\SI$ 
 such that $\WT \subttt \WT'$ by showing:
 
 \begin{enumerate}
 \item $\WT_3 \subttt \WT_1$
 \item $\pi_3^n. \WT_2 \subttt \pi_1^n. \WT_2$, $n \geq 0$  
 \end{enumerate}

\begin{enumerate}
\item The tree simulation for $\WT_3 \subttt \WT_1$ is

\[
\begin{array}{lcl}
\Rel & = & \{(\WT_3,\WT_1), \\
& & 
((\tout\pp{\ell_3}{\ST_3})^3.\WT_3,\tout\pp{\ell_3}{\ST_3}.\WT_1),\\
& & 
 ((\tout\pp{\ell_3}{\ST_3})^2.\WT_3, \WT_1),  \\
& & 
(\tout\pp{\ell_3}{\ST_3}.\WT_3, \tin\pp{\ell_1}{\ST_1}.\WT_1), \\
& & 
(\WT_3,( \tin\pp{\ell_1}{\ST_1}.\tin\pp{\ell_1}{\ST_1})^n.\WT_1),
\\
& & 
((\tout\pp{\ell_3}{\ST_3})^3.\WT_3,( \tin\pp{\ell_1}{\ST_1}.\tin\pp{\ell_1}{\ST_1})^n.\tout\pp{\ell_3}{\ST_3}.\WT_1),\\
& &
 ((\tout\pp{\ell_3}{\ST_3})^2.\WT_3, ( \tin\pp{\ell_1}{\ST_1}.\tin\pp{\ell_1}{\ST_1})^n.\WT_1), 
 \\
& & (\tout\pp{\ell_3}{\ST_3}.\WT_3, ( \tin\pp{\ell_1}{\ST_1}.\tin\pp{\ell_1}{\ST_1})^n.\tin\pp{\ell_1}{\ST_1}.\WT_1) 
 \mid n \geq 1
\}
\end{array}
\]

where $( \tin\pp{\ell_1}{\ST_1}.\tin\pp{\ell_1}{\ST_1})^n \equiv \underbrace{\tin\pp{\ell_1}{\ST_1}.\tin\pp{\ell_1}{\ST_1}. \ldots \tin\pp{\ell_1}{\ST_1}.\tin\pp{\ell_1}{\ST_1}. }_{\text 2n}\; \; n\in \mathbb{N}$ and
 $( \tout\pp{\ell_3}{\ST_3})^i \equiv \underbrace{\tout\pp{\ell_3}{\ST_3}. \ldots \tout\pp{\ell_3}{\ST_3} }_{\text i} \; \; i = 1,2,3$.

{\bf Proof.} The trees $\WT_3$ and $\WT_1$ are related by the following coinductive derivations:\\
for $n \geq 1$
\[
  \infer=[\mbox{\rulename{ref-in}}]{%
   \WT_3 \subttt  (\tin\pp{\ell_1}{\ST_1}.\tin\pp{\ell_1}{\ST_1})^{n}.\WT_1
  }{%
    \infer=[\mbox{\rulename{ref-$\BC$}, {\scriptsize $\BContext\pp = (\tin\pp{\ell_1}{\S_1}.\tin\pp{\ell_1}{\S_1})^{n}$}}]{%
   (\tout\pp{\ell_3}{\ST_3})^3.\WT_3 \subttt  (\tin\pp{\ell_1}{\ST_1}.\tin\pp{\ell_1}{\ST_1})^{n}.\tout\pp{\ell_3}{\ST_3}.\WT_1
  }{%
    \infer=[\mbox{\rulename{ref-$\BC$}, {\scriptsize $\BContext\pp = (\tin\pp{\ell_1}{\S_1}.\tin\pp{\ell_1}{\S_1})^{n} \tin\pp{\ell_1}{\S_1}$}}]{%
  (\tout\pp{\ell_3}{\ST_3})^2.\WT_3 \subttt  (\tin\pp{\ell_1}{\ST_1}.\tin\pp{\ell_1}{\ST_1})^{n}.\WT_1
  }{%
\infer=[\mbox{\rulename{ref-$\BC$}, {\scriptsize $\BContext\pp = (\tin\pp{\ell_1}{\S_1}.\tin\pp{\ell_1}{\S_1})^{n+1}$}}]{%
    \tout\pp{\ell_3}{\ST_3}.\WT_3 \subttt  (\tin\pp{\ell_1}{\ST_1}.\tin\pp{\ell_1}{\ST_1})^{n}.\tin\pp{\ell_1}{\ST_1}.\WT_1
    }{%
       \WT_3 \subttt (\tin\pp{\ell_1}{\ST_1}.\tin\pp{\ell_1}{\ST_1})^{n+1}.\WT_1
    }%
  }%
 }%
}%
\]

\[
  \infer=[\mbox{\rulename{ref-in}}]{%
   \WT_3 \subttt \WT_1
  }{%
  \infer=[\mbox{\rulename{ref-out}}]{%
   (\tout\pp{\ell_3}{\ST_3})^3.\WT_3 \subttt\tout\pp{\ell_3}{\ST_3}.\WT_1
  }{%
   \infer=[\mbox{\rulename{ref-$\BC$}, {\scriptsize $\BContext\pp = \tin\pp{\ell_1}{\S_1}$}}]{%
  (\tout\pp{\ell_3}{\ST_3})^2. \WT_3 \subttt \WT_1
  }{%
    \infer=[\mbox{\rulename{ref-$\BC$}, {\scriptsize $\BContext\pp = \tin\pp{\ell_1}{\S_1}.\tin\pp{\ell_1}{\S_1}$}}]{%
    \tout\pp{\ell_3}{\ST_3}.\WT_3 \subttt \tin\pp{\ell_1}{\ST_1}.\WT_1
    }{%
       \WT_3 \subttt \tin\pp{\ell_1}{\ST_1}.\tin\pp{\ell_1}{\ST_1}.\WT_1
    }%
  }%
 }%
}%
\]

\item  Case $\pi_3^n. \WT_2 \subttt \pi_1^n. \WT_2$, $n \geq 0$.\\
If $n=0$, then $\WT_2 \subttt \WT_2$ holds by reflexivity of $\subttt$, Lemma~\ref{lem:transitivityW}.\\
In case $n > 0$, we first show that 
\begin{equation}\label{eq:sub}
\pi_3^n. \WT_2 \subttt \pi_1.\pi_3^{n-1}.\WT_2
\end{equation}
The tree simulation is
\[
\begin{array}{lcl}
\Rel & = & \{(\pi_3^n. \WT_2, \pi_1.\pi_3^{n-1}.\WT_2), \\
& & 
((\tout\pp{\ell_3}{\ST_3})^3. \pi_3^{n-1}.\WT_2, \tout\pp{\ell_3}{\ST_3}.\pi_3^{n-1}.\WT_2),\\
& & 
((\tout\pp{\ell_3}{\ST_3})^2. \pi_3^{n-1}.\WT_2, \pi_3^{n-1}.\WT_2),\\
& & 
(\tout\pp{\ell_3}{\ST_3}. \pi_3^{n-1}.\WT_2,  \pi_1. \tout\pp{\ell_3}{\ST_3}.\pi_3^{n-2}.\WT_2),\\
& & 
(\pi_3^{n-1}.\WT_2,  \pi_1.\pi_3^{n-2}.\WT_2),\\
&  & \vdots \\
& & 
(\pi_3.\WT_2,  \pi_1.\WT_2),\\
& & 
((\tout\pp{\ell_3}{\ST_3})^3.\WT_2, \tout\pp{\ell_3}{\ST_3}.\WT_2),\\
& & 
((\tout\pp{\ell_3}{\ST_3})^2.\WT_2, \WT_2),\\
& & 
(\tout\pp{\ell_3}{\ST_3}.\WT_2, \WT_2),\\
& & 
(\WT_2, \WT_2)\}\\
\end{array}
\]

The refinement $\pi_3^n. \WT_2 \subttt \pi_1.\pi_3^{n-1}.\WT_2$ is derived by the following coinductive derivation

\[
\infer=[\mbox{\rulename{ref-in}}]{%
\pi_3^n.\WT_2 \subttt \pi_1.\pi_3^{n-1}.\WT_2\\
  }{%
\infer={%
    \vdots\\
    }{%
\infer=[\mbox{\rulename{ref-in}}]{%
\pi_3^2.\WT_2 \subttt \pi_1.\pi_3.\WT_2
    }{%
\infer=[\mbox{\rulename{ref-out}}]{%
    (\tout\pp{\ell_3}{\S_3})^3.\pi_3.\WT_2 \subttt \tout\pp{\ell_3}{\S_3}.\pi_3.\WT_2
    }{%
\infer=[\mbox{\rulename{ref-$\BC$}, {\scriptsize $\BContext\pp = \tin\pp{\ell_1}{\S_1}$}}]{%
    (\tout\pp{\ell_3}{\S_3})^2.\pi_3.\WT_2 \subttt \pi_3.\WT_2
    }{%
\infer=[\mbox{\rulename{ref-$\BC$}, {\scriptsize $\BContext\pp = \tin\pp{\ell_1}{\S_1}$}}]{%
    \tout\pp{\ell_3}{\S_3}.\pi_3.\WT_2 \subttt \pi_1.\tout\pp{\ell_3}{\S_3}.\WT_2
    }{%
  \infer=[\mbox{\rulename{ref-in}}]{%
   \pi_3.\WT_2 \subttt \pi_1.\WT_2
  }{%
  \infer=[\mbox{\rulename{ref-out}}]{%
   (\tout\pp{\ell_3}{\S_3})^3.\WT_2 \subttt\tout\pp{\ell_3}{\S_3}.\WT_2
  }{%
  \infer=[\mbox{\rulename{ref-$\BC$}, {\scriptsize $\BContext\pp = \tin\pp{\ell_1}{\S_1}$}}]{%
  (\tout\pp{\ell_3}{\S_3})^2. \WT_2 \subttt \WT_2
  }{%
    \infer=[\mbox{\rulename{ref-$\BC$}, {\scriptsize $\BContext\pp = \tin\pp{\ell_1}{\S_1}$}}]{%
    \tout\pp{\ell_3}{\S_3}.\WT_2 \subttt \WT_2
    }{%
       \WT_2 \subttt \WT_2
    }%
  }%
 }%
}%
}%
}%
}%
}%
}%
}%
\]

This coinductive derivation proves all refinements
\[
\pi_3^{n-k}.\WT_2 \subttt \pi_1.\pi_3^{n-k-1}.\WT_2, \quad k = 0, \ldots, n-1.
\]
By $k$ consecutive application first of \rulename{ref-out} and then \rulename{ref-in}, it holds that
\[
\pi_1^k\pi_3^{n-k}.\WT_2 \subttt \pi_1^{k+1}.\pi_3^{n-k-1}.\WT_2, \quad k = 0, \ldots, n-1.
\]
which means that 
\[
\begin{array}{ll}
\pi_3^n.\WT_2 \subttt \pi_1.\pi_3^{n-1}.\WT_2, & \\
\pi_1. \pi_3^{n-1}.\WT_2 \subttt \pi_1^2.\pi_3^{n-2}.\WT_2, & \\
\vdots & \\
 \pi_1^k\pi_3^{n-k}.\WT_2 \subttt \pi_1^{k+1}.\pi_3^{n-k-1}.\WT_2, & \\
 \vdots & \\
\pi_1^{n-1}\pi_3.\WT_2 \subttt \pi_1^{n}.\WT_2 &
\end{array}
\]

then $\pi_3^n. \WT_2 \subttt \pi_1^n. \WT_2$ follows by transitivity of $\subttt$, Lemma~\ref{lem:transitivityW}.\\\\

This concludes the proof that $\T \subt \T'$,  which could not be given by
the answer (Yes or No) by the algorithm in \cite[Example 3.21]{BravettiCLYZ19L}. 
\end{enumerate}
\end{example}

\section{Appendix of Section \ref{subsec:typesystem}}
\label{app:typesystem}
\subsection{Proof of Theorem~\ref{lem:subtyping-preserves-liveness}}
\label{app:proofs-subtyping-liveness}%
The main aim of this section is to prove
Theorem~\ref{lem:subtyping-preserves-liveness}.

\begin{lemma}\label{lem:struct-equiv-and-liveness}
If $\Gamma$ is live and $\Gamma \equiv \Gamma'$ then $\Gamma'$ is also live.
\end{lemma}
\begin{proof}
  From the the definition of $\Gamma \equiv \Gamma'$, %
  $\Gamma$ and $\Gamma'$ perform exactly the same reductions %
  (\ie, they are strongly bisimilar). %
  Therefore, the result follows by Definition~\ref{def:env-liveness}.
\end{proof}

\lemMovePreservesLiveness*
\begin{proof}
  By contradiction, assume that $\Gamma$ is live, 
  while $\Gamma'$ is \emph{not} live. Then, by \Cref{def:env-path-fairness},
  there is a fair but non-live path $(\Gamma'_i)_{i \in I}$
  such that $\Gamma'_0 = \Gamma'$.
  But then, consider the path that starts with $\Gamma$,
  reaches $\Gamma'$ as the first step, and then follows $(\Gamma'_i)_{i \in I}$
  --- more formally, the path $(\Gamma_j)_{j \in J}$
  such that $\Gamma_0 = \Gamma$, %
  and $\forall j \in J: j \ge 1$ implies $\Gamma_j = \Gamma'_{j-1}$.
  Such a path is fair, but not live: therefore,
  by \Cref{def:env-path-fairness}, we obtain that $\Gamma$ is not live
  --- contradiction.  Hence, we conclude that $\Gamma'$ is live.
\end{proof}

In the rest of this section, %
we will use %
the following alternative formulation of Def.~\ref{def:env-liveness}, %
that is more handy to construct proofs by coinduction.

\begin{definition}[Coinductive liveness]
  \label{def:coinductive-liveness}%
  \emph{$\varphi$ is a $\pp$-liveness property} iff, %
  whenever $\varphi(\Gamma)$:
  \begin{itemize}
  \item {\textnormal{\rulename{LP$\&$}}}\; %
    $\Gamma(\pp) = (\tqueue_\pp, \TT_\pp)$ %
    with $\TT_\pp = \texternal_{i \in I}{\tin\pq{\ell_i}{\S_i}.{\TT_i}}$ %
    \;implies that, for all fair paths $(\Gamma_j)_{j \in J}$
    such that $\Gamma_0 = \Gamma$, %
    \; $\exists h \in J, k \in I$ such that %
    $\Gamma \reds \Gamma_{h-1} \red \Gamma_h$, %
    with:%
    \begin{enumerate}
    \item%
      $\Gamma_{h-1}(\pp) = (\tqueue_\pp, \TT_\pp)$ \;and\; %
      $\Gamma_{h-1}(\pq) = (\tout\pp{\ell_k}\S\cdot \tqueue', \TT_\pq)$;%
    \item%
      $\Gamma_h(\pp) = (\tqueue_\pp,\TT_k)$ \;and\; %
      $\Gamma_h(\pq) = (\tqueue',\TT_\pq)$;
    \end{enumerate}
  \item {\textnormal{\rulename{LP$\oplus$}}}\; %
    $\Gamma(\pp) = (\tout\pq\ell\S\cdot \tqueue, \TT_\pp)$ %
    \;implies that, for all fair paths $(\Gamma_j)_{j \in J}$
    such that $\Gamma_0 = \Gamma$, %
    \; $\exists h \in J, k \in I$ such that %
    $\Gamma \reds \Gamma_{h-1} \red \Gamma_h$, %
    with:%
    \begin{enumerate}
    \item%
      $\Gamma_{h-1}(\pp) = (\tout\pq\ell\S\cdot \tqueue'_\pp, \TT'_\pp)$ \;and\; %
      $\Gamma_{h-1}(\pq) = (\tqueue_{\pq}, \texternal_{i \in I}{\tin\pp{\ell_i}{\S_i}.{\TT_i}})$;%
    \item%
      $\Gamma_h(\pp) = (\tqueue'_\pp,\TT'_\pp)$ \;and\; %
      $\Gamma_h(\pq) = (\tqueue_\pq,\TT_k)$;
    \end{enumerate}
  \item {\textnormal{\rulename{LP$\red$}}}\; %
    $\Gamma \red \Gamma'$ \;implies\; $\varphi(\Gamma')$.
  \end{itemize}
  We say that \emph{$\Gamma$ is $\pp$-live} iff $\varphi(\Gamma)$ %
  for some $\pp$-liveness property $\varphi$.
  We say that \emph{$\Gamma$ is live} iff $\Gamma$ %
  is $\pp$-live for all $\pp \in \dom{\Gamma}$.
\end{definition}

\begin{definition}[SISO tree projections]
  \label{def:siso-tree-projection}%
  The \emph{projections} of a SISO tree $\WT$ %
  are SISO trees coinductively defined as follows:

  \noindent%
  \begin{minipage}{0.5\linewidth}
    \[\small%
      \begin{array}{@{}r@{\;}c@{\;}l@{}}
        \projOut{\tend}{\pp} &=& \tend%
        \\%
        \projOut{(\tout\pp\ell\S.\WT')}{\pp} &=&%
        \tout\pp\ell\S.(\projOut{\WT'}{\pp})%
        \\%
        \projOut{(\tout\pq\ell\S.\WT')}{\pp} &=&%
        \left\{%
        \begin{array}{@{}l@{\;\;}l@{}}%
          \projOut{\WT'}{\pp} &%
          \text{if $\pq \!\neq\! \pp \land \pp \!\in\! \participant{\WT'}$}%
          \\%
          \tend &%
          \text{if $\pq \!\neq\! \pp \land \pp \!\not\in\! \participant{\WT'}$}%
        \end{array}
        \right.%
        \\%
        \projOut{(\tin\pq\ell\S.\WT')}{\pp} &=&%
        \left\{%
        \begin{array}{@{}l@{\;\;}l@{}}%
          \projOut{\WT'}{\pp} &%
          \text{if $\pp \!\in\! \participant{\WT'}$}%
          \\%
          \tend &%
          \text{if $\pp \!\not\in\! \participant{\WT'}$}%
        \end{array}
        \right.
      \end{array}
    \]
  \end{minipage}
  \quad%
  \begin{minipage}{0.5\linewidth}
    \[\small
      \begin{array}{@{}r@{\;}c@{\;}l@{}}
        \projIn{\tend}{\pp} &=& \tend%
        \\%
        \projIn{(\tin\pp\ell\S.\WT')}{\pp} &=&%
        \tin\pp\ell\S.(\projIn{\WT'}{\pp})%
        \\%
        \projIn{(\tin\pq\ell\S.\WT')}{\pp} &=&%
        \left\{%
        \begin{array}{@{}l@{\;\;}l@{}}%
          \projIn{\WT'}{\pp} &%
          \text{if $\pq \!\neq\! \pp \land \pp \!\in\! \participant{\WT'}$}%
          \\%
          \tend &%
          \text{if $\pq \!\neq\! \pp \land \pp \!\not\in\! \participant{\WT'}$}%
        \end{array}
        \right.%
        \\%
        \projIn{(\tout\pq\ell\S.\WT')}{\pp} &=&%
        \left\{%
        \begin{array}{@{}l@{\;\;}l@{}}%
          \projIn{\WT'}{\pp} &%
          \text{if $\pp \!\in\! \participant{\WT'}$}%
          \\%
          \tend &%
          \text{if $\pp \!\not\in\! \participant{\WT'}$}%
        \end{array}
        \right.
      \end{array}
    \]
  \end{minipage}
\end{definition}

\newcommand{\subtttStrict}{\sqsubseteq}
\begin{definition}
  \label{def:strict-siso-refinement}%
  The \emph{strict refinement $\subtttStrict$} %
  between SISO trees 
  is coinductively defined as:
  {\rm%
  \[
    \cinfer[\rulename{str-in}]{\ST' \subs \ST  \quad  \WT \subtttStrict \WT'}{
      \tin\pp{\ell}{\ST}.\WT
      \subtttStrict
      \tin\pp{\ell}{\ST'}.\WT'
    }
    \qquad%
    \cinfer[\rulename{str-out}]{\ST \subs \ST'  \quad  \WT \subtttStrict \WT'}
    {
      \tout\pp{\ell}{\ST}.\WT
      \subtttStrict
      \tout\pp{\ell}{\ST'}.\WT'
    }
    \qquad%
    \cinfer[\rulename{str-end}]{}
    {\tend \subtttStrict \tend}
  \]
  }%
\end{definition}

By Def.~\ref{def:strict-siso-refinement}, %
$\subtttStrict$ is a sub-relation of $\subttt$ %
that does not allow to change the order of inputs nor outputs. %
And by Prop.~\ref{lem:siso-tree-projections-equal-subt} below, %
the refinement $\subttt$ does not alter the order of inputs nor outputs %
from/to a given participant $\pp$: %
the message reorderings allowed by $\subttt$ can only alter %
interactions targeting different participants, %
or outputs \wrt inputs to a same participant.

\begin{proposition}
  \label{lem:siso-tree-projections-equal-subt}%
  For all $\WT$ and $\WT'$ and $\pp$, %
  if $\WT \subttt \WT'$, %
  then\; $(\projOut{\WT}{\pp}) \subtttStrict (\projOut{\WT'}{\pp})$ 
  \;and\; %
  $(\projIn{\WT}{\pp}) \subtttStrict (\projIn{\WT'}{\pp})$.
\end{proposition}
\begin{proof}
  By coinduction on the derivation of $\WT \subttt \WT'$.
\end{proof}

\begin{proposition}
  \label{lem:subt-reductions}%
  Take any $\pp$-live $\Gamma$ with $\Gamma(\pp) = (\tqueue, \TT)$. %
  Take $\Gamma' = \Gamma\mapUpdate{\pp}{(\tqueue',\TT')}$ %
  with $\tqueue'{\cdot}\TT' \subt \tqueue{\cdot}\TT$. %
  Then, for any fair path $(\Gamma'_j)_{j \in J'}$ with $\Gamma'_0 = \Gamma'$:
  \begin{enumerate}
  \item\label{item:subt-reductions:wi}%
    for all $n$, %
    the first $n$ inputs/outputs of $\pp$ along $(\Gamma'_j)_{j \in J'}$ %
    match the first $n$ input/output actions %
    of some $\WT' \in \llbracket\UT'\rrbracket_{\SI}$, %
    with $\UT' \in \llbracket\TT'\rrbracket_{\SO}$;%
  \item\label{item:subt-reductions:w}%
    there is a fair path $(\Gamma_j)_{j \in J}$ with $\Gamma_0 = \Gamma$ %
    such that, for all $n$, %
    the first $n$ inputs/outputs of $\pp$ along $(\Gamma_j)_{j \in J}$ %
    match the first $n$ input/output actions %
    of some $\WT \in \llbracket\VT\rrbracket_{\SO}$, %
    with $\VT \in \llbracket\TT\rrbracket_{\SI}$; and%
  \item\label{item:subt-reductions:subt}%
    $\tqueue'{\cdot}\WT' \subttt \tqueue{\cdot}\WT$.%
  \end{enumerate}
\end{proposition}
\begin{proof}
  Before proceeding, with a slight abuse of notation, %
  we ``rewind'' $\Gamma$ and $\Gamma'$, %
  \ie, we consider %
  $\Gamma$ and $\Gamma'$ in the statement to be defined such that:
  \begin{align}
    \label{eq:subt-reductions:rewind}%
    &\Gamma(\pp) = (\temptyqueue, \tqueue{\cdot}\TT) \qquad%
    \Gamma'(\pp) = (\temptyqueue, \tqueue'{\cdot}\TT')%
  \end{align}
  \ie, the outputs in $\tqueue$ and $\tqueue'$ are not yet queued; %
  instead, the queues of $\Gamma(\pp)$ and $\Gamma'(\pp)$ are empty, %
  and the outputs $\tqueue$ and $\tqueue'$ %
  are prefixes of the respective types, and are about to be sent. %
  We will ``undo'' this rewinding at the end of the proof, %
  to obtain the final result.

  Since $\tqueue'{\cdot}\TT' \subt \tqueue{\cdot}\TT$ (by hypothesis), %
  by Def.~\ref{def:subtyping} we have:
  \begin{align}
    \label{eq:subtt-reductions:subt-pre}%
    \forall \UT'\in\llbracket \tqueue'{\cdot}\TT'\rrbracket_\SO \quad \forall \VT \in \llbracket \tqueue{\cdot}\TT \rrbracket_\SI \quad \exists \WT_2 \in\llbracket \UT' \rrbracket_\SI \quad \exists \WT_1 \in \llbracket \VT \rrbracket_\SO \qquad \WT_2 \subttt\WT_1%
  \end{align}
  Observe that $\UT'$ and $\VT$ in \eqref{eq:subtt-reductions:subt-pre} %
  are quantified over sets of %
  session trees beginning with a same sequence of singleton selections %
  ($\tqueue'$ and $\tqueue$, respectively). %
  Therefore, such sequences of selections appear at the beginning %
  of all SISO trees extracted from any such $\UU'$ and $\V'$, which means:
  \begin{align}
    \label{eq:subtt-reductions:subt}%
    \forall \UT'\in\llbracket \TT'\rrbracket_\SO \quad \forall \VT \in \llbracket \TT \rrbracket_\SI \quad \exists \WT_a \in\llbracket \UT' \rrbracket_\SI \quad \exists \WT_b \in \llbracket \VT \rrbracket_\SO \qquad \tqueue'{\cdot}\WT_a \subttt \tqueue{\cdot}\WT_b%
  \end{align}

  \newcommand{\delayed}[1]{\mathit{delayed({#1})}}%
  \newcommand{\delayedEmpty}{\epsilon}%
  \newcommand{\pathEmpty}{\epsilon}%
  \newcommand{\GammaSteps}[1]{\gamma({#1})}%
  \newcommand{\tryFire}[5]{\operatorname{tryFire}\!\left({#1},{#2},{#3},{#4},{#5}\right)}%
  \newcommand{\treePairs}[1]{\mathcal{W}({#1})}%
  \newcommand{\pActions}[1]{\operatorname{{\pp}-actions}({#1})}%
  \newcommand{\pActionsi}[1]{\operatorname{{\pp}-actions}'({#1})}%
  \newcommand{\pActionsEmpty}{\epsilon}%

  We now define a procedure that, %
  using the first $m$ steps of a fair path $(\Gamma'_j)_{j \in J'}$, %
  constructs the beginning of a fair path $(\Gamma_j)_{j \in J}$; %
  also, the procedure ensures that: %
  \begin{enumerate*}[label=\emph{(\arabic*)}]
  \item%
    in the first $m$ steps of $(\Gamma'_j)_{j \in J'}$, %
    participant $\pp$ follows $\tqueue'$ and then a prefix of some $\WT_a$; %
  \item%
    the path $(\Gamma_j)_{j \in J}$ is constructed so that %
    $\pp$ follows $\tqueue$ and then some $\WT_b$ such that %
    $\tqueue'{\cdot}\WT_a \subttt \tqueue{\cdot}\WT_b$; %
    and %
  \item%
    such $\WT_a$ and $\WT_b$ are quantified %
    in \eqref{eq:subtt-reductions:subt}.
  \end{enumerate*}
  In particular, the procedure lets each participant %
  in $\Gamma'\setminus\pp = \Gamma\setminus\pp$ %
  fire the same sequences of actions %
  along $(\Gamma'_j)_{j \in J'}$ and $(\Gamma_j)_{j \in J}$. %
  The difference is that some actions may be delayed in %
  $(\Gamma_j)_{j \in J}$, %
  because $\pp$ in $\Gamma'$ may anticipate some outputs and inputs %
  (thanks to subtyping) \wrt $\Gamma$, %
  and unlock some participants in $\Gamma'$ earlier than $\Gamma$: %
  the procedure remembers any delayed actions (and their order), %
  and fires them as soon as they become enabled, %
  thus ensuring fairness.%

  For the procedure, we use:
  \begin{itemize}
  \item%
    $\pActionsi{i}$: %
    sequence of input/output actions performed by $\pp$ %
    in $\Gamma'$ when it reaches $\Gamma'_i$. %
    We begin with $\pActionsi{0} = \pActionsEmpty$;
  \item%
    $\GammaSteps{i}$: number of reduction steps %
    constructed along $(\Gamma_j)_{j \in J}$ %
    when $\Gamma'$ has reached $\Gamma'_i$. %
    We begin with $\GammaSteps{0} = 0$;
  \item%
    $\pActions{i}$: %
    sequence of input/output actions performed by $\pp$ %
    in $\Gamma$ when $\Gamma'$ reaches $\Gamma'_i$ %
    (note that, at this stage, $\Gamma$ has reached $\Gamma_{\GammaSteps{i}}$). %
    We begin with $\pActions{0} = \pActionsEmpty$;%
  \item%
    $\treePairs{i}$: %
    set of SISO tree pairs $(\WT_a, \WT_b)$ %
    such that, when $\Gamma'$ has reached $\Gamma'_i$, %
    and $\Gamma$ has reached $\Gamma_{\GammaSteps{i}}$, %
    the sequence $\pActionsi{i}$ %
    matches a prefix of $\tqueue'{\cdot}\WT_a$, %
    and the sequence $\pActions{i}$ %
    matches a prefix of $\tqueue{\cdot}\WT_b$, %
    and $\tqueue{\cdot}\WT_a \subttt \tqueue'{\cdot}\WT_b$. %
    We begin with $\treePairs{0}$ %
    containing all pairs $\WT_a$ and $\WT_b$ quantified in %
    \eqref{eq:subtt-reductions:subt};
  \item%
    $\delayed{i}$: sequence of reduction labels %
    that have been fired by $\Gamma'$ when it reaches $\Gamma'_i$, %
    but have not (yet) been fired by $\Gamma$ %
    when it reaches $\Gamma_{\GammaSteps{i}}$. %
    Labels in this sequence will be fired with the highest priority. %
    We begin with $\delayed{0} = \delayedEmpty$.
  \end{itemize}

  We also use the following function:

  \[
    \tryFire{d}{\Gamma}{d'}{f}{s} =%
    \begin{cases}
      \text{if $d = \delayedEmpty$ then $(d', f, s)$}%
      \\%
      \text{else if $\Gamma \redLabel{head(d)} \Gamma'$ %
        then $\tryFire{tail(d)}{\Gamma'}{d'}{f{\cdot}head(d)}{s{\cdot}\Gamma}$}%
      \\%
      \text{else $\tryFire{tail(d)}{\Gamma}{d'{\cdot}head(d)}{f}{s}$} %
    \end{cases}
  \]

  The function $\tryFire{d}{\Gamma}{d'}{f}{s}$ %
  tries to fire the environment reduction labels in the sequence $d$ %
  from $\Gamma$. %
  The other parameters are used along recursive calls, %
  to build the triplet that is returned by the function:
  \begin{itemize}
  \item%
    $d'$: a sequence of labels that have \emph{not} been fired. %
    It is extended each time the topmost label in $d$ cannot be fired;
  \item%
    $f$: a sequence of labels that \emph{have} been fired. %
    It is extended each time the topmost label in $d$ is fired; and
  \item%
    $s$: a sequence of typing environments %
    reducing from one into another through the sequence of labels $f$.
    It is extended each time $f$ is extended (see above).
  \end{itemize}

  When $(\Gamma'_j)_{j \in J'}$ performs a step $m+1$, %
  with a label $\alpha$ such that %
  $\Gamma'_m \redLabel{\;\alpha\;} \Gamma'_{m+1}$, %
  we proceed as follows:
  \begin{enumerate}
  \item\label{item:subt-reductions:not-p-gammai}
    if $\alpha$ does \emph{not} involve an input/output by $\pp$, %
    \ie, $\alpha = \recvLabel{\pq}{\pr}{\ell}$ %
    or $\alpha = \sendLabel{\pq}{\pr}{\ell}$ for some $\pq,\pr \neq \pp$:%
    \begin{enumerate}        
    \item\label{item:subt-reductions:p-actions-i-gamma-red}%
      $\pActionsi{m+1} = \pActionsi{m}$
    \end{enumerate}
  \item%
    otherwise (\ie, if $\alpha = \sendLabel{\pp}{\pq}{\ell}$ %
    or $\alpha = \recvLabel{\pp}{\pq}{\ell}$): %
    \begin{enumerate}
    \item%
      $\pActionsi{m+1} = \pActionsi{m} {\cdot} \alpha$;
    \end{enumerate}
  \item%
    $(d', f, s) =%
    \tryFire{\delayed{m}}{\Gamma_{\GammaSteps{m}}}{\delayedEmpty}{\delayedEmpty}{\pathEmpty}$
    \hfill{(\ie, try to fire each delayed action)}%
  \item%
    $\Gamma^* = \text{if $|s| > 0$ then $last(s)$ else %
      $\Gamma_{\GammaSteps{m}}$}$
    \hfill(\ie, $\Gamma^*$ is the latest env.~reached from $\Gamma$)
  \item%
    $\forall i = 1..|s|: \Gamma_{\GammaSteps{m}+i} = s(i)$ %
    \hfill(\ie, we extend the path of $\Gamma$ to reach $\Gamma^*$)
  \item%
    if $\Gamma^* \redLabel{\;\alpha\;} \Gamma''$ %
    (for some $\Gamma''$), %
    we add $\Gamma''$ to the path of $\Gamma$, as follows:
    \begin{enumerate}
    \item%
      $\Gamma_{\GammaSteps{m}+|s|+1} = \Gamma''$
    \item%
      $\GammaSteps{m+1} = \GammaSteps{m} + |s| + 1$
    \item%
      $\delayed{m+1} = d'$
    \item\label{item:subt-reductions:not-p}
      if $\alpha$ does \emph{not} involve an input/output by $\pp$, %
      \ie, $\alpha = \recvLabel{\pq}{\pr}{\ell}$ %
      or $\alpha = \sendLabel{\pq}{\pr}{\ell}$ for some $\pq,\pr \neq \pp$:%
      \begin{enumerate}        
      \item%
        $\pActions{m+1} = \pActions{m}$ %
        extended with all labels in $f$ involving $\pp$;
      \end{enumerate}
    \item%
      otherwise (\ie, if $\alpha = \sendLabel{\pp}{\pq}{\ell}$ %
      or $\alpha = \recvLabel{\pp}{\pq}{\ell}$): %
      \begin{enumerate}
      \item%
        $\pActions{m+1} = \pActions{m}$ %
        extended with all labels in $f$ involving $\pp$, followed by $\alpha$;
      \end{enumerate}
    \end{enumerate}
  \item%
    otherwise (\ie, if there is no $\Gamma''$ such that %
    $\Gamma^* \redLabel{\;\alpha\;} \Gamma''$), %
    we add $\alpha$ to the delayed actions, as follows:
    \begin{enumerate}
    \item%
      $\pActions{m+1} = \pActions{m}$ %
      extended with all labels in $f$ involving $\pp$%
    \item%
      $\GammaSteps{m+1} = \GammaSteps{m} + |s|$%
    \item%
      $\delayed{m+1} = d' {\cdot} \alpha$%
    \end{enumerate}
  \item%
    $\treePairs{m+1} = %
    \left\{%
      (\WT_a, \WT_b) \in \treePairs{m}%
      \;\middle|\;%
      \begin{array}{@{}l@{}}
        \text{$\WT_a$ matches $\pActionsi{m+1}$}%
        \\%
        \text{$\WT_b$ matches $\pActions{m+1}$}%
      \end{array}
      \;\right\}$%
  \end{enumerate}

  The procedure has the following invariants, for all $i \ge 0$: %
  \begin{enumerate}[label=\emph{(i\arabic*)}]
  \item\label{item:subt-reductions:inv-pactionsi}%
    $\pActionsi{i}$ is the sequence of inputs/outputs of $\pp$ %
    fired along the transitions from $\Gamma'_0 = \Gamma'$ to $\Gamma'_i$ %
    (by construction);
  \item\label{item:subt-reductions:inv-pactions}%
    $\pActions{i}$ is the sequence of inputs/outputs of $\pp$ %
    fired along the transitions from %
    $\Gamma_0 = \Gamma$ to $\Gamma_{\GammaSteps{i}}$ %
    (by construction);
  \item\label{item:subt-reductions:inv-treepairs}%
    $\forall (\WT_a, \WT_b) \in \treePairs{i}$: %
    $\pActionsi{i}$ matches the beginning of $\WT_a$
    and $\pActions{i}$ matches the beginning of $\WT_b$
    (by construction);
  \item\label{item:subt-reductions:inv-pairs-exist}%
    $\treePairs{i} \neq \emptyset$. %
    We prove this claim by induction on $i$. %
    When $i=0$, the result is immediate: %
    by definition, %
    $\treePairs{0}$ contains all pairs $(\WT_a, \WT_b)$ quantified in %
    \eqref{eq:subtt-reductions:subt}, and at least one such pair exists %
    (otherwise, the hypothesis $\tqueue'{\cdot}\TT' \subt \tqueue{\cdot}\TT$ 
    would be contradicted). 
    
    In  the inductive case $i = m+1$ (for some $m$), %
    by the induction hypothesis we have that $\treePairs{m} \neq \emptyset$, %
    and (by item~\ref{item:subt-reductions:inv-treepairs} above) %
    for each $(\WT_a, \WT_b) \in \treePairs{i}$, %
    $\pActionsi{m}$ matches the beginning of $\WT_a$
    and $\pActions{m}$ matches the beginning of $\WT_b$. %
    Then, by contradiction, assume that $\treePairs{m+1} = \emptyset$:
    this means that $\Gamma(\pp)$ or $\Gamma'(\pp)$,
    has performed an input or output that cannot be matched by
    any pair in of trees $\treePairs{i}$,
    which means that
    $\tqueue'{\cdot}\TT' \not\subt \tqueue{\cdot}\TT$ --- contradiction.
    Therefore, we conclude $\treePairs{m+1} \neq \emptyset$.%
  \end{enumerate}

  We now obtain our thesis, %
  by invariants %
  \ref{item:subt-reductions:inv-pactionsi}--\ref{item:subt-reductions:inv-pairs-exist} above: %
  by taking any fair path $(\Gamma'_j)_{j \in J'}$ with $\Gamma'_0 = \Gamma'$, %
  and applying the procedure above for any %
  $n = |\pActionsi{m}|$ for some $m \ge n$ such that $m \in J'$ %
  (\ie, for any number $n$ of reductions of $\pp$ that are performed %
  along the path, within $m$ steps), %
  we find some $\WT'$ such that $\tqueue'{\cdot}\WT'$ matches $\pActionsi{m}$, %
  and we construct the beginning of a fair path $(\Gamma_j)_{j \in J}$ %
  where $\pp$ behaves according to some $\WT$ %
  such that $\tqueue{\cdot}\WT$ matches $\pActions{m}$, %
  and such that $\tqueue'{\cdot}\WT' \subttt \tqueue{\cdot}\WT$; %
  and by increasing $n$ and $m$, %
  we correspondingly extend the sequences %
  $\pActionsi{m}$ and $\pActions{m}$. %

  To conclude the proof, we need to undo %
  the ``rewinding'' in~\eqref{eq:subt-reductions:rewind}. %
  Consider any path %
  $(\Gamma'_j)_{j \in J'}$,
  and the corresponding $(\Gamma_j)_{j \in J}$ %
  obtained with $\Gamma',\Gamma$ rewinded as %
  in~\eqref{eq:subt-reductions:rewind}: %
  we can undo the rewinding of $\tqueue'$ and $\tqueue$ by:
  \begin{enumerate}
  \item%
    choosing a path $(\Gamma'_j)_{j \in J'}$ %
    that fires the outputs in $\tqueue'$ in its first reductions, %
    thus reaching the ``original'' $\Gamma'$ from the statement;%
  \item%
    then, for such a path of $\Gamma'$, %
    the procedure gives us the beginning of a corresponding live path %
    $(\Gamma_j)_{j \in J}$ that fires all outputs in $\tqueue$, %
    within $k$ steps (for some $k$). %
    By induction on $k$ and $\tqueue$, we can reorder the first $k$ actions %
    of $(\Gamma_j)_{j \in J}$ %
    so that the outputs in $\tqueue$ are fired first, %
    thus reaching the ``original'' $\Gamma$ from the statement;%
  \item%
    after the outputs in $\tqueue$ and $\tqueue'$ are fired %
    along such paths of $\Gamma'$ and $\Gamma$, %
    we have that %
    $\pp$ follows SISO trees $\WT_a$ and $\WT_b$ %
    extracted respectively from $\TT'$ and $\TT$, %
    and we have $\tqueue{\cdot}\WT_a \subttt \tqueue'{\cdot}\WT_b$: %
    therefore, such paths satisfy the statement.
  \end{enumerate}
\end{proof}

\begin{proposition}
  \label{lem:subttt-no-output}%
  Take any $\pp$-live $\Gamma$ with $\Gamma(\pp) = (\tqueue, \TT)$. %
  Take $\Gamma' = \Gamma\mapUpdate{\pp}{(\tqueue',\TT')}$ %
  with $\tqueue'{\cdot}\TT' \subt \tqueue{\cdot}\TT$. %
  Assume that there is a fair path $(\Gamma'_j)_{j \in J'}$ %
  with $\Gamma'_0 = \Gamma'$ such that, %
  for some $\pq, \ell, \S_\pr$:%
  \begin{align}
    \label{eq:lem:subt-no-output:no-out-gammai}%
    &\forall j \in J', \tqueue_\pr, \TT_\pr:\;%
    \Gamma'_j(\pr) \neq (\tout{\pq}{\ell}{\S_\pr}{\cdot}\tqueue_{\pr}, \TT_\pr)%
  \end{align}
  Then, there is a fair path $(\Gamma_j)_{j \in J}$ %
  with $\Gamma_0 = \Gamma$ such that, %
  \begin{align}
    \label{eq:lem:subt-no-output:no-out-gamma}%
    &\forall j \in J, \tqueue_\pr, \TT_\pr:\;%
    \Gamma_j(\pr) \neq (\tout{\pq}{\ell}{\S_\pr}{\cdot}\tqueue_{\pr}, \TT_\pr)%
  \end{align} 
\end{proposition}
\begin{proof}
  Take the path $(\Gamma'_j)_{j \in J'}$. %
  By Prop.~\ref{lem:subt-reductions} %
  there is a SISO tree $\WT'$ %
  describing the first $n$ reductions %
  (for any $n$) %
  of $\pp$ in $\Gamma'$, %
  and a corresponding SISO tree $\WT$ %
  such that $\tqueue'{\cdot}\WT' \subttt \tqueue{\cdot}\WT$ %
  describing the reductions of $\pp$ in $\Gamma$ %
  along a fair path $(\Gamma_j)_{j \in J}$. %
  By Prop.~\ref{lem:siso-tree-projections-equal-subt}, %
  $\tqueue'{\cdot}\WT'$ contains the same sequences of per-participant %
  inputs and outputs of $\tqueue{\cdot}\WT$; %
  moreover, by Def.~\ref{def:ref}, $\WT'$ %
  can perform the outputs appearing in $\WT$, possibly earlier. %
  And by the path construction in Prop.~\ref{lem:subt-reductions}, %
  each participant $\pq \in \dom{\Gamma}=\dom{\Gamma'}$ %
  (with $\pq \neq \pp$) %
  can fire along $(\Gamma'_j)_{j \in J'}$ %
  at least the same outputs and the same inputs %
  (in the same respective order) %
  that it fires along $(\Gamma_j)_{j \in J}$.
  Now, observe that by hypothesis %
  \eqref{eq:lem:subt-no-output:no-out-gammai}, %
  along the fair path $(\Gamma'_j)_{j \in J'}$, %
  participant $\pr$ %
  never produces an output $\tout{\pq}{\ell}{\S_\pr}$; %
  but then, %
  along $(\Gamma_j)_{j \in J}$, %
  participant $\pr$ never produces the output %
  $\tout{\pq}{\ell}{\S_\pr}$, either. %
  Therefore, we obtain \eqref{eq:lem:subt-no-output:no-out-gamma}.%
\end{proof}

\begin{proposition}
  \label{lem:subttt-no-input}%
  Take any $\pp$-live $\Gamma$ with $\Gamma(\pp) = (\tqueue, \TT)$. %
  Take $\Gamma' = \Gamma\mapUpdate{\pp}{(\tqueue',\TT')}$ %
  with $\tqueue'{\cdot}\TT' \subt \tqueue{\cdot}\TT$. %
  Assume that there is a fair path $(\Gamma'_j)_{j \in J'}$ %
  with $\Gamma'_0 = \Gamma$ such that, %
  for some $\pq, I, \ell_i, \S_{\pr, i}, \TT_{\pr, i} \;(i \in I)$:%
  \begin{align}
    \label{eq:lem:subt-no-input:no-in-gammai}%
    &\forall j \in J', \tqueue_\pr:\; %
      \Gamma'_j(\pr) \neq \left(\tqueue_{\pr}, \texternal_{i\in I}\tin{\pq}{\ell_i}{\S_{\pr, i}}.\TT_{\pr, i}\right)%
  \end{align}
  Then, there is a fair path $(\Gamma_j)_{j \in J}$ %
  with $\Gamma_0 = \Gamma$ such that:
  \begin{align}
    \label{eq:lem:subt-no-input:no-in-gamma}%
    &\forall j \in J, \tqueue_\pr:\; %
      \Gamma_j(\pr) \neq \left(\tqueue_{\pr}, \texternal_{i\in I}\tin{\pq}{\ell_i}{\S_{\pr, i}}.\TT_{\pr, i}\right)%
  \end{align} 
\end{proposition}
\begin{proof}
  Similar to Prop.~\ref{lem:subttt-no-output}. %
\end{proof}

\begin{definition}[Queue output prefixing]
  \label{notation:output-prefixing}%
  We write $\tqueue{\cdot}\TT$ %
  for the session tree obtained %
  by prefixing $\TT$ with the sequence of singleton internal choices %
  matching the sequence of outputs in $\tqueue$.
\end{definition}

\newcommand{\liveSubP}{\mathcal{P}}%
\begin{lemma}
  \label{lem:subtyping-preserves-p-live}%
  If $\Gamma,\pp{:}(\tqueue, \TT)$ is $\pp$-live %
  and $\tqueue'{\cdot}\TT' \subt \tqueue{\cdot}\TT$, %
  then $\Gamma,\pp{:}(\tqueue', \TT')$ is $\pp$-live.
\end{lemma}
\begin{proof}
  \newcommand{\powerset}[1]{\mathscr{P}\!\mathit{owerset}\!\left({#1}\right)}%
  \newcommand{\liveP}{\mathord{\mathcal{L}}}%
  \newcommand{\liveSubPi}{\mathcal{P}'}%
  \newcommand{\redNotPStar}[1]{\mathrel{\xRightArrow{#1}{}^{\!\!*}}}%

  Let $\liveP$ be the set of all $\pp$-live typing contexts, %
  i.e., the largest $\pp$-liveness property %
  by Def.~\ref{def:coinductive-liveness}. %
  Consider the following property:%

  \begin{align}
    &\label{eq:subttt-liveness:prop:def}%
    \liveSubP = \liveP \,\cup\, \liveP'%
    \quad%
    \text{where}\;%
    \liveP' \,=\,%
    \left\{%
      \Gamma\mapUpdate{\pp}{(\tqueue', \TT')}
    \;\middle|\;%
      \begin{array}{@{}l@{}}
        \Gamma \in \liveP\\%
        \Gamma(\pp) = (\tqueue, \TT)\\%
        \tqueue'{\cdot}\TT' \subt \tqueue{\cdot}\TT%
      \end{array}
    \right\}%
  \end{align}

  We now prove that $\liveSubP$ is a $\pp$-liveness property %
  --- i.e., %
  we prove that each element of $\liveSubP$ %
  satisfies the clauses of Def.~\ref{def:coinductive-liveness}. %
  Since all elements of $\liveP$ trivially satisfy the clauses, %
  we only need to examine the elements of $\liveP'$: %
  to this purpose, %
  we consider each $\Gamma' \in \liveP'$, %
  and observe that, by \eqref{eq:subttt-liveness:prop:def}, there exist %
  $\Gamma, \tqueue, \TT, \tqueue', \TT'$ such that:
  \begin{align}
    &\label{eq:subttt-liveness:prop:gammai}%
      \Gamma' = \Gamma\mapUpdate{\pp}{(\tqueue', \TT')}%
    \\%
    &\label{eq:subttt-liveness:prop:gammap}%
      \Gamma(\pp) = (\tqueue, \TT)%
    \\%
    &\label{eq:subttt-liveness:prop:sub}
      \tqueue'{\cdot}\TT' \subt \tqueue{\cdot}\TT%
  \end{align}
  By \eqref{eq:subttt-liveness:prop:sub} and Def.~\ref{def:subtyping},
  we also have:
  \begin{align}
    &\label{eq:subttt-liveness:prop:siso-trees}%
      \forall \UT'\in\llbracket \tqueue'{\cdot}\TT'\rrbracket_\SO \quad
      \forall \VT \in \llbracket \tqueue{\cdot}\TT \rrbracket_\SI \quad
      \exists \WT_1 \in\llbracket \UT' \rrbracket_\SI \quad
      \exists \WT_2 \in \llbracket \VT \rrbracket_\SO \qquad \WT_1 \subttt \WT_2
  \end{align}
  Observe that %
  by the definitions of $\llbracket \cdot\rrbracket_\SO$ %
  and $\llbracket \cdot\rrbracket_\SI$ on page~\pageref{page:tree-decomp}, %
  and by Def.~\ref{def:ref}, %
  all relations $\WT_1 \subttt \WT_2$ %
  in \eqref{eq:subttt-liveness:prop:siso-trees} %
  are yielded by a same refinement rule. %
  We proceed by cases on such a rule.

  \begin{itemize}
  \item%
    Case \rulename{ref-end}.\quad%
    In this case, we have %
    $\tqueue'{\cdot}\TT' \subttt \tqueue{\cdot}\TT = \tend$, %
    which means %
    $\tqueue' = \tqueue = \temptyqueue$ %
    and $\TT' = \TT = \tend$. %
    Therefore, by \eqref{eq:subttt-liveness:prop:gammap} and \eqref{eq:subttt-liveness:prop:gammai}, %
    $\Gamma' = \Gamma \in \liveP$, and thus, %
    we conclude that %
    $\Gamma'$ satisfies the clauses of Def.~\ref{def:coinductive-liveness}.
  \item%
    Case \rulename{ref-in}.\quad%
    In this case, we have:
    \begin{align}
      &\label{eq:subttt-liveness:sub2-in:tqueue}%
        \tqueue = \tqueue' = \temptyqueue%
        \\%
      &\nonumber%
        \exists \pq, I, I', \ell_i, \ST_i, \ST'_i, \TT_i, \TT'_i%
        \;\text{such that:}%
      \\%
      &\label{eq:subttt-liveness:sub2-in:tt}%
        \qquad%
        \TT' = \texternal_{i \in I \cup I'}\tin\pq{\ell_i}{\ST'_i}.\TT'_i \quad\text{and}\quad%
        \TT = \texternal_{i \in I}\tin\pq{\ell_i}{\ST_i}.\TT_i%
      \\%
      &\label{eq:subttt-liveness:sub2-in:st-wt}%
        \qquad%
        \forall i \in I: \ST_i \subs \ST'_i  \quad\text{and}\quad  \TT'_i \subttt \TT_i%
    \end{align}
    We now show that $\Gamma'$ satisfies all clauses %
    of Def.~\ref{def:coinductive-liveness}:
    \begin{itemize}
    \item%
      clause \rulename{LP$\&$}.\quad%
      Since $\Gamma$ is $\pp$-live, %
      we know that, %
      for all fair paths $(\Gamma_j)_{j \in J}$
      such that $\Gamma_0 = \Gamma$, %
      \; $\exists h \in J, k \in I$ such that %
      $\Gamma \reds \Gamma_{h-1} \red \Gamma_h$, %
      with:%
      \begin{enumerate}
      \item%
        $\Gamma_{h-1}(\pp) = (\tqueue, \TT)$ \;and\; %
        $\Gamma_{h-1}(\pq) = (\tout\pp{\ell}{\ST_\pq}\cdot \tqueue_\pq, \TT_\pq)$ %
        with $\ell = \ell_i$ and $\ST_\pq \subt \ST_i$ for some $i \in I$
      \item%
        $\Gamma_h(\pp) = (\tqueue,\TT_i)$ \;and\; %
        $\Gamma_h(\pq) = (\tqueue_\pq,\TT_\pq)$.
      \end{enumerate} 
      Now, for all such $(\Gamma_j)_{j \in J}$, %
      we can construct a path corresponding $(\Gamma'_j)_{j \in J'}$ %
      such that: %
      \begin{itemize}
      \item%
         $\Gamma'_0 = \Gamma'$
      \item
        $\forall n \in 0..h-1: %
        \Gamma'_n = \Gamma_n\mapUpdate{\pp}{(\tqueue', \TT')}$; %
      \item%
        $\Gamma'_h = \Gamma_h\mapUpdate{\pp}{(\tqueue', \TT'_i)}$ %
        (i.e., the queue of $\Gamma_h(\pp)$ %
        is preserved in $\Gamma'_h(\pp)$);
      \item%
        the rest of the path after the $h$-th reduction is arbitrary (but fair).
      \end{itemize}
      Observe that $(\Gamma'_j)_{j \in J'}$ %
      is fair, reproduces the first $h$ steps of $(\Gamma_j)_{j \in J}$, %
      and triggers the top-level input of $\Gamma'(\pp)$. %
      Also, observe that \emph{all} fair paths from $\Gamma'$ %
      eventually trigger the top-level input of $\Gamma'(\pp)$: %
      in fact, if (by contradiction) we assume that %
      there is a fair path from $\Gamma'$ %
      that never triggers $\Gamma'(\pp)$'s input, %
      then (by inverting the construction above) %
      we would find a corresponding fair path of $\Gamma$ %
      that never triggers $\Gamma(\pp)$'s input %
      --- i.e., we would conclude that $\Gamma$ is \emph{not} $\pp$-live %
      (contradiction).
      Thus, we conclude that $\Gamma'$ satisfies %
      clause~\rulename{LP$\&$} of Def.~\ref{def:coinductive-liveness};
    \item%
      clause \rulename{LP$\oplus$}.\quad%
      The clause is vacuously satisfied;
    \item%
      clause \rulename{LP$\red$}.\quad%
      Assume $\Gamma' \red \Gamma''$.  We have two possibilities:
      \begin{enumerate}[label=(\alph*)]
      \item\label{item:subttt-livenes:move-not-p}%
        the reduction does \emph{not} involve $\pp$. %
        Then, there is a corresponding reduction $\Gamma \red \Gamma'''$ %
        with $\Gamma'' = \Gamma'''\mapUpdate{\pp}{(\tqueue', \TT')}$. %
        Observe that $\Gamma'''$ is $\pp$-live, %
        and thus, $\Gamma''' \in \liveP$; %
        therefore, by \eqref{eq:subttt-liveness:prop:def}, %
        we have $\Gamma'' \in \liveP' \subseteq \liveSubP$. %
        Thus, we conclude that $\Gamma'$ satisfies %
        clause~\rulename{LP$\red$} of Def.~\ref{def:coinductive-liveness};
      \item\label{item:subttt-livenes:move-p}%
        the reduction \emph{does} involve $\pp$. %
        There are three sub-cases:
        \begin{enumerate}[label=(\roman*)]
        \item%
          $\pp$ is enqueuing an output toward some participant $\pr$. %
          This case is impossible, by \eqref{eq:subttt-liveness:sub2-in:tt};
        \item%
          $\pp$ is receiving an input from $\pq$, %
          i.e., %
          $\Gamma'(\pq) = \Gamma(\pq) = %
          (\tout\pp{\ell_i}{\S_\pq}\cdot \tqueue_\pq, \TT_\pq)$ %
          with $\S_\pq \subs \S'_i$ and %
          $\Gamma''(\pp) = (\tqueue', \TT'_i)$ (for some $i \in I$). %
          Notice that, since $\Gamma$ is $\pp$-live, %
          the same output from $\pq$ can be received by %
          $\pp$ in $\Gamma$, %
          and thus, there is a corresponding reduction $\Gamma \red \Gamma'''$ %
          where $\Gamma'''(\pp) = (\tqueue, \TT_i)$ %
          and $\Gamma'' = \Gamma'''\mapUpdate{\pp}{(\tqueue', \TT'_i)}$. %
          Observe that $\Gamma'''$ is also $\pp$-live, %
          and thus, $\Gamma''' \in \liveP$; %
          therefore, by \eqref{eq:subttt-liveness:sub2-in:st-wt} %
          and \eqref{eq:subttt-liveness:prop:def}, %
          we have $\Gamma'' \in \liveP' \subseteq \liveSubP$. %
          Hence, we conclude that $\Gamma'$ satisfies %
          clause~\rulename{LP$\red$} of Def.~\ref{def:coinductive-liveness};
        \item%
          one of $\pp$'s queued outputs in $\tqueue'$ %
          is received by another participant. %
          This case is impossible, %
          by \eqref{eq:subttt-liveness:sub2-in:tqueue}.
        \end{enumerate}
      \end{enumerate}
    \end{itemize}
  \item%
    Case \rulename{ref-$\AC$}.\quad%
    In this case, we have:
    \begin{align}
      &\label{eq:subttt-liveness:sub2-a:tqueue}%
        \tqueue = \tqueue' = \temptyqueue%
        \\%
      &\label{eq:subttt-liveness-sub2-a:tti}%
        \exists I, I', \ell_i, \ST'_i, \TT'_i:\; %
        \TT' = \texternal_{i \in I \cup I'}\tin\pq{\ell_i}{\ST'_i}.\TT'_i%
        \\%
      &\nonumber%
        \text{for all $\WT_1, \WT_2$ in \eqref{eq:subttt-liveness:prop:siso-trees}, }
        \exists i \in I, \AContext{\pq}, \ell_i, \ST_i, \ST'_i, \WT, \WT', \;\text{such that:}%
      \\%
      &\label{eq:subttt-liveness:sub2-a:tt}%
        \qquad%
        \WT_1 = \tin\pq{\ell_i}{\ST'_i}.\WT' \quad\text{and}\quad%
        \WT_2 = \ACon{\pq}{\tin\pq{\ell_i}{\ST_i}.\WT}%
      \\%
      &\label{eq:subttt-liveness:sub2-a:st-wt-actions}%
        \qquad%
        \ST_i \subs \ST'_i \quad\text{and}\quad  \WT' \subttt \ACon{\pq}{\WT}%
        \quad\text{and}\quad%
        \actions{\WT'}=\actions{\ACon\pq{\WT}}
    \end{align}
    We now show that $\Gamma'$ satisfies all clauses %
    of Def.~\ref{def:coinductive-liveness}:
    \begin{itemize}
    \item%
      clause \rulename{LP$\&$}.\quad%
      We proceed by contradiction: %
      we show that if there is a fair path $(\Gamma'_j)_{j \in J'}$
      (with $\Gamma'_0 = \Gamma'$) %
      that violates clause \rulename{LP$\&$}, %
      then there is a corresponding fair path $(\Gamma_j)_{j \in J}$
      (with $\Gamma_0 = \Gamma$) %
      that violates the same clause, %
      which would lead to the absurd conclusion that %
      $\Gamma$ is \emph{not} $\pp$-live. %
      Such a hypothetical path $(\Gamma'_j)_{j \in J'}$ %
      consists of a series of transitions %
      $\Gamma'_{j-1} \redLabel{\;\alpha_j\;} \Gamma'_j$ (for $j \in J'$) %
      where, $\forall j \in J'$, $\alpha_j$ does \emph{not} involve $\pp$ %
      (by \eqref{eq:subttt-liveness:sub2-a:tqueue}). %
      This means that:
      \begin{align}
        &\label{eq:subttt-liveness:sub2-a:no-in-from-q}
          \forall j \in J': \not\exists \TT_\pq, \tqueue_\pq, \S_\pq:%
          \Gamma'_j(\pq) = (\tout{\pp}{\ell_i}{\S_\pq}{\cdot}\tqueue_{\pq}, \TT_\pq)%
          \text{ with } \S_\pq \subs \S'_i \text{ (for any $i \in I$)}%
      \end{align}
      But then, by Prop.~\ref{lem:subttt-no-output}, %
      there is a fair path $(\Gamma_j)_{j \in J}$ with $\Gamma_0 = \Gamma$ %
      where $\pp$ reduces according %
      to some $\WT_2$ in \eqref{eq:subttt-liveness:sub2-a:tt}, %
      and such that:
      \begin{align}
        &\label{eq:subttt-liveness:sub2-a:no-in-from-q-gamma}%
          \forall j \in J: \not\exists \TT_\pq, \tqueue_\pq, \S_\pq:%
          \Gamma_j(\pq) = (\tout{\pp}{\ell_i}{\S_\pq}{\cdot}\tqueue_{\pq}, \TT_\pq)%
          \text{ with } \S_\pq \subs \S_i \text{ (for any $i \in I$)}%
      \end{align}
      Now, observe that %
      the fair path $(\Gamma_j)_{j \in J}$ is constructed %
      using Prop.~\ref{lem:subt-reductions}, %
      and therefore, will eventually attempt to fire %
      the input $\tin{\pq}{\ell_i}{\ST_i}$ of $\WT_2$ %
      in \eqref{eq:subttt-liveness:sub2-a:tt} %
      --- but no suitable output will be available, by %
      \eqref{eq:subttt-liveness:sub2-a:no-in-from-q-gamma}: %
      thus, we obtain that $\Gamma$ is not $\pp$-live --- contradiction. %
        Therefore, we conclude that $\Gamma'$ %
        satisfies clause $\rulename{LP$\&$}$ %
        of Def.~\ref{def:coinductive-liveness};%
    \item%
      clause \rulename{LP$\oplus$}.\quad%
      The clause is vacuously satisfied;
    \item%
      clause \rulename{LP$\red$}.\quad%
      Assume $\Gamma' \red \Gamma''$.  We have two possibilities:
      \begin{enumerate}[label=(\alph*)]
      \item%
        the reduction does \emph{not} involve $\pp$. %
        The proof is similar to case %
        \rulename{ref-in}\rulename{LP$\red$}\ref{item:subttt-livenes:move-not-p} above;
      \item\label{item:subttt-livenes:sub2-ina-move-p}%
        the reduction \emph{does} involve $\pp$. %
        There are three sub-cases:
        \begin{enumerate}[label=(\roman*)]
        \item%
          $\pp$ is enqueuing an output toward some participant $\pr$. %
          This case is impossible, by \eqref{eq:subttt-liveness-sub2-a:tti};%
        \item%
          $\pp$ is receiving an input from $\pq$, %
          i.e., for some $i \in I$:
          \begin{align}
            \label{eq:subttt-liveness:sub2-a:red:pq}%
            &\Gamma'(\pq) = \Gamma(\pq) = %
            (\tout\pp{\ell_i}{\S_\pq}\cdot \tqueue_\pq, \TT_\pq)%
            \;\text{ with $\S_\pq \subs \S'_i$}%
            \\%
            \label{eq:subttt-liveness:sub2-a:red:gammaii}%
            &\Gamma' \redRecv{\pp}{\pq}{\ell_i} \Gamma'' %
            \;\text{ with $\Gamma''(\pp) = (\tqueue', \TT'_i)$}
          \end{align}%
          By Prop.~\ref{lem:subt-reductions}, %
          the fair paths of $\Gamma$ %
          that match some witness %
          $\WT_2$ in \eqref{eq:subttt-liveness:sub2-a:tt} %
          eventually perform %
          the reduction $\alpha$ above. %
          But then, consider the session tree $\TT^*$ %
          that only has SISO trees %
          similar to $\WT_2$ in \eqref{eq:subttt-liveness:prop:siso-trees}, %
          except that each one performs the input $\tin{\pq}{\ell_j}{\ST}$ %
          immediately --- \ie:
          \begin{align}
            &\label{eq:subttt-liveness:sub2-a:siso-trees-2}%
              \hspace{-5mm}\begin{array}{@{}l@{}}
                \forall \UT'\in\llbracket \tqueue'{\cdot}\TT'\rrbracket_\SO \quad
                \forall \VT^* \in \llbracket \tqueue{\cdot}\TT^* \rrbracket_\SI \\
                \exists \WT_1=\tin\pq{\ell_i}{\ST'_i}.\WT' \in\llbracket \UT' \rrbracket_\SI \quad
                \exists \WT^*=\tin\pq{\ell_i}{\ST_i}.\ACon{\pq}{\WT} \in \llbracket \VT^* \rrbracket_\SO \quad \WT_1 \subttt \WT^*
                \end{array}
          \end{align}
          And now, consider the typing environment %
          $\Gamma^* = \Gamma\mapUpdate{\pp}{(\tqueue, \TT^*)}$. %
          Such $\Gamma^*$ is $\pp$-live: %
          it realises (part of) the fair paths of $\Gamma$, %
          except that it consumes $\pq$'s queued output earlier %
          --- and thus, we have $\Gamma^* \in \liveP$. %
          Now, consider $\Gamma'''$ such that %
          $\Gamma^* \redRecv{\pp}{\pq}{\ell_i} \Gamma'''$: %
          by clause \rulename{LP$\red$} of Def.~\ref{def:coinductive-liveness},
          we have $\Gamma''' \in \liveP$. %
          Also, we have $\Gamma'''(\pp) = (\tqueue, \TT^{**})$ %
          such that:
          \begin{align}
            \nonumber%
             &\hspace{-5mm}\begin{array}{@{}l@{}}
                \forall \UT'\in\llbracket \tqueue'{\cdot}\TT'_i\rrbracket_\SO \quad
                \forall \VT^{**} \in \llbracket \tqueue{\cdot}\TT^{**} \rrbracket_\SI \\
                \exists \WT_1=\WT' \in\llbracket \UT' \rrbracket_\SI \quad
                \exists \WT^{*}=\ACon{\pq}{\WT} \in \llbracket \VT^{**} \rrbracket_\SO \quad \WT_1 \subttt \WT^*
                \end{array}
            &\text{(by Def.~\ref{def:typing-env-reductions} %
              and  \eqref{eq:subttt-liveness:sub2-a:st-wt-actions})}%
          \end{align}
          Therefore, by Def.~\ref{def:subtyping}, %
          we have $\Gamma'''(\pp) = (\tqueue, \TT^{**})$, %
          and %
          $\Gamma'' = \Gamma'''\mapUpdate{\pp}{(\tqueue', \TT'_i)}$, %
          with $\tqueue'{\cdot}\TT'_i \subt \tqueue{\cdot}\TT^{**}$
          --- hence, by \eqref{eq:subttt-liveness:sub2-a:st-wt-actions} %
          and \eqref{eq:subttt-liveness:prop:def}, %
          we also have $\Gamma'' \in \liveP' \subseteq \liveSubP$. %
          Thus, we conclude that $\Gamma'$ satisfies %
          clause~\rulename{LP$\red$} of Def.~\ref{def:coinductive-liveness};
        \item%
          one of $\pp$'s queued outputs in $\tqueue'$ %
          is received by another participant. %
          This case is impossible, %
          by \eqref{eq:subttt-liveness:sub2-a:tqueue}.
        \end{enumerate}
      \end{enumerate}
    \end{itemize}
  \item%
    Case \rulename{ref-out}.\quad%
    In this case, we have:
    \begin{align}
      &\nonumber%
        \exists \pq, \ell, \tqueue, \tqueue', \ST, \ST', \TT_1, \TT_2 \;\text{such that:}%
      \\%
      &\label{eq:subttt-liveness:sub2-out:tt-tqueue}%
        \qquad%
        \tqueue'{\cdot}\TT' = \tout\pq{\ell}{\ST'}.\TT_1 \quad\text{and}\quad%
        \tqueue{\cdot}\TT = \tout\pq{\ell}{\ST}.\TT_2%
      \\%
      &\label{eq:subttt-liveness:sub2-out:st-wt}%
        \qquad%
        \ST' \subs \ST  \quad\text{and}\quad  \TT_1 \subt \TT_2%
    \end{align}
    We now show that $\Gamma'$ satisfies all clauses %
    of Def.~\ref{def:coinductive-liveness}:
    \begin{itemize}
    \item%
      clause \rulename{LP$\&$}.\quad%
      If $\TT'$ does \emph{not} begin with an external choice %
      $\tin{\pr}{\ell_\pr}{\ST'_\pr}$, the clause is vacuously satisfied. %
      In the case where %
      $\TT' = \texternal_{i \in I}\tin{\pr}{\ell_{\pr, i}}{\ST'_{\pr, i}}.\TT_i$, %
      the proof is similar to case %
      \rulename{ref-$\AC$}\rulename{LP$\&$} above;
      Thus, we conclude that $\Gamma'$ satisfies %
      clause~\rulename{LP$\&$} of Def.~\ref{def:coinductive-liveness};
    \item%
      clause \rulename{LP$\oplus$}.\quad%
      We proceed by contradiction: %
      we show that if there is a fair path $(\Gamma'_j)_{j \in J'}$
      (with $\Gamma'_0 = \Gamma'$) %
      that violates clause \rulename{LP$\oplus$}, %
      then there is a corresponding fair path $(\Gamma_j)_{j \in J}$
      (with $\Gamma_0 = \Gamma$) %
      that violates the same clause, %
      which would lead to the absurd conclusion that %
      $\Gamma$ is \emph{not} $\pp$-live. %
      We need to consider any output at the head of %
      the queue of $\Gamma'(\pp)$ up-to reordering via $\equiv$: %
      let such an output be %
      $\tout\pr{\ell_\pr}{\ST'_\pr}$ (for some $\pr$). %
      Consider the hypothetical non-$\pp$-live path $(\Gamma'_j)_{j \in J'}$: %
      it must consist of a series of transitions %
      $\Gamma'_{j-1} \redLabel{\;\alpha_j\;} \Gamma'_j$ (for $j \in J'$) %
      where, $\forall j \in J'$, $\alpha_j \neq \recvLabel{\pr}{\pp}{\ell}$. %
      Hence:
      \begin{align}
        &\label{eq:subttt-liveness:sub2-out:no-in-from-r}
          \forall j \in J':%
          \not\exists \tqueue_\pr, I, \ell_i, \S_{\pr,i}, \TT_{\pr, i} (i \in I):\; %
          \Gamma'_j(\pr) = \left(\tqueue_{\pr}, \texternal_{i\in I}\tin{\pp}{\ell_i}{\S_{\pr, i}}.\TT_{\pr, i}\right)%
        \\%
        &\nonumber%
          \qquad\qquad\text{with } %
          \ell_i = \ell_\pr \text{ and } \S'_\pr \subs \S_{\pr, i} %
          \text{ (for some $i \in I$})%
      \end{align}
      But then, by \eqref{eq:subttt-liveness:sub2-out:no-in-from-r} %
      and Prop.~\ref{lem:subttt-no-input}, %
      there is a fair path $(\Gamma_j)_{j \in J}$ with $\Gamma_0 = \Gamma$ %
      such that:
      \begin{align}
        &\label{eq:subttt-liveness:sub2-out:no-in-from-r2}%
          \forall j \in J:%
          \not\exists \tqueue_\pr, I, \ell_i, \S_{\pr,i}, \TT_{\pr, i} (i \in I):\; %
          \Gamma_j(\pr) = \left(\tqueue_{\pr}, \texternal_{i\in I}\tin{\pp}{\ell_i}{\S_{\pr, i}}.\TT_{\pr, i}\right)%
        \\%
        &\nonumber%
          \qquad\qquad\text{with } %
          \ell_i = \ell_\pr \text{ and } \S'_\pr \subs \S_{\pr, i} %
          \text{ (for some $i \in I$})%
      \end{align}
      Now, observe that %
      the path $(\Gamma_j)_{j \in J}$ is constructed %
      using Prop.~\ref{lem:subt-reductions}; %
      therefore, by \eqref{eq:subttt-liveness:prop:gammai}, %
      \eqref{eq:subttt-liveness:prop:gammap}, %
      \eqref{eq:subttt-liveness:prop:sub}, %
      and Prop.~\ref{lem:siso-tree-projections-equal-subt}, %
      an output $\tout\pr{\ell_\pr}{\ST_\pr}$ (with $\ST'_\pr \subt \ST_\pr$) %
      occurs along $(\Gamma_j)_{j \in J}$; %
      and by \eqref{eq:subttt-liveness:sub2-out:no-in-from-r2}, %
      such an output will never be consumed: %
      this means that $\Gamma$ is not $\pp$-live --- contradiction. %
      Thus, we conclude that $\Gamma'$ satisfies %
      clause~\rulename{LP$\oplus$} of Def.~\ref{def:coinductive-liveness};

    \item%
      clause \rulename{LP$\red$}.\quad%
      Assume $\Gamma' \red \Gamma''$.  We have two possibilities:
      \begin{enumerate}[label=(\alph*)]
      \item%
        the reduction does \emph{not} involve $\pp$. %
        The proof is similar to case %
        \rulename{ref-in}\rulename{LP$\red$}\ref{item:subttt-livenes:move-not-p} above;
      \item\label{item:subttt-liveness:sub2-out:red-p}%
        the reduction \emph{does} involve $\pp$. %
        There are three sub-cases:
        \begin{enumerate}[label=(\roman*)]
        \item\label{item:subttt-liveness:sub2-out:red-p:p-enq-r}%
          $\pp$ is enqueuing an output toward some participant $\pr$. %
          In this case, we have: %
          \begin{align}
            \label{item:subttt-liveness:sub2-out:red-p:enq}%
            &\TT' \;=\;%
              \tinternal_{i \in I}\tout{\pr}{\ell_{i}}{\ST'_{i}}.\TT'_i %
            \\%
            \label{item:subttt-liveness:sub2-out:red-p:enq:gammaii}%
            &\Gamma' \redSend{\pp}{\pr}{\ell_i} \Gamma'' = %
            \Gamma'\mapUpdate{\pp}{%
              (\tqueue'{\cdot}\tout{\pr}{\ell_{i}}{\ST'_{i}}, \TT'_i)
            }%
            \;\text{ (for some $i \in I$)}%
            \hspace{-50mm}%
            &\text{(by Def.~\ref{def:typing-env-reductions})}%
            \\%
            \label{item:subttt-liveness:sub2-out:red-p:enq:gammaii-gamma}%
            &\Gamma'' = %
            \Gamma''\mapUpdate{\pp}{%
              (\tqueue'{\cdot}\tout{\pr}{\ell_{i}}{\ST'_{i}}, \TT'_i)
            }%
            \hspace{-20mm}%
            &\text{(by \eqref{item:subttt-liveness:sub2-out:red-p:enq:gammaii} %
              and \eqref{eq:subttt-liveness:prop:gammai})}
            \\%
            \label{item:subttt-liveness:sub2-out:red-p:enq:subt}
            &
              \tqueue'{\cdot}\tout{\pr}{\ell_{i}}{\ST'_{i}}{\cdot}\TT'_i%
              \subt \tqueue'{\cdot}\TT'%
            &\text{(by induction on $\tqueue'$, %
              using \eqref{item:subttt-liveness:sub2-out:red-p:enq} %
              and Def.~\ref{def:subtyping})}%
            \\%
            \label{item:subttt-liveness:sub2-out:red-p:enq:subt-tra}
            &\tqueue'{\cdot}\tout{\pr}{\ell_{i}}{\ST'_{i}}{\cdot}\TT'_i%
              \subt \tqueue{\cdot}\TT%
            &\text{(by \eqref{item:subttt-liveness:sub2-out:red-p:enq:subt}, %
              \eqref{eq:subttt-liveness:prop:sub}, %
              and Lemma~\ref{lem:transitivity})}
            \\%
            \label{item:subttt-liveness:sub2-out:red-p:enq:gammaii-in-l}
            &\Gamma'' \in \liveP' \subseteq \liveSubP%
            &\text{(by \eqref{item:subttt-liveness:sub2-out:red-p:enq:subt-tra} and \eqref{item:subttt-liveness:sub2-out:red-p:enq:gammaii-gamma})}
          \end{align}
          Therefore, by \eqref{item:subttt-liveness:sub2-out:red-p:enq:gammaii-in-l}, %
          we conclude that $\Gamma'$ satisfies %
          clause~\rulename{LP$\red$} of Def.~\ref{def:coinductive-liveness};
        \item\label{item:subttt-liveness:sub2-out:red-p:recv}%
          $\pp$ is receiving an input from some $\pr$. %
          In this case, we have:%
          \begin{align}
            \label{eq:subttt-liveness:sub2-out:red:p-in:r}%
            &\exists I, \ell_i, \ST'_i, \TT'_i:\; %
            \TT' = \texternal_{i \in I}\tin\pr{\ell_i}{\ST'_i}.\TT'_i%
            \\%
            \label{eq:subttt-liveness:sub2-out:red:alpha}%
            &\Gamma' \redLabel{\alpha} \Gamma''%
            \;\text{ with $\alpha = \recvLabel{\pp}{\pr}{\ell_i}$ %
            (for some $i \in I$)}%
            \\%
            \label{eq:subttt-liveness:sub2-out:red:p-in:queues}%
            &\tqueue = \temptyqueue \;\text{ or }\; %
              head(\tqueue) = \tout\pq{\ell}{\ST'}
            &\text{(by \eqref{eq:subttt-liveness:sub2-out:tt-tqueue})}%
          \end{align}
          Now, from %
          \eqref{eq:subttt-liveness:sub2-out:red:p-in:queues}, %
          \eqref{eq:subttt-liveness:sub2-out:tt-tqueue} %
          and \eqref{eq:subttt-liveness:sub2-out:st-wt}, %
          \begin{align}
            &\nonumber%
              \text{letting }\;%
              \tqueue_0' \text{ such that } %
              \tqueue' = \tout\pq{\ell}{\ST'}.\tqueue_0' %
              \;\text{ and }\; %
              \tqueue_0 = \begin{cases}
                \temptyqueue & \text{if $\tqueue = \temptyqueue$} \\
                tail(\tqueue) & \text{ otherwise}%
              \end{cases}
            \\%
            &\nonumber%
              \text{for all $\WT_1, \WT_2$ in \eqref{eq:subttt-liveness:prop:siso-trees}, }
              \exists i \in I, \DContext{\pr}, \ell_i, \ST_i, \ST'_i, \WT, \WT', \;\text{such that:}%
            \\%
            &\label{eq:subttt-liveness:sub2-out:tt}%
              \qquad%
              \WT_1 = \tout\pq{\ell}{\ST'}.\tqueue_0'.\tin\pr{\ell_i}{\ST'_i}.\WT' \quad\text{and}\quad%
              \WT_2 = \tout\pq{\ell}{\ST}.\tqueue_0.\DCon{\pr}{\tin\pr{\ell_i}{\ST_i}.\WT}%
            \\%
            &\label{eq:subttt-liveness:sub2-out:st-wt-d}%
              \qquad%
              \ST_i \subs \ST'_i \quad\text{and}\quad  %
              \tqueue_0'.\tin\pr{\ell_i}{\ST'_i}.\WT'%
              \subttt \tqueue_0.\DCon{\pr}{\tin\pr{\ell_i}{\ST_i}.\WT}%
          \end{align}
          where $\DContext{\pr}$ is a sequence of outputs to any participant, %
          or inputs from any participant except $\pr$, %
          for which \eqref{eq:subttt-liveness:sub2-out:st-wt-d} %
          can be derived with $0$ or more instances of %
          \rulename{ref-${\AC}$} or \rulename{ref-$\BC$}. %
          Notice that, by induction on $\tqueue'_0$ and %
          $\tqueue_0.\DContext{\pr}$, %
          for each pair of SISO trees %
          related in \eqref{eq:subttt-liveness:sub2-out:st-wt-d} %
          we prove that:
          \begin{align}
            &\label{eq:subttt-liveness:sub2-out:wt-sub-swap}%
              \qquad%
              \tqueue_0'.\tin\pr{\ell_i}{\ST'_i}.\WT'%
              \subttt \tqueue_0.\tin\pr{\ell_i}{\ST_i}.\DCon{\pr}{\WT}%
          \end{align}

          Now, by Prop.~\ref{lem:subt-reductions}, %
          the fair paths of $\Gamma$ %
          that match some witness %
          $\WT_2$ in \eqref{eq:subttt-liveness:sub2-out:tt} %
          eventually perform %
          the reduction $\alpha$ %
          in \eqref{eq:subttt-liveness:sub2-out:red:alpha}. %
          But then, consider the session tree $\TT^*$ %
          that only has SISO trees %
          like the RHS of \eqref{eq:subttt-liveness:sub2-out:wt-sub-swap}: %
          such trees are similar to $\WT_2$ in %
          \eqref{eq:subttt-liveness:sub2-out:tt}, %
          except that each one performs the input $\tin{\pr}{\ell_i}{\S_i}$ %
          earlier. %
          And now, consider the typing environment %
          $\Gamma^* = \Gamma\mapUpdate{\pp}{(\tqueue, \TT^*)}$. %
          Such $\Gamma^*$ is $\pp$-live: %
          it realises (part of) the fair paths of $\Gamma$, %
          except that it consumes $\pr$'s queued output earlier %
          --- and thus, we have $\Gamma^* \in \liveP$. %
          Now, consider $\Gamma'''$ such that %
          $\Gamma^* \redRecv{\pp}{\pr}{\ell_i} \Gamma'''$: %
          by clause \rulename{LP$\red$} of Def.~\ref{def:coinductive-liveness},
          we have $\Gamma''' \in \liveP$. %
          Also, we have $\Gamma'''(\pp) = (\tqueue, \TT^{**})$ %
          such that:
          \begin{align}
            \nonumber%
             &\hspace{-5mm}\begin{array}{@{}l@{}}
                \forall \UT'\in\llbracket \tqueue'{\cdot}\TT'_i\rrbracket_\SO \quad
                \forall \VT^{**} \in \llbracket \tqueue{\cdot}\TT^{**} \rrbracket_\SI \\
                \exists \WT_1=\tqueue_0'.\WT' \in\llbracket \UT' \rrbracket_\SI \quad
                \exists \WT^{*}=\tqueue_0.\DCon{\pr}{\WT} \in \llbracket \VT^{**} \rrbracket_\SO \quad \WT_1 \subttt \WT^*
                \end{array}
            &\text{(by Def.~\ref{def:typing-env-reductions} %
              and  \eqref{eq:subttt-liveness:sub2-out:wt-sub-swap})}%
          \end{align}
          Therefore, by Def.~\ref{def:subtyping}, %
          we have $\Gamma'''(\pp) = (\tqueue, \TT^{**})$, %
          and %
          $\Gamma'' = \Gamma'''\mapUpdate{\pp}{(\tqueue', \TT'_i)}$, %
          with $\tqueue'{\cdot}\TT'_i \subt \tqueue{\cdot}\TT^{**}$
          --- hence, by \eqref{eq:subttt-liveness:sub2-a:st-wt-actions} %
          and \eqref{eq:subttt-liveness:prop:def}, %
          we also have $\Gamma'' \in \liveP' \subseteq \liveSubP$. %
          Thus, we conclude that $\Gamma'$ satisfies %
          clause~\rulename{LP$\red$} of Def.~\ref{def:coinductive-liveness};
        \item\label{item:subttt-liveness:sub2-out:red-p:recv-qr}%
          one of $\pp$'s queued outputs in $\tqueue'$ %
          is received by another participant. %
          We have two possibilities:
          \begin{enumerate}[label=\emph{(\alph*)},start=3]
          \item\label{item:subttt-liveness:sub2-out:red-p:recv-r}%
            the received output is $\tout{\pr}{\ell_\pr}{\ST'_\pr}$ %
            (for some $\pr \neq \pq$).\quad %
            In this case, we have:%
            \begin{align}
              \label{eq:subttt-liveness:sub2-out:red:recv-r:out-r-cong}%
              &\exists I, \ell_\pr, \ST'_\pr \tqueue'':%
                \tqueue' \equiv \tout\pr{\ell_\pr}{\ST'_\pr}{\cdot}\tqueue''
              \\%
              \label{eq:subttt-liveness:sub2-out:red:recv-r:out-r}%
              &\exists I, \ell_\pr, \ST'_\pr, \tqueue'_0, \tqueue'_1:%
                \tqueue' = \tqueue'_0 {\cdot} \tout\pr{\ell_\pr}{\ST'_\pr}{\cdot}\tqueue'_1 \;\text{with $\pr! \not\in \actions{\tqueue'_0}$}%
              &\text{(by \eqref{eq:subttt-liveness:sub2-out:red:recv-r:out-r-cong} and queue congruence)}%
              \\%
              \label{eq:subttt-liveness:sub2-out:red:recv-r:alpha}%
              &\Gamma' \redLabel{\alpha} \Gamma''%
                \;\text{ with $\alpha = \recvLabel{\pr}{\pp}{\ell_\pr}$}%
              \\%
              &\nonumber%
                \text{for all $\WT_1, \WT_2$ in \eqref{eq:subttt-liveness:prop:siso-trees}, }
                \BContext{\pr}, \ell_\pr, \ST_\pr, \ST'_\pr, \WT, \WT'', \;\text{such that:}%
              \\%
              &\label{eq:subttt-liveness:sub2-out:recv-r:tt}%
                \qquad%
                \WT_1 = \tqueue'_0.\tout\pr{\ell_\pr}{\ST'_\pr}.\tqueue_1'.\WT' \quad\text{and}\quad%
                \WT_2 = \BCon{\pr}{\tout\pr{\ell_\pr}{\ST_\pr}.\WT''}%
              \\%
              &\label{eq:subttt-liveness:sub2-out:recv-r:st-wt}%
                \qquad%
                \ST'_\pr \subs \ST_\pr \quad\text{and}\quad  %
                \tqueue'_0.\tout\pr{\ell_\pr}{\ST'_\pr}.\tqueue_1'.\WT'%
                \subttt \BCon{\pr}{\tout\pr{\ell_\pr}{\ST_\pr}.\WT''}%
            \end{align}
            Therefore, by \eqref{eq:subttt-liveness:sub2-out:recv-r:st-wt}, %
            and by induction on $\tqueue'_0$ %
            using \eqref{eq:subttt-liveness:sub2-out:red:recv-r:out-r}, %
            \begin{align}
              &\label{eq:subttt-liveness:sub2-out:recv-r:st-wt-pr1}%
                \tout\pr{\ell_\pr}{\ST'_\pr}.\tqueue'_0.\tqueue_1'.\WT'%
                \subttt%
                \tqueue'_0.\tout\pr{\ell_\pr}{\ST'_\pr}.\tqueue_1'.\WT'%
                \hspace{-10mm}%
            \end{align}
          Notice that, by induction on $\BContext{\pr}$, %
          for each pair of SISO trees %
          related in \eqref{eq:subttt-liveness:sub2-out:recv-r:st-wt} %
          we prove that:
          \begin{align}
            &\label{eq:subttt-liveness:sub2-out:recv-r:wt-sub-swap}%
              \qquad%
              \tqueue_0'.\tout\pr{\ell_\pr}{\ST'_\pr}.\tqueue_1'.\WT'%
              \subttt \tout\pr{\ell_\pr}{\ST_\pr}.\BCon{\pr}{\WT''}%
          \end{align}
          From which we get:
          \begin{align}
            &\label{eq:subttt-liveness:sub2-out:recv-r:st-wt-pr2}%
              \tout\pr{\ell_\pr}{\ST'_\pr}.\tqueue'_0.\tqueue_1'.\WT'%
              \subttt \tout\pr{\ell_\pr}{\ST_\pr}.\BCon{\pr}{\WT''}%
              \hspace{-3mm}%
            &\text{(by \eqref{eq:subttt-liveness:sub2-out:recv-r:st-wt-pr1},
              \eqref{eq:subttt-liveness:sub2-out:recv-r:wt-sub-swap}, %
              and Lemma~\ref{lem:transitivity})}%
            \\%
            &\label{eq:subttt-liveness:sub2-out:recv-r:st-wt-pr2-cont}%
              \tqueue'_0.\tqueue_1'.\WT'%
              \subttt \BCon{\pr}{\WT''}%
            &\text{(by \eqref{eq:subttt-liveness:sub2-out:recv-r:st-wt-pr2} %
              and \rulename{ref-out})}
          \end{align}

          Now, by Prop.~\ref{lem:subt-reductions}, %
          the fair paths of $\Gamma$ %
          that match some witness %
          $\WT_2$ in \eqref{eq:subttt-liveness:sub2-out:recv-r:tt} %
          eventually perform %
          the reduction $\alpha$ %
          in \eqref{eq:subttt-liveness:sub2-out:red:recv-r:alpha}. %
          But then, consider the session tree $\TT^*$ %
          that only has SISO trees %
          like the RHS of \eqref{eq:subttt-liveness:sub2-out:recv-r:wt-sub-swap}: %
          such trees are similar to %
          \eqref{eq:subttt-liveness:sub2-out:recv-r:tt}, %
          except that each one performs the output $\tout{\pr}{\ell_\pr}{\S_\pr}$ %
          earlier. %
          And now, consider the typing environment %
          $\Gamma^* = \Gamma\mapUpdate{\pp}{(\tqueue, \TT^*)}$. %
          Such $\Gamma^*$ is $\pp$-live: %
          it realises (part of) the fair paths of $\Gamma$, %
          except that performs the output to $\pr$ earlier %
          --- and thus, we have $\Gamma^* \in \liveP$. %
          Now, consider $\Gamma'''$ such that %
          $\Gamma^* \redRecv{\pr}{\pp}{\ell_\pr} \Gamma'''$: %
          by clause \rulename{LP$\red$} of Def.~\ref{def:coinductive-liveness},
          we have $\Gamma''' \in \liveP$. %
          Also, we have $\Gamma'''(\pp) = (\tqueue, \TT^{**})$ %
          such that:
          \begin{align}
            \nonumber%
             &\hspace{-5mm}\begin{array}{@{}l@{}}
                \forall \UT'\in\llbracket \tqueue'{\cdot}\TT'_i\rrbracket_\SO \quad
                \forall \VT^{**} \in \llbracket \tqueue{\cdot}\TT^{**} \rrbracket_\SI \\
                \exists \WT_1=\tqueue_0'.\tqueue_1'.\WT' \in\llbracket \UT' \rrbracket_\SI \quad
                \exists \WT^{*}=\BCon{\pr}{\WT''} \in \llbracket \VT^{**} \rrbracket_\SO \quad \WT_1 \subttt \WT^*
                \end{array}
            &\text{(by Def.~\ref{def:typing-env-reductions} %
              and  \eqref{eq:subttt-liveness:sub2-out:recv-r:st-wt-pr2-cont})}%
          \end{align}
          Therefore, by Def.~\ref{def:subtyping}, %
          we have $\Gamma'''(\pp) = (\tqueue, \TT^{**})$, %
          and %
          $\Gamma'' = \Gamma'''\mapUpdate{\pp}{(\tqueue', \TT'_i)}$, %
          with $\tqueue'{\cdot}\TT'_i \subt \tqueue{\cdot}\TT^{**}$
          --- hence, by \eqref{eq:subttt-liveness:sub2-a:st-wt-actions} %
          and \eqref{eq:subttt-liveness:prop:def}, %
          we also have $\Gamma'' \in \liveP' \subseteq \liveSubP$. %
          Thus, we conclude that $\Gamma'$ satisfies %
          clause~\rulename{LP$\red$} of Def.~\ref{def:coinductive-liveness};

          \item%
            the received output is $\tout{\pq}{\ell}{\ST'}$ %
            from \eqref{eq:subttt-liveness:sub2-out:tt-tqueue}.\quad%
            The proof is similar to %
            case~\ref{item:subttt-liveness:sub2-out:red-p:recv-r} above, %
            letting $\pr = \pq$ --- %
            but the development is simpler, %
            since we have $\tqueue'_0 = \temptyqueue$, %
            and $\BContext{\pq}$ is empty;
          \end{enumerate}
        \end{enumerate}
      \end{enumerate}
    \end{itemize}
  \item%
    case \rulename{ref-$\BC$}.\quad%
    In this case, we have:
    \begin{align}
      &\label{eq:subttt-liveness-ref-b:tti}%
        \exists I, \ell_i, \ST'_i, \TT'_i:\; %
        \tqueue'{\cdot}\TT' = \tinternal_{i \in I}\tout\pq{\ell_i}{\ST'_i}.\TT'_i%
        \\%
      &\nonumber%
        \text{for all $\WT_1, \WT_2$ in \eqref{eq:subttt-liveness:prop:siso-trees}, }
        \exists i \in I, \BContext{\pq}, \ell_i, \ST_i, \ST'_i, \WT, \WT', \;\text{such that:}%
      \\%
      &\label{eq:subttt-liveness:ref-b:tt}%
        \qquad%
        \WT_1 = \tout\pq{\ell_i}{\ST'_i}.\WT' \quad\text{and}\quad%
        \WT_2 = \BCon{\pq}{\tout\pq{\ell_i}{\ST_i}.\WT}%
      \\%
      &\label{eq:subttt-liveness:ref-b:st-wt-actions}%
        \qquad%
        \ST'_i \subs \ST_i \quad\text{and}\quad  \WT' \subttt \BCon{\pq}{\WT}%
        \quad\text{and}\quad%
        \actions{\WT'}=\actions{\BCon\pq{\WT}}%
    \end{align}
    We now show that $\Gamma'$ satisfies all clauses %
    of Def.~\ref{def:coinductive-liveness}:
    \begin{itemize}
    \item%
      clause \rulename{LP$\&$}.\quad%
      The proof is similar to case \rulename{ref-out}\rulename{LP$\&$} above;%
    \item%
      clause \rulename{LP$\oplus$}.\quad%
      The proof is similar to case \rulename{ref-out}\rulename{LP$\oplus$} above;%
    \item%
      clause \rulename{LP$\red$}.\quad%
      Assume $\Gamma' \red \Gamma''$.  We have two possibilities:
      \begin{enumerate}[label=(\alph*)]
      \item%
        the reduction does \emph{not} involve $\pp$. %
        The proof is similar to case %
        \rulename{ref-in}\rulename{LP$\red$}\ref{item:subttt-livenes:move-not-p} above;
      \item%
        the reduction \emph{does} involve $\pp$. %
        There are three sub-cases:
        \begin{enumerate}[label=(\roman*)]
        \item%
          $\pp$ is enqueuing an output toward some participant $\pr$. %
          In this case, we have: %
          \begin{align}
            \nonumber%
            &\TT' \;=\;%
              \tinternal_{i \in J}\tout{\pr}{\ell_{j}}{\ST'_{j}}.\TT''_j%
          \end{align}
          and the proof is similar to case %
          \rulename{ref-out}\rulename{LP$\red$}\ref{item:subttt-liveness:sub2-out:red-p}\ref{item:subttt-liveness:sub2-out:red-p:p-enq-r} above;
        \item%
          $\pp$ is receiving an input from some $\pr$. %
          The proof is similar to case %
          \rulename{ref-out}\rulename{LP$\red$}\ref{item:subttt-liveness:sub2-out:red-p}\ref{item:subttt-liveness:sub2-out:red-p:recv} above;
        \item%
          one of $\pp$'s queued outputs in $\tqueue'$
          is received by another participant. %
          The proof is similar to case %
          \rulename{ref-out}\rulename{LP$\red$}\ref{item:subttt-liveness:sub2-out:red-p}\ref{item:subttt-liveness:sub2-out:red-p:recv-qr} above.
        \end{enumerate}
      \end{enumerate}
    \end{itemize}
  \end{itemize}
  Summing up: %
  we have proved that each element of $\liveSubP$ %
  in \eqref{eq:subttt-liveness:prop:def} %
  satisfies the clauses of Def.~\ref{def:coinductive-liveness} %
  for participant $\pp$, %
  which implies that $\liveSubP$ %
  is a $\pp$-liveness property. %
  Moreover, %
  by \eqref{eq:subttt-liveness:prop:def}, %
  we have that for any $\Gamma$, %
  if $\Gamma,\pp{:}(\tqueue, \TT)$ is $\pp$-live %
  and $\TT' \subt \TT$, %
  then $\Gamma,\pp{:}(\tqueue, \TT') \in \liveSubPi \subseteq \liveSubP$. %
  Therefore, we conclude that $\Gamma,\pp{:}(\tqueue, \TT')$ is $\pp$-live.
\end{proof}

\begin{proposition}
  \label{lem:subtyping-preserves-live}%
  Assume that $\Gamma$ is live, %
  and that $\Gamma(\pp) = (\tqueue, \TT)$. %
  If $\tqueue'{\cdot}\TT' \subt \tqueue{\cdot}\TT$, %
  then $\Gamma\mapUpdate{\pp}{(\tqueue', \TT')}$ is live.
\end{proposition}
\begin{proof}
  By Def.~\ref{def:coinductive-liveness}, %
  we need to show that $\Gamma'$ is $\pq$-live for all participants %
  $\pq \in \dom{\Gamma'} = \dom{\Gamma}$. %
  By hypothesis, we know that $\Gamma$ is $\pq$-live for all participants %
  $\pq \in \dom{\Gamma}$. %
  By Lemma~\ref{lem:subtyping-preserves-p-live}, %
  we know that $\Gamma'$ is $\pp$-live %
  --- and in particular, %
  $\Gamma'$ is in the $\pp$-liveness property $\liveSubP$ %
  defined in \eqref{eq:subttt-liveness:prop:def}. %
  We are left to prove that $\Gamma'$ is $\pq$-live %
  for all other participants $\pq \neq \pp$.
  
  By contradiction, assume that $\Gamma'$ is \emph{not} $\pq$-live %
  for some $\pq \in \dom{\Gamma'}$. %
  This means that we can find some $\Gamma'''$ such that %
  $\Gamma' \reds \Gamma'''$ %
  and $\Gamma'''$ is not $\pq$-live because %
  it violates clause~\rulename{LP$\&$} or $\rulename{LP$\oplus$}$ %
  of Def.~\ref{def:coinductive-liveness}. %
  Now, consider the $\pp$-liveness property $\liveSubP$ %
  in \eqref{eq:subttt-liveness:prop:def}: %
  since $\liveSubP$ contains $\Gamma$, %
  it also contains $\Gamma'''$ %
  (by iteration of clause~\rulename{LP$\red$} %
  of Def.~\ref{def:coinductive-liveness}), %
  and by \eqref{eq:subttt-liveness:prop:def}, %
  there exists a corresponding $\pp$-live $\Gamma''$ such that:
  \begin{align}
    \label{eq:subttt-preserves-live:gammaiii}%
    &\Gamma''' \setminus \pp = \Gamma'' \setminus \pp%
    \\%
    \label{eq:subttt-preserves-live:sub}%
    &\Gamma'''(\pp) = (\tqueue''', \TT''') \;\text{ and }\; %
    \Gamma''(\pp) = (\tqueue'', \TT'') \;\text{ such that }\; %
    \tqueue'''{\cdot}\TT''' \subt \tqueue''{\cdot}\TT''%
  \end{align}
  Let us examine the two (non-mutually exclusive) cases %
  that can make $\Gamma'''$ not $\pq$-live:
  \begin{itemize}
  \item%
    $\Gamma'''$ violates clause~\rulename{LP$\&$}
    of Def.~\ref{def:coinductive-liveness}.\quad %
    This means that in $\Gamma'''$ there is a participant $\pr$ %
    with a top-level external choice from some participant $\pq$, %
    but some fair path of $\Gamma'''$ never enqueues a corresponding %
    output by $\pq$. %
    In this case, similarly to %
    the proof of Prop.~\ref{lem:subtyping-preserves-p-live} %
    (case \rulename{ref-$\AC$}\rulename{LP$\&$}), %
    we conclude that $\Gamma''$ is \emph{not} $\pr$-live, %
    hence not live --- contradiction;%
  \item%
    $\Gamma'''$ violates clause~\rulename{LP$\oplus$}
    of Def.~\ref{def:coinductive-liveness}.\quad %
    This means that in $\Gamma'''$ there is a participant $\pr$ %
    with a top-level queued output toward participant $\pq$, %
    but some fair path of $\Gamma'''$ where $\pq$ never consumes the message. %
    In this case, similarly to %
    the proof of Prop.~\ref{lem:subtyping-preserves-p-live} %
    (case \rulename{ref-out}\rulename{LP$\oplus$}), %
    we conclude that $\Gamma''$ is \emph{not} $\pr$-live, %
    hence not live --- contradiction.%
  \end{itemize}
  Summing up: if we assume that $\Gamma'$ is \emph{not} $\pq$-live %
  for some $\pq \in \dom{\Gamma'}$, %
  then we derive a contradiction. %
  Therefore, we obtain that $\Gamma'$ is $\pq$-live %
  for all $\pq \in \dom{\Gamma'}$. %
  Thus, by Def.~\ref{def:coinductive-liveness}, %
  we conclude that $\Gamma'$ is live.
\end{proof}

\lemSubPreservesLiveness*
\begin{proof}
  Assume $\dom{\Gamma} \cap \dom{\Gamma'} = \{\pp_1,\pp_2,\ldots,\pp_n\}$. %
  We first first show that:%
  \[
    \Gamma_i \;=\; \Gamma\mapUpdate{\pp_1}{\Gamma'(\pp_1)}\ldots\mapUpdate{\pp_2}{\Gamma'(\pp_2)}
    \ldots\mapUpdate{\pp_i}{\Gamma'(\pp_i)} %
    \;\text{ is live, for all $i \in 0..n$}%
  \]
  
  We proceed by induction on $i \in 0..n$. %
  The base case $i = 0$ is trivial: %
  we apply no updates to $\Gamma$, which is live by hypothesis.

  In the inductive case $i = m+1$, %
  we have (by the induction hypothesis) that %
  $\Gamma_m$ is live. %
  By Definition of $\Gamma' \subt \Gamma$ %
  (see page~\pageref{page:gamma-subt-gammai}), %
  we know that %
  $\Gamma'(\pp_i) \subt \Gamma(\pp_i)$.
  Therefore, by Prop.~\ref{lem:subtyping-preserves-live}, %
  we obtain that $\Gamma_i$ is live.

  To conclude the proof, %
  consider the set %
  $\dom{\Gamma'} \setminus \dom{\Gamma_n} %
  = \dom{\Gamma'} \setminus \dom{\Gamma} = \{\pq_1,\pq_2,\ldots,\pq_k\}$: %
  it contains all participants that are in $\Gamma'$, but not in $\Gamma$. %
  By Definition of $\Gamma' \subt \Gamma$ %
  (see page~\pageref{page:gamma-subt-gammai}) %
  we know that %
  $\forall i \in 1..k: \Gamma'(\pq_i) \equiv (\temptyqueue, \tend)$. %
  Therefore, if we extend $\Gamma_n$ by adding %
  an entry $\pq_i: (\temptyqueue, \tend)$ for each $i \in 1..k$, %
  we obtain an environment $\Gamma''$ such that %
  $\Gamma'' \equiv \Gamma_n$ and $\Gamma'' \equiv \Gamma'$ %
  --- hence, $\Gamma_n \equiv \Gamma'$. %
  Therefore, since $\Gamma_n$ is live, %
  by Lemma~\ref{lem:struct-equiv-and-liveness} %
  we conclude that $\Gamma'$ is live.
\end{proof}
\subsection{
Proofs of Subject Reduction and Type Safety}
\label{subsec:app:typesystem}

\begin{lemma}[Typing Inversion]
\label{lemma:typeinversion}
Let $\Theta \vdash \PP:\TT$: Then,
\begin{enumerate}
\item\label{proc1} $\PP=\mu X.\PP_1$ implies $\Theta, X:\TT_1 \vdash \PP_1:\TT_1$ and $\TT_1\subt \TT$ for some $\TT_1;$
\item\label{proc2} $\PP=\sum_{i\in I}\tin\pq{\ell_i}{\x_i}.{\PP_i}$ implies
\begin{enumerate}
  \item $\texternal_{i\in I}\tin\pq{\ell_i}{\S_i}.{\TT_i} \subt \TT$ and 
  \item $\forall i\in I$   $\Theta, x_i:\S_i   \vdash \PP_i: \TT_i$;
\end{enumerate}
\item\label{proc3} $\PP=\procout\pq\ell\e\PP_1$ implies
\begin{enumerate}
  \item $\tout\pq\ell{\S_1}.\TT_1 \subt \TT$ and 
  \item $\Theta \vdash \e:\S$ and $\S_1\subs \S$ and
  \item $\Theta \vdash \PP_1: \TT_1;$
\end{enumerate}
\item\label{proc4} $\PP=\cond\e{\PP_1}{\PP_2}$ implies 
\begin{enumerate}
  \item $\Theta  \vdash \e:\tbool$ and 
  \item $\Theta  \vdash  \PP_1:  \TT$
  \item $\Theta  \vdash  \PP_2:  \TT$
\end{enumerate}
\end{enumerate}

\noindent%
Let\; $\vdash \h:\tqueue$. Then:
\begin{enumerate}[resume]
\item\label{lem:inversion-queue-empty}%
  $\h = \emptyqueue$ \;implies\; $\tqueue = \temptyqueue$
\item\label{lem:inversion-queue-not-empty}%
  $\h = (\pq,\ell(\val)) \!\cdot\! \h_1$ \;implies\; %
  $\vdash \val:\S$ \;and\; $\tqueue \equiv \tout\pq\ell\S \!\cdot\! \tqueue'$ \;and\; $\vdash \h_1:\tqueue'.$
\end{enumerate}

\noindent%
Let $\Gamma\vdash \N$. Then:
\begin{enumerate}[resume]
\item\label{sess} If  $\Gamma \vdash \prod_{i\in I} (\pa{\pp_i}{\PP_i}\pc \pa{\pp_i}{\h_i})$ then
  \begin{enumerate}
    \item $ \forall i\in I:$\; $\vdash \PP_i:\TT_i$ and
    \item $ \forall i\in I:$\; $\vdash \h_i:\tqueue_i$ and
    \item $\Gamma=\{\pp_i:(M_i,\TT_i): i\in I\}$
  \end{enumerate}
\end{enumerate}
\end{lemma}
\begin{proof}
The proof is by induction on type derivations.
\end{proof}

\begin{lemma}[Typing congruence]
\label{lemma:type_congruence}
\begin{enumerate}
\item
  If $\Theta \vdash \PP: \TT$ and $\PP\equiv\PQ$ then $\Theta\vdash \PQ:\TT.$ %
\item 
  If $\vdash \h_1:\tqueue_1$ and $\h_1\equiv \h_2$ then there is $\tqueue_2$ such that $\tqueue_1\equiv\tqueue_2$ and $ \vdash\h_2:\tqueue_2.$
\item
  If $\Gamma' \vdash \N'$ and $\N\equiv \N'$, then there is $\Gamma$ such that $\Gamma \equiv\Gamma'$  and $\Gamma \vdash \N$.
\end{enumerate}
\end{lemma}

\begin{proof}
  The proof is by case analysis.
\end{proof}

\begin{lemma}[Substitution]
  \label{lem:substitution}%
  If $\Theta, x:\S \vdash \PP: \TT$ %
  and\; $\Theta\vdash \val: \S$, %
  \;then\; %
  $\Theta \vdash \PP\sub{\val}{\x}: \TT$. %
\end{lemma}
\begin{proof}
By structural induction on $\PP$.
\end{proof}

\begin{lemma}\label{lem:sr:eval}
 If $\emptyset \vdash \e:\S$ then there is $\val$ such that $\eval\e\val.$
\end{lemma}

\begin{proof}
The proof is by induction on the derivation of $\emptyset \vdash \e:\S.$
\end{proof}

\begin{lemma}\label{lem:queue->queue-type}
Let $\vdash \h: \tqueue$. If 
$\h \not\equiv \msg{\pp}{-}{-} \cdot \h'$ 
then 
$\tqueue\not\equiv\tout\pp{-}{-}\cdot{\tqueue'}$.
\end{lemma}

\begin{proof}
The proof is by contrapositive: assume $\tqueue\equiv\tout\pp{-}{-}\cdot{\tqueue'}$. 
Then, by induction on the derivation of $\tqueue\equiv\tout\pp{-}{-}\cdot{\tqueue'}$, 
we may show that $\h \equiv \msg{\pp}{-}{-} \cdot \h'$ .
\end{proof}

\begin{lemma}\label{lem:minimal-type-of-process}
If $\Theta \vdash \PP: \TT$ then there is $\TT'$ such that 
$\TT'\subt \TT$ and $\Theta \vdash \PP: \TT'$ and $\actions{\TT'}\subseteq\actions{\PP}$.
\end{lemma}
\begin{proof}
By induction on $\Theta \vdash \PP: \TT$. 
The only interesting case is \rulename{t-cond}. 
In this case, we have that 
$\Theta \vdash \cond\e{\PP_1}{\PP_2}:\TT$ 
is derived from 
$\Theta \vdash \e:\tbool$, 
$\Theta  \vdash \PP_1:\TT$ and $ \Theta  \vdash \PP_2:\TT$. 
By induction hypothesis we derive that there exist $\TT_1$ and $\TT_2$  
such that 
\begin{align}
& \Theta  \vdash \PP_1:\TT_1 
 \;\text{ and }\;  \TT_1\subt\TT 
 \;\text{ and }\;  \actions{\TT_1}\subseteq\actions{\PP_1}\label{eq:P1-minimal-T1}\\
& \Theta  \vdash \PP_2:\TT_2 
 \;\text{ and }\; \TT_2\subt\TT 
 \;\text{ and }\; \actions{\TT_2}\subseteq\actions{\PP_2}\label{eq:P2-minimal-T2}
\end{align}

We will show that there is $\TT'$ such that 
$\TT_1 \subt \TT' \subt \TT$ and 
$\TT_2 \subt \TT' \subt \TT$ and
\[
\actions{\TT'}
\;\subseteq\; \actions{\TT_1} \cup \actions{\TT_2} 
\;(\subseteq\; \actions{\PP_1}\cup \actions{\PP_1}
\;=\; \actions{\PP})
\]

Let us now expand the derivations of $\TT_1\subt\TT$ and $\TT_2\subt\TT$ 
given in~\eqref{eq:P1-minimal-T1} and \eqref{eq:P2-minimal-T2}:
\begin{align}
 &  \forall \UT_1\in\llbracket \TT_1\rrbracket_\SO \quad \forall \VT' \in \llbracket \TT \rrbracket_\SI \quad \exists \WT_1 \in\llbracket \UT_1 \rrbracket_\SI \quad \exists \WT_1' \in \llbracket \VT' \rrbracket_\SO \qquad \WT_1 \subttt\WT_1'\label{eq:derivation-of-T1-subt-T}\\
 &     \forall \UT_2\in\llbracket \TT_2\rrbracket_\SO \quad \forall \VT' \in \llbracket \TT \rrbracket_\SI \quad \exists \WT_2 \in\llbracket \UT_2 \rrbracket_\SI \quad \exists \WT_2' \in \llbracket \VT' \rrbracket_\SO \qquad \WT_1 \subttt\WT_2'\label{eq:derivation-of-T2-subt-T}
\end{align}
Let us consider the sets:
\begin{itemize}
\item $A_1$ as the set of all actions occurring in any $\UT_1$ in~\eqref{eq:derivation-of-T1-subt-T};
\item $A_2$ as the set of all actions occurring in any $\UT_2$ in~\eqref{eq:derivation-of-T2-subt-T};
\item $A'$  as the set of all actions occurring in any $\VT'$ in~\eqref{eq:derivation-of-T1-subt-T} 
(or equivalently in~\eqref{eq:derivation-of-T2-subt-T}).
\end{itemize}

Notice that $\actions{\TT_1}=A_1, \actions{\TT_2}=A_2$ and $\actions{\TT}=A'$.

Furthermore, let us also consider the sets: 
\begin{itemize}
\item $B_1$  as the set of all actions occurring in any $\WT_1$ selected in~\eqref{eq:derivation-of-T1-subt-T};
\item $B_2$  as the set of all actions occurring in any $\WT_2$ selected in~\eqref{eq:derivation-of-T2-subt-T};
\item $B_1'$ as the set of all actions occurring in any $\WT_1'$ selected in~\eqref{eq:derivation-of-T1-subt-T};
\item $B_2'$ as the set of all actions occurring in any $\WT_2'$ selected in~\eqref{eq:derivation-of-T2-subt-T}.
\end{itemize}

Now we have:
\begin{align}
& B_1 \subseteq A_1\label{eq:B1-contained-in-A1}\\
& B_2 \subseteq A_2\label{eq:B2-contained-in-A2}\\
& B_1' \cup B_2' \subseteq A'\\
& B_1 = B_1' \quad  (\text{from}\; \WT_1 \subttt \WT_1')\label{eq:B1=B1'}\\
& B_2 = B_2' \quad  (\text{from}\; \WT_2 \subttt \WT_2')\label{eq:B2=B2'}
\end{align}

Therefore, by~\eqref{eq:B1=B1'} and~\eqref{eq:B2=B2'} we have 
\begin{equation}
B_1 \cup B_2 = B_1' \cup B_2'
\end{equation}
and thus, by~\eqref{eq:B1-contained-in-A1},~\eqref{eq:B2-contained-in-A2},~\eqref{eq:P1-minimal-T1}  
and~\eqref{eq:P2-minimal-T2}, we obtain 
\begin{equation}
B_1' \cup B_2' 
\;\subseteq\; A_1 \cup A_2 
\;\subseteq\; \actions{\PP_1} \cup \actions{\PP_2} 
\;=\; \actions{\PP} 
\end{equation}

Now, consider the set $D = A' \setminus (A_1 \cup A_2)$: 
it contains all actions that occur in $\TT$, but not in $\TT_1$ nor $\TT_2$.  
Consider any action $\alpha \in D$: it must belong to some SISO tree $\WT''$ 
which was not selected neither as $\WT_1'$ in~\eqref{eq:derivation-of-T1-subt-T}, 
nor as $\WT_2'$ in~\eqref{eq:derivation-of-T2-subt-T}.  
Therefore, it must belong to some action of some SO tree in 
$\llbracket \VT' \rrbracket_\SO$ that is never selected by $\TT_1$ nor $\TT_2$.  
This means that $\alpha$ belongs to some internal choice branches of $\TT$ that are never selected by 
$\TT_1$ nor $\TT_2$.  
Therefore, if we prune $\T$ (i.e., the syntactic type with tree $\TT$) by removing all such internal choice branches, 
we get a session type $\T''$ with tree $\TT'' \subt \TT$ such that:
 \begin{align}
 &  \forall \UT_1\in\llbracket \TT_1\rrbracket_\SO \quad \forall \VT' \in \llbracket \TT'' \rrbracket_\SI \quad \exists \WT_1 \in\llbracket \UT_1 \rrbracket_\SI \quad \exists \WT_1' \in \llbracket \VT' \rrbracket_\SO \qquad \WT_1 \subttt\WT_1'\label{eq:derivation-of-T1-subt-T''}\\
 &     \forall \UT_2\in\llbracket \TT_2\rrbracket_\SO \quad \forall \VT' \in \llbracket \TT'' \rrbracket_\SI \quad \exists \WT_2 \in\llbracket \UT_2 \rrbracket_\SI \quad \exists \WT_2' \in \llbracket \VT' \rrbracket_\SO \qquad \WT_1 \subttt\WT_2'\label{eq:derivation-of-T2-subt-T''}
\end{align}
Hence, $\TT_1 \subt \TT''$ and $\TT_2\subt \TT''$.

Now let $A''=\actions{\TT''}$.
If we compute $D'=A''\setminus(A_1 \cup A_2)$ (i.e., the set of all actions that occur in $\TT''$, 
but not in $\TT_1$ nor $\TT_2$) 
using~\eqref{eq:derivation-of-T1-subt-T''} and~\eqref{eq:derivation-of-T2-subt-T''}, 
we obtain $D' \subset D$, because $\alpha$ (and possibly some other actions in $D$) have been removed by the pruning.  
By iterating the procedure (i.e., by induction on the number of actions in $D$), noticing that $D$ is finite (because $\TT$ is syntax-derived from some $\T$), we will eventually find some $\T'$ with tree $\TT'$ such that:
\begin{align*}
& \TT' \subt \TT   \;\text{ and }\;   \TT_1 \subt \TT'   \;\text{ and }\;   \TT_2 \subt \TT'\\
& \actions{\TT'} \subseteq \actions{\TT_1} \cup \actions{\TT_2}\\
& \Theta \vdash \PP : \T'\\
& \actions{\TT'} \subseteq \actions{\PP} 
\end{align*}
\end{proof}

\begin{lemma}[Error subject reduction]\label{lem:error-subject-reduction}
  If $\N \red \error$ then there is no type environment $\Gamma$ such that  $\Gamma$ is live and $ \Gamma \vdash \N.$
\end{lemma}

\begin{proof}
  By induction on derivation $\N \red \error$.
\\
Base cases:
\\
\underline{\rulename{err-mism}}: 
     We have
    \begin{align}
      &\N = \pa\pp\sum\limits_{j\in J} \procin\pq{\ell_j(\x_j)}\PP_j \pc \pa\pp\h_\pp \pc \pa\pq\PP_\pq \pc \pa\pq (\pp,\ell(\val))\!\cdot \!\h \pc \N_1
      \label{eq:err-mismatch-sess}%
      \\%
      &\forall j\in J: \ell \neq \ell_j       \label{eq:err-labels}%
      \\%
      &\N_1 = \prod_{i\in I} (\pa{\pp_i}{\PP_i} \pc \pa{\pp_i}{\h_i})
    \end{align}
    Assume to the contrary that there exists $\Gamma$  such that:
    \begin{align}
    &\Gamma \vdash \N%
    &
    \\%
    \label{lem:sr:err:gamma-live}%
    &\text{$\Gamma$ is live}%
    &
  \end{align}
  By Lemma~\ref{lemma:typeinversion}.\ref{sess},
  \begin{align}
    &   \vdash \sum\limits_{j\in J} \procin\pq{\ell_j(\x_j)}\PP_j:\TT \label{eq:err:sr-receive-1}\\
    &   \vdash \h_\pp:\tqueue_\pp \label{eq:err:sr-receive-12}\\
    &  \vdash \PP_\pq:\TT_\pq \label{eq:err:sr-receive-11}\\
    &  \vdash (\pp,\ell(\val))\!\cdot\!\h:\tqueue  \label{eq:err:sr-receive-w} \\
    &  \forall i\in I:\quad  \vdash \PP_i:\TT_i \label{eq:err:sr-receive-3}\\
    &  \forall i\in I:\quad  \vdash \h_i:\tqueue_i \label{eq:err:sr-receive-4}\\
    &  \Gamma = \{\pp:(\tqueue_\pp,\TT), \pq:(\tqueue, \TT_\pq)\}
    \cup \{\pp_i:(\tqueue_i,\TT_i):i\in I\} \label{eq:err:sr-receive-5}%
  \end{align}
  By Lemma~\ref{lemma:typeinversion}.\ref{proc2},
  there are $\TT'_j, \S'_j$ (for $j \in J$) such that:
    \begin{align}
      \label{eq:err:sr-receive-extc-t}%
      & \texternal_{j\in J}\tin\pq{\ell_j}{\S'_j}.{\TT'_j} \subt \TT\\
      \label{eq:sr-receive-cont-t}%
      & \forall j \in J:\quad  x_j:\S'_j\vdash \PP_j: \TT'_j
    \end{align}
    By Lemma~\ref{lemma:typeinversion}.\ref{lem:inversion-queue-not-empty},
    there are $\tqueue', \S$ such that:
    \begin{align}
      &\vdash \val:\S \;\text{ and }\;
      \tqueue \equiv \tout\pp{\ell}{\S} \!\cdot\! \tqueue' 
       \;\text{ and }\;\vdash \h: \tqueue'
      \label{eq:err:sr-receive-7}
    \end{align}
    Now, let:
    \begin{align}
      \label{eq:err:sr-receive-8}%
      &\Gamma' = \{
        \pp:(\tqueue_\pp,\texternal_{j\in J}\tin\pq{\ell_j}{\S'_j}.{\TT'_j}),
        \pq:(\tout\pp{\ell}{\S} \!\cdot\! \tqueue',\TT_\pq)
      \}
      \cup \{\pp_i:(\tqueue_i,\TT_i):i \in I\}
    \end{align}
    Then, we have:
    \begin{align}
      \label{eq:err:sr-receive-9}%
      &\Gamma' \subt \Gamma%
      &\text{(by 
        \eqref{eq:err:sr-receive-7}, \eqref{eq:err:sr-receive-extc-t}, \eqref{eq:err:sr-receive-5}, \eqref{eq:err:sr-receive-8})}%
      \\%
      \label{eq:err:sr-receive-10}%
      &\text{$\Gamma'$ is live}%
      &\text{(by \eqref{eq:err:sr-receive-9} %
        and Lemma~\ref{lem:subtyping-preserves-liveness})}.%
    \end{align}
    By \eqref{eq:err-labels}, \eqref{eq:err:sr-receive-10} %
    and Definition~\ref{def:env-liveness} we get a contradiction.
    \\
\item[] \underline{\rulename{err-eval}}: We have 
    \begin{align}
      &\N = \pa\pp\cond\e{\PP_1}{\PP_2}\pc \pa\pp\h \pc  \N_1
      \\%
      &\not\exists \val: \eval \e \val \label{eq:sr:err:eval}
      \\%
      &\N_1 = \prod_{i\in I} (\pa{\pp_i}{\PP_i} \pc \pa{\pp_i}{\h_i})
    \end{align}
    Assume to the contrary that there exists $\Gamma$  such that:
    \begin{align}
    &\Gamma \vdash \N%
    &
    \\%
    \label{lem:sr:err:gamma-live}%
    &\text{$\Gamma$ is live}%
    &
  \end{align}
  By Lemma~\ref{lemma:typeinversion}.\ref{sess},
  \begin{align}
    &   \vdash \cond\e{\PP_1}{\PP_2}:\TT \label{eqv:err:proof1}\\
    &  \vdash \h:\tqueue  \label{eqv:err:proof2} \\
    &  \forall i\in I \;\;\;  \vdash \PP_i:\TT_i \label{eqv:err:proof3}\\
    &  \forall i\in I \;\;\; \vdash \h_i:\tqueue_i \label{eqv:err:proof4}\\
    &  \Gamma = \{\pp:(\tqueue, \TT)\} \cup \{\pp_i:(\tqueue_i,\TT_i):i\in I\} \label{eqv:err:proof5}%
  \end{align}
  From~(\ref{eqv:err:proof1}) and by Lemma~\ref{lemma:typeinversion}.\ref{proc4}:
    \begin{align}
      & \vdash  \PP_1:  \TT\\
      & \vdash  \PP_2:  \TT\\
      &  \vdash \e:\tbool
    \end{align} 
    By Lemma~\ref{lem:sr:eval}, there is a value $v$ such that $\eval\e\val,$ which leads to contradiction with assumption \eqref{eq:sr:err:eval}.

\item[] \underline{\rulename{err-eval2}}: We have 
    \begin{align}
      &\N = \pa\pp\procout{\pq}{\ell}{\e}{\PP} \pc \pa\pp\h  \pc  \N_1
      \\%
      &\not\exists \val: \eval \e \val \label{eq:sr:err:eval2}
      \\%
      &\N_1 = \prod_{i\in I} (\pa{\pp_i}{\PP_i} \pc \pa{\pp_i}{\h_i})
    \end{align}
    Assume to the contrary that there exists $\Gamma$  such that:
    \begin{align}
    &\Gamma \vdash \N%
    &
    \\%
    \label{lem:sr:err2:gamma-live}%
    &\text{$\Gamma$ is live}%
    &
  \end{align}
  By Lemma~\ref{lemma:typeinversion}.\ref{sess},
  \begin{align}
    &   \vdash \pa\pp\procout{\pq}{\ell}{\e}{\PP}:\TT \label{eqv:err2:proof1}\\
    &  \vdash \h:\tqueue  \label{eqv:err2:proof2} \\
    &  \forall i\in I \;\;\;  \vdash \PP_i:\TT_i \label{eqv:err2:proof3}\\
    &  \forall i\in I \;\;\; \vdash \h_i:\tqueue_i \label{eqv:err2:proof4}\\
    &  \Gamma = \{\pp:(\tqueue, \TT)\} \cup \{\pp_i:(\tqueue_i,\TT_i):i\in I\} \label{eqv:err2:proof5}%
  \end{align}
  From~(\ref{eqv:err2:proof1}) and by Lemma~\ref{lemma:typeinversion}.\ref{proc3}:
    \begin{align}
      & \tout\pq\ell{\S_1}.\TT_1 \subt \TT\\
      & \vdash \e:\S\\
      &  \S_1\subs \S
    \end{align} 
    By Lemma~\ref{lem:sr:eval}, there is a value $v$ such that $\eval\e\val,$ which leads to contradiction with assumption \eqref{eq:sr:err:eval2}.
\\
\underline{\rulename{err-ophn}}: We have:
    \begin{align}
      &\N = \pa\pp\PP \pc \pa\pp\h_\pp \pc \pa\pq\PP_\pq \pc \pa\pq (\pp,\ell(\val))\!\cdot \!\h \pc \N_1 \  \label{eq:err-orph-msg-sess}\\
      &q?\not\in\actions\PP \label{eq:err-orph-negin}
      \\%
      &\N_1 = \prod_{i\in I} (\pa{\pp_i}{\PP_i} \pc \pa{\pp_i}{\h_i})\label{eq:err-orph-msg-sess-tail}
    \end{align}
  By Lemma~\ref{lemma:typeinversion}.\ref{sess},
  \begin{align}
    &  \vdash \PP:\TT_{\pp} \label{eq:err-receive-1'}\\
    &  \vdash (\pp,\ell(\val))\!\cdot\!\h:\tqueue_{\pq}  \label{eq:err-receive-w} \\
    &   \vdash \PP_\pq:\TT_\pq \label{eq:err-receive-11'}\\
    &   \vdash \h_\pp:\tqueue_\pp \label{eq:err-receive-12'}\\
    &  \forall i\in I:\quad  \vdash \PP_i:\TT_i \label{eq:err-receive-3'}\\
    &  \forall i\in I:\quad  \vdash \h_i:\tqueue_i \label{eq:err-receive-4'}\\
    &  \Gamma = \{\pp:(\tqueue_\pp,\TT_{\pp}), \pq:(\tqueue_{\pq}, \TT_\pq)\}
    \cup \{\pp_i:(\tqueue_i,\TT_i):i\in I\} \label{eq:err-receive-5'}%
  \end{align}
    By Lemma~\ref{lemma:typeinversion}.\ref{lem:inversion-queue-not-empty},
    there are $\tqueue', \S'$ such that:
    \begin{align}
      &\vdash \val:\S' \;\text{ and }\;
      \tqueue_{\pq} \equiv \tout\pp{\ell_k}{\S'} \!\cdot\! \tqueue'
      \;\text{ and }\; \vdash \h: \tqueue'
       \label{eq:err-receive-7'}
    \end{align}
    Let 
    \begin{equation}\label{eq:gamma'}
    \Gamma'=\{\pp:(\tqueue_\pp,\TT_{\pp}), \pq:(\tout\pp{\ell_k}{\S'} \!\cdot\! \tqueue', \TT_\pq)\}
    \cup \{\pp_i:(\tqueue_i,\TT_i):i\in I\}
    \end{equation}
    Then we have 
    \begin{align}
    & \Gamma \equiv \Gamma'\label{eq:gamma'-is-equiv-to-gamma}
    &\text{(by \eqref{eq:err-receive-5'} and \eqref{eq:gamma'})}\\
    & \Gamma' \text{ is live}\label{eq:gamma'-is-live}
    &\text{(by Lemma~\ref{lem:struct-equiv-and-liveness})}
    \end{align}
    By \eqref{eq:err-receive-1'} and Lemma~\ref{lem:minimal-type-of-process} 
    we have that there is $\TT'$ such that 
    $\TT'\subt\TT_{\pp}$ and 
    $\vdash \PP: \TT'$ and 
    $\actions{\TT'}\subseteq\actions{\PP}$.
    Now let
        \begin{equation}\label{eq:gamma'}
    \Gamma''=\{\pp:(\tqueue_\pp,\TT'), \pq:(\tout\pp{\ell_k}{\S'} \!\cdot\! \tqueue', \TT_\pq)\}
    \cup \{\pp_i:(\tqueue_i,\TT_i):i\in I\}
    \end{equation} 
    Then we have $\Gamma'' \subt \Gamma'$, and hence, 
    $\Gamma''$ is live (by Lemma~\ref{lem:subtyping-preserves-liveness}). 
    This gives a contradiction with \eqref{eq:err-orph-negin}, 
    $\actions{\TT'}\subseteq\actions{\PP}$ and 
    Definition~\ref{def:env-liveness} 
    (since $\pq? \notin\actions{\TT'}$ holds, the message in the queue of $\pq$ will never be received
    in any reduction of $\Gamma''$).
\\
\underline{\rulename{err-strv}:} We have
    \begin{align}
      &\N = \pa\pp\sum\limits_{j\in J} \procin\pq{\ell_j(\x_j)}\PP_j \pc \pa\pp\h_\pp \pc \pa\pq\PP_\pq \pc \pa\pq\h_\pq \pc \N_1
      \label{eq:err-starv-sess}
      \\%
      & \pp!\not\in \actions{\PP_\pq}\label{eq:err-startv-no-prefix-in-process}\\
      & \h_\pq \not\equiv \msg{\pp}{-}{-} \cdot \h_\pq'\label{eq:err-startv-no-prefix-in-queue}\\
      &\N_1 = \prod_{i\in I} (\pa{\pp_i}{\PP_i} \pc \pa{\pp_i}{\h_i})\label{eq:err-starv-sess-tail}
    \end{align}
  Assume to the contrary that there exists $\Gamma$  such that:
    \begin{align}
    &\Gamma \vdash \N%
    &
    \\%
    \label{lem:sr:err:gamma-live}%
    &\text{$\Gamma$ is live}%
    &
  \end{align}

  By Lemma~\ref{lemma:typeinversion}.\ref{sess},
  \begin{align}
    &   \vdash \sum\limits_{j\in J} \procin\pq{\ell_j(\x_j)}\PP_j:\TT_\pp \label{eq:err-receive-1}\\
    &  \vdash \h_\pq:\tqueue_\pq  \label{eq:err-receive-w} \\
    &   \vdash \PP_\pq:\TT_\pq \label{eq:err-receive-11}\\
    &   \vdash \h_\pp:\tqueue_\pp \label{eq:err-receive-12}\\
    &  \forall i\in I:\quad  \vdash \PP_i:\TT_i \label{eq:err-receive-3}\\
    &  \forall i\in I:\quad  \vdash \h_i:\tqueue_i \label{eq:err-receive-4}\\
    &  \Gamma = \{\pp:(\tqueue_\pp,\TT_\pp), \pq:(\tqueue_\pq, \TT_\pq)\}
    \cup \{\pp_i:(\tqueue_i,\TT_i):i\in I\} \label{eq:err-receive-5}%
  \end{align}
  By Lemma~\ref{lemma:typeinversion}.\ref{proc2},
  there are $\TT'_j, \S'_j$ (for $j \in J$) such that:
    \begin{align}
      \label{eq:err-receive-extc-t}%
      & \texternal_{j\in J}\tin\pq{\ell_j}{\S'_j}.{\TT'_j} \subt \TT_\pp\\
      \label{eq:err-receive-cont-t}%
      & \forall j \in J:\quad  x_j:\S'_j\vdash \PP_j: \TT'_j
    \end{align}
    Now, let:
    \begin{align}
      \label{eq:err-starv-receive-8}%
      &\Gamma' = \{
        \pp:(\tqueue_\pp,\texternal_{j\in J}\tin\pq{\ell_j}{\S'_j}.{\TT'_j}),
        \pq:(\tqueue_\pq,\TT_\pq)
      \}
      \cup \{\pp_i:(\tqueue_i,\TT_i):i \in I\}
    \end{align}
    Then, we have:
    \begin{align}
      \label{eq:err-receive-9}%
      &\Gamma' \subt \Gamma%
      &\text{(by  \eqref{eq:err-receive-5}, \eqref{eq:err-receive-extc-t}, \eqref{eq:err-starv-receive-8})}%
      \\%
      \label{eq:err-receive-10}%
      &\text{$\Gamma'$ is live}%
      &\text{(by \eqref{eq:err-receive-9} %
        and Lemma~\ref{lem:subtyping-preserves-liveness})}%
    \end{align}
    By \eqref{eq:err-startv-no-prefix-in-queue}, \eqref{eq:err-receive-w} 
    and Lemma~\ref{lem:queue->queue-type} we have that 
    $\tqueue_{\pq}\not\equiv\tout\pp{-}{-}\cdot{\tqueue'}$. 
    By \eqref{eq:err-receive-11} and Lemma~\ref{lem:minimal-type-of-process} 
    there is $\TT'$ such that $\TT'\subt \TT_\pq$ and $\vdash \PP_\pq:\TT'$ 
    and $\actions{\TT'}\subseteq \actions{\PP_\pq}$.
    Now let 
    \begin{align}
      \label{eq:err-starv-gamma''}%
      &\Gamma'' = \{
        \pp:(\tqueue_\pp,\texternal_{j\in J}\tin\pq{\ell_j}{\S'_j}.{\TT'_j}),
        \pq:(\tqueue_\pq,\TT')
      \}
      \cup \{\pp_i:(\tqueue_i,\TT_i):i \in I\}
    \end{align}   
    Then we have $\Gamma'' \subt \Gamma'$, and hence, 
    $\Gamma''$ is live (by Lemma~\ref{lem:subtyping-preserves-liveness}). 
    This gives a contradiction with \eqref{eq:err-startv-no-prefix-in-process} and  
    $\actions{\TT'}\subseteq\actions{\PP_\pq}$ and $\tqueue_{\pq}\not\equiv\tout\pp{-}{-}\cdot{\tqueue'}$ and 
    Definition~\ref{def:env-liveness} 
    (since $\pp! \notin\actions{\TT'}$ and $\tqueue_{\pq}\not\equiv\tout\pp{-}{-}\cdot{\tqueue'}$ hold, 
    none of the active inputs in $\pp$ will be activated
    in any reduction of $\Gamma''$).  
\\\\
Inductive step:
\\
 \rulename{r-struct}. Assume that $\N \red \error$ is derived from:
    \begin{align}
      &\N \equiv \N_1 \label{lem:sr:struct1} \\
      &\N_1 \red \error
    \end{align}
  By the induction hypothesis, there is no live $\Gamma_1$ such that $\Gamma_1\vdash \N_1.$
  Assume on the contrary that there is live $\Gamma$ such that $\Gamma \vdash \N$. 
  Then, by Lemma~\ref{lemma:type_congruence}.3, there is $\Gamma_2$ such that 
  $\Gamma\equiv\Gamma_2$ and $\Gamma_2\vdash \N_1.$ 
  Since $\Gamma$ is live, by Lemma~\ref{lem:struct-equiv-and-liveness}
  we obtain $\Gamma_2$ is live, which is a contradiction with the induction hypothesis.
\end{proof}

\begin{lemma}\label{lem:type-contexts-equiv->subt}
If $\Gamma\equiv\Gamma'$ then $\Gamma\subt\Gamma'$.
\end{lemma}
\begin{proof}
Directly by the definitions of $\Gamma\equiv\Gamma'$ and $\Gamma\subt\Gamma'$.
\end{proof}

\theoremSR*

\begin{proof}
  Assume:
  \begin{align}
    \label{lem:sr:hyp:typing}%
    &\Theta \cdot \Gamma \vdash \N%
    &\text{(by hypothesis)}%
    \\%
    \label{lem:sr:hyp:gamma-live}%
    &\text{$\Gamma$ is live}%
    &\text{(by hypothesis)}%
    \\%
    \label{lem:sr:hyp:sess-move}%
    &\N \red \N'
    &\text{(by hypothesis)}%
  \end{align}
  The proof is by induction on the derivation of $\N \red \N'.$ 
\\
  Base cases:
  \\
\underline{\rulename{r-send}:} We have
    \begin{align}
      &\N = \pa\pp\procout\pq\ell\e\PP\pc \pa\pp\h \pc  \N_1
      \\%
      &\eval \e \val%
      \\%
      \label{eqv:proof0}%
      &\N'=\pa\pp\PP\pc\pa\pp\h\cdot (\pq,\ell(\val))\pc \N_1%
      \\%
      &\N_1 = \prod_{i\in I} (\pa{\pp_i}{\PP_i} \pc \pa{\pp_i}{\h_i})
    \end{align}
  
  By Lemma~\ref{lemma:typeinversion}.\ref{sess},
  \begin{align}
    &  \vdash \procout\pq\ell\e\PP:\TT \label{eqv:proof1}\\
    &  \vdash \h:\tqueue  \label{eqv:proof2} \\
    &  \forall i\in I \;\;\; \vdash \PP_i:\TT_i \label{eqv:proof3}\\
    &  \forall i\in I \;\;\; \vdash \h_i:\tqueue_i \label{eqv:proof4}\\
    &  \Gamma = \{\pp:(\tqueue, \TT)\} \cup \{\pp_i:(\tqueue_i,\TT_i):i\in I\} \label{eqv:proof5}%
  \end{align}
  By Lemma~\ref{lemma:typeinversion}.\ref{proc3},
  there are $\TT',\S'$ such that:
    \begin{align}
      & \tout\pq\ell\S.\TT'\subt \TT  \label{eqv:proof7}\\
      &  \vdash \e:\S' \quad\text{and}\quad%
      \S'\subs \S  \label{eqv:proof8}\\
      & \vdash \PP: \TT'  \label{eqv:proof9}
    \end{align}
    Now, let:
    \begin{align}
      \label{eqv:proof9-2}%
      &\Gamma'' =  \{\pp:(\tqueue,  \tout\pq\ell{\S}.\TT')\} \cup \{\pp_i:(\tqueue_i,\TT_i):i\in I\}
      \\
      \label{eqv:proof11-2}%
      &\Gamma' = \{\pp:(\tqueue \cdot  \tout\pq\ell{\S}, \TT')\} \cup \{\pp_i:(\tqueue_i,\TT_i):i\in I\}
    \end{align}
    Then, we conclude:
    \begin{align} 
      \label{eqv:proof10}%
      &\Gamma'' \subt \Gamma%
      &\text{%
        (by \eqref{eqv:proof9-2}, \eqref{eqv:proof7}, and \eqref{eqv:proof5})
      }%
      \\
      \label{eqv:proof11}%
      &\text{$\Gamma''$ is live}
      &\text{(by \eqref{eqv:proof10} %
        and Lemma~\ref{lem:subtyping-preserves-liveness})%
      }%
      \\
      \label{eqv:proof12}%
      & \Gamma'' \red \Gamma'%
      &\text{(by \eqref{eqv:proof9-2}, \eqref{eqv:proof11-2}, %
        and \rulename{e-send} of Def.~\ref{def:typing-env-reductions})}%
      \\%
      \nonumber%
      &\text{$\Gamma'$ is live}%
      &\hspace{-20mm}\text{%
        (by \eqref{eqv:proof11}, \eqref{eqv:proof12}, and %
        Proposition~\ref{lem:move-preserves-liveness})%
      }%
      \\%
      \nonumber%
      &  \Gamma' \vdash \N'%
      &\text{%
        (\eqref{eqv:proof0}, \eqref{eqv:proof11-2}, and
        Def.~\ref{def:type-system})%
      }
    \end{align}
\underline{\rulename{r-rcv}:} We have
    \begin{align}
      &\N = \pa\pp\sum\limits_{j\in J} \procin\pq{\ell_j(\x_j)}\PP_j \pc \pa\pp\h_\pp \pc \pa\pq\PP_\pq \pc \pa\pq (\pp,\ell_k(\val))\!\cdot \!\h \pc \N_1
      \quad\text{(for some $k \in J$)}%
      \\%
      \label{eq:sr-receive-cont-sess}%
      &\N' = \pa\pp \PP_k\sub{\val}{\x_k} \pc \pa\pp\h_\pp \pc  \pa\pq\PP_\pq \pc  \pa\pq \h \pc \N_1
      \\%
      &\N_1 = \prod_{i\in I} (\pa{\pp_i}{\PP_i} \pc \pa{\pp_i}{\h_i})
    \end{align}
  By Lemma~\ref{lemma:typeinversion}.\ref{sess},
  \begin{align}
    &   \vdash \sum\limits_{j\in J} \procin\pq{\ell_j(\x_j)}\PP_j:\TT \label{eq:sr-receive-1}\\
    &  \vdash (\pp,\ell_k(\val))\!\cdot\!\h:\tqueue  \label{eq:sr-receive-w} \\
    &  \vdash \PP_\pq:\TT_\pq \label{eq:sr-receive-11}\\
    &   \vdash \h_\pp:\tqueue_\pp \label{eq:sr-receive-12}\\
    &  \forall i\in I:\quad  \vdash \PP_i:\TT_i \label{eq:sr-receive-3}\\
    &  \forall i\in I:\quad  \vdash \h_i:\tqueue_i \label{eq:sr-receive-4}\\
    &  \Gamma = \{\pp:(\tqueue_\pp,\TT), \pq:(\tqueue, \TT_\pq)\}
    \cup \{\pp_i:(\tqueue_i,\TT_i):i\in I\} \label{eq:sr-receive-5}%
  \end{align}
  By Lemma~\ref{lemma:typeinversion}.\ref{proc2},
  there are $\TT'_j, \S'_j$ (for $j \in J$) such that:
    \begin{align}
      \label{eq:sr-receive-extc-t}%
      & \texternal_{j\in J}\tin\pq{\ell_j}{\S'_j}.{\TT'_j} \subt \TT\\
      \label{eq:sr-receive-cont-t}%
      & \forall j \in J:\quad x_j:\S'_j\vdash \PP_j: \TT'_j
    \end{align}
    By Lemma~\ref{lemma:typeinversion}.\ref{lem:inversion-queue-not-empty},
    there are $\tqueue', \S'$ such that:
    \begin{align}
      &\vdash \val:\S' \;\text{ and }\;
      \tqueue \equiv \tout\pp{\ell_k}{\S'} \!\cdot\! \tqueue' \label{eq:sr-receive-7}
    \end{align}
    Now, let:
    \begin{align}
      \label{eq:sr-receive-8}%
      &\Gamma'' = \{
        \pp:(\tqueue_\pp,\texternal_{j\in J}\tin\pq{\ell_j}{\S'_j}.{\TT'_j}),
        \pq:(\tout\pp{\ell_k}{\S'} \!\cdot\! \tqueue',\TT_\pq)
      \}
      \cup \{\pp_i:(\tqueue_i,\TT_i):i \in I\}
    \end{align}
    Then, we have:
    \begin{align}
      \label{eq:sr-receive-9}%
      &\Gamma'' \subt \Gamma%
      &\text{(by \eqref{eq:sr-receive-8}, %
        \eqref{eq:sr-receive-7}, %
        \eqref{eq:sr-receive-extc-t},
        \eqref{eq:sr-receive-5})}%
      \\%
      \label{eq:sr-receive-10}%
      &\text{$\Gamma''$ is live}%
      &\text{(by \eqref{eq:sr-receive-9} %
        and Lemma~\ref{lem:subtyping-preserves-liveness})}%
    \end{align}
    Observe that we also have:
    \begin{align}
      \label{eq:sr-receive-subs}%
      \S' \subs \S'_k%
    \end{align}
    To prove \eqref{eq:sr-receive-subs}, assume (by contradiction) %
    that $\S' \not\subs \S'_k$. %
    Then, the premise of~\rulename{e-rcv} %
    in Def.~\ref{def:typing-env-reductions} %
    does not hold, %
    and thus, %
    $\pp$'s external choice in $\Gamma''$ %
    cannot possibly synchronise with $\pq$'s queue type. %
    But then, %
    by Def.~\ref{def:env-liveness}, %
    $\Gamma''$ is \emph{not} live, which gives a contradiction with~(\ref{eq:sr-receive-10}). %
    Therefore, it must be the case that %
    \eqref{eq:sr-receive-subs} holds.

    Now, let:
    \begin{align}
      \label{eq:sr-receive-11}%
      &\Gamma' = \{{
        \pp:(\tqueue_\pp, \TT'_k),
        \pq:(\tqueue', \TT_\pq)
      }\}
      \cup \{\pp_i:(\tqueue_i,\TT_i):i \in I\}
    \end{align}

    And we conclude:
    \begin{align}
      \label{eq:sr-receive-12}%
      & \Gamma'' \red \Gamma'%
      &\text{(by \eqref{eq:sr-receive-8}, \eqref{eq:sr-receive-subs}, %
        \eqref{eq:sr-receive-11}, %
        and rule \rulename{e-rcv} of Def.~\ref{def:typing-env-reductions})}
      \\%
      \nonumber%
      &\text{$\Gamma'$ is live}%
      &\hspace{-20mm}\text{%
        (by \eqref{eq:sr-receive-10}, \eqref{eq:sr-receive-12}, and %
        Proposition~\ref{lem:move-preserves-liveness})%
      }%
      \\%
      \nonumber%
      &  \Gamma' \vdash \N'%
      &\text{%
        (by \eqref{eq:sr-receive-cont-t}, %
        \eqref{eq:sr-receive-subs},
        \eqref{eq:sr-receive-cont-sess}, %
        Lemma~\ref{lem:substitution}, %
        \eqref{eq:sr-receive-11}, %
        and Def.~\ref{def:type-system})%
      }
    \end{align}
\\        
\underline{\rulename{r-cond-T} (\rulename{r-cond-F})}: We have:
    \begin{align}
      &\N = \pa\pp\cond\e\PP\PQ \pc \pa\pp\h \pc  \N_1
      \label{eq:sr-cond-cont-sess} \\%
      &\N' = \pa\pp \PP  \pc  \pa\pp \h \pc \N_1 \quad (\N' = \pa\pp \PQ  \pc  \pa\pp \h \pc \N_1)
      \label{eq:sr-cond-sess} \\%
      &\N_1 = \prod_{i\in I} (\pa{\pp_i}{\PP_i} \pc \pa{\pp_i}{\h_i})
    \end{align}
  By Lemma~\ref{lemma:typeinversion}.\ref{sess},
  \begin{align}
    &   \vdash \cond\e\PP\PQ :\TT \label{eq:sr-cond-1}\\
    &  \vdash \h:\tqueue  \label{eq:sr-cond-2} \\
    &  \forall i\in I:\quad  \vdash \PP_i:\TT_i \label{eq:sr-cond-3}\\
    &  \forall i\in I:\quad  \vdash \h_i:\tqueue_i \label{eq:sr-cond-4}\\
    &  \Gamma = \{\pp:(\tqueue,\TT)\}  \cup \{\pp_i:(\tqueue_i,\TT_i):i\in I\} \label{eq:sr-cond-5}%
  \end{align}
  By Lemma~\ref{lemma:typeinversion}.\ref{proc4}:
    \begin{align}
      & \vdash \PP: \TT       \label{eq:sr-cond-6}\\
      & \vdash \PQ: \TT       \label{eq:sr-cond-7}\\
      & \vdash \e:\tbool       \label{eq:sr-cond-8}
    \end{align}
    Then, letting $\Gamma' = \Gamma'' = \Gamma$, we have:
    \begin{align}
      &  \Gamma'' \subt \Gamma & \text{(by reflexivity of $\subt$)}
      \\%
      &  \Gamma'' \reds \Gamma'%
      \\%
      &  \Gamma' \vdash \N'%
      &\text{%
        (by    \eqref{eq:sr-cond-sess}, \eqref{eq:sr-cond-2}, %
        \eqref{eq:sr-cond-6}  or \eqref{eq:sr-cond-7},
        and Def.~\ref{def:type-system})%
      }
    \end{align}
Inductive step:
\\
 \underline{\rulename{r-struct}} Assume that $\N \red \N'$ is derived from:
    \begin{align}
      &\N \equiv \N_1 \label{lem:sr:struct1} \\
      &\N_1 \red \N_1' \\
      &\N' \equiv \N_1'  \label{lem:sr:struct3}
    \end{align}
    
\noindent
From \eqref{lem:sr:hyp:typing}, \eqref{lem:sr:struct1}, by Lemma~\ref{lemma:type_congruence},  there is $\Gamma_1$ such that
   \begin{align}
      & \Gamma_1 \equiv \Gamma  \label{lem:sr:struct4} \\
      &  \Gamma_1 \vdash \N_1  \label{lem:sr:struct5}.
    \end{align}
  
 \noindent 
  By induction hypothesis, there is are live type environments %
  $\Gamma_1', \Gamma''$ 
  such that:
  \begin{align}
      & \Gamma'' \subt \Gamma_1%
        \;\text{ and }\; \Gamma'' \red^* \Gamma_1'   \label{lem:sr:struct6} \\
      & \Gamma_1'\vdash\N_1'  \label{lem:sr:struct7}
  \end{align}

 \noindent 
 Now, by \eqref{lem:sr:struct3} and Lemma~\ref{lemma:type_congruence}, 
 there is a live environment $\Gamma'$ such that
    \begin{align}
      & \Gamma' \equiv \Gamma'_1 
       \;\text{ and }\; \Gamma' \vdash \N' \label{lem:sr:struct8}
    \end{align}
    
\noindent
We conclude:
    \begin{align}
      & \Gamma'' \subt \Gamma %
      & (\text{by \eqref{lem:sr:struct6}, \eqref{lem:sr:struct4}, Lemma~\ref{lem:type-contexts-equiv->subt} and transitivity of} \subt) \\
      &  \Gamma'' \red^* \Gamma' & \text{(by \eqref{lem:sr:struct6}, \eqref{lem:sr:struct8} and rule \rulename{e-struct} of Def.~\ref{def:typing-env-reductions})}
    \end{align}
\end{proof}

We may now show Type Safety and Progress results, that are both corollaries of Subject Reduction and Error subject reduction results.

\corollaryProgress*
\begin{proof}
  Assume $\Gamma \vdash \N$ with $\Gamma$ live, %
  and $\N \red^* \N'$. %
  By Theorem~\ref{thm:SR} (subject reduction), %
  there is some live $\Gamma'$ such that:
  \begin{align}
    \label{eq:cor-progress:gammai}
    &\Gamma' \vdash \N'%
  \end{align}
  This implies that $\N'$ cannot be an (untypable) $\error$ %
  --- which is the first part of the thesis. %
  For the \emph{``also\ldots''} part of the statement, %
  we have two possibilities:
  \begin{enumerate}
  \item%
    $\N'\equiv\pa\pp\inact \pc \pa\pp\emptyqueue$. %
    This is the thesis; or,%
  \item%
    $\N'\not\equiv\pa\pp\inact \pc \pa\pp\emptyqueue$. %
    We have two sub-cases:
    \begin{enumerate}
    \item%
      there is $\N''$ such that $\N' \red \N'' \neq \error$. %
      This is the thesis; or, %
    \item%
      there is no $\N''$ such that $\N' \red \N'' \neq \error$. %
      This case is impossible, because it means that either:
      \begin{enumerate}
      \item%
        $\not\!\exists \N'': \N' \red \N''$. %
        This case is impossible. %
        In fact, it would imply that $\N'$ cannot reduce by %
        rules \rulename{r-send}, %
        \rulename{r-rcv}, \rulename{r-cond-t}, or \rulename{r-cond-f} %
        (possibly via \rulename{r-struct}) %
        in Table~\ref{tab:main:reduction}. %
        Since $\N'\not\equiv\pa\pp\inact \pc \pa\pp\emptyqueue$, %
        this can only happen because $\N'$ %
        has some process stuck on an external choice %
        without a matching message, %
        or has some expression %
        (in conditionals or outputs) %
        that cannot be evaluated. %
        But then, by the rules in Table~\ref{tab:error}, %
        we have $\N' \red \error$ %
        by at least one %
        of the rules \rulename{err-mism}, \rulename{err-dlock}, %
        \rulename{err-eval} or \rulename{err-eval2}. %
        This leads to case~\ref{item:cor-progress:ni-error} below;%
      \item\label{item:cor-progress:ni-error}%
        $\N' \red \error$. %
        This case is impossible. %
        In fact, if we admit this case, %
        by \eqref{eq:cor-progress:gammai} and the contrapositive of %
        Lemma~\ref{lem:error-subject-reduction}, %
        we have that $\Gamma'$ is not live, %
        which means %
        (by the contrapositive of Proposition~\ref{lem:move-preserves-liveness}) %
        that $\Gamma$ is not live --- contradiction.
      \end{enumerate}
    \end{enumerate}
  \end{enumerate}
\end{proof}

\section{Appendix of Section~\ref{sec:op}}
\label{ap:op}

\label{subsec:app:preciseness}

\subsection{Proof of \Cref{lem:subttt-negation}}
\label{sec:subttt-negation}

Our aim is to prove that
a pair of SISO trees is \emph{not} related by the coinducteve relation $\subttt$ \emph{if and only if} they are related by the inductive relation $\nsubttt$.

\newcommand{\PhiApp}[1]{\Phi\!\left({#1}\right)}%
We adopt an approach inspired by \citeN{BHLN12,BHLN17}, %
by introducing the notions of \emph{rule function} %
and \emph{failing derivation} for our SISO refinement:%
\begin{enumerate}
\item
  we define a \emph{rule function $\Phi$} %
  (\Cref{def:rule_function} below), %
  to associate a judgment of the form ``$\WT \subttt \WT'$'' %
  \;to its coinductive premise in \Cref{def:ref} (if any), 
  so that $\WT \subttt \WT'$ can be derived %
  (by \Cref{def:ref}) %
  if and only if %
  \begin{enumerate*}[label=\emph{(\arabic*)}]
  \item%
    $\PhiApp{\WT \subttt \WT'}$ yields ``$\true$,'' or
  \item%
    $\PhiApp{\WT \subttt \WT'}$ yields
    a derivable judgement %
    according to the rules in \Cref{def:ref};
  \end{enumerate*}
\item
  then, in \Cref{def:derivation}
  we use the rule function $\Phi$ to define sequences of judgements
  representing \emph{successful or failing derivations}
  according to the rules in \Cref{def:ref};
\item
  finally, in \Cref{lem:failing-nsubttt-ind}
  we show that if we have a failing derivation,
  then for each one of its judgements ``$\WT \subttt \WT'$'',
  we can derive\; $\WT \nsubttt \WT'$
  \;under the rules in \Cref{tab:negationW}.
  Symmetrically, in \Cref{lem:nsubttt-failing-ind}
  we prove that for each judgement $\WT \nsubttt \WT'$ in a derivation
  based on \Cref{tab:negationW} %
  we can construct a failing derivation for ``$\WT \subttt \WT'$''.
  These results lead to the proof of \Cref{lem:subttt-negation}.
\end{enumerate}

Before proceeding,
it is handy to introduce some auxiliary rules for $\nsubttt$
(\Cref{def:aux-ref}) that are sound \wrt the rules in \Cref{tab:negationW}
(as per \Cref{lem:aux-ref-sound} below).
Such auxiliary rules will only be used in the proofs of this section.

\begin{definition}[Auxiliary SISO refinement rules]
  \label{def:aux-ref}
  We define the following additional rules for the SISO refinement relation:
  \[
    \inferruleR[\rulename{n-${\AC}$-act}]{
      \S' \subs\S & \actions{\WT} \neq \actions{\ACon\pp\WT'}
    }{
      \tin\pp{\ell}{\ST}.\WT \nsubttt \ACon\pp{{\tin\pp{\ell}{\ST'}.\WT'}}
    }
    \qquad
    \inferruleR[\rulename{n-${\BC}$-act}]{
      \S\subs\S' & \actions{\WT} \neq \actions{\BCon\pp\WT'}
    }{
      \tout\pp{\ell}{\ST}.\WT \nsubttt \BCon\pp{{\tout\pp{\ell}{\ST'}.\WT'}}
    }
  \]
\end{definition}

\begin{proposition}[Soundness of the auxiliary SISO refinement rules]
  \label{lem:aux-ref-sound}
  Assume that\, $\WT_1 \nsubttt \WT_2$ holds by one of the auxiliary rules
  introduced in \Cref{def:aux-ref}.
  Then,\, $\WT_1 \nsubttt \WT_2$ can also be derived
  by only using the rules in \Cref{tab:negationW}.
\end{proposition}
\begin{proof}
  We have two cases.
  \begin{itemize}
  \item $\WT_1 \nsubttt \WT_2$ holds by \rulename{n-${\AC}$-act}.\quad
    In this case, we have:
    \begin{align}
      \label{eq:lem:aux-ref-sound:ac-act:premises-wt}
      &\WT_1 = \tin\pp{\ell}{\ST}.\WT
      \quad\text{and}\quad
      \WT_2 = \ACon\pp{{\tin\pp{\ell}{\ST'}.\WT'}}
      \\
      \label{eq:lem:aux-ref-sound:ac-act:premises-s}
      &\ST' \subs\ST
      \\
      \label{eq:lem:aux-ref-sound:ac-act:premises-act}
      &\actions{\WT} \neq \actions{\ACon\pp\WT'}
    \end{align}
    Observe that, by \eqref{eq:lem:aux-ref-sound:ac-act:premises-act},
    we have two (non-mutually exclusive) cases:
    \begin{enumerate}[label=(\alph*)]
    \item
      for some $\pp$, an action $\pp!$ (resp.~$\pp?$)
      occurs in $\WT$ but not in $\ACon\pp\WT'$.
      In this case, assume that a prefix $\tout{\pp}{\ell'}{\ST''}$
      (resp.~$\tin{\pp}{\ell'}{\ST''}$) corresponding to such an action
      first occurs at the $n$\textsuperscript{th} position in $\WT$.
      Then, there is a derivation of $\WT \not\subttt \ACon\pp\WT'$
      such that either:
      \begin{enumerate}[label=(\roman*)]
      \item%
        after less than $n-1$ applications of the inductive rules
        \rulename{n-inp-$\WT$}, \rulename{n-${\AC}$-$\WT$},
        \rulename{n-out-$\WT$} or \rulename{n-${\BC}$-$\WT$},
        the derivation reaches an axiom of \Cref{tab:negationW}, or
      \item
        after exactly $n-1$ applications of such inductive rules,
        the derivation reaches the axiom \rulename{n-out}
        (resp.~\rulename{n-inp}).
      \end{enumerate}
      Therefore, in both cases,
      we have that $\WT \not\subttt \ACon\pp\WT'$ is derivable
      by only using the rules in \Cref{tab:negationW}.
      And from this, by \eqref{eq:lem:aux-ref-sound:ac-act:premises-wt}
      and \eqref{eq:lem:aux-ref-sound:ac-act:premises-s},
      we conclude that $\WT_1 \nsubttt \WT_2$ holds by rule
      \rulename{n-${\BC}$-$\WT$} (resp.~\rulename{n-${\AC}$-$\WT$}),
      with a derivation that only uses the rules in \Cref{tab:negationW};
    \item
      for some $\pp$, an action $\pp!$ (resp.~$\pp?$)
      occurs in $\ACon\pp\WT'$ but not in $\WT$.
      In this case, assume that a prefix $\tout{\pp}{\ell'}{\ST''}$
      (resp.~$\tin{\pp}{\ell'}{\ST''}$) corresponding to such an action
      first occurs at the $n$\textsuperscript{th} position in $\ACon\pp\WT'$.
      Then, there is a derivation of $\WT \not\subttt \ACon\pp\WT'$
      such that either:
      \begin{enumerate}[label=(\roman*)]
      \item%
        after less than $n-1$ applications of the inductive rules
        \rulename{n-inp-$\WT$}, \rulename{n-${\AC}$-$\WT$},
        \rulename{n-out-$\WT$} or \rulename{n-${\BC}$-$\WT$},
        the derivation reaches an axiom of \Cref{tab:negationW}, or
      \item
        after exactly $n-1$ applications of such inductive rules,
        the derivation reaches the axiom \rulename{n-out-R}
        (resp.~\rulename{n-inp-R}).
      \end{enumerate}
      Therefore, in both cases,
      we have that $\WT \not\subttt \ACon\pp\WT'$ is derivable
      by only using the rules in \Cref{tab:negationW}.
      And from this, by \eqref{eq:lem:aux-ref-sound:ac-act:premises-wt}
      and \eqref{eq:lem:aux-ref-sound:ac-act:premises-s},
      we conclude that $\WT_1 \nsubttt \WT_2$ holds by rule
      \rulename{n-${\BC}$-$\WT$} (resp.~\rulename{n-${\AC}$-$\WT$}),
      with a derivation that only uses the rules in \Cref{tab:negationW};
    \end{enumerate}
  \item $\WT_1 \nsubttt \WT_2$ holds by \rulename{n-${\BC}$-act}.\quad
    The proof is similar to the previous case.
  \end{itemize}
\end{proof}

\begin{definition}
\label{def:rule_function}
  The \emph{rule function $\Phi$} %
  maps judgments of the form ``$\WT_1 \subttt \WT_2$''
  to judgment of the same form, or ``$\true$,'' or ``$\false,$'' %
  as follows:
  \[
    \PhiApp{\WT_1 \subttt \WT_2} \;=\;
    \begin{cases}
      \WT \subttt \WT'
      & \text{if }
      \begin{cases}
        \WT_1 = \tin\pp{\ell}{\ST}.\WT \\
        \text{and } \WT_2 = \tin\pp{\ell}{\ST'}.\WT' \\
        \text{and } \ST' \subs \ST
      \end{cases}
      \\
      \WT \subttt \ACon{\pp}{\WT'}
      & \text{if }
      \begin{cases}
        \WT_1 = \tin\pp{\ell}{\ST}.\WT \\
        \text{and } \WT_2 = \ACon\pp{{\tin\pp{\ell}{\ST'}.\WT'}} \\
        \text{and } \ST' \subs \ST \\
        \text{and } \actions{\WT}=\actions{\ACon\pp{\WT'}}
      \end{cases}
      \\%
      \WT \subttt \WT'
      & \text{if }
      \begin{cases}
        \WT_1 =  \tout\pp{\ell}{\ST}.\WT \\
        \text{and } \WT_2 = \tout\pp{\ell}{\ST'}.\WT' \\
        \text{and } \ST \subs \ST'
      \end{cases}
      \\
      \WT \subttt \BCon\pp{\WT'}
      & \text{if }
      \begin{cases}
        \WT_1 = \tout\pp{\ell}{\ST}.\WT \\
        \text{and } \WT_2 = \BCon\pp {\tout\pp{\ell}{\ST'}.\WT'} \\
        \text{and } \ST \subs \ST' \\
        \text{and } \actions{\WT}=\actions{\BCon\pp{\WT'}}
      \end{cases}
      \\
      \true & \text{if $\WT_1 = \WT_2 = \tend$} \\
      \false & \text{in all other cases}
    \end{cases}
  \]
\end{definition}

\begin{definition}
\label{def:derivation}
\label{def:failing_derivation}
  Consider %
  a (possibly infinite) sequence $(J_i)_{i \in I}$, %
  where $I$ contains consecutive natural numbers starting with $1$. %
  We say that the sequence is a \emph{derivation of\, $\WT \subttt \WT'$} %
  iff it satisfies the following requirements:
  \begin{enumerate}
  \item\label{item:deriv:j}
    for all $i \in I$,
    $J_i$ is a judgement of the form ``$\WT_i \subttt \WT'_i$'';
  \item\label{item:deriv:first} $J_1$ is\; %
    ``$\WT \subttt \WT'$'';
  \item\label{item:deriv:m} for all $i \in I$,\; $i > 1$ \;implies\; %
    $J_i = \PhiApp{\WT_{i-1} \subttt \WT'_{i-1}}$.
  \end{enumerate}
  We say that a derivation $(J_i)_{i \in I}$ is \emph{successful} %
  iff it is infinite, %
  or if it has a finite length $n$ and $\PhiApp{J_n} = \true$.

  We say that a derivation $(J_i)_{i \in I}$ is \emph{failing} %
  iff it has a finite length $n$ and $\PhiApp{J_n} = \false$.
\end{definition}

\begin{proposition}
\label{lem:failing-deriv}
   The refinement\, $\WT \subttt \WT'$
   is derivable (resp.~\emph{not} derivable)
   with the rules of \Cref{def:ref} %
   iff there is a successful (resp.~failing) derivation %
   of\, $\WT \subttt \WT'$.
\end{proposition}
\begin{proof}
  Straightforward by \Cref{def:derivation}.
\end{proof}

\begin{corollary}
  \label{lem:failing-deriv:n-forms}%
  Assume that the sequence of judgements %
  $(J_1,\ldots,J_n)$ (for some $n \geq 1$) is a failing derivation
  of\, $\WT_1 \subttt \WT'_1$. 
  \;Then, %
  $J_n$ is a judgement of the form ``$\WT_n \subttt \WT'_n$''
  \;such that $\WT_n \nsubttt \WT'_n$ holds
  by at least one of the axioms in \Cref{tab:negationW}
  or \Cref{def:aux-ref}, \ie:
  \emph{%
    \rulename{n-out}, \rulename{n-inp}, \rulename{n-out-R}, \rulename{n-inp-R},
    \rulename{n-inp-$\ell$}, \rulename{n-inp-$\S$},
    \rulename{n-${\AC}$-$\ell$}, \rulename{n-${\AC}$-$\S$},
    \rulename{n-i-o-1}, \rulename{n-i-o-2},
    \rulename{n-out-$\ell$}, \rulename{n-out-$\S$},
    \rulename{n-${\BC}$-$\ell$}, \rulename{n-${\BC}$-$\S$},
    \rulename{n-${\AC}$-act}, \rulename{n-${\BC}$-act}.
  }%
\end{corollary}
\begin{proof}
  By \Cref{lem:failing-deriv}, we have $\PhiApp{J_n} = \emptyset$. %
  Therefore, we obtain the thesis from \Cref{def:rule_function}, 
  by inspecting all cases that falsify all its ``if\ldots'' clauses:
  each case corresponds to either at least one axiom in \Cref{tab:negationW},
  or at least one additional axiom in \Cref{def:aux-ref}.
\end{proof}

\begin{proposition}
  \label{lem:failing-nsubttt-ind}
  Assume that $(J_1,\ldots,J_n)$ (for some $n \geq 1$) is a failing derivation
  of\, $\WT \subttt \WT'$. 
  Then, for all $i \in 1..n$, we have that\, $\WT_i \nsubttt \WT'_i$ %
  is derivable by the rules in \Cref{tab:negationW}.
\end{proposition}
\begin{proof}
  The proof is by induction on $n$, %
  decreasing toward $1$:
  \begin{itemize}
  \item base case $n$.\quad
    In this case, by \Cref{lem:failing-deriv:n-forms}, we have either:
    \begin{enumerate}[label=(\alph*)]
    \item%
      $\WT_n \nsubttt \WT'_n$ holds
      by some axiom in \Cref{tab:negationW},
      which is the thesis; or
    \item
      $\WT_n \nsubttt \WT'_n$ holds
      by some axiom in \Cref{def:aux-ref}.
      In this case, by \Cref{lem:aux-ref-sound},
      there is a derivation of $\WT_n \nsubttt \WT'_n$ 
      that only uses the rules in \Cref{tab:negationW}, which is the thesis;
    \end{enumerate}
  \item inductive case $i < n$.
    By \Cref{def:failing_derivation},
    we know that\; $J_{i+1} = \PhiApp{\WT_i \subttt \WT'_i} \neq \false$.
    \;Therefore, by \Cref{def:rule_function},
    we have one of the following possibilities:
    \begin{itemize}
      \item
        $\WT_i = \tin\pp{\ell}{\ST}.\WT$
        and $\WT'_i = \tin\pp{\ell}{\ST'}.\WT'$
        and $\ST' \subs \ST$.
        In this case, by \Cref{def:failing_derivation,def:rule_function},
        we also have $J_{i+1} = \PhiApp{\WT_i \subttt \WT'_i} =
        \text{``$\WT \subttt \WT'$''}$ ---
        hence, by the induction hypothesis,
        $\WT \nsubttt \WT'$ is derivable by the rules in \Cref{tab:negationW}.
        Thus, we conclude that
        $\WT_i \nsubttt \WT'_i$ is derivable by \rulename{n-inp-$\WT$},
        only using the rules in  \Cref{tab:negationW};
      \item
        $\WT_i = \tin\pp{\ell}{\ST}.\WT$
        and $\WT'_i = \ACon\pp{{\tin\pp{\ell}{\ST'}.\WT'}}$
        and $\ST' \subs \ST$
        and $\actions{\WT}=\actions{\ACon\pp{\WT'}}$.
        In this case, by \Cref{def:failing_derivation,def:rule_function},
        we also have $J_{i+1} = \PhiApp{\WT_i \subttt \WT'_i} =
        \text{``$\WT \subttt \ACon{\pp}{\WT'}$''}$ ---
        hence, by the induction hypothesis,
        $\WT \nsubttt \ACon{\pp}{\WT'}$
        is derivable by the rules in \Cref{tab:negationW}.
        Thus, we conclude that
        $\WT_i \nsubttt \WT'_i$ is derivable by \rulename{n-$\AC$-$\WT$},
        only using the rules in  \Cref{tab:negationW};
      \item
        $\WT_i = \tout\pp{\ell}{\ST}.\WT$
        and $\WT'_i = \tout\pp{\ell}{\ST'}.\WT'$
        and $\ST \subs \ST'$.
        In this case, by \Cref{def:failing_derivation,def:rule_function},
        we also have $J_{i+1} = \PhiApp{\WT_i \subttt \WT'_i} =
        \text{``$\WT \subttt \WT'$''}$ ---
        hence, by the induction hypothesis,
        $\WT \nsubttt \WT'$ is derivable by the rules in \Cref{tab:negationW}.
        Thus, we conclude that
        $\WT_i \nsubttt \WT'_i$ is derivable by \rulename{n-out-$\WT$},
        only using the rules in  \Cref{tab:negationW};
      \item
        $\WT_i = \tout\pp{\ell}{\ST}.\WT$
        and $\WT'_i = \BCon\pp{{\tout\pp{\ell}{\ST'}.\WT'}}$
        and $\ST \subs \ST'$
        and $\actions{\WT}=\actions{\ACon\pp{\WT'}}$.
        In this case, by \Cref{def:failing_derivation,def:rule_function},
        we also have $J_{i+1} = \PhiApp{\WT_i \subttt \WT'_i} =
        \text{``$\WT \subttt \BCon{\pp}{\WT'}$''}$ ---
        hence, by the induction hypothesis,
        $\WT \nsubttt \ACon{\pp}{\WT'}$
        is derivable by the rules in \Cref{tab:negationW}.
        Thus, we conclude that
        $\WT_i \nsubttt \WT'_i$ is derivable by \rulename{n-$\BC$-$\WT$},
        only using the rules in  \Cref{tab:negationW};
      \item
        $\WT_i = \WT'_i = \tend$. This case is impossible,
        as it would imply
        $J_{i+1} = \PhiApp{\WT_i \subttt \WT'_i} = \true$,
        hence $(J_1,\ldots,J_n)$ would \emph{not} be a failing derivation
        according to \Cref{def:failing_derivation},
        thus contradicting the hypothesis in the statement.
    \end{itemize}
  \end{itemize}
\end{proof}

\begin{proposition}
  \label{lem:nsubttt-failing-ind}
  Assume that $\WT \nsubttt \WT'$ holds, by the rules in \Cref{tab:negationW};
  also assume that its derivation has $n \geq 1$ rule applications,
  counting upwards %
  (\ie, the $1$\textsuperscript{st} rule concludes\, $\WT \nsubttt \WT'$,
  while the $n$\textsuperscript{th} rule is the axiom),
  with the $i$\textsuperscript{th} rule concluding\,
  $\WT_i \nsubttt \WT'_i$ (for $i \in 1..n$).
  Then,\; $\PhiApp{\WT_n \subttt \WT'_n} = \false$;
  \;moreover, for all $i \in 1..n{-}1$,
  we have either\; $\PhiApp{\WT_i \subttt \WT'_i} = \false$, \;or\;
  $\PhiApp{\WT_i \subttt \WT'_i} = \text{``$\WT_{i+1} \subttt \WT'_{i+1}$''}$.
\end{proposition}
\begin{proof}
    The judgement $\WT_n \nsubttt \WT'_n$ holds by at least one axiom among
    \rulename{n-out}, \rulename{n-inp}, \rulename{n-out-R}, \rulename{n-inp-R},
    \rulename{n-inp-$\ell$}, \rulename{n-inp-$\S$},
    \rulename{n-${\AC}$-$\ell$}, \rulename{n-${\AC}$-$\S$},
    \rulename{n-i-o-1}, \rulename{n-i-o-2},
    \rulename{n-out-$\ell$}, \rulename{n-out-$\S$},
    \rulename{n-${\BC}$-$\ell$}, \rulename{n-${\BC}$-$\S$}.
    We proceed by cases on such axioms, inspecting the corresponding shapes
    of $\WT_n$ and $\WT'_n$: in all cases we conclude %
    $\PhiApp{\WT_n \subttt \WT'_n} = \false$.
    For any other $m < n$, the judgement $\WT_m \nsubttt \WT'_m$
    holds by one of the inductive rules in \Cref{tab:negationW}.
    Hence, we have the following possibilities:
    \begin{itemize}
    \item%
      \rulename{n-inp-$\WT$}.\quad
      In this case, we have $\WT_m = \tin\pp{\ell}{\ST}.\WT_{m+1}$,
      \;$\WT'_m = \tin\pp{\ell}{\ST'}.\WT_{m+1}$,
      \;and\; $\WT_{m+1} \nsubttt \WT'_{m+1}$ \;and\; $\ST' \subttt \ST$.
      Therefore, by \Cref{def:rule_function} we conclude
      $\PhiApp{\WT_m \subttt \WT'_m} =
      \text{``$\WT_{m+1} \subttt \WT'_{m+1}$''}$;
    \item%
      \rulename{n-${\AC}$-$\WT$}.\quad
      In this case, we have $\WT_m = \tin\pp{\ell}{\ST}.\WT_*$,
      \;$\WT'_m = \ACon{\pp}{\tin\pp{\ell}{\ST'}.\WT'_*}$,
      \;and\; $\ST' \subttt \ST$,
      \;and\; $\WT_* \nsubttt \ACon{\pp}{\WT'_*}$.
      Observe that the latter is the $m+1$\textsuperscript{th}
      step of the derivation, \ie, we have\; $\WT_{m+1} \nsubttt \WT'_{m+1}$
      with\; $\WT_{m+1}=\WT_*$ \;and\; $\WT'_{m+1} = \ACon{\pp}{\WT'_*}$.
      \;Therefore, we have the following sub-cases:
      \begin{itemize}
      \item%
        $\actions{\WT_{m+1}} \neq \actions{\WT'_{m+1}}$.\quad
        In this case, by \Cref{def:rule_function} we conclude
        $\PhiApp{\WT_m \subttt \WT'_m} = \false$;
      \item%
        $\actions{\WT_{m+1}} = \actions{WT'_{m+1}}$.\quad
        In this case, 
        by \Cref{def:rule_function} we conclude
        $\PhiApp{\WT_m \subttt \WT'_m} =
        \text{``$\WT_{m+1} \subttt \WT'_{m+1}$''}$;
      \end{itemize}
    \item%
      \rulename{n-out-$\WT$}.\quad
      Similar to case \rulename{n-inp-$\WT$} above;
    \item%
      \rulename{n-${\BC}$-$\WT$}.\quad
      Similar to case \rulename{n-$\AC$-$\WT$} above.
    \end{itemize}
\end{proof}

\lemNegationW*
\begin{proof}
  ($\implies$) Assume that $\WT \subttt \WT'$ is \emph{not} derivable:
  then, by \Cref{lem:failing-deriv} and \Cref{def:failing_derivation},
  there is a failing derivation $(J_1,\ldots,J_n)$ (for some finite $n \geq 1$)
  where $J_1$ is the judgement\, ``$\WT \subttt \WT'$''. %
  Then, by \Cref{lem:failing-nsubttt-ind}, we obtain the thesis.%

  ($\impliedby$) Assume that $\WT \nsubttt \WT'$ holds,
  by the rules in \Cref{tab:negationW}, %
  with a derivation with $n \geq 1$ rule applications.
  Then, by \Cref{lem:nsubttt-failing-ind},
  there is a failing derivation for ``$\WT \nsubttt \WT'$''
  (of length $n$ or less),
  and by \Cref{lem:failing-deriv}, we conclude that\,
  $\WT \nsubttt \WT'$ is \emph{not} derivable.
\end{proof}

\subsection{Regular representatives for subtyping negation}\label{sec:regular_representatives_appendix}

In the sequel, we will always consider only regular representatives of SO and SI trees that appear in the definition of the negation of subtyping. Before we adopt that approach, we will prove that whenever there exist a pair of (possibly non-regular) representatives $\UT\in \llbracket \ttree{\T} \rrbracket_\SO$ and $\VT'\in\llbracket \ttree{\T'} \rrbracket_\SI$ with $\UT \nsubttt \VT',$ there is also a pair of regular representatives $\ttree{\UU_1}\in \llbracket \ttree{\T} \rrbracket_\SO$ and $\ttree{\V_1'}\in \llbracket \ttree{\T'} \rrbracket_\SI$ such that $\ttree{\UU_1}\nsubttt\ttree{\V_1'}.$

We start by proving that for each irregular tree $\UT \in  \llbracket \ttree{\T} \rrbracket_\SO$ there is a regular tree  $\UT_1 \in  \llbracket \ttree{\T} \rrbracket_\SO$ such that $\UT$ and $\UT_1$ overlap in at least top $n$ levels, for a given $n.$  For that purpose, we introduce two auxiliary functions, $\regU{\UT}{i}{\T}{\T'}$ and $\mufree{\T}.$
The function $\regU{\UT}{i}{\T}{\T'},$ with $\UT\in \llbracket \ttree{\T} \rrbracket_\SO,$ follows  in the tree of $\T$  the pattern determined by top  $i$ levels of $\UT$ and extracts (step by step) a type $\UU$  with the tree $\ttree{\UU}$ that follows the same pattern. The type $\UU$ might not be unique, but all such types have the same top $i$ levels. Each step of the procedure applies one of the three options that  are introduced and  clarified along the following lines.
\begin{itemize}
\item[(1)]
If $\T = \mu\ty.\T_1,$  then 
\begin{align*}
   \regU{\UT}{i}{\mu\ty.\T_1}{\T'} = & \mu\ty. \regU{\UT}{i}{\T_1}{\mu\ty.\T_1}, \text{ for every }i\geq 0;
\end{align*}
The function goes behind $\mu\ty$  and in the same time  the forth parameter keeps the information on the form of the $\mu$ type (it might be needed later on for unfolding).
\item[(2)]
If  $\T\neq \mu\ty.\T_1,$ for any $\ty$ and $\T_1,$ and $i>0$, then
{\small
 \begin{align*}
    \regU{\UT}{i}{\T}{\T'} = & \left\{
    \begin{array}{ll}
         \tend & ,\UT=\tend\\
         \tout\pp{\ell_j}{\S_j}. \regU{\UT_j}{i-1}{\T_j}{\T'} & ,\UT=\tout\pp{\ell_j}{\S_j}.\UT_j, j\in K,\\
                                                                                 & \T=\tinternal_{k\in K}\tout\pp{\ell_k}{\S_k}.{\T_k},\\
         \texternal_{k\in K} \tin\pp{\ell_k}{\S_k}. \regU{\UT_k}{i-1}{\T_k}{\T'} & ,\UT=\texternal_{k\in K}\tin\pp{\ell_k}{\S_k}.{\UT_k}, \\
                                                                    & \T=\texternal_{k\in K}\tin\pp{\ell_k}{\S_k}.{\T_k},\\
         \regU{\UT}{i}{\T_1}{\T'} &  ,\T=\ty, \T'=\mu\ty.\T_1 \\
    \end{array}
    \right.
 \end{align*}}
The function extracts from $\T$ the prefix for $\UU$ that will induce the same level in its tree as  $\UT.$ If $\T=\ty,$ it first applies unfolding, recovering the form  for substitution from the forth parameter.
 \item[(3)]
 If  $\T\neq \mu\ty.\T_1,$ for any $\ty$ and $\T_1,$ and $i=0$, then
\\
 $ \regU{\UT}{0}{\T}{\T'} =  \UU_1$  for some $\ttree{\UU_1}\in \llbracket \ttree{\T\subst{\T'''}{\ty_1}} \rrbracket_\SO,$ where $\T'=\mu\ty_1.\T''$ and $\T'''=\mu\ty_2.\T''\subst{\ty_2}{\ty_1}$ (we choose here a fresh name $\ty_2$). The choice of $\UU_1$ is not unique and we will always choose one that satisfies $\actions{\UT} \supseteq \actions{\ttree{\UU_1}}.$
\end{itemize}

One can notice that the previous procedure might  create some terms of the form $\mu\ty.\UU'$ with $\ty \not\in \UU'.$ These terms are cleaned up by
$ \mufree{\T}$, that is defined as follows:

 \[
   \mufree{\T}=\left\{
   \begin{array}{ll}
        \tend &  \T=\tend \\
        \tinternal_{k\in K}\tout\pp{\ell_k}{\S_k}.\mufree{\T_k} &  \T= \tinternal_{k\in K} \tout\pp{\ell_k}{\S_k}.\T_k \\
        \texternal_{k\in K}\tin\pp{\ell_k}{\S_k}.\mufree{\T_k}  &  \T= \texternal_{k\in K} \tin\pp{\ell_k}{\S_k}.\T_k \\
        \mu\ty.\T' & \T=\mu\ty.\T', \ty\in \T' \\
         \mufree{\T'} & \T=\mu\ty.\T', \ty \not \in \T' \\
   \end{array} \right.
\]

\begin{lemma}
\label{lem:regU}
  If $\UT\in \llbracket \ttree{\T} \rrbracket_\SO$  then there is $\UU_1$ such that $\ttree{\UU_1}\in   \llbracket \ttree{\T} \rrbracket_\SO$ and $\ttree{\UU_1}$ overlaps with $\UT$ at top $n$ levels.
 \end{lemma}

\begin{proof} 
  If $\UT$ is finite, then we  choose $\UU_1=\UT.$
 If $\UT$ is infinite, then we choose $\UU_1=\mufree{\regU{\UT}{n}{\T}{\T}}.$ 
\end{proof}

In the following two examples we illustrate the procedure on some interesting cases.

\begin{example}
  Take $\T=\mu\ty_1.(\tout\pp{\ell_1}{\S_1}.\ty_1 \& \tout\pp{\ell_2}{\S_2}.\mu\ty_2.\tout\pp{\ell_3}{\S_3}.\ty_2)$ and choose $\UT\in \llbracket \ttree{\T} \rrbracket_\SO$ such that  $\UT=\tout\pp{\ell_1}{\S_1}.\tout\pp{\ell_2}{\S_2}.\tout\pp{\ell_3}{\S_3}\ldots$ We show here that the procedure introduced above gives a regular $\UU_1$ that overlaps with  $\UT$ (at least) in the top 3 levels and $\ttree{\UU_1}\in \llbracket \ttree{\T} \rrbracket_\SO$.
  \begin{align*}
      \UU'_1 = & \regU{\UT}{ 3}{\T}{\T} \\
               = & \mu\ty_1. \regU{\UT}{3}{\tout\pp{\ell_1}\ty_1 \& \tout\pp{\ell_2}{\S_2}.\mu\ty_2.\tout\pp{\ell_3}{\S_3}.\ty_2}{\T}\\
               = & \mu\ty_1.\tout\pp{\ell_1}{\S_1}. \regU{\tout\pp{\ell_2}{\S_2}.\tout\pp{\ell_3}{\S_3}\ldots}{2}{\ty_1}{\T}\\
               = & \mu\ty_1.\tout\pp{\ell_1}{\S_1}. \regU{\tout\pp{\ell_2}{\S_2}.\tout\pp{\ell_3}{\S_3}\ldots}{2}{\tout\pp{\ell_1}{\S_1}.\ty_1 \& \tout\pp{\ell_2}{\S_2}.\mu\ty_2.\tout\pp{\ell_3}{\S_3}.\ty_2}{\T}\\
               = & \mu\ty_1.\tout\pp{\ell_1}{\S_1}.\tout\pp{\ell_2}{\S_2}. \regU{\tout\pp{\ell_3}{\S_3}\ldots}{1}{\mu\ty_2.\tout\pp{\ell_3}{\S_3}.\ty_2}{\T}\\
               = & \mu\ty_1.\tout\pp{\ell_1}{\S_1}.\tout\pp{\ell_2}{\S_2}. \mu\ty_2.\regU{\tout\pp{\ell_3}{\S_3}\ldots}{1}{\tout\pp{\ell_3}{\S_3}.\ty_2}{\mu\ty_2.\tout\pp{\ell_3}{\S_3}.\ty_2}\\
               = & \mu\ty_1.\tout\pp{\ell_1}{\S_1}.\tout\pp{\ell_2}{\S_2}. \mu\ty_2.\tout\pp{\ell_3}{\S_3}.\regU{\tout\pp{\ell_3}{\S_3}\ldots}{0}{\ty_2}{\mu\ty_2.\tout\pp{\ell_3}{\S_3}.\ty_2}\\
               = & \mu\ty_1.\tout\pp{\ell_1}{\S_1}.\tout\pp{\ell_2}{\S_2}. \mu\ty_2.\tout\pp{\ell_3}{\S_3}.  \mu\ty_3.\tout\pp{\ell_3}{\S_3}. \ty_3.
  \end{align*}
  After erasure of the meaningless $\mu$ terms, we get
    \begin{align*}
      \UU_1 = & \mufree{\UU'} = \mufree{\mu\ty_1.\tout\pp{\ell_1}{\S_1}.\tout\pp{\ell_2}{\S_2}. \mu\ty_2.\tout\pp{\ell_3}{\S_3}. \mu\ty_3.\tout\pp{\ell_3}{\S_3}. \ty_3} \\
               = & \mufree{\tout\pp{\ell_1}{\S_1}.\tout\pp{\ell_2}{\S_2}. \mu\ty_2.\tout\pp{\ell_3}{\S_3}. \mu\ty_3.\tout\pp{\ell_3}{\S_3}. \ty_3}  \\
               = & \tout\pp{\ell_1}{\S_1}.\tout\pp{\ell_2}{\S_2}.  \mufree{\mu\ty_2.\tout\pp{\ell_3}{\S_3}. \mu\ty_3.\tout\pp{\ell_3}{\S_3}. \ty_3}  \\
               = & \tout\pp{\ell_1}{\S_1}.\tin\pp{\ell_2}{\S_2}. \tout\pp{\ell_3}{\S_3}. \mufree{\mu\ty_3.\tout\pp{\ell_3}{\S_3}. \ty_3}\\
               = & \tout\pp{\ell_1}{\S_1}.\tin\pp{\ell_2}{\S_2}. \tout\pp{\ell_3}{\S_3}. \mu\ty_3.\tout\pp{\ell_3}{\S_3}. \ty_3.
    \end{align*}
    
\end{example}

\begin{example}
  Take $\T=\mu\ty.(\tout\pp{\ell_1}{\S_1}.\tin\pp{\ell_4}{\S_4}.\ty \& \tout\pp{\ell_2}{\S_2}.\ty)$ and choose $\UT\in \llbracket \ttree{\T} \rrbracket_\SO$ such that  $\UT=\tout\pp{\ell_1}{\S_1}\ldots$ We consutruct here a regular $\UU_1$ that overlaps with  $\UT$ at the top level and $\ttree{\UT_1}\in \llbracket \ttree{\T} \rrbracket_\SO$.
  \begin{align*}
      \UU_1'  = & \regU{\UT}{ 1}{\T}{\T} \\
               = & \mu\ty. \regU{\UT}{1}{\tout\pp{\ell_1}{\S_1}.\tin\pp{\ell_4}{\S_4}.\ty \& \tout\pp{\ell_2}{\S_2}.\ty}{\T}\\
               = & \mu\ty. \tout\pp{\ell_1}{\S_1}. \regU{\UT}{0}{\tin\pp{\ell_4}{\S_4}.\ty}{\T} 
  \end{align*}
  We can now choose any $\UU_2$ such that  
  \begin{align*}
       \ttree{\UU_2} \in \llbracket \ttree{\tin\pp{\ell_4}{\S_4}.\mu\ty_1.(\tout\pp{\ell_1}{\S_1}.\tin\pp{\ell_4}{\S_4}.\ty \& \tout\pp{\ell_2}{\S_2}.\ty_1)} \rrbracket_\SO.
  \end{align*}
   For example, with
  $\UU_2=\tin\pp{\ell_4}{\S_4}.\mu\ty_2. \tout\pp{\ell_2}{\S_2}.\ty_2$ we get
    \begin{align*}
      \UU_1 = & \mufree{\mu\ty. \tout\pp{\ell_1}{\S_1}.\UU_2} \\
               = & \mufree{\mu\ty.\tout\pp{\ell_1}{\S_1}.\tin\pp{\ell_4}{\S_4}.\mu\ty_2. \tout\pp{\ell_2}{\S_2}.\ty_2} \\
               = & \tout\pp{\ell_1}{\S_1}.\tin\pp{\ell_4}{\S_4}.\mu\ty_2. \tout\pp{\ell_2}{\S_2}.\ty_2.
  \end{align*}
\end{example}

\begin{lemma}
\label{lem:regV}
  If $\VT'\in \llbracket \ttree{\T'} \rrbracket_\SO$  then there is $\V'_1$ such that $\ttree{\V'_1}\in   \llbracket \ttree{\T'} \rrbracket_\SI$ and $\ttree{\V'_1}$ overlaps with $\VT'$ at top $n$ levels.
 \end{lemma}
 
 \begin{proof}
   The construction is analogous to the one from the previous lemma.
 \end{proof}

\begin{corollary}
\label{lem:reg-iregU}
   Let $\UT\in \llbracket \ttree{\T} \rrbracket_\SO$ and $\VT'\in \llbracket \ttree{\T'} \rrbracket_\SI$ be such that $\UT\not\subt \VT'.$
   Then, there are $\UU_1$ and $\V'_1$ such that $\ttree{\UU_1}\in \llbracket \ttree{\T} \rrbracket_\SO$ and $\ttree{\V_1'}\in \llbracket \ttree{\T'} \rrbracket_\SI$   
   and $\ttree{\UU_1} \not\subt \ttree{\V_1'}.$
\end{corollary}

\begin{proof}
If $\UT \not\subt \VT'$ was derived in $n$ steps $(n\geq 1)$, there  is $k$ such that  prefixes from  top $n$ levels of $\UT$ that appear in $\VT'$ are placed in the top $k$ levels of $\VT'$ (those that are considered for the negation derivation). By Lemma~\ref{lem:regU} and Lemma~\ref{lem:regV}, there are  $\UU_1$ and $\V_1'$ such that $\ttree{\UU_1} \in  \llbracket \ttree{\T} \rrbracket_\SO$ and $\ttree{\V_1'} \in  \llbracket \ttree{\T'} \rrbracket_\SI$  such that $\ttree{\UU_1}$ ovelaps with $\UT$ in top $n$ levels and $\ttree{\V_1}$ ovelaps with $\VT'$ in top $k$ levels. It can be derived in  $n$ steps that $\ttree{\UU_1}\not\subt\ttree{\V_1'}.$ 
\end{proof}

 \paragraph*{Step2: characteristic process}
The proof that characteristic process of $\UU$ is typable by $\ttree{\UU}$ is exactly the same as in the case of synchronous multiparty sessions (See \cite{GhilezanJPSY19}).
We consider only single-output processes and for such processes there is no difference in typing rules. The whole proof is replicated here, adapted to single-output processes.

\begin{lemma}[Strengthening]
  If $ \Theta, X:\UT'\vdash \PP: \UT$ and $X\not\in \fv\{\PP\}$ then $\Theta \vdash \PP: \UT.$
\end{lemma}

\begin{lemma}[Weakening]
  If $\Theta \vdash \PP:\UT$ and $X\not\in \dom\Theta$ then $\Theta ,X:\UT' \vdash \PP: \UT.$
\end{lemma}

\begin{lemma}
\label{lem:tvs}
  If $\Theta, X_{\ty}:\UT_1 \vdash \CP{\UU}:\ttree{\UU\sigma}$ where $\UT_1 = \ttree{\UU_1}$ (for some $\UU_1$), 
  and  $ \sigma=\{\nicefrac{\Theta(X_{\ty'})}{\ty'} \;|\; \ty'\in \fv(\UU)\},$  then $\Theta, X_{\ty}:\UT_2 \vdash \CP{\UU}: \ttree{\UU\sigma'},$ 
  for any $\UT_2 = \ttree{\UU_2}$ (for some $\UU_2$) and  $\sigma'= (\sigma \setminus \subst{\UU_1}{\ty}) \cup \subst{\UU_2}{\ty}.$
\end{lemma}

\begin{proof} By induction on the structure of $\UU$.
\end{proof}

\begin{lemma}
\label{lem:cp}
    For every (possibly open) type $\UU$,  there are $\Theta$ and $\sigma$ such that  $\mathrm{dom}(\Theta)= \{X_{\ty} \,\;|\;\, \ty\in \fv(\UU)\}$ 
    and  $\Theta \vdash \CP{\UU} : \ttree{\UU\sigma},$ where $\sigma$ is a substitution  such that 
    $\sigma=\{\nicefrac{\UU_{\!\ty}}{\ty} \;\;|\;\; \ty\in \fv(\UU) \text{ and } \Theta(X_{\ty}) = \ttree{\UU_{\!\ty}} \}.$
\end{lemma}

\begin{proof} By induction on the structure of $\UU$.
\begin{itemize}
\item $\UU \equiv \tend:$  $\CP\tend=\inact$ and, by \rulename{t-$\inact$}, $\vdash \CP\tend: \tend.$
 \item $\UU \equiv \ty:$ $\CP{\ty}=X_{\ty}$\\
 By \rulename{t-var}, $X_\ty:\ttree{\UU'} \vdash X_\ty: \ttree{\UU'}$ for any $\UU'.$ For $\sigma=\subst{\UU'}{\ty}$, we have
   \begin{equation*}
        X_\ty:\ttree{\UU'} \vdash \CP{\ty}: \ttree{\UU\sigma}.
   \end{equation*}
\item $\UU \equiv \texternal\limits_{i\in I} \pp?\ell_i(\S_i).\UU_i:$ $\CP{\UU}=\sum\limits_{i\in I} \procin  \pp {\ell_i(\x_i)}\cond{\exprt{\x_i}{\S_i}}{\CP{\UU_i}}{\CP{\UU_i}}$\\
    By the induction hypothesis, %
  $\Theta_i \vdash \CP{\UU_i}: \ttree{\UU_i \sigma_i},$ where $\sigma_i=\{\nicefrac{\UU_{\ty}}{\ty} \;|\;\ty\in \fv(\UU_i)$ and $\ttree{\UU_{\ty}} = \Theta_i(X_{\ty})\}$ for some $\Theta_i$, and every $i\in I$. 
  Let us denote by $\Theta$ the environment consisting of assignments $X_{\ty}:\ttree{\UU_{\ty}}$ for arbitrarily chosen $\UU_{\ty}$, where $X_{\ty}\in \dom {\Theta_i}$ for some $i\in I$.   By typing rules, $\Theta, \x_i:\S_i \vdash \exprt{\x_i}{\S_i}:{\tbool},$ for every $i\in I$. By Lemma\nobreakspace \ref {lem:tvs}, \rulename{t-cond} and weakening, for each $i \in I$, we have the judgements:
  \[
   \Theta,\x_i:\S_i \vdash  \cond{\exprt{\x_i}{\S_i}}{\CP{\UU_i}}{\CP{\UU_i}}: {\ttree{\UU_i\sigma_i'}}
 \]
 where
 $ \sigma_i'=
  \left\{\nicefrac{\UU_{\ty}}{\ty} \;\middle|\;
  \begin{array}{@{}l@{}}
    \ty\in \fv(\UU_i)  \text{ and\; } %
    \\%
   \ttree{\UU_{\ty}} \!=\! \Theta(X_\ty)
  \end{array}
  \right\}.%
  $
  Now, by \rulename{t-ext}, we have 
  \begin {equation*}
     \Theta\vdash \sum\limits_{i\in I} \procin  \pp {\ell(\x_i)}\cond{\exprt{\x_i}{\S_i}}{\CP{\UU_i}}{\CP{\UU_i}}: \texternal\limits_{i\in I} \pp?\ell(\S_i).\ttree{\UU_i\sigma_i'}.
   \end{equation*}
   We conclude this case by remarking that 
   $
     \texternal\limits_{i\in I} \pp?\ell(\S_i).\ttree{\UU_i\sigma_i'} = \ttree{ \UU \sigma} 
  $
    for 
    $\sigma=\cup_{i\in I} \sigma_i'=
     \{\nicefrac{\UU_{\ty}}{\ty} \;\;|\;\;
     \ty\in \fv(\UU) \text{ and }
     \ttree{\UU_\ty} = \Theta(X_{\ty})
     \}.
  $
 \\
\item $\UU\equiv  \pp!\ell(\S).\UU':$ $\CP{\UU}= \procout  \pp{\ell}{\valt{\S}}{\CP{\UU'}}$\\ 
  By the induction hypothesis,  $\Theta'\vdash \CP{\UU'}: \ttree{\UU' \sigma'},$ where $\sigma'=\{\nicefrac{\UU_{\ty}}{\ty} \;|\;\ty\in \fv(\UU')$  and $\ttree{\UU_{\ty}} = \Theta'(X_{\ty})\}$ for some $\Theta'$. Let us denote by $\Theta$ the environment consisting of assignments $X_{\ty}:\ttree{\UU_{\ty}}$ for arbitrarily chosen $\UU_{\ty}$, where $X_{\ty}\in \dom {\Theta'}$.

 By  \rulename{t-out},
   \[
       \Theta'\vdash {\pp!\ell(\valt{S}).\CP{\UU}}:{\pp!\ell(\S).\ttree{\UU'\sigma'}}.
    \]

\item $\UU \equiv \mu\ty.\UU':$ $\CP{\UU}=\mu X_{\ty}.\CP{\UU'}$
 \\
 By induction hypothesis, there is $\Theta'$ such that 
 \begin{equation*}
   \Theta' \vdash \CP{\UU'}: \ttree{\UU'\sigma'}
   \qquad\text{
     where 
     $\sigma'=\{\nicefrac{\UU_{\ty'}}{\ty'} \;|\;
     \ty'\in \fv(\UU')
     \text{ and }
     \ttree{\UU_{\ty'}} = \Theta'(X_{\ty'})
     \}$
   }
  \end{equation*}
 We have two cases:
 \begin{itemize}
   \item[(i)] $\ty \not\in \fv(\UU')$.\quad
     In this case, $X_{\ty}\not\in \fv(\CP{\UU'})$ and $\Theta'', X_{\ty}:\ttree{\UU''} \vdash \CP{\UU'} : \ttree{\UU'\sigma'}$, for some $\UU''$
     (either  $\Theta'=\Theta'', X_{\ty}:\ttree{\UU''}$ or it is obtained by weakening of $\Theta'$). By \rulename{t-rec}, we get $\Theta'' \vdash \mu X_{\ty}.\CP{\UU'}: \ttree{\UU'\sigma}$ with $\sigma'=\sigma.$  
   \item[(ii)] $\ty \in \fv(\UU')$.\quad%
     In this case, %
     $X_{\ty}\in \fv(\CP{\UU'})$ and $\Theta'=\Theta'', X_{\ty}:\ttree{\UU''}$ for some $\UU''.$ By Lemma\nobreakspace \ref {lem:tvs},
    \begin{equation*}
      \begin{array}{c}
        \Theta'',X_{\ty}:\ttree{\UU'\sigma'} \vdash \CP{\UU'}: \ttree{\UU'\sigma''}
        \\%
        \text{where}\quad%
        \sigma''\,=\,
        \left\{\nicefrac{\UU_{\ty'}}{\ty'}
        \;\;\middle|\;\;
        \ty'\in \fv(\UU') \setminus \{\ty\} \text{ and } \ttree{\UU_{\ty'}} = \Theta''(X_{\ty'})
        \right\}
        \cup \subst{\UU'\sigma'}{\ty}
      \end{array}
    \end{equation*}
    i.e., %
    the difference between $\sigma'$ and $\sigma''$  is that $\sigma''$ contains $\nicefrac{\UU'\sigma'}{\ty}$
    instead of $\nicefrac{\UU''}{\ty}$.   Then, by \rulename{t-rec}, we conclude 
    $
      \Theta'' \vdash \mu X_{\ty}.\CP{\UU'}: \ttree{\UU'\sigma}
    $ 
     where
     $\sigma\,=\,\{\nicefrac{\UU_{\ty}}{\ty} \;|\; \ty\in \fv(\UU)$
     and
     $\ttree{\UU_{\ty}} = \Gamma''(X_{\ty}) \} \,=\, \sigma'' \setminus \subst{\UU'\sigma'}{\ty}.$
 \end{itemize}
\end{itemize}
\end{proof}

\begin{lemma}\label{lem:sub:properties}
Let $\TT$ be a session type tree. Then
\begin{itemize}
  \item[(i)]  $\forall \UT \in \llbracket \TT\rrbracket_\SO$ $\UT\subt \TT$
  \item[(ii)] $\forall \VT' \in \llbracket \TT'\rrbracket_\SI$ $\TT'\subt \VT'$
\end{itemize}
\end{lemma}

\begin{proof}
  Follows  from the definition of decompositions.
\end{proof}

\propCP*

\begin{proof} 
  As a  direct consequence of Lemma\nobreakspace \ref {lem:cp}
 we get $\vdash \CP{\UU}:\UU.$ 
 Since by Lemma~\ref{lem:sub:properties} for every $\UU$ with $\ttree{\UU}\in \llbracket \ttree{\T} \rrbracket_\SO,$ $\ttree{\UU} \subt \ttree{\T},$ by \rulename{t-sub}, $\vdash \CP{\UU}:\T.$
\end{proof}
 \paragraph*{Step3: characteristic session}

 \propCharacteristicSession*
 
 \begin{proof}
   Follows directly from the construction of the characteristic session types and the definition of the live typing environments.
 \end{proof}

\propCharacteristicSessionTi*

\begin{proof}
By Proposition~\ref{prop:live}, $\vdash \Q_1:\V'$ implies 
there is a live typing environment $\Gamma''$ such that 
  \[
       \Gamma'' \vdash \pa\pr \Q_1 \pc \pa\pr\EmptyQueue \pc \M_{\pr,\V'}
  \]
where $\Gamma''= \Gamma', \pr:(\emptyqueue, \V')$. 
By Lemma~\ref{lem:sub:properties} we have $\T' \subt \V'$. 
Since for $\Gamma= \Gamma', \pr:(\emptyqueue, \T')$ 
we have $\Gamma \subt \Gamma''$, by 
Lemma~\ref{lem:subtyping-preserves-liveness} we obtain that $\Gamma$ is also live. 
Hence, if $\vdash \Q:\T',$ then we may show 
  \[
       \Gamma \vdash \pa\pr \Q \pc \pa\pr\EmptyQueue \pc \M_{\pr,\V'}
  \]
and $\Gamma$ is live.
\end{proof}

 \paragraph*{Step4: completeness}

 \begin{table}[th!]
$
{\small
\begin{array}{@{}l@{}}
     
     \inferrule[\rulename{n-UV-out-act}]
  {\pp! \not\in \actions{\VT' }}{\tout\pp\ell\S.\UT \not\subt \VT'}
  \qquad
       \inferrule[\rulename{n-UV-inp-act}]
  {\pp? \not\in \actions{\VT' }}{\texternal_{i \in I}\tin\pp{\ell_i}{\S_i}.{\UT_i} \not\subt \VT'}
\qquad
       \inferrule[\rulename{n-UV-out-act-R}]
  {\pp! \not\in \actions{\UT}}{\UT \not\subt \tinternal_{i \in I}\tout\pp{\ell_i}{\S_i}.{\VT'_i}}
  \qquad
       \inferrule[\rulename{n-UV-inp-act-R}]
  {\pp? \not\in \actions{\UT}}{\UT \not\subt \tin\pp\ell\S.\VT'}
\\\\
\inferrule[\rulename{n-UV-inp}]
     {\forall i\in I : \ell_i\neq \ell \;\vee\; \ST \not\subs \ST_i \;\vee\; \UT_i \not\subt \VT'}
     {\texternal_{i \in I}\tin\pp{\ell_i}{\S_i}.{\UT_i} \not\subt\tin\pp{\ell}{\ST}.\VT'}
 \qquad
 \inferrule[\rulename{n-UV-${\AC}$}]
     {\forall i\in I : \ell_i\neq \ell \;\vee\; \ST \not\subs \ST_i \;\vee\; \UT_i \not\subt \ACon\pp\VT'}
     {\texternal_{i \in I}\tin\pp{\ell_i}{\S_i}.{\UT_i} \not\subt \ACon\pp{{\tin\pp{\ell}{\ST}.\VT'}}}
\\\\
  \inferrule[\rulename{n-UV-in-out-1}]
{ }{\texternal_{i \in I}\tin\pp{\ell_i}{\S_i}.{\UT_i} \not\subt \tinternal_{j \in J}\tout\pq{\ell_j}{\S_j}.{\VT'_j}}
  \qquad
  \inferrule[\rulename{n-UV-in-out-2}]
 { }{\texternal_{i \in I}\tin\pp{\ell_i}{\S_i}.{\UT_i} \not\subt \ACon\pp\tinternal_{j \in J}\tout\pq{\ell_j}{\S_j}.{\VT'_j}}
  \\\\
\inferrule[\rulename{n-UV-out}]
   {\forall i \in I : \ell\neq \ell_i \;\vee\; \ST \not\subs \ST_i \;\vee\; \UT \not\subt \VT'_i}
   {\tout\pp{\ell}{\ST}.\UT \not\subt \tinternal_{i \in I}\tout\pp{\ell_i}{\S_i}.{\VT'_i}}
\qquad
   \inferrule[\rulename{n-UV-$\mathcal{C}$}]
   {\forall n\in N \, \forall i \in I_n : \ell\neq\ell_i \;\vee\; \ST\not\subs\ST_i \;\vee\; \UT\not\subt (\proj{\CContext\pp}{n})[\VT_i']}
   {\tout\pp{\ell}{\ST}.\UT \not\subt \CContext\pp[\tinternal_{i\in I_n}\tout\pp{\ell_i}{\S_i}.{\VT'_i}]^{n\in N}}
\end{array}}$
\\
\caption{\label{tab:negationUandV'}Shapes of unrelated $\UT$ and $\VT'$ type trees}
\end{table}
  
In order to describe the shape of $\VT'$ type when $\UT\not\subt\VT'$ 
is derived using cases that involve context $\BContext\pp$ 
(for the corresponding projections that satisfy $\WT\nsubttt\WT'$) 
we define context $\CContext\pp$ with holes, that (as $\VT$) have only single inputs, 
and, in which there are no outputs on $\pp$. 

\[
\CContext\pp :: = 
[ \; ]^{n} \sep
 \tin\pq{\ell}{\ST}.\CContext\pp \sep \tinternal_{i\in I}\tout\pr{\ell_i}{\ST_i}.\CContext\pp \internal
\tinternal_{i\in I'} \tout\pr{\ell_i}{\ST_i}.\VT'_i
\qquad
\pr\not=\pp \text{ and } \pp !\not\in\actions{\VT'_i}
\]
Note that, in case of selection, context $\CContext\pp$ may have holes only in some branches 
while the rest of the branches contain no outputs on $\pp$. 
This allows us to determine all the ``first'' outputs appearing in some $\VT'$ type tree.
We write $\CContext\pp[\;]^{n\in N}$ to denote a context 
with  holes  indexed by elements of $N$
and $\CContext\pp[\VT'_n]^{n\in N}$ to denote the same context when the hole 
$[ \; ]^n$ has been filled with $\VT'_n$.
We index the holes in contexts in order to distinguish them. 
For the rest of the paper we assume $\CContext\pp$ are nonempty, i.e., 
it always holds that $\CContext\pp\neq [ \; ]$.

Furthermore, for a context $\CContext\pp[\;]^{n\in N}$ we may apply 
a mapping $\proj{\CContext\pp}{n}$ that  produces the projection of the context 
into the path that leads to the hole indexed with $n$, 
i.e., it produces the corresponding $\BContext\pp.[\;]^n$.
 
\begin{lemma}\label{lemm:negationUandV}
If $\neg(\UT\subt\VT')$ then $\UT\not\subt\VT'$ can be derived 
by the inductive rules given in Table~\ref{tab:negationUandV'}.
\end{lemma}
\begin{proof}
If $\neg(\UT\subt\VT')$ then $\forall\WT\in\llbracket \UT \rrbracket_\SI$ 
and $\forall\WT'\in\llbracket \VT' \rrbracket_\SO$ holds $\WT\nsubttt\WT'$.
Now the proof continues by case analysis of the last applied rules for $\WT\nsubttt\WT'$.
	\begin{itemize}
	\item Case \rulename{n-out}: 
	If $\forall\WT\in\llbracket \UT \rrbracket_\SI$ 
	and $\forall\WT'\in\llbracket \VT' \rrbracket_\SO$ $\WT\nsubttt\WT'$ is derived by \rulename{n-out}, 
	then from $\WT=\tout\pp\ell\ST.\WT_1$ we conclude $\UT=\tout\pp\ell\ST.\UT_1$, 
	by the definition of $\llbracket \; \rrbracket_\SI$. 
	Since $\forall\WT'\in\llbracket \VT' \rrbracket_\SO$ holds $\pp!\not\in\actions{\WT'}$, 
	we may conclude $\pp!\not\in\actions{\VT'}$ (by definition of $\llbracket \; \rrbracket_\SO$). 
	Hence, we get that $\UT$ and $\VT'$ satisfy the clauses of \rulename{n-UV-out-act}.
	\item Cases \rulename{n-inp}, \rulename{n-out-R} and \rulename{n-inp-R}: 
	By a similar reasoning as in the previous case we may show that 
	$\UT$ and $\VT'$ satisfy the clauses of \rulename{n-UV-inp-act}, 
	\rulename{u-UV-out-act-R} and  \rulename{n-UV-inp-act-R}, respectively.
	\item Cases \rulename{n-inp-$\ell$}, \rulename{n-inp-$\ST$} and \rulename{n-inp-$\WT$}: 
	Assume $\exists\WT\in\llbracket \UT \rrbracket_\SI$ 
	and $\exists\WT'\in\llbracket \VT' \rrbracket_\SO$ such that $\WT\nsubttt\WT'$ is derived by 
	\rulename{n-inp-$\ell$}, \rulename{n-inp-$\ST$} or \rulename{n-inp-$\WT$}. 
	Then, from $\WT= \tin\pp\ell\ST.\WT_1$ and definition of $\llbracket \; \rrbracket_\SI$ 
	we may conclude $\UT=\texternal_{i\in I}\tin\pp{\ell_i}{\ST_i}.\UT_i$. 
	Also, from $\WT'= \tin\pp{\ell'}{\ST'}.\WT'_1$ and definition of 
	$\llbracket \; \rrbracket_\SO$, we conclude $\VT'=\tin\pp{\ell'}{\ST'}.\VT'_1$. 
	Since $\forall\WT\in\llbracket \UT \rrbracket_\SI$ 
	and $\forall\WT'\in\llbracket \VT' \rrbracket_\SO$ holds $\WT\nsubttt\WT'$ 
	we distinguish three subcases:
		\begin{itemize}
		\item $\forall\WT\in\llbracket \UT \rrbracket_\SI$ $\forall\WT'\in\llbracket \VT' \rrbracket_\SO$  
		$\WT\nsubttt\WT'$ is derived by \rulename{n-inp-$\ell$};
		\item $\exists\WT\in\llbracket \UT \rrbracket_\SI$ $\exists\WT'\in\llbracket \VT' \rrbracket_\SO$  
		$\WT\nsubttt\WT'$ is derived by \rulename{n-inp-$\ST$}, but 
		could not be derived by \rulename{n-inp-$\ell$};
		\item $\exists\WT\in\llbracket \UT \rrbracket_\SI$ $\exists\WT'\in\llbracket \VT' \rrbracket_\SO$  
		$\WT\nsubttt\WT'$ is derived by \rulename{n-inp-$\WT$}, but could not be derived by 
		\rulename{n-inp-$\ell$} and \rulename{n-inp-$\WT$}.
		\end{itemize}
	In the first subcase we conclude that $\forall i \in I : \ell_i\neq \ell'$; 
	in the second 
	$\exists i \in I : \ell_i= \ell'$ but $\ST'\not\subs\ST_i$; 
	while in the third case 
	$\exists i \in I : \ell_i= \ell'$ and $\ST'\subs\ST_i$ but 
	$\forall\WT_1 \in \llbracket \UT_i \rrbracket_\SI  \forall \WT'_1\in \llbracket \VT'_1 \rrbracket_\SO$ 
	holds $\WT_1\nsubttt \WT'_1$. 
	Thus, we derived that $\forall i\in I : \ell_i\neq \ell' \;\vee\; \ST'\not\subs\ST_i 
	\;\vee\; \UT_i\not\subt\VT'$, which are the clauses of \rulename{n-UV-inp}.
	\item Cases \rulename{n-${\AC}$-$\ell$}, \rulename{n-${\AC}$-$\ST$} 
	and \rulename{n-${\AC}$-$\WT$}: 
	Follow by a similar reasoning as in the previous item, 
	only deriving that $\UT$ and $\VT'$ satisfy rule \rulename{n-UV-${\AC}$}.
	\item Cases \rulename{n-i-o-1} and \rulename{n-i-o-2}: 
	Assume $\exists\WT\in\llbracket \UT \rrbracket_\SI$ 
	and $\exists\WT'\in\llbracket \VT' \rrbracket_\SO$ such that $\WT\nsubttt\WT'$ 
	is derived by \rulename{n-i-o-1} or \rulename{n-i-o-2}. 
	Using the definition of $\llbracket \; \rrbracket_\SO$ and $\llbracket \; \rrbracket_\SI$ 
	we conclude $\UT= \texternal_{i\in  I} \tin\pp{\ell_i}{\ST_i}.\UT_i$ and, 
	$\VT'=\tinternal_{j\in J} \tout\pq{\ell_j}{\ST_j}.\VT'_j$ or 
	$\VT'=\AContext\pp.\tinternal_{j\in J} \tout\pq{\ell_j}{\ST_j}.\VT'_j$,
	i.e., $\UT$ and $\VT'$ satisfy \rulename{n-UV-in-out-1} or \rulename{n-UV-in-out-2}, respectively.
	\item Cases \rulename{n-out-$\ell$}, \rulename{n-out-$\ST$} and \rulename{n-out-$\WT$}: 
	Follow by a similar reasoning as in the item with \rulename{n-inp-$\ell$}, 
	\rulename{n-inp-$\ST$} and \rulename{n-inp-$\WT$}, 
	only deriving that $\UT$ and $\VT'$ satisfy rule \rulename{n-UV-out}.
	\item Cases \rulename{n-${\BC}$-$\ell$}, \rulename{n-${\BC}$-$\ST$} 
	and \rulename{n-${\BC}$-$\WT$}: 
	Assume $\exists\WT\in\llbracket \UT \rrbracket_\SI$ 
	and $\exists\WT'\in\llbracket \VT' \rrbracket_\SO$ such that $\WT\nsubttt\WT'$ 
	is derived by \rulename{n-${\BC}$-$\ell$}, \rulename{n-${\BC}$-$\ST$} 
	or \rulename{n-${\BC}$-$\WT$}. 
	Then, $\WT=\tout\pp\ell\ST.\WT_1$ and 
	$\WT'=\BContext\pp.\tout\pp{\ell'}{\ST'}.\WT'_1$. 
	By definition of $\llbracket \; \rrbracket_\SI$ we directly obtain 
	$\UT=\tout\pp\ell\ST.\UT_1$. 
	By definition of $\llbracket \; \rrbracket_\SO$ we obtain 
	$\pp!\in\actions{\VT'}$ and $\VT'\neq \tinternal_{i \in I} \tout\pp{\ell_i}{\ST_i}.\VT'_i$, 
	i.e., $\forall\WT'\in \llbracket \VT' \rrbracket_\SO$ either 
	$\WT'=\BContext\pp.\tout{\pp}{\ell''}{\ST''}.\WT''_1$ or $\pp!\not\in\actions{\WT'}$. 
	Thus, we may conclude 
	$\VT'=\CContext\pp[\tinternal_{i \in I_n}\tout\pp{\ell_i}{\ST_i}.\VT'_i]^{n\in N}$. 
	Furthermore, since $\forall\WT\in \llbracket \UT \rrbracket_\SI$  
	(that have form $\WT=\tout\pp\ell\ST.\WT_1$) and 
	$\forall\WT' \in \llbracket \VT' \rrbracket_\SO$  (that have form 
	$\WT'=(\proj{\CContext\pp}{n}) [\tout{\pp}{\ell_i}{\ST_i}.\WT'_i]$ 
	or $\pp !\notin\actions{\WT'}$) 
	hold $\WT\nsubttt\WT'$, we may distinguish three cases:
		\begin{itemize}
		\item $\forall\WT=\tout\pp\ell\ST.\WT_1$ 
		$\forall\WT'=(\proj{\CContext\pp}{n}) [\tout{\pp}{\ell_i}{\ST_i}.\WT'_i]$  
		$\WT\nsubttt\WT'$ is derived by \rulename{n-${\BC}$-$\ell$};
		\item $\exists\WT=\tout\pp\ell\ST.\WT_1$ 
		$\exists \WT'=(\proj{\CContext\pp}{n}) [\tout{\pp}{\ell_i}{\ST_i}.\WT'_i]$  
		$\WT\nsubttt\WT'$ is derived by \rulename{n-${\BC}$-$\ST$}, but could not 
		be derived by \rulename{n-${\BC}$-$\ell$};
		\item $\exists\WT=\tout\pp\ell\ST.\WT_1$ 
		$\exists \WT'=(\proj{\CContext\pp}{n}) [\tout{\pp}{\ell_i}{\ST_i}.\WT'_i]$  
		$\WT\nsubttt\WT'$ is derived by \rulename{n-${\BC}$-$\WT$}, but could not 
		be derived by \rulename{n-${\BC}$-$\ell$} or \rulename{n-${\BC}$-$\ST$}.
		\end{itemize}
	In the first case we get that $\forall n\in N\forall i \in I_n \; \ell\neq \ell_i$; 
	in the second that $\exists n \in N \exists i \in I_n$ such that $\ell=\ell_i$, 
	and $\ST\not\subs\ST_i$; and in the third that $\exists n \in N \exists i \in I_n$ 
	such that $\ell=\ell_i$ and $\ST\subs\ST_i$, but 
	$\WT\nsubttt(\proj{\CContext\pp}{n}) [\tout{\pp}{\ell_i}{\ST_i}.\WT'_i]$. 
	Notice that in the third case $\forall\WT\in \llbracket \tout\pp\ell\ST.\UT_1 \rrbracket_\SI$ 
	and $\forall \WT'\in \llbracket(\proj{\CContext\pp}{n}) [\tout{\pp}{\ell}{\ST_i}.\VT'_i]\rrbracket_\SO$, 
	where $\ST\subs\ST_i$, we have $\WT\nsubttt\WT'$ is derived by \rulename{n-${\BC}$-$\WT$}. 
	Then, we may conclude that $\forall\WT_1\in \llbracket \UT_1 \rrbracket_\SI$ 
	and $\forall \WT'_1\in \llbracket(\proj{\CContext\pp}{n}) [\VT'_i]\rrbracket_\SO$, 
	we have $\WT_1\nsubttt\WT'_1$, i.e., we obtain $\UT_1 \not\subt (\proj{\CContext\pp}{n}) [\VT'_i]$. 
	Therefore, we concluded that $\UT=\tout\pp\ell\ST.\UT_1$, and
	$\VT'=\CContext\pp[\tinternal_{i \in I_n}\tout\pp{\ell_i}{\ST_i}.\VT'_i]^{n\in N}$, 
	and that $\forall n\in N\forall i \in I_n : \ell\neq\ell_i \;\vee\; \ST\not\subs\ST_i 
	\;\vee\; \UT_1 \not\subt (\proj{\CContext\pp}{n}) [\VT'_i]$, that are the clauses of \rulename{n-UV-$\mathcal{C}$}.
\end{itemize}
\end{proof}

\begin{lemma}
\label{lemma:error:val}
  If $\S \not \subs \S'$ there is no $\val$ such that $ \eval {\exprt{\valt{S}}{S'}}\val.$
\end{lemma}

\begin{proof}
   By case analysis, we consider expression ${\exprt{\valt{S}}{S'}}$:
   \begin{itemize}
      \item $\tint\not\subs \tnat:$ $\exprt{-1}{\tnat} = (\fsucc{-1}>0);$
       \item $\tbool\not\subs \tnat:$ $\exprt{\true}{\tnat} = (\fsucc{\true} >0);$
       \item $\tnat\not\subs \tbool:$ $\exprt{1}{\tbool} = \neg 1;$
       \item $\tint\not\subs \tbool:$ $\exprt{-1}{\tbool} = \neg(-1);$
       \item $\tbool\not\subs \tint:$ $\exprt{\true}{\tint} = (\finv{\true}>0).$
   \end{itemize}
   In each case, the expression is undefined since the following  expressions are undefined: $\fsucc{-1}, \fsucc{\true},  \neg 1,\neg(-1),  \finv{\true}.$
\end{proof}

\begin{lemma}
\label{lemma:true:val}
  If $\S  \subs \S'$ then $\eval {\exprt{\valt{S}}{S'}}\true$ or $\eval {\exprt{\valt{S}}{S'}}\false.$
\end{lemma}

\begin{proof}
    By case analysis:
    \begin{itemize}
       \item $\tnat \subs \tnat:$ $\exprt{1}{\tnat} = \eval{(\fsucc{1} >0)}{\true};$
       \item $\tint \subs \tint:$ $\exprt{-1}{\tint} = \eval{(\finv{-1}>0)}{\true};$
       \item $\tnat\subs\tint:$ $\exprt{1}{\tint} = \eval{(\finv{1}>0)}{\false}.$
    \end{itemize}
\end{proof}

\begin{lemma}\label{lemm:for-completeness-induction-on-A-context}
Let $\participant{\AContext\pp.\V_1'}\subseteq\{\pp_k:1\leq k\leq m\}$ and 
\[
\MM =  \pa\pr\PP \pc\pa\pr\EmptyQueue \pc 
              \prod\limits_{1\leq k \leq m} 
	             (\pa{\pp_k}\CP{\cyclic{\AContext\pp.\V_1'}{\pp_k}} \pc\pa{\pp_k}\h_k)
\]
Then, 
\[
\MM \red^* \pa\pr\PP \pc\pa\pr\EmptyQueue \pc 
             \prod\limits_{1\leq k \leq m} 
	             (\pa{\pp_k}\CP{\cyclic{\V_1'}{\pp_k}} \pc\pa{\pp_k}\h'_k)
\]
where for all $1\leq k \leq m$: 
$\h'_k\equiv \h_k\cdot(\pr,\ell_{1}(\valt{\S_{1}}))\cdot\ldots\cdot(\pr,\ell_{n}(\valt{\S_{n}}))$
when 
\[
\AContext\pp=\AContext{\pp_k}_1.\tin{\pp_k}{\ell_{1}}{\S_{1}}.\AContext{\pp_k}_2.\ldots\AContext{\pp_k}_n.\tin{\pp_k}{\ell_{n}}{\S_{n}}.\AContext{\pp_k}_{n+1}
\] 
where instead of $\AContext{\pp_k}_i$ contexts there could also be empty contexts, and if $\pp_k?\notin\actions{\AContext\pp}$ then $\h'_k\equiv\h_k$.
\end{lemma}

\begin{proof}
By induction on context $\AContext\pp$. 
We detail only the case of inductive step. 
Without loss of generality assume $\AContext\pp=\tin{\pp_1}{\ell}{\S}.\AContext\pp_1$, 
and let $\pp_1=\pp_{m+1}$. 
Then, 
\begin{align*}
\MM = & \pa\pr{\PP} \pc\pa\pr\EmptyQueue \pc \\
      & \pa{\pp_1} \CP{\tout{\pr}{\ell}{\S}.\tout{\pp_{2}}{\ell}{\tbool}.\tin{\pp_{m}}{\ell}{\tbool}.\cyclic{\AContext\pp_1.\V_1'}{\pp_1}} \pc\pa{\pp_1}\h_1 \pc \\
      & \prod\limits_{2\leq k \leq m} 
	             (\pa{\pp_k}\CP{\tin{\pp_{k-1}}{\ell}{\tbool}.\tout{\pp_{k+1}}{\ell}{\tbool}.\cyclic{\AContext\pp_1.\V_1'}{\pp_k}} \pc\pa{\pp_k}\h_k)\\
\red^*& \pa\pr{\PP} \pc\pa\pr\EmptyQueue \pc \\
      & \pa{\pp_1} \CP{\cyclic{\AContext\pp_1.\V_1'}{\pp_1}} \pc\pa{\pp_1}\h''_1  \pc \\
      & \prod\limits_{2\leq k \leq m} 
	             (\pa{\pp_k}\CP{\cyclic{\AContext\pp_1.\V_1'}{\pp_k}} \pc\pa{\pp_k}\h''_k)\\
	=: \MM_1
\end{align*}
where $\h''_1=\h_1\cdot (\pr,\ell(\valt{\S}))$ and $\h''_k=\h_k$, and where in $\MM$ first participant $\pp_1$ enqueues the message and then all participants (except $\pr$) participate in the cycle. Applying induction hypothesis on $\MM_1$ we obtain 
\[
\MM_1 \red^* \pa\pr{\PP} \pc\pa\pr\EmptyQueue \pc    
            \prod\limits_{1\leq k \leq m} 
	             (\pa{\pp_k}\CP{\cyclic{\V_1'}{\pp_k}} \pc\pa{\pp_k}\h'_k)
\]
where for all $1\leq k \leq m$: 
$\h'_k\equiv \h''_k\cdot(\pr,\ell_{1}(\valt{\S_{1}}))\cdot\ldots\cdot(\pr,\ell_{n}(\valt{\S_{n}}))$
when 
\[
\AContext\pp_1=\AContext{\pp_k}_1.\tin{\pp_k}{\ell_{1}}{\S_{1}}.\AContext{\pp_k}_2.\ldots\AContext{\pp_k}_n.\tin{\pp_k}{\ell_{n}}{\S_{n}}.\AContext{\pp_k}_{n+1}
\] 
where instead of $\AContext{\pp_k}_i$ contexts there could also be empty contexts, and if $\pp_k?\notin\actions{\AContext\pp_1}$ then $\h'_k\equiv\h''_k$. 
Since $\AContext\pp=\tin{\pp_1}{\ell}{\S}.\AContext\pp_1$ and 
$\h'_1\equiv \h_1\cdot (\pr,\ell(\valt{\S}))\cdot(\pr,\ell_{1}(\valt{\S_{1}}))\cdot\ldots\cdot(\pr,\ell_{n}(\valt{\S_{n}}))$, or $\h'_1\equiv\h_1\cdot (\pr,\ell(\valt{\S}))$ in case $\pp_1?\notin\actions{\AContext\pp_1}$, we may conclude the proof.

\end{proof}

\begin{proposition}
\label{thm:completenessUandV'}
  Let $\UU$ and $\V'$ be session types such that  $\ttree{\UU} \not \subt \ttree{\V'}$ and $\participant{\V'}\subseteq\{\pp_k:1\leq k\leq m\}$  and $\UU_{\pp_k}=\cyclic{\V'}{\pp_k}$.
  If
   \begin{eqnarray*}
        \MM & \equiv & \pa\pr \CP{\UU} \pc  \pa\pr \EmptyQueue \pc \prod\limits_{1\leq k \leq m} \left(\pa{\pp_k} \CP{\UU_{\pp_k}}\pc \pa{\pp_k} \EmptyQueue \right)\\
        \MM' & \equiv & \pa\pr \PP \pc  \pa\pr \Queue_\pp \pc \prod\limits_{1\leq k \leq m} \left( \pa{\pp_k} \PP_k \pc \pa{\pp_k} \Queue_{k} \right),
 \end{eqnarray*}
 where $\pa\pr \CP{\UU} \pc  \pa\pr \EmptyQueue \red^*\pa\pr \PP \pc  \pa\pr \Queue_\pp$ and  
 \[
 \prod\limits_{1\leq k \leq m} \left(\pa{\pp_k} \CP{\UU_{\pp_k}}\pc \pa{\pp_k} \EmptyQueue \right) 
 \red^* 
 \prod\limits_{1\leq k \leq m} \left( \pa{\pp_k} \PP_k \pc \pa{\pp_k} \Queue_{k} \right)
 \] 
 then,  $\MM\red^*\MM'$ and $\MM'\red^*\error.$
\end{proposition}

\begin{proof}
The proof is by induction on the derivation of 
$\ttree{\UU}\not\subt\ttree{\V'}$.    
We extensively use notation $\UT=\ttree{\UU}$ and $\VT'=\ttree{\V'}$. 
The cases for the last rule applied are derived from Table~\ref{tab:negationUandV'}.

We first consider the cases with $\actions{\UT}\neq  \actions{\VT'}$.
\\\\
\fbox{$\actions{\UT}\neq \actions{\VT'}$} 
\\\\
\underline{$\rulename{n-UV-out-act}:$} $\UT=\tout\pp\ell\S.\UT_1$ and $\pp!\not\in\actions{\VT'}$.
\\
In this case, if $\pp\in\participant{\V'}$ by definition of characteristic session type and characteristic process, $\pr?\not\in \CP{\UU_\pp}$.  Thus,
\begin{align*}
  \MM \equiv & \pa\pr\CP{\UU} \pc\pa\pr\EmptyQueue \pc \pa\pp\CP{\UU_\pp} \pc\pa\pp\EmptyQueue \pc\MM'_1 &\\
          \equiv & \pa\pr\CP{\tout\pp\ell\S.\UU_1} \pc\pa\pr\EmptyQueue \pc \pa\pp \CP{\UU_\pp} \pc\pa\pp\EmptyQueue \pc\MM'_1 & \\
            \red  &  \pa\pr\CP{\UU_1} \pc\pa\pr (\pp,\ell(\valt{\S})) \pc \pa\pp \CP{\UU_\pp} \pc\pa\pp\EmptyQueue \pc\MM'_1 
                       \red \error  & \text{(by \rulename{err-ophn})}
\end{align*}
If $\pp\notin\participant{\V'}$ we use $\MM\equiv \MM \pc \pa\pp\inact \pc\pa\pp\EmptyQueue$ and derive the analogous proof as above.
\underline{$\rulename{n-UV-inp-act}:$} $\UT=\texternal_{i\in I}\tin\pp{\ell_i}{\S_i}.\UT_i$ and $\pp?\not\in\actions{\VT'}$.
\\
In this case, if $\pp\in\participant{\V'}$ 
by definition of characteristic session type and characteristic process, $\pr!\not\in \CP{\UU_\pp}$. Thus,
\begin{align*}
  \MM \equiv & \pa\pr\CP{\UU} \pc\pa\pr\EmptyQueue \pc \pa\pp\CP{\UU_\pp} \pc\pa\pp\EmptyQueue \pc\MM'_1 &\\
          \equiv & \pa\pr\CP{\texternal_{i\in I}\tin\pp{\ell_i}{\S_i}.\UU_i} \pc\pa\pr\EmptyQueue \pc \pa\pp \CP{\UU_\pp} \pc\pa\pp\EmptyQueue \pc\MM'_1 & \\
          \equiv & \pa\pr \sum_{i\in I}\procin{\pp}{\ell_i(\x_i)}{\PP_i} \pc\pa\pr\EmptyQueue \pc \pa\pp \CP{\UU_\pp} \pc\pa\pp\EmptyQueue \pc\MM'_1 
                       \red \error  &  \text{(by \rulename{err-strv})}
\end{align*}
If $\pp\notin\participant{\V'}$ we use $\MM\equiv \MM \pc \pa\pp\inact \pc\pa\pp\EmptyQueue$ and derive the analogous proof as above.
\underline{$\rulename{n-UV-out-act-R}:$} $\pp!\not\in\actions{\UT}$ and 
$\VT'=\tinternal_{i\in I}\tout\pp{\ell_i}{\S_i}.\VT'_i$.
\\
In this case, by definition of characteristic process, $\pp!\not\in \CP{\UU}$.
\begin{align*}
  \MM \equiv & \pa\pr\CP{\UU} \pc\pa\pr\EmptyQueue \pc \pa\pp\CP{\cyclic{\tinternal_{i\in I}\tout\pp{\ell_i}{\S_i}.\V'_i}{\pp}} \pc\pa\pp\EmptyQueue \pc\MM'_1 &\\
          \equiv & \pa\pr\CP{\UU} \pc\pa\pr\EmptyQueue \pc \pa\pp \sum_{i\in I}\procin{\pp}{\ell_i(\x_i)}{\PP_i} \pc\pa\pp\EmptyQueue \pc\MM'_1 \red \error  &  \text{(by \rulename{err-strv})}
\end{align*}
\underline{$\rulename{n-UV-inp-act-R}:$} $\pp?\not\in\actions{\UT}$ and $\VT'=\tin\pp\ell\S.\VT'_1$.
\\
In this case, by definition of characteristic process, $\pp?\not\in \CP{\UU}$.
\begin{align*}
  \MM \equiv & \pa\pr\CP{\UU} \pc\pa\pr\EmptyQueue \pc \pa\pp\CP{\cyclic{\tin\pp\ell\S.\V'_1}{\pp}} \pc\pa\pp\EmptyQueue \pc\MM'_1 &\\
          \equiv & \pa\pr \CP{\UU} \pc\pa\pr\EmptyQueue \pc \pa\pp \procout\pr\ell{\valt{\S}}\PP \pc\pa\pp\EmptyQueue \pc\MM'_1 &\\ 
           \red   & \pa\pr\CP{\UU} \pc\pa\pr\EmptyQueue \pc \pa\pp \PP \pc\pa\pp (\pr,\ell(\valt{\S})) \pc\MM'_1  \red \error  &  \text{(by \rulename{err-ophn})}
\end{align*}
In the following cases, $\UT$ type tree is rooted with an external choice.
\\\\
\fbox{$\UT= \texternal_{i\in I}\tin {\pp}{\ell_i}{\S_i}.\UT_i$}   
\\\\
In these cases, we have 
\begin{align*}
   \CP{\UU}  \equiv  &  \sum_{i\in I}\procin{\pp}{\ell_i(\x_i)}{\PP_i}, \text{ where }\\
       \PP_i  \equiv  & \cond{\exprt{\x_i}{\S_i}}{\CP{\UU_i}}{\CP{\UU_i}}
\end{align*}
According to  Table~\ref{tab:negationUandV'}, we distinguish four cases (not already considered), depending on the form of $\VT'.$
\\
\underline{\rulename{n-UV-inp}:} $\VT'=\tin\pp{\ell}{\S}.\VT_1'$ and 
$\forall i\in I : \ell_i\neq \ell \;\vee\; \S\not\subs\S_i \;\vee\; \UT_i \not\subt \VT'_1$.
\\
Now we have
\begin{align*}
  \MM \equiv   & \pa\pr\sum_{i\in I}\procin{\pp}{\ell_i(\x_i)}{\PP_i} \pc\pa\pr\EmptyQueue \pc \pa\pp \CP{\cyclic{\tin\pp{\ell}{\S}.\V_1'}{\pp}} \pc\pa\pp\EmptyQueue \pc\MM'_1 &\\
  =  & \pa\pr\sum_{i\in I}\procin{\pp}{\ell_i(\x_i)}{\PP_i} \pc\pa\pr\EmptyQueue \pc \pa\pp \procout\pr{\ell}{\valt{\S}}\PP' \pc\pa\pp\EmptyQueue \pc\MM'_1 &\\
\red & \pa\pr\sum_{i\in I}\procin{\pp}{\ell_i(\x_i)}{\PP_i} \pc\pa\pr\EmptyQueue \pc \pa\pp\PP' \pc\pa\pp (\pr,\ell(\valt{\S})) \pc\MM'_1 \\
=: & \MM_1
\end{align*}
We now distinguish three cases.
	\begin{itemize}
	\item $\forall i\in I : \ell_i\neq \ell$: 
	 Session $\MM_1$ reduces to $\error$ by \rulename{err-mism}.
	\item $\exists i\in I : \ell_i=\ell \;\wedge\; \S\not\subs\S_i$:
	\[
	\MM_1 \red \pa\pr\cond{\exprt{\valt{\S}}{\S_i}}{\CP{\UU_i}}{\CP{\UU_i}} \pc   
	\pa\pr\EmptyQueue \pc  \pa\pp\PP' \pc \pa\pp\EmptyQueue \pc\MM'_1
	\]
	By Lemma~\ref{lemma:error:val} and \rulename{err-eval},  the session reduces to $\error$.
	\item $\exists i\in I : \ell_i=\ell \;\wedge\; \S\subs\S_i \;\wedge\; \UT_i\not\subt \VT'_1$:
	\[
	\MM_1 \red \pa\pr\cond{\exprt{\valt{\S}}{\S_i}}{\CP{\UU_i}}{\CP{\UU_i}} \pc   
	\pa\pr\EmptyQueue \pc  \pa\pp\PP' \pc \pa\pp\EmptyQueue \pc\MM'_1
	\]
	By Lemma~\ref{lemma:true:val}, we further derive 
	\[
	\MM_1 \red^* \pa\pr\CP{\UU_i} \pc   
	\pa\pr\EmptyQueue \pc  \pa\pp\PP' \pc \pa\pp\EmptyQueue \pc\MM'_1
	\]
	where (assuming $\pp=\pp_1=\pp_{m+1}$)
	\begin{align*}
	\PP' \equiv     & \tout{\pp_{2}}{\ell}{\true}.\tin{\pp_{m}}{\ell}{\x}.\\
	 & \cond{\exprt{\valt{\x}}{\tbool}} {\cyclic{\V'_1}{\pp_1}}{\cyclic{\V'_1}{\pp_1}}\\
	\MM'_1 \equiv    & \prod\limits_{2\leq k \leq m} 
	             (\pa{\pp_k} \tin{\pp_{k-1}}{\ell}{\x}. \cond{\exprt{\valt{\x}}{\tbool}}
	             {\PQ_k}
	             {\PQ_k}
	             \pc\pa{\pp_k}\EmptyQueue )
	\end{align*}
	where $\PQ_k=\tout{\pp_{k+1}}{\ell}{\true}.\CP{\cyclic{\V'_1}{\pp_k}}$.
	Hence, we have 
	\begin{align*}
    \MM_1 \red^* & \pa\pr\CP{\UU_i} \pc   
	\pa\pr\EmptyQueue \pc  \pa\pp\cyclic{\V'_1}{\pp} \pc \pa\pp\EmptyQueue \pc\\
	&          \prod\limits_{2\leq k \leq m} 
	           \left(\pa{\pp_k}\CP{\cyclic{\V'_1}{\pp_k}}\pc\pa{\pp_k}\EmptyQueue \right)
	\end{align*}
	Since $\UT_i\not\subt \VT'_1$, by induction hypothesis  the session reduces to $\error$.
	\end{itemize}
\underline{\rulename{n-UV-${\AC}$}:}  
$\VT'=\AContext\pp.\tin\pp{\ell}{\S}.\VT_1'$ and 
$\forall i\in I : \ell_i\neq \ell \;\vee\; \S\not\subs\S_i \;\vee\; \UT_i \not\subt \AContext\pp.\VT'_1$.
\\
Assuming $\pp_1=\pp_{m+1}=\pp$, we have 
\begin{align*}
\MM \equiv & \pa\pr\sum_{i\in I}\procin{\pp}{\ell_i(\x_i)}{\PP_i} \pc\pa\pr\EmptyQueue \pc 
             \pa\pp \CP{\cyclic{\AContext\pp.\tin\pp{\ell}{\S}.\V_1'}{\pp}} \pc\pa\pp\EmptyQueue \pc\\
           & \prod\limits_{2\leq k \leq m} 
	             (\pa{\pp_k}\CP{\cyclic{\AContext\pp.\tin\pp{\ell}{\S}.\V_1'}{\pp_k}} \pc\pa{\pp_k}\EmptyQueue)
\end{align*}
By Lemma~\ref{lemm:for-completeness-induction-on-A-context}, we have  
\begin{align*}
\MM \red^*& \pa\pr\sum_{i\in I}\procin{\pp}{\ell_i(\x_i)}{\PP_i} \pc\pa\pr\EmptyQueue \pc 
             \pa\pp \CP{\cyclic{\tin\pp{\ell}{\S}.\V_1'}{\pp}} \pc\pa\pp\EmptyQueue \pc\\
           & \prod\limits_{2\leq k \leq m} 
	             (\pa{\pp_k}\CP{\cyclic{\tin\pp{\ell}{\S}.\V_1'}{\pp_k}} \pc\pa{\pp_k}\h_k)
\end{align*}
where for all $2\leq k \leq m$: 
$\h_k=(\pr,\ell_{1}(\valt{\S_{1}}))\cdot\ldots\cdot(\pr,\ell_{n}(\valt{\S_{n}}))$
when 
\[
\AContext\pp=\AContext{\pp_k}_1.\tin{\pp_k}{\ell_{1}}{S_{1}}.\AContext{\pp_k}_2.\ldots\AContext{\pp_k}_n.\tin{\pp_k}{\ell_{n}}{S_{n}}.\AContext{\pp_k}_{n+1}
\] 
where instead of $\AContext{\pp_k}_i$ contexts there could also be empty contexts, and if $\pp_k?\notin\actions{\AContext\pp}$ then $\h_k=\EmptyQueue$.
Using the last observation, we have 
\begin{align*}
\MM \red^*& \pa\pr\sum_{i\in I}\procin{\pp}{\ell_i(\x_i)}{\PP_i} \pc\pa\pr\EmptyQueue \pc 
             \pa\pp \procout\pr{\ell}{\valt{\S}}\PP' \pc\pa\pp\EmptyQueue \pc\\
           & \prod\limits_{2\leq k \leq m} 
	             (\pa{\pp_k}\CP{\cyclic{\tin\pp{\ell}{\S}.\V_1'}{\pp_k}} \pc\pa{\pp_k}\h_k)\\
	\red    & \pa\pr\sum_{i\in I}\procin{\pp}{\ell_i(\x_i)}{\PP_i} \pc\pa\pr\EmptyQueue \pc 
             \pa\pp \PP' \pc\pa\pp(\pr,\ell(\valt{\S})) \pc\\
           & \prod\limits_{2\leq k \leq m} 
	             (\pa{\pp_k}\CP{\cyclic{\tin\pp{\ell}{\S}.\V_1'}{\pp_k}} \pc\pa{\pp_k}\h_k)
\end{align*}
We now distinguish three cases
	\begin{itemize}
	\item $\forall i\in I : \ell_i\neq \ell$;
	\item $\exists i\in I : \ell_i=\ell \;\wedge\; \S\not\subs\S_i$;
	\item $\exists i\in I : \ell_i=\ell \;\wedge\; \S\subs\S_i \;\wedge\; \UT_i\not\subt \AContext\pp.\VT'_1$.
	\end{itemize}
In the first two cases $\MM$ reduces to $\error$ by 
the same arguments that are presented for the case of \rulename{n-UV-inp}. 
For the third case, again using the same arguments as in the case of \rulename{n-UV-inp}, 
we have that 
\begin{align*}
\MM \red^* & \pa\pr\CP{\UU_i} \pc   
	\pa\pr\EmptyQueue \pc  \pa\pp\cyclic{\V'_1}{\pp} \pc \pa\pp\EmptyQueue \pc\\
	&          \prod\limits_{2\leq k \leq m} 
	           (\pa{\pp_k}\CP{\cyclic{\V'_1}{\pp_k}}\pc\pa{\pp_k}\h_k )\\
	&   =: \MM'
\end{align*}
Now let 
\begin{align*}
\MM'' = & \pa\pr\CP{\UU_i} \pc   
	\pa\pr\EmptyQueue \pc  \pa\pp\cyclic{\AContext\pp.\V'_1}{\pp} \pc \pa\pp\EmptyQueue \pc\\
	   &          \prod\limits_{2\leq k \leq m} 
	           (\pa{\pp_k}\CP{\cyclic{\AContext\pp.\V'_1}{\pp_k}}\pc\pa{\pp_k}\EmptyQueue )
\end{align*}
By Lemma~\ref{lemm:for-completeness-induction-on-A-context} we obtain $\MM'' \red^* \MM'$.
Since $\UT_i\not\subt \AContext\pp.\VT'_1$ 
by induction hypothesis $\MM''\red^* \MM'\red^* \error$, and hence, $\MM\red^*\MM'\red^*\error$.
\\\\
\underline{\rulename{n-UV-in-out-1}:} $\VT'=\tinternal_{j\in J}\tout\pq{\ell_j}{\S_j}.\VT'_j$. 
\\
Assuming 
$\pp_1=\pp_{m+1}=\pq$
we have
\begin{align*}
\MM \equiv  & \pa\pr\CP{\UU} \pc\pa\pr\EmptyQueue \pc \pa\pq\CP{\cyclic{\V'}{\pq}} \pc\pa\pq\EmptyQueue \pc\MM'_1 &\\
    =  & \pa\pr\sum_{i\in I}\procin{\pp}{\ell_i(\x_i)}{\PP_i}  \pc\pa\pr\EmptyQueue \pc  
         \pa\pq \sum_{j\in J}\procin\pr{\ell(\x_j)}{\PP_j^\pq}\pc\pa\pq\EmptyQueue \pc\MM'_1                
\end{align*}  
where 
\[
\MM'_1 \equiv \prod\limits_{2\leq k \leq m} 
	             (\pa{\pp_k} \sum_{j\in J}\tin{\pp_{k-1}}{\ell_j}{\x_j}.\PP^{\pp_k}_j\pc\pa{\pp_k}\EmptyQueue)
\]
The session reduces to $\error$ by \rulename{err-dlock}.
\\\\
\underline{\rulename{n-UV-in-out-2}:} 
$\VT'=\AContext\pp.\tinternal_{j\in J}\tout\pq{\ell_j}{\S_j}.\VT'_j$.
\\
Let us first assume $\pq\neq\pp$.
Denoting 
$\pp_1=\pp_{m+1}=\pq$ and $\pp_2=\pp$ 
we have 
\begin{align*}
\MM \equiv  & \pa\pr\sum_{i\in I}\procin{\pp}{\ell_i(\x_i)}{\PP_i} \pc\pa\pr\EmptyQueue \pc\\
      &  \prod\limits_{1\leq k \leq m} 
	             (\pa{\pp_k}\CP{\cyclic{\AContext\pp.\tinternal_{j\in J}\tout\pq{\ell_j}{\S_j}.\V'_j}{\pp_k}} \pc\pa{\pp_k}\EmptyQueue)           
\end{align*}  
By Lemma~\ref{lemm:for-completeness-induction-on-A-context} we obtain
\begin{align*}
\MM \red^*& \pa\pr\sum_{i\in I}\procin{\pp}{\ell_i(\x_i)}{\PP_i} \pc\pa\pr\EmptyQueue \pc \\
           & \prod\limits_{1\leq k \leq m} 
	             (\pa{\pp_k}\CP{\cyclic{\tinternal_{j\in J}\tout\pq{\ell_j}{\S_j}.\V'_j}{\pp_k}} \pc\pa{\pp_k}\h_k)\\
	       & =:\MM'
\end{align*}
where for all $1\leq k \leq m$: 
$\h_k=(\pr,\ell_{1}(\valt{\S_{1}}))\cdot\ldots\cdot(\pr,\ell_{n}(\valt{\S_{n}}))$
when 
\[
\AContext\pp=\AContext{\pp_k}_1.\tin{\pp_k}{\ell_{1}}{S_{1}}.\AContext{\pp_k}_2.\ldots\AContext{\pp_k}_n.\tin{\pp_k}{\ell_{n}}{S_{n}}.\AContext{\pp_k}_{n+1}
\] 
where instead of $\AContext{\pp_k}_i$ contexts there could also be empty contexts, and if $\pp_k?\notin\actions{\AContext\pp}$ then $\h_k=\EmptyQueue$. 
Since $\pp?\notin\actions{\AContext\pp}$, for the last derived session we have 
\begin{align*}
\MM' =  & \pa\pr\sum_{i\in I}\procin{\pp}{\ell_i(\x_i)}{\PP_i}  \pc\pa\pr\EmptyQueue \pc  
         \pa\pq \sum_{j\in J}\procin\pr{\ell(\x_j)}{\PP_j^\pq}\pc\pa\pq\h_1 \pc \\
       &  \pa\pp \sum_{j\in J}\tin{\pq}{\ell_j}{\x_j}.\PP^{\pp_k}_j \pc\pa\pp\EmptyQueue
       \pc \prod\limits_{3\leq k \leq m} 
	     (\pa{\pp_k} \sum_{j\in J}\tin{\pp_{k-1}}{\ell_j}{\x_j}.\PP^{\pp_k}_j\pc\pa{\pp_k}\h_k)
\end{align*}
the session reduces to $\error$ by \rulename{err-dlock}. 
In case $\pq=\pp$, the proof follows similar lines.
\\\\
In the following cases, $\UT$ type tree is rooted with an internal choice.
\\\\
\fbox{$\UT= \tout {\pp}{\ell}{\S}.\UT_1$}   
\\\\
In these case, we have $\CP{\UU}  \equiv   \procout{\pp}{\ell}{\valt{S}}\CP{\UU_1}.$
\\
Depending on the form of $\VT'$, we distinguish two more cases (according to rules in Table~\ref{tab:negationUandV'}).
\\\\
\underline{\rulename{n-UV-out}:} $\VT'=\tinternal_{i\in I}\tout\pp{\ell_i}{\S_i}.\VT'_i$ and 
$\forall i\in I : \ell\neq \ell_i \;\vee\; \S\not\subs\S_i \;\vee\; \UT_1 \not\subt \VT'_i$.
\\
Now we have
\begin{align*}
  \MM \equiv   & \pa\pr\procout{\pp}{\ell}{\valt{S}}\CP{\UU_1} \pc\pa\pr\EmptyQueue \pc 
  \pa\pp \CP{\cyclic{\tinternal_{i\in I}\tout\pp{\ell_i}{\S_i}.\V'_i}{\pp}} \pc\pa\pp\EmptyQueue \pc\MM'_1 &\\
  =  & \pa\pr\procout{\pp}{\ell}{\valt{S}}\CP{\UU_1} \pc\pa\pr\EmptyQueue \pc \pa\pp \sum_{i\in I}\procin{\pr}{\ell_i(\x_i)}{\PP_i} \pc\pa\pp\EmptyQueue \pc\MM'_1 &\\
\red & \pa\pr\CP{\UU_1} \pc\pa\pr(\pp,\ell(\valt{\S})) \pc \pa\pp\sum_{i\in I}\procin{\pr}{\ell_i(\x_i)}{\PP_i} \pc\pa\pp \EmptyQueue \pc\MM'_1 
\end{align*}
Now the proof proceeds following the same lines as in the case of \rulename{n-UV-inp}.
\\\\
\underline{\rulename{n-UV-$\mathcal{C}$}}:  
$\VT'=\CContext\pp[\tinternal_{i\in I_n}\tout\pp{\ell_i}{\S_i}.\VT'_i]^{n\in N}$ and 
$\forall n\in N \, \forall i \in I_n : 
\ell\neq\ell_i \;\vee\; \ST\not\subs\ST_i 
\;\vee\; \UT_1\not\subt (\proj{\CContext\pp}{n})[\VT_i']$.
\\
Let us denote
$\pp_1=\pp_{m+1}=\pp$.
We have 
\begin{align*}
  \MM \equiv   & \pa\pr\procout{\pp}{\ell}{\valt{S}}\CP{\UU_1} \pc\pa\pr\EmptyQueue \pc 
  \pa\pp \CP{\cyclic{\CContext\pp[\tinternal_{i\in I_n}\tout\pp{\ell_i}{\S_i}.\V'_i]^{n\in N}}{\pp}} \pc\pa\pp\EmptyQueue  \\
  & \pc \prod\limits_{2\leq k \leq m} 
	             (\pa{\pp_k}\CP{\cyclic{\CContext\pp[\tinternal_{i\in I_n}\tout\pp{\ell_i}{\S_i}.\V'_i]^{n\in N}}{\pp_k}} \pc\pa{\pp_k}\EmptyQueue) 
\end{align*}

By induction on $\CContext\pp [\;]^{n\in N}$ we may show that either $\MM \red^* \error$ or there is $n\in N$ such that
\begin{align*}
  \MM \red^*   & \pa\pr\CP{\UU_1'} \pc\pa\pr(\pp,\ell(\valt{\S}))\cdot\h_r \pc 
  \pa\pp \CP{\cyclic{\tinternal_{i\in I_n}\tout\pp{\ell_i}{\S_i}.\V'_i}{\pp}} \pc\pa\pp\h_\pp  \\
       & \pc \prod\limits_{2\leq k \leq m} 
	     (\pa{\pp_k}\CP{\cyclic{\tinternal_{i\in I_n}\tout\pp{\ell_i}{\S_i}.\V'_i}{\pp_k}} \pc\pa{\pp_k}\h_k)\\
    =: & \MM_1 
\end{align*}
where there exist output-only context $\BContext{\pr}$  
and $\AContext\pr$ context (where instead of $\pr$ we could use any other fresh name) such that 
\begin{align*}
   \MM_2 = & \pa\pr\CP{\tout\pp{\ell}{\S}.\BContext{\pr}.\UU_1'}  \pc\pa\pr\EmptyQueue \pc 
  \pa\pp \CP{\cyclic{\AContext{\pr}.\tinternal_{i\in I_n}\tout\pp{\ell_i}{\S_i}.\V'_i}{\pp}} \pc\pa\pp\EmptyQueue  \\
  & \pc \prod\limits_{2\leq k \leq m} 
	     (\pa{\pp_k}\CP{\cyclic{\AContext{\pr}.\tinternal_{i\in I_n}\tout\pp{\ell_i}{\S_i}.\V'_i}{\pp_k}} \pc\pa{\pp_k}\EmptyQueue)\\
	     \red^* & \MM_1  
\end{align*}
and where %
$\BContext{\pr}.\UT_1' \not\subt \AContext{\pr}.\VT'_i$ 
can be derived from $\UT_1\not\subt (\proj{\CContext\pp}{n})[\VT_i']$ by applying the rules given in Table~\ref{tab:negationUandV'} from conclusion to premises. 

Note that the above $\AContext\pr$ is actually derived from $(\proj{\CContext\pp}{n})$, 
by taking out all the outputs (that are transformed into the inputs in the 
characteristic session), since they have found the appropriate outputs in $\UT_1$, 
and also some inputs (that are transformed into outputs in the characteristic session), that have found the appropriate inputs in $\UT_1$: these pairs of actions enabled session $\MM$ to reduce to $\MM_1$. 
Along these lines $\BContext{\pr}.\UT_1'$ is derived from $\UT_1$.

The above also implies that by the assumption 
$\BContext{\pr}.\UT_1' \subt \AContext{\pr}.\VT'_i$ 
we could derive $\UT_1\subt (\proj{\CContext\pp}{n})[\VT_i']$.

Since
\[
\CP{\cyclic{\tinternal_{i\in I_n}\tout\pp{\ell_i}{\S_i}.\V'_i}{\pp}} = 
\sum_{i\in I}\procin{\pr}{\ell_i(\x_i)}{\PP_i}
\]
we distinguish three cases
	\begin{itemize}
	\item $\forall i\in I : \ell\neq \ell_i$;
	\item $\exists i\in I : \ell=\ell_i \;\wedge\; \S\not\subs\S_i$;
	\item $\exists i\in I : \ell=\ell_i \;\wedge\; \S\subs\S_i \;\wedge\; \UT_1\not\subt (\proj{\CContext\pp}{n})[\VT_i']$.
	\end{itemize}
In the first two cases $\MM_1$ reduces to $\error$ using similar arguments as for \rulename{n-UV-inp}. 
For the third case we have that 
\begin{align*}
  \MM_1 \red^*   & \pa\pr\CP{\UU_1'} \pc\pa\pr\h_r \pc 
  \pa\pp \CP{\cyclic{\V'_i}{\pp}} \pc\pa\pp\h_\pp  \\
       & \pc \prod\limits_{2\leq k \leq m} 
	     (\pa{\pp_k}\CP{\cyclic{\V'_i}{\pp_k}} \pc\pa{\pp_k}\h_k)\\
	     =: \MM'
\end{align*}
Similarly as in the case of \rulename{n-UV-${\AC}$}, we have 
\begin{align*}
   \MM'' = & \pa\pr\CP{\BContext{\pr}.\UU_1'}  \pc\pa\pr\EmptyQueue \pc 
  \pa\pp \CP{\cyclic{\AContext{\pr}.\V'_i}{\pp}} \pc\pa\pp\EmptyQueue  \\
  & \pc \prod\limits_{2\leq k \leq m} 
	     (\pa{\pp_k}\CP{\cyclic{\AContext{\pr}.\V'_i}{\pp_k}} \pc\pa{\pp_k}\EmptyQueue)\\
	     \red^* & \MM'  
\end{align*}
Since $\BContext{\pr}.\UT_1' \not\subt \AContext{\pr}.\VT'_i$ 
can be derived from $\UT_1\not\subt (\proj{\CContext\pp}{n})[\VT_i']$ by applying the rules given in 
Table~\ref{tab:negationUandV'} from conclusion to premises, we may apply induction hypothesis 
and obtain $\MM''\red^*\MM'\red^* \error$. Hence, $\MM\red^*\MM'\red^* \error$.
\end{proof}

\begin{proposition}
\label{prop:completenessT}
  Let $\T$ and $\T'$ be session types such that $\T \not\subt\T'$. %
  Then, there are $\UU$ and $\V'$ %
  with $\ttree{\UU} \in \llbracket \ttree{\T} \rrbracket_\SO$ %
  and $\ttree{\V'} \in \llbracket \ttree{\T'} \rrbracket_\SI$ %
  and $\UU \not\subt \V'$ such that:

  \smallskip%
  \centerline{\(%
    \pa\pr \CP{\UU} \pc  \pa\pr \EmptyQueue \pc \prod\limits_{1\leq k \leq m} \pa{\pp_k} \left(\CP{\UU_{\pp_k}}\pc \pa{\pp_k} \EmptyQueue \right) \red^* \error,
    \)}%

 where $\participant{\V'}\subseteq\{\pp_k:1\leq k\leq m\}$  and $\UU_{\pp_k}=\cyclic{\V'}{\pp_k}$.
\end{proposition}

\begin{proof}
 If  $\ttree{\T} \not \subt \ttree{\T'}$,  there are $\UT$ and $\VT$ such that $\UT \in \llbracket \ttree{\T} \rrbracket_\SO$ and $\VT' \in \single{ \ttree{\T'} }$ and $\UT \not\subt \VT'.$ By Corollary~\ref{lem:reg-iregU}, there are $\UU$ and $\V'$ such that $\ttree{\UU} \in \llbracket \ttree{\T} \rrbracket_\SO$ and $\ttree{\V'} \in \single{ \ttree{\T'} }$ and $\ttree{\UU} \not\subt \ttree{\V'}.$ Now the proof follows by Proposition~\ref{thm:completenessUandV'}.
\end{proof}

\thCompleteness*

\begin{proof}
Let $\T$ and $\T'$ be such that $\ttree{\T}\not\subt \ttree{\T'}.$ 
Then, by Proposition~\ref{prop:completenessT} there are $\UU$ and $\V'$ with $\ttree{\UU} \in \llbracket \ttree{\T} \rrbracket_\SI$ and $\ttree{\V'} \in \llbracket \ttree{\T'} \rrbracket_\SO$ and $\ttree{\UU}\not\subt \ttree{\V'},$
 such that,
 \begin{equation}\label{eq:charact_sess_red_to_err}
 \pa\pr \CP{\UU} \pc  \pa\pr \EmptyQueue \pc \M_{\pr,\V'} \red^* \error,
 \end{equation}      
 where 
 \[
 \M_{\pr,\V'}= \prod\limits_{1\leq k \leq m} \pa{\pp_k} \left(\CP{\UU_{\pp_k}}\pc \pa{\pp_k} \EmptyQueue \right)
 \]
 and
 $\participant{\V'}\subseteq\{\pp_k:1\leq k\leq m\}$  and $\UU_{\pp_k}=\cyclic{\V'}{\pp_k}$.

By Proposition~\ref{prop:live-ti}, $\vdash \Q_1:\T'$ implies there is a live typing environment $\Gamma$ such that 
  \[
       \Gamma \vdash \pa\pr \Q_1 \pc \pa\pr\EmptyQueue \pc \M_{\pr,\V'}
  \]

Since by Proposition~\ref{cp} $\vdash \CP{\UU}:\T,$
we conclude the proof by~(\ref{eq:charact_sess_red_to_err}).
\end{proof}

\end{document}